\documentclass{article}
\usepackage[margin=.8in,left=.8in]{geometry}

\usepackage{amsthm}
\usepackage{amssymb}
\usepackage{amsmath}
\usepackage{mathtools}
\usepackage{adjustbox}
\usepackage{xcolor}
\usepackage{stmaryrd}    
\usepackage{hyperref}
\hypersetup{
  colorlinks = true,
  allcolors=.
}
\usepackage{xfrac}
\usepackage{hhline} 
\usepackage{cleveref}

\usepackage{tikz}
\usetikzlibrary{backgrounds, decorations.pathmorphing}
\usetikzlibrary{cd}

\usepackage{enumitem} \setlist{nosep}

\newcommand{\underoverset}[3]{\underset{#1}{\overset{#2}{#3}}}

\usepackage[titletoc]{appendix}

\usepackage[safe]{tipa} 

\definecolor{darkblue}{rgb}{0.05,0.25,0.65}
\definecolor{greenii}{RGB}{20,140,10}
\definecolor{darkgreen}{rgb}{0.00,0.85,0.1}
\definecolor{lightgray}{rgb}{0.9,0.9,0.9}
\definecolor{orangeii}{RGB}{200,100,5}
\definecolor{darkyellow}{rgb}{.91,.91,0}
\definecolor{lightbackgroundgray}{RGB}{245,245,245}

\usepackage[skins]{tcolorbox}

\usepackage{multirow}

\DeclareMathAlphabet{\mathpzc}{OT1}{pzc}{m}{it} 
\newcommand\mathscr[1]{\scalebox{1.1}{$\mathpzc{#1}$}}

\usepackage{mathptmx}
\usepackage{amsmath}
\usepackage{graphicx}
\DeclareRobustCommand{\coprod}{\mathop{\text{\fakecoprod}}}
\newcommand{\fakecoprod}{%
  \sbox0{$\prod$}%
  \smash{\raisebox{\dimexpr.9625\depth-\dp0}{\scalebox{1}[-1]{$\prod$}}}%
  \vphantom{$\prod$}%
}

\newcommand{\boldDelta}{\mbox{$\Delta$\hspace{-6.3pt}$\Delta$}}

\newcommand{\Differential}{
  \mathrm{d}
}

\newcommand{\DifferentialForms}[3]{
  \Omega^{#2}_{\mathrm{dR}}
  #1(
    #3
  #1)
}

\newcommand{\deRhamDifferential}{\mathrm{d}}
\newcommand{\DeRhamDifferential}{\deRhamDifferential}

\newcommand{\DeRhamComplex}[2]{
  \DifferentialForms{#1}{\bullet}
    {#2}
}

\newcommand{\NeutralElement}{
  \mathrm{e}
}

\newcommand{\NaturalNumbers}{
  \mathbb{N}
}

\newcommand{\Integers}{
  \mathbb{Z}
}

\newcommand{\RationalNumbers}{
  \mathbb{Q}
}

\newcommand{\RealNumbers}{
  \mathbb{R}
}

\newcommand{\DualTorus}[1]{\widehat{\mathbb{T}}{}^{#1}}

\newcommand{\ImaginaryUnit}{
  \mathrm{i}
}

\newcommand{\ComplexNumbers}{\mathbb{C}}

\newcommand{\CyclicGroup}[1]{\mathbb{Z}_{#1}}

\newcommand{\ZTwo}{
  \CyclicGroup{2}
}

\newcommand{\OrthogonalGroup}{
  \mathrm{O}
}

\newcommand{\UnitaryGroup}{
  \mathrm{U}
}

\newcommand{\SpecialUnitaryLieAlgebra}[1]{\mathfrak{su}_{#1}}
\newcommand{\suTwo}{\SpecialUnitaryLieAlgebra{2}}

\newcommand{\HighestWeightVector}{v^0}

\newcommand{\SLTwoZ}{\mathrm{SL}(2,\Integers)}

\newcommand{\suAffine}[2]{\widehat{\SpecialUnitaryLieAlgebra{#1}}^{\raisebox{-2.5pt}{\scalebox{.73}{\hspace{-1pt}$#2$}}}}

\newcommand{\suTwoAffine}[1]{\suAffine{2}{#1}}

\newcommand{\UnitaryOperator}{U}

\newcommand{\FredholmOperators}{
  \mathrm{Fred}
}

\newcommand{\FredholmOperator}{
  F
}

\newcommand{\UH}{
  \UnitaryGroup(\mathscr{H})
}

\newcommand{\CircleGroup}{
  \UnitaryGroup(1)
}

\newcommand{\SmoothManifold}{
  \TopologicalSpace
}

\newcommand{\ComplexPlane}{\ComplexNumbers}

\newcommand{\FlatConnectionForm}{\omega_1}

\newcommand{\NumberOfPunctures}{N}
\newcommand{\NumberOfProbeBranes}{n}

\newcommand{\ConfigurationSpace}[1]{  \underset{
    \scalebox{.65}{$
      \{1,\cdots,#1\}
    $}
  }
  {\mathrm{Conf}}
}

\newcommand{\weight}{\mathrm{w}}

\newcommand{\level}{k}
\newcommand{\Level}{\level}
\newcommand{\ShiftedLevel}{\kappa}
\newcommand{\Denominator}{r}

\newcommand{\ConformalBlocks}{\mathrm{CnfBlck}}

\newcommand{\TopologicalSpace}{
  \mathrm{X}
}

\newcommand{\Maps}[3]{
  \mathrm{Map}
  #1(
    #2
    ,\,
    #3
  #1)
}

\newcommand{\PointedMaps}[3]{
  \mathrm{Map}^{\ast/}
  #1(
    #2
    ,\,
    #3
  #1)
}

\newcommand{\stable}{\mathrm{stbl}}

\newcommand{\homotopy}{\mathrm{htpy}}
\newcommand{\continuous}{\mathrm{cnts}}

\newcommand{\Homs}[3]{
  \mathrm{Hom}
  #1(
    #2
    ,\,
    #3
  #1)
}

\newcommand{\TED}{$\mathrm{TED}$}

\newcommand{\HomotopyQuotient}[2]{
  #1 \!\sslash\! #2
}

\usepackage[new]{old-arrows}   

\def\acts{\raisebox{1.4pt}{\;\rotatebox[origin=c]{90}{$\curvearrowright$}}\hspace{.5pt}}

\makeatletter
\newif\if@sup
\newtoks\@sups
\def\append@sup#1{\edef\act{\noexpand\@sups={\the\@sups #1}}\act}%
\def\reset@sup{\@supfalse\@sups={}}%
\def\mk@scripts#1#2{\if #2/ \if@sup ^{\the\@sups}\fi \else%
  \ifx #1_ \if@sup ^{\the\@sups}\reset@sup \fi {}_{#2}%
  \else \append@sup#2 \@suptrue \fi%
  \expandafter\mk@scripts\fi}
\def\tensor#1#2{\reset@sup#1\mk@scripts#2_/}
\def\multiscripts#1#2#3{\reset@sup{}\mk@scripts#1_/#2%
  \reset@sup\mk@scripts#3_/}
\makeatother

\makeatletter
\newbox\slashbox \setbox\slashbox=\hbox{$/$}
\def\itex@pslash#1{\setbox\@tempboxa=\hbox{$#1$}
  \@tempdima=0.5\wd\slashbox \advance\@tempdima 0.5\wd\@tempboxa
  \copy\slashbox \kern-\@tempdima \box\@tempboxa}
\def\slash{\protect\itex@pslash}
\makeatother

\def\clap#1{\hbox to 0pt{\hss#1\hss}}
\def\mathllap{\mathpalette\mathllapinternal}
\def\mathrlap{\mathpalette\mathrlapinternal}
\def\mathclap{\mathpalette\mathclapinternal}
\def\mathllapinternal#1#2{\llap{$\mathsurround=0pt#1{#2}$}}
\def\mathrlapinternal#1#2{\rlap{$\mathsurround=0pt#1{#2}$}}
\def\mathclapinternal#1#2{\clap{$\mathsurround=0pt#1{#2}$}}

\let\oldroot\root
\def\root#1#2{\oldroot #1 \of{#2}}
\renewcommand{\sqrt}[2][]{\oldroot #1 \of{#2}}

\DeclareSymbolFont{symbolsC}{U}{txsyc}{m}{n}
\SetSymbolFont{symbolsC}{bold}{U}{txsyc}{bx}{n}
\DeclareFontSubstitution{U}{txsyc}{m}{n}

\DeclareSymbolFont{stmry}{U}{stmry}{m}{n}
\SetSymbolFont{stmry}{bold}{U}{stmry}{b}{n}

\DeclareFontFamily{OMX}{MnSymbolE}{}
\DeclareSymbolFont{mnomx}{OMX}{MnSymbolE}{m}{n}
\SetSymbolFont{mnomx}{bold}{OMX}{MnSymbolE}{b}{n}
\DeclareFontShape{OMX}{MnSymbolE}{m}{n}{
    <-6>  MnSymbolE5
   <6-7>  MnSymbolE6
   <7-8>  MnSymbolE7
   <8-9>  MnSymbolE8
   <9-10> MnSymbolE9
  <10-12> MnSymbolE10
  <12->   MnSymbolE12}{}

\crefformat{section}{\S#2#1#3} 
\crefformat{subsection}{\S#2#1#3}
\crefformat{subsubsection}{\S#2#1#3}

\theoremstyle{italics}
\newtheorem{theorem}{Theorem}[section]
\newtheorem{lemma}[theorem]{Lemma}
\newtheorem{fact}[theorem]{Fact}
\newtheorem{conjecture}[theorem]{Conjecture}
\newtheorem{proposition}[theorem]{Proposition}

\theoremstyle{definition}

\newtheorem{example}[theorem]{Example}

\newtheorem{remark}[theorem]{Remark}

\usepackage{amsfonts}
\usepackage{colortbl}

\renewcommand{\emph}{\textit}

\begin{document}

\title{
{ Anyonic} Topological Order in
  \\
  Twisted Equivariant Differential (TED)
  K-Theory
}

\author{
  Hisham Sati${}^{\ast \dagger}$,
  \;\;
  Urs Schreiber${}^{\ast}$
}

\bigskip

\maketitle

\begin{abstract}
While the classification of non-interacting crystalline topological insulator phases by equivariant K-theory has become widely accepted, its generalization to anyonic interacting phases -- hence to phases with
 topologically ordered ground states supporting topological braid quantum gates -- has remained wide open.

 \smallskip

 On the contrary, the success of K-theory with classifying
non-interacting phases seems to have tacitly been perceived as precluding a K-theoretic classification of interacting topological order; and instead a mix of other
proposals has been explored.
However, only K-theory connects closely to the actual physics of valence electrons;
and self-consistency demands that any other proposal must connect to K-theory.

\smallskip

Here we provide a detailed argument for the classification of
  symmetry protected/enhanced
  $\suTwo$-anyonic topological order, specifically in interacting 2d semi-metals,  by the
  {\it twisted equivariant differential} (TED) K-theory
  of {\it configuration spaces of points} in the complement of nodal points inside the crystal's Brillouin torus orbi-orientifold.

  \smallskip
  \noindent
  We argue, in particular, that:

  \smallskip
  \begin{itemize}[leftmargin=.6cm]

  \item[(1)] topological 2d semi-metal phases modulo global mass terms are classified by the {\it flat differential} twisted equivariant K-theory of the complement of the nodal points;

  \item[(2)]
  $n$-electron interacting phases
  are classified by the K-theory of configuration spaces of $n$ points in the Brillouin torus;

  \item[(3)] the somewhat neglected twisting of equivariant K-theory by ``inner local systems'' reflects the effective ``fictitious'' gauge interaction
  of Chen, Wilczeck, Witten \& Halperin (1989), which turns fermions into anyonic quanta;

  \item[(4)] the induced
  $\suTwo$-anyonic topological order is reflected in the {\it twisted} Chern classes of the interacting valence bundle over configuration space,
  constituting the {\it hypergeometric integral construction} of monodromy braid representations.
\end{itemize}

\smallskip

\noindent
A tight dictionary relates these arguments to those for classifying defect brane charges in string theory \cite{SS22AnyonicDefectBranes}, which we
expect to be the
 images of momentum-space $\suTwo$-anyons under a non-perturbative version of the AdS/CMT correspondence.

\end{abstract}

\medskip

\tableofcontents

\vfill

\hrule
\vspace{4pt}

{
\footnotesize
\noindent
\def\arraystretch{1}
\tabcolsep=0pt
\begin{tabular}{ll}
${}^*$\,
&
Mathematics, Division of Science; and
\\
&
Center for Quantum and Topological Systems,
\\
&
NYUAD Research Institute,
\\
&
New York University Abu Dhabi, UAE.
\end{tabular}
\hfill
\adjustbox{raise=-15pt}{
\href{https://ncatlab.org/nlab/show/Center+for+Quantum+and+Topological+Systems}{
\includegraphics[width=3cm]{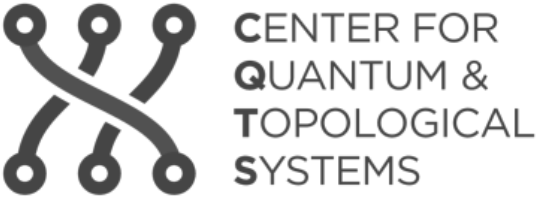}}
}
}

\vspace{1mm}
\noindent ${}^\dagger$The Courant Institute for Mathematical Sciences, NYU, NY

\newpage

\section{Introduction}

 A profound
 and celebrated conjecture in condensed matter theory (reviewed in \cref{TEDKClassifiesTopologicalPhasesOfMatter}) says
 that symmetry-protected/enhanced ``topological phases'' of {\it non-interacting}\footnote{\label{NonInteracting}
 Here ``non-interacting'' means that the
 screened/dressed electrons in the crystalline material may be well approximated as not interacting with each other, but just with
 the effective classical Coulomb field of the crystal lattice (e.g. \cite[\S 4.5]{Li06}, cf. Fact \ref{VacuaOfTheRelativisticElectronPositronFieldInBackground}). Technically, this means that the ground state of the crystalline material
 is well approximated by filling the lowest single-electron {\it Bloch states} (recalled as Fact \ref{BlochFloquetTheory} below).
 We go beyond this approximation in \cref{InteractingPhasesAndTEDKOfConfigurationSpaces}} gapped crystalline
 materials -- ``topological insulators'' -- are classified, up to adiabadic deformations (Rem. \ref{QuantumAdiabaticTheorem}),
 by the twisted equivariant topological K-theory of their Brillouin tori, orbi-orientifolded by the crystallographic point group
 and/or by CPT quantum symmetry groups.
 There is some experimental evidence and some theoretical arguments for this conjecture, which have been repeated in a wealth of
 publications on the subject, but there have remained conceptual gaps (see Rem. \ref{ComparisonToExistingLiteratureForKClassification})
 large enough that one of the most highly cited sources for the statement actually dis-claimed it:

\vspace{-2mm}
\begin{quote}
{\it
Although [K-theory] is used in the condensed matter literature, it is not clear to us that it is well motivated.}
\cite[{\href{https://arxiv.org/pdf/1208.5055.pdf\#page=57}{p. 57}}]{FreedMoore12}
\end{quote}

To start with, on this point we observe below
(see Fact \ref{VacuaOfTheRelativisticElectronPositronFieldInBackground}) that
careful mathematical analysis of free-electron field theory in background fields shows  \cite{KS77}\cite{CareyHurstOBrien82} that (relativistic) quantum ground states of free electrons in Coulumb potentials are classified by {\it Fredholm operators}, and hence that valence electron states in crystals ought to be classified by Bloch families of Fredhom operators (Conjecture \ref{FamiliesOfRelativisticBlochVacua}).
But it is a classical result that families of Fredholm operators are, in turn, classified by K-theory; in fact it is the Fredholm-operator picture which most naturally supports twisted equivariant K-theory (\cite{AtiyahSegal04}\cite{SS21EPB}\cite{SS22AnyonicDefectBranes}\cite{SS22TED}). We suggest that this logic is what ties K-theory to the classification of topological insulators (Fact \ref{KTheoryClassificationOfTopologicalPhasesOfMatter}).

\medskip

At the same time, much of the thrust in the field of topological phases of matter lies {\it beyond topological insulator-phases} \cite{TurnerVishwanath13}  in semi-metallic and/or topologically ordered interacting phases (recalled in \cref{TEDKDescribesRealisticAnyonSpecies}): It is topological order which supports the most drastic expected application of topological physics -- namely to topological quantum computation by braiding of anyonic defects (Rem. \ref{QuantumAdiabaticTheorem}, \cite{SS22TQC}); and it is topological semi-metal phases which may be the most realistic substrate for topological order, namely for momentum-space anyons (Rem. \ref{MomentumSpaceAnyons}).

\medskip

In summary, this means that a major open problem in the condensed matter theory of topological phases has been to provide a {\it detailed argument} for a classification of crystalline topological phases {\it subsuming  all} of: (1) free, (2) semi-metallic, (3) interacting, and (4) topologically ordered phases, such that -- at least on free and globally gapped phases -- it reduces to the twisted equivariant K-theory of the orbi-orientifolded Brioullin torus.

\medskip

Here we offer such a detailed argument, leading to the conjecture (Conjectures \ref{KTheoryClassificationOfTopologicalOrder}, \ref{ClassificationOfSPTOrderInTEDK}) that the complete classification of crystalline topological phases of matter -- subsuming the traditional folklore for topological insulators but encompassing also interacting SPT semi-metal phases and their anyonic topological order -- is given by the twisted equivariant differential (TED) K-theory of configuration spaces of points inside the complement of nodal points in the orbi-orientifolded Brillouin torus. We phrase this as a conjecture to stay true to the standards in mathematical physics, but the level of supporting arguments we provide seems to be no less than what supports related statements in topological condensed matter theory.

\medskip

The key mathematical insight which makes this work is the recent result of the authors \cite{SS22AnyonicDefectBranes}: A ``well-known'',
but previously somewhat neglected, extra twisting of equivariant K-theory by ``inner local systems'' makes the twisted equivariant K-theory of
``internal'' $\CyclicGroup{\ShiftedLevel} \subset \CircleGroup$-symmetry (\cref{InternalSymmetriesAndInnerLocalSystemTwists})
have (twisted) Chern characters in de Rham cohomology with local
systems of coefficients.
Furthermore, on configuration spaces this twisted cohomology constitutes the ``hypergeometric integral-construction''
of anyonic braid group representations. This is the content of the concluding section \cref{AnyonicTopologicalOrderAndInnerLocalSysyemTEDK}.

\medskip
In unwinding this mathematical statement we find a neat match between the various physical phenomena involved and the fine detail
of the inner workings of stacky Fredholm-operator TED-K theory.
For instance, beyond the now famous reflection of the ``10-fold way'' of CPT quantum symmetries in the internal twists of KR-theory (Fact. \ref{ListOfCPTTwistings}, which we review in streamlined form in \cref{CrystallographicSymmetriesAndOrbifoldKTheory}) we find that the traditional ``fictitious gauge field'', which encodes the effective interactions
of anyonic quanta (\hyperlink{TableOfTwistedEquivariances}{\it Table 2}), is identified with the ``inner local system''-twist of TED-K theory; and the
{\it logarithmic} conformal block structure of topologically ordered ground states emerges
from the ``delocalized'' direct sum nature of
equivariant K-theory (see below \hyperlink{RelationBetweenFractionalLevelAndLogarithmicCFT}{\it Figure 13}).

\medskip

In \cite{SS22AnyonicDefectBranes}, these same mathematical results were matched to phenomena expected for defect branes in string theory,
as part of the authors' program of understanding ``M-theory'' in terms of the generalized cohomology of Cohomotopy moduli stacks (see \cite{SS22ConfigurationSpaces}\cite{CSS21}).
Indeed, there is a tight dictionary (\hyperlink{RosettaStone}{\it Table 1}) relating condensed matter theory (CMT) to stringy brane physics
via TED-K-theory, reminiscent of the expectations in the AdS/CMT correspondence (Rem. \ref{AnalogyWithDBraneCharge}).

\medskip

By way of outlook, the comprehensive reflection which we establish, of crystalline topological phases of matter in the mathematics
of TED-K-theory of configuration spaces, provides new theoretical leverage for further investigations, notably into topological quantum computation (Rem. \ref{QuantumAdiabaticTheorem}).
We are discussing this in \cite{SS22TQC}.

\medskip

\noindent{\bf Outline:}
\begin{itemize}[leftmargin=.6cm]

\item
[{\bf \cref{TEDKClassifiesTopologicalPhasesOfMatter}}] reviews the expected/accepted K-theoretic classification of crystalline
free fermion topological phases. Here we prepare the ground for  \cref{TEDKDescribesRealisticAnyonSpecies} by phrasing all
K-theory constructions
in terms of the cohesive moduli stack of Fredholm operators, developed in \cite{SS21EPB}\cite{SS22AnyonicDefectBranes}\cite{SS22TED}. The dictionary (\hyperlink{RosettaStone}{\it Table 1}) between this
twisted equivariant K-theory and the physics of valence electrons (as well as that of stable D-branes,
see Rem. \ref{AnalogyWithDBraneCharge})
%
%
%
is so close that readers familiar with one of the sides may regard it as providing
the explanation of the other(s).

\item[{\bf \cref{TEDKDescribesRealisticAnyonSpecies}}] presents our argument for generalizing the K-theory classification
{\it beyond topological insulators}, namely to (2d) topological semi-metals and to their interacting topologically ordered phases.
\end{itemize}

\medskip
\medskip


\medskip

In closing this introduction we recall the following general principle, which appears in several guises in the main text.
\begin{remark}[\bf Adiabatic transformations of parameterized quantum systems]
  \label{QuantumAdiabaticTheorem}
  A {\it parameterized} quantum system is a set of quantum systems continuously parameterized by a ``classical'' parameter space. A basic example are time-dependent quantum systems, dependent on a time classical parameter. The examples of interest here are:

  \begin{itemize}

  \item[1.] Bloch states parameterized by a Bloch momentum (Fact \ref{BlochFloquetTheory});

  \item[2.] topologically ordered phases of matter parameterized by the (time-dependent) positions of anyonic defects (\cref{InteractingPhasesAndTEDKOfConfigurationSpaces}).

  \end{itemize}

  \medskip

  \noindent
  The {\it quantum adiabatic theorem}
  (\cite{BornFock28}\cite{Kato50}\cite{Nenciu80})
  states that in the limit of arbitrarily slow movement of its parameters (relative to the relaxation time of the quantum system), a parameterized quantum system remains in a joint eigenstate of a given commuting set of quantum observables, even if the eigenvalues change significantly. This means that the paths in the parameter space of a parameterized quantum system define, by adiabatic movement of the system along these paths, unitary transformations of the system's eigenspaces, compatible under composition of paths. These are known as (non-abelian) {\it Berry phases} (\cite{Berry84}\cite{WilczekZee84}, review includes \cite[\S 2]{Stanescu20}\cite[\S IV.C]{CayssolFuchs20}).

  \medskip
  The adiabatic parameter-action on ground-states of quantum materials is one model for quantum computation (``adiabatic quantum computation'', e.g.  \cite{AlbashLidar16}\cite{CLBFN15}). When these {\it adiabatic quantum gates} depend only on the isotopy classes of paths in parameter space, and when these are non-trivial -- such as when the parameter space is a configuration space
  \eqref{ConfigurationSpaceOfPoints}
  of positions/momenta of {\it defect anyons} in a 2d crystal/Brillouin torus (\hyperlink{NotionsOfAnyons}{\it Table 5})    -- then we are dealing with {\it topological quantum computation} via {\it braid quantum gates}.

  \smallskip

   This is a popular idea
  (review includes \cite{Stanescu20}\cite{FieldSimula18}\cite{RowellWang18}, going back to \cite{Kitaev03}\cite{FKLW01}\cite{FLW02}\cite{NSSFS08}) but has remained somewhat elusive, both theoretically
  (cf. \cite[p. 8]{Kitaev06}\cite[p. 7-8]{SarmaFreedmanNayak15}) as well as experimentally (\cite{KouwenhovenEtAl21} \cite{KouwenhovenEtAl22}).
  The new theory presented
  in \cref{TEDKDescribesRealisticAnyonSpecies} may help elucidate the issue.

\end{remark}

\vspace{-.5cm}

\begin{center}
\hypertarget{AdiabaticBraiding}{}
\begin{tikzpicture}

  \shade[right color=lightgray, left color=white]
    (3,-3)
      --
      node[above, yshift=-1pt, sloped]{
        \scalebox{.7}{
          \color{darkblue}
          \bf
          Brillouin torus
        }
      }
    (-1,-1)
      --
    (-1.21,1)
      --
    (2.3,3);

  \draw[]
    (3,-3)
      --
    (-1,-1)
      --
    (-1.21,1)
      --
    (2.3,3)
      --
    (3,-3);

\draw[-Latex]
  ({-1 + (3+1)*.3},{-1+(-3+1)*.3})
    to
  ({-1 + (3+1)*.29},{-1+(-3+1)*.29});

\draw[-Latex]
    ({-1.21 + (2.3+1.21)*.3},{1+(3-1)*.3})
      --
    ({-1.21 + (2.3+1.21)*.29},{1+(3-1)*.29});

\draw[-Latex]
    ({2.3 + (3-2.3)*.5},{3+(-3-3)*.5})
      --
    ({2.3 + (3-2.3)*.49},{3+(-3-3)*.49});

\draw[-latex]
    ({-1 + (-1.21+1)*.53},{-1 + (1+1)*.53})
      --
    ({-1 + (-1.21+1)*.54},{-1 + (1+1)*.54});

  \begin{scope}[rotate=(+8)]
   \draw[dashed]
     (1.5,-1)
     ellipse
     ({.2*1.85} and {.37*1.85});
   \begin{scope}[
     shift={(1.5-.2,{-1+.37*1.85-.1})}
   ]
     \draw[->, -Latex]
       (0,0)
       to
       (180+37:0.01);
   \end{scope}
   \begin{scope}[
     shift={(1.5+.2,{-1-.37*1.85+.1})}
   ]
     \draw[->, -Latex]
       (0,0)
       to
       (+37:0.01);
   \end{scope}
   \begin{scope}[shift={(1.5,-1)}]
     \draw (.43,.65) node
     { \scalebox{.8}{$
       \sfrac{\weight_I}{\ShiftedLevel}
     $} };
  \end{scope}
  \draw[fill=white, draw=gray]
    (1.5,-1)
    ellipse
    ({.2*.3} and {.37*.3});
  \draw[line width=3.5, white]
   (1.5,-1)
   to
   (-2.2,-1);
  \draw[line width=1.1]
   (1.5,-1)
   to node[above, yshift=-3pt, pos=.85]{
     \;\;\;\;\;\;\;\;\;\;\;\;\;
     \rotatebox[origin=c]{7}
     {
     \scalebox{.7}{
     \color{orangeii}
     \bf
     \colorbox{white}{nodal} point
     }
     }
   }
   (-2.2,-1);
  \draw[
    line width=1.1
  ]
   (1.5+1.2,-1)
   to
   (3.5,-1);
  \draw[
    line width=1.1,
    densely dashed
  ]
   (3.5,-1)
   to
   (4,-1);

  \draw[line width=3, white]
   (-2,-1.3)
   to
   (0,-1.3);
  \draw[-latex]
   (-2,-1.3)
   to
   node[below, yshift=+3pt]{
     \scalebox{.7}{
       \rotatebox{+7}{
       \color{darkblue}
       \bf
       time
       }
     }
   }
   (0,-1.3);
  \draw[dashed]
   (-2.7,-1.3)
   to
   (-2,-1.3);

 \draw
   (-3.15,-.8)
   node{
     \scalebox{.7}{
       \rotatebox{+7}{
       \color{greenii}
       \bf
       braiding
       }
     }
   };

  \end{scope}

  \begin{scope}[shift={(-.2,1.4)}, scale=(.96)]
  \begin{scope}[rotate=(+8)]
  \draw[dashed]
    (1.5,-1)
    ellipse
    (.2 and .37);
  \draw[fill=white, draw=gray]
    (1.5,-1)
    ellipse
    ({.2*.3} and {.37*.3});
  \draw[line width=3.1, white]
   (1.5,-1)
   to
   (-2.3,-1);
  \draw[line width=1.1]
   (1.5,-1)
   to
   (-2.3,-1);
  \draw[line width=1.1]
   (1.5+1.35,-1)
   to
   (3.6,-1);
  \draw[
    line width=1.1,
    densely dashed
  ]
   (3.6,-1)
   to
   (4.1,-1);
  \end{scope}
  \end{scope}

  \begin{scope}[shift={(-1,.5)}, scale=(.7)]
  \begin{scope}[rotate=(+8)]
  \draw[dashed]
    (1.5,-1)
    ellipse
    (.2 and .32);
  \draw[fill=white, draw=gray]
    (1.5,-1)
    ellipse
    ({.2*.3} and {.32*.3});
  \draw[line width=3.1, white]
   (1.5,-1)
   to
   (-1.8,-1);
\draw
   (1.5,-1)
   to
   (-1.8,-1);
  \draw
    (5.23,-1)
    to
    (6.4-.6,-1);
  \draw[densely dashed]
    (6.4-.6,-1)
    to
    (6.4,-1);
  \end{scope}
  \end{scope}

\draw (2.41,-2.3) node
 {
  \scalebox{1}{
    $\DualTorus{2}$
  }
 };

\draw (1.73,-1.06) node
 {
  \scalebox{.8}{
    $k_{{}_{I}}$
  }
 };

\begin{scope}
[ shift={(-2,-.55)}, rotate=-82.2  ]

 \begin{scope}[shift={(0,-.15)}]

  \draw[]
    (-.2,.4)
    to
    (-.2,-2);

  \draw[
    white,
    line width=1.1+1.9
  ]
    (-.73,0)
    .. controls (-.73,-.5) and (+.73-.4,-.5) ..
    (+.73-.4,-1);
  \draw[
    line width=1.1
  ]
    (-.73+.01,0)
    .. controls (-.73+.01,-.5) and (+.73-.4,-.5) ..
    (+.73-.4,-1);

  \draw[
    white,
    line width=1.1+1.9
  ]
    (+.73-.1,0)
    .. controls (+.73,-.5) and (-.73+.4,-.5) ..
    (-.73+.4,-1);
  \draw[
    line width=1.1
  ]
    (+.73,0+.03)
    .. controls (+.73,-.5) and (-.73+.4,-.5) ..
    (-.73+.4,-1);

  \draw[
    line width=1.1+1.9,
    white
  ]
    (-.73+.4,-1)
    .. controls (-.73+.4,-1.5) and (+.73,-1.5) ..
    (+.73,-2);
  \draw[
    line width=1.1
  ]
    (-.73+.4,-1)
    .. controls (-.73+.4,-1.5) and (+.73,-1.5) ..
    (+.73,-2);

  \draw[
    white,
    line width=1.1+1.9
  ]
    (+.73-.4,-1)
    .. controls (+.73-.4,-1.5) and (-.73,-1.5) ..
    (-.73,-2);
  \draw[
    line width=1.1
  ]
    (+.73-.4,-1)
    .. controls (+.73-.4,-1.5) and (-.73,-1.5) ..
    (-.73,-2);

 \draw
   (-.2,-3.3)
   to
   (-.2,-2);
 \draw[
   line width=1.1,
   densely dashed
 ]
   (-.73,-2)
   to
   (-.73,-2.5);
 \draw[
   line width=1.1,
   densely dashed
 ]
   (+.73,-2)
   to
   (+.73,-2.5);

  \end{scope}
\end{scope}

\begin{scope}[shift={(-5.6,-.75)}]

  \draw[line width=3pt, white]
    (3,-3)
      --
    (-1,-1)
      --
    (-1.21,1)
      --
    (2.3,3)
      --
    (3, -3);

  \shade[right color=lightgray, left color=white, fill opacity=.7]
    (3,-3)
      --
    (-1,-1)
      --
    (-1.21,1)
      --
    (2.3,3);

  \draw[]
    (3,-3)
      --
    (-1,-1)
      --
    (-1.21,1)
      --
    (2.3,3)
      --
    (3, -3);

\draw (1.73,-1.06) node
 {
  \scalebox{.8}{
    $k_{{}_{I}}$
  }
 };

\draw[-Latex]
  ({-1 + (3+1)*.3},{-1+(-3+1)*.3})
    to
  ({-1 + (3+1)*.29},{-1+(-3+1)*.29});

\draw[-Latex]
    ({-1.21 + (2.3+1.21)*.3},{1+(3-1)*.3})
      --
    ({-1.21 + (2.3+1.21)*.29},{1+(3-1)*.29});

\draw[-Latex]
    ({2.3 + (3-2.3)*.5},{3+(-3-3)*.5})
      --
    ({2.3 + (3-2.3)*.49},{3+(-3-3)*.49});

\draw[-latex]
    ({-1 + (-1.21+1)*.53},{-1 + (1+1)*.53})
      --
    ({-1 + (-1.21+1)*.54},{-1 + (1+1)*.54});

  \begin{scope}[rotate=(+8)]
   \draw[dashed]
     (1.5,-1)
     ellipse
     ({.2*1.85} and {.37*1.85});
   \begin{scope}[
     shift={(1.5-.2,{-1+.37*1.85-.1})}
   ]
     \draw[->, -Latex]
       (0,0)
       to
       (180+37:0.01);
   \end{scope}
   \begin{scope}[
     shift={(1.5+.2,{-1-.37*1.85+.1})}
   ]
     \draw[->, -Latex]
       (0,0)
       to
       (+37:0.01);
   \end{scope}
  \draw[fill=white, draw=gray]
    (1.5,-1)
    ellipse
    ({.2*.3} and {.37*.3});
 \end{scope}

   \begin{scope}[shift={(-.2,1.4)}, scale=(.96)]
  \begin{scope}[rotate=(+8)]
  \draw[dashed]
    (1.5,-1)
    ellipse
    (.2 and .37);
  \draw[fill=white, draw=gray]
    (1.5,-1)
    ellipse
    ({.2*.3} and {.37*.3});
\end{scope}
\end{scope}

  \begin{scope}[shift={(-1,.5)}, scale=(.7)]
  \begin{scope}[rotate=(+8)]
  \draw[dashed]
    (1.5,-1)
    ellipse
    (.2 and .32);
  \draw[fill=white, draw=gray]
    (1.5,-1)
    ellipse
    ({.2*.3} and {.37*.3});
\end{scope}
\end{scope}

\begin{scope}
[ shift={(-2,-.55)}, rotate=-82.2  ]

 \begin{scope}[shift={(0,-.15)}]

 \draw[line width=3, white]
   (-.2,-.2)
   to
   (-.2,2.35);
 \draw
   (-.2,.5)
   to
   (-.2,2.35);
 \draw[dashed]
   (-.2,-.2)
   to
   (-.2,.5);

\end{scope}
\end{scope}

\begin{scope}
[ shift={(-2,-.55)}, rotate=-82.2  ]

 \begin{scope}[shift={(0,-.15)}]

 \draw[
   line width=3, white
 ]
   (-.73,-.5)
   to
   (-.73,3.65);
 \draw[
   line width=1.1
 ]
   (-.73,.2)
   to
   (-.73,3.65);
 \draw[
   line width=1.1,
   densely dashed
 ]
   (-.73,.2)
   to
   (-.73,-.5);
 \end{scope}
 \end{scope}

\begin{scope}
[ shift={(-2,-.55)}, rotate=-82.2  ]

 \begin{scope}[shift={(0,-.15)}]

 \draw[
   line width=3.2,
   white]
   (+.73,-.6)
   to
   (+.73,+3.7);
 \draw[
   line width=1.1,
   densely dashed]
   (+.73,-0)
   to
   (+.73,+-.6);
 \draw[
   line width=1.1 ]
   (+.73,-0)
   to
   (+.73,+3.71);
\end{scope}
\end{scope}

\end{scope}

\draw
  (-2.2,-4.2) node
  {
    \scalebox{1.2}{
      $
       \mathllap{
          \raisebox{1pt}{
            \scalebox{.58}{
              \color{darkblue}
              \bf
              \def\arraystretch{.9}
              \begin{tabular}{c}
                some ground state for
                \\
                fixed defect positions
                \\
                $k_1, k_2, \cdots$
                at time
                {\color{purple}$t_1$}
              \end{tabular}
            }
          }
          \hspace{-5pt}
       }
        \big\vert
          \psi({\color{purple}t_1})
        \big\rangle
      $
    }
  };

\draw[|->]
  (-1.3,-4.1)
  to
  node[
    sloped,
    yshift=5pt
  ]{
    \scalebox{.7}{
      \color{greenii}
      \bf
      Berry phase
      unitary transformation
    }
  }
  node[
    sloped,
    yshift=-5pt,
    pos=.4
  ]{
    \scalebox{.7}{
      \color{greenii}
      \bf
      {\color{black}=}
      adiabatic quantum gate
      }
  }
  (+2.4,-3.4);

\draw
  (+3.2,-3.85) node
  {
    \scalebox{1.2}{
      $
        \underset{
          \raisebox{-7pt}{
            \scalebox{.55}{
              \color{darkblue}
              \bf
              \def\arraystretch{.9}
               \begin{tabular}{c}
              another ground state for
                \\
                fixed defect positions
                \\
                $k_1, k_2, \cdots$
                at time
                {\color{purple}$t_2$}
              \end{tabular}
            }
          }
        }{
        \big\vert
          \psi({\color{purple}t_2})
        \big\rangle
        }
      $
    }
  };

\end{tikzpicture}

\vspace{.2cm}
\begin{minipage}{14cm}
  \footnotesize
  {\bf Figure 1 -- Adiabatic braid quantum gate}.
  Schematically indicated is the unitary transformation
  induced on the
  topologically ordered ground state
  (as discussed below in \cref{AnyonicTopologicalOrderAndInnerLocalSysyemTEDK})
  of
  an effectively 2-dimensional topological semi-metal
  (as in \cref{BerryPhasesAndDifferentialKTheory})
  under
  adiabatic braiding
  (Rem. \ref{QuantumAdiabaticTheorem})
  of nodal points in the Brillouin torus (Rem. \ref{MomentumSpaceAnyons}).
\end{minipage}

\end{center}

\newpage

\section{TED-K classifies free topological phases}
\label{TEDKClassifiesTopologicalPhasesOfMatter}

The following is a joint review, with some new developments, of:

\begin{itemize}[leftmargin=.6cm]

\item[\bf 1.]
the basic idea of non-interacting gapped {\it topological phases} of crystalline quantum materials (topological insulators),

\item[\bf 2.]
the general concept of {\it twisted equivariant KR-theory}, and

\item[\bf 3.]
the classification of symmetry protected crystalline topological phases by the twisted equivariant K-theory
of the material's crystallographic orbi-orientifold Brillouin torus.

\end{itemize}

\medskip

\noindent
{\bf The basic idea of topological phases of quantum materials} is the following syllogism:

\begin{itemize}[leftmargin=.8cm]

\item[\bf (A)]
The topic of {\it topology}
(e.g. \cite{Munkres00}\cite{tomDieck08})
is, in essence, invariance under ``gentle'' (continuous) deformations.

\item[\bf (B)]
The eponymous hallmark of {\it quantum physics}
(e.g. \cite[\S 1]{Nakahara03}\cite{Landsman17})
is transitions occurring in potentially discrete quanta of energies.

\item[\bf (C)]
Therefore, a quantum material's ground state whose possible excitations are separated by an energy {\it gap} is  invariant under deformations which are ``gentle'' (technically: {\it adiabatic}, Rem. \ref{QuantumAdiabaticTheorem}) in that they do not bring in energy above the gap. The global properties of such {\it gapped systems} in their ground state should hence be well-described by topology.
\end{itemize}

\medskip

In itself, the phenomenon of energy gaps in quantum physics is not exotic, on the contrary: For practically isolated atoms, such as those in dilute gases, the energy gaps between their excitations are famous since the dawn of quantum theory. But since there is no external parameter to tune these gapped ground states,  there is no non-trivial topological property associated with them.

\medskip
However, when many atoms {\it condense} to form tightly packed solid matter such as crystalline materials, then
the electron orbitals of the individual atomic sites in the crystal overlap to form mixed states whose energy gaps generically shrink away (e.g. \cite[Fig. 4.3, 4.17]{Li06}).
Hence if very special conditions are met so that a condensed matter system retains an energy gap, then it may have a degenerate ground state below that gap which is characterized by non-trivial topological properties (such as a non-trivial valence bundle).
These are the exotic
{\it topological phases of matter} of interest here
(e.g.  \cite{FruchartCapentier13}\cite{MoessnerMoore21}).

\medskip

We now elaborate on all this in more detail.

\subsection{Relativistic electron vacua and topological K-theory}
\label{RelativisticElectronVacuaAndTopologicalKTheory}

\medskip
\noindent
{\bf Crystalline Brillouin torus of quasi-momenta.}
In the spirit of Klein's Erlanger program, the geometry of (ideal) $d$ dimensional crystals is essentially the group and representation theory of their symmetry groups: Such a {\it crystallographic group} $G_{\mathrm{chr}}$ (e.g. \cite{Hilton1903}\cite[\S 2]{Miller72}\cite{Farkas81}\cite{Engel86}) -- is an extension by a full lattice $\Integers^d \,\simeq\, \Lambda \hookrightarrow \RealNumbers^d$ (e.g. \cite[\S 1.2.1]{EngelMichelSenechal04}) -- representing translations along the {\it crystal lattice} -- of a finite {\it point group} $G \subset \mathrm{O}(d)$ of orthogonal transformations:\footnote{
Historically, the original notion of ``crystallographic group'' was less explicit, defined to be any discrete subgroup $G_{\mathrm{cr}} \subset \RealNumbers^d \rtimes \OrthogonalGroup(d)$ such that the corresponding quotient group is compact (review in \cite[\S III]{Farkas81}).
That this implies {\bf (a)} the translations sub-group being a full lattice and {\bf (b)}
the point group being finite, as shown in \eqref{CrystallographicGroup}, is known as {\it Bieberbach's first theorem} (\cite[\S III]{Bieberbach1910}\cite[Thm. 14]{Farkas81}\cite[Thm. I 3.1]{Charlap86},
see also \cite[Thm. 2.3]{Tolcachier20}).}
\vspace{-2mm}
\begin{equation}
  \label{CrystallographicGroup}
  \begin{tikzcd}[row sep=6pt]
    &[-28pt]
    1
    \ar[d]
    &
    1
    \ar[d]
    &[-30pt]
    \\
    \mbox{
      \tiny \bf
      \color{darkblue}
      \def\arraytsretch{.9}
      \begin{tabular}{c}
        Crystal
        \\
        lattice
        \\
        \color{black}
        (full)
      \end{tabular}
    }
    &
    \Lambda
    \ar[d, hook]
    \ar[r, hook]
    \ar[d, hook]
    &
    \RealNumbers^d
    \ar[d, hook]
    &
    \mbox{
      \tiny \bf
      \def\arraytsretch{.9}
      \begin{tabular}{c}
        Euclidean
        \\
        \color{darkblue}
        translation
        \\
        group
      \end{tabular}
    }
    \\
    \mbox{
      \tiny \bf
      \color{darkblue}
      \def\arraytsretch{.9}
      \begin{tabular}{c}
        Crystallographic
        \\
        group
        \\
        \color{black}
        (discrete)
      \end{tabular}
    }
    &
    G_{\mathrm{cr}}
    \ar[d, ->>]
    \ar[r, hook]
    &
    \RealNumbers^d \rtimes \OrthogonalGroup(d)
    \ar[d, ->>]
    &
    \mbox{
      \tiny \bf
      \def\arraytsretch{.9}
      \begin{tabular}{c}
        Euclidean
        \\
        \color{darkblue}
        isometry
        \\
        group
      \end{tabular}
    }
    \\
    \mbox{
      \tiny \bf
      \color{darkblue}
      \def\arraytsretch{.9}
      \begin{tabular}{c}
        Point
        \\
        group
        \\
        \color{black}
        (finite)
      \end{tabular}
    }
    &
    G_{\mathrm{pt}}
    \ar[r, hook]
    \ar[d]
    &
    \OrthogonalGroup(d)
    \ar[d]
    &
    \mbox{
      \tiny \bf
      \def\arraytsretch{.9}
      \begin{tabular}{c}
        Euclidean
        \\
        \color{darkblue}
        rotation
        \\
        group
      \end{tabular}
    }
    \\
    &
    1
    &
    1
\end{tikzcd}
\end{equation}
\vspace{-2mm}

The corresponding {\it dual lattice} (e.g. \cite[p. 311]{ReedSimon78}\cite[\S 2.2.2]{Tong17},  in CMT often: ``reciprocal lattice'' \cite[p. 27]{Kittel53}) characterizes the space of distinguishable wave-vectors/momenta (of electrons, phonons, ...) in the crystal, which is the {\it Brillouin torus} (e.g. \cite[p. 57]{FreedMoore12}):
\begin{equation}
  \label{BrillouinTorus}
  \overset{
    \mathclap{
    \raisebox{3pt}{
      \tiny
      \color{darkblue}
      \bf
      \def\arraystretch{.9}
      \begin{tabular}{c}
        Brillouin
        \\
        torus
      \end{tabular}
    }
    }
  }{
  \DualTorus{d}
  }
  \;\;\;\coloneqq\;\;\;
  \overset{
    \mathclap{
    \raisebox{3pt}{
      \tiny
      \color{darkblue}
      \bf
      \def\arraystretch{.9}
      \begin{tabular}{c}
        all
        \\
        Euclidean
        momenta
      \end{tabular}
    }
    }
  }{
  \Homs{\big}
    { \Lambda }{\RealNumbers}
  }
  \big/
  \overset{
    \mathclap{
    \raisebox{3pt}{
      \tiny
      \color{darkblue}
      \bf
      \def\arraystretch{.9}
      \begin{tabular}{c}
        trivial
        \\
        lattice momenta
      \end{tabular}
     }
    }
  }{
  \Homs{\big}
    { \Lambda }{ \Integers }
  }
  \;\simeq\;
  \overset{
    \raisebox{3pt}{
      \tiny
      \color{darkblue}
      \bf
      Pontrjagin dual group
    }
  }{
  \Homs{\big}
    {
      \Lambda
    }
    { \CircleGroup }
  }
  \,.
\end{equation}
Notice that the Brillouin torus inherits an action of the point group $G_{\mathrm{pt}}$; we come back to this below in \cref{CrystallographicSymmetriesAndOrbifoldKTheory}:
\begin{equation}
  \label{EquivariantBrillouinTorus}
  G_{\mathrm{pt}} \acts \,
  \DualTorus{d}
  \;\;
  =
  \;\;
  \Homs{\big}
   { G_{\mathrm{pt}} \acts \, \Lambda }
   { \CircleGroup }
\end{equation}

\newpage
\noindent
{\bf Bloch states of electrons and
the valence vector bundle.}
By lattice-translation invariance,
the energies of excitations of a crystal
depend {\it independently} on:

\begin{itemize}[leftmargin=*]

\item[{\bf 1\rlap{.}}]
their wave-vector $k$, which ranges through
the {\it Brillouin torus} $\DualTorus{d}$ \eqref{BrillouinTorus};

\item[{\bf 2\rlap{.}}]
their internal degrees of freedom, such as the atomic sites in a unit cell, their atomic orbitals and the spin degrees of
freedom of the electrons that are involved in the excitation.
\end{itemize}

\vspace{-0cm}
\begin{center}
\hspace{-.55cm}
\hypertarget{ElectronWaveFunctionInPotential}{}
\begin{tabular}{ll}
\begin{minipage}{8.5cm}
  \footnotesize
  \noindent
  {\bf Figure 2.}
  Schematics of a
  Bloch wave state in a
  periodic (crystalline)
  Coulomb potential background (e.g. \cite[Fig. 2.1]{Roessler04}\cite[Fig. 4.5]{Li06} \cite[Fig. 4.1, 4.9]{Vanderbilt18}).
  The actual Bloch wave functions are
  periodic only up to a complex phase
  (not shown here)
  depending
  on the inverse wavelength $k$
  (see Fact \ref{BlochFloquetTheory}).

\vspace{.5mm}
  The possible positronic admixture
  indicated in the figure
  (cf. \cite[(2.3)]{KS77} \cite[\S 1.4.6]{Thaller92})
  is
  traditionally first disregarded in
  solid state physics, but later implicitly re-invoked to account for spin-orbit coupling (cf. Ex. \ref{HaldaneModel}).

  \vspace{.5mm}
  We find that
  taking positronic contributions to the valence bundle into
  into account
  (see Fact \ref{VacuaOfTheRelativisticElectronPositronFieldInBackground})
  is crucial
  for bringing out the
  expected K-theoretic classification of topological insulators
  (see Fact \ref{KTheoryClassificationOfTopologicalPhasesOfMatter}).
\end{minipage}
&
\raisebox{-46pt}{
\fbox{
\begin{tikzpicture}[scale=1.2]

\begin{scope}[xscale=.85]
\clip (-.3-.65,.1) rectangle (5.8+.65,-1.8);

\begin{scope}[xscale=.5]
\begin{scope}
\draw
  (-.5,0)
  .. controls
    (1,0) and (1.2,-1) ..
  (1.3,-2);
\end{scope}
\begin{scope}[xscale=-1, shift={(1,0)}]
\draw[dashed]
  (-.5,0)
  .. controls
    (1,0) and (1.2,-1) ..
  (1.3,-2);
\end{scope}
\end{scope}

\begin{scope}[xscale=.5, shift={(4,0)}]
\begin{scope}
\draw
  (-.5,0)
  .. controls
    (1,0) and (1.2,-1) ..
  (1.3,-2);
\end{scope}
\begin{scope}[xscale=-1, shift={(1,0)}]
\draw
  (-.5,0)
  .. controls
    (1,0) and (1.2,-1) ..
  (1.3,-2);
\end{scope}
\end{scope}

\begin{scope}[xscale=.5, shift={(8,0)}]
\begin{scope}
\draw
  (-.5,0)
  .. controls
    (1,0) and (1.2,-1) ..
  (1.3,-2);
\end{scope}
\begin{scope}[xscale=-1, shift={(1,0)}]
\draw
  (-.5,0)
  .. controls
    (1,0) and (1.2,-1) ..
  (1.3,-2);
\end{scope}
\end{scope}

\begin{scope}[xscale=.5, shift={(12,0)}]
\begin{scope}
\draw[dashed]
  (-.5,0)
  .. controls
    (1,0) and (1.2,-1) ..
  (1.3,-2);
\end{scope}
\begin{scope}[xscale=-1, shift={(1,0)}]
\draw
  (-.5,0)
  .. controls
    (1,0) and (1.2,-1) ..
  (1.3,-2);
\end{scope}
\end{scope}

\end{scope}

\begin{scope}[shift={(+.64,-.6)}, gray]

\begin{scope}[orangeii, shift={(-1.8,0)}, xscale=-1]
\clip (-1,-.2) rectangle (-.3,1);
\begin{scope}[xscale=1, dashed]
\draw
  (-1,0)
    .. controls
    (-.5,0) and (-.5,1) ..
  (0,1);
\end{scope}
\end{scope}

\begin{scope}[orangeii, xscale=.8]
\begin{scope}[xscale=1, line width=2.5pt, white]
\draw
  (-1,0)
    .. controls
    (-.5,0) and (-.5,1) ..
  (0,1);
\end{scope}
\begin{scope}[orangeii, xscale=1]
\draw
  (-1,0)
    .. controls
    (-.5,0) and (-.5,1) ..
  (0,1);
\end{scope}
\begin{scope}[orangeii, xscale=-1, line width=2.5pt, white]
\draw
  (-1,0)
    .. controls
    (-.5,0) and (-.5,1) ..
  (0,1);
\end{scope}
\begin{scope}[xscale=-1]
\draw
  (-1,0)
    .. controls
    (-.5,0) and (-.5,1) ..
  (0,1);
\end{scope}
\end{scope}

\draw
  (.8,0) to (.9,0);

\begin{scope}[orangeii, shift={(1.7,0)}, xscale=.8]
\begin{scope}[xscale=1, line width=2.5pt, white]
\draw
  (-1,0)
    .. controls
    (-.5,0) and (-.5,1) ..
  (0,1);
\end{scope}
\begin{scope}[xscale=1]
\draw
  (-1,0)
    .. controls
    (-.5,0) and (-.5,1) ..
  (0,1);
\end{scope}
\begin{scope}[orangeii, xscale=-1, line width=2.5pt, white]
\draw
  (-1,0)
    .. controls
    (-.5,0) and (-.5,1) ..
  (0,1);
\end{scope}
\begin{scope}[xscale=-1]
\draw
  (-1,0)
    .. controls
    (-.5,0) and (-.5,1) ..
  (0,1);
\end{scope}
\end{scope}

\draw
  (.8+1.7,0) to (.9+1.7,0);

\begin{scope}[orangeii, shift={(1.7+1.7,0)}, xscale=.8]
\begin{scope}[xscale=1, line width=2.5pt, white]
\draw
  (-1,0)
    .. controls
    (-.5,0) and (-.5,1) ..
  (0,1);
\end{scope}
\begin{scope}[xscale=1]
\draw
  (-1,0)
    .. controls
    (-.5,0) and (-.5,1) ..
  (0,1);
\end{scope}
\begin{scope}[xscale=-1, line width=2.5pt, white]
\draw
  (-1,0)
    .. controls
    (-.5,0) and (-.5,1) ..
  (0,1);
\end{scope}
\begin{scope}[ xscale=-1]
\draw
  (-1,0)
    .. controls
    (-.5,0) and (-.5,1) ..
  (0,1);
\end{scope}
\end{scope}

\draw
  (.8+3.4,0) to (.9+3.4,0);

\begin{scope}[orangeii, shift={(1.7+1.7+1.7,0)}, xscale=.8]
\begin{scope}[xscale=1, line width=2.5pt, white]
\draw
  (-1,0)
    .. controls
    (-.5,0) and (-.5,1) ..
  (0,1);
\end{scope}
\clip (-1,-.2) rectangle (-.3,1);
\begin{scope}[xscale=1, dashed]
\draw
  (-1,0)
    .. controls
    (-.5,0) and (-.5,1) ..
  (0,1);
\end{scope}

\end{scope}

\end{scope}

\draw
  (2.85,-1.25) node
  {
    \begin{tabular}{l}
    $\;\;\;\;V$
    \\
      \tiny \bf
      \color{darkblue}
      \def\arraystretch{.9}
      \begin{tabular}{c}
        Coulomb
        \\
        potential
      \end{tabular}
    \end{tabular}
  };

\draw
  (3.2,+.75) node
  {
    \begin{tabular}{l}
      $\left\vert
        \hspace{-5pt}
        \raisebox{2pt}{
        \scalebox{.7}{$
          \begin{array}{cc}
            u_{\uparrow}, v_\uparrow
            \\
            u_{\downarrow}, v_\downarrow
          \end{array}
        $}
        }
        \hspace{-7pt}
        \right\rangle$
        \hspace{-10pt}
      \tiny
      \color{orangeii}
      \bf
      \def\arraystretch{.9}
      \begin{tabular}{c}
        electron/positron
        \\
        wavefunction
      \end{tabular}
    \end{tabular}
  };

\end{tikzpicture}
}
}
\end{tabular}
\end{center}

\begin{fact}[\bf Energy levels of free electron/positrons in the Coulomb potential of a crystal lattice]
\label{BlochFloquetTheory}
The Hilbert space of states of a single electron propagating in the
Coulomb potential of the lattice of nuclei inside a crystal (\cite[p. 312]{ReedSimon78})
may be identified with the space of square integrable sections of an infinite-rank Hilbert space bundle $\DualTorus{d} \times \mathscr{B}$ over the Brillouin torus $\DualTorus{d}$ \eqref{BrillouinTorus}, such that the total Hamiltonian is the direct integral
\cite[(137)]{ReedSimon78}
of Hamiltonians $H_k$ acting on the fibers
(the ``Bloch-Floquet transform'',
e.g. \cite[Thm. XIII.99]{ReedSimon78}\cite[(D.19-22)]{FreedMoore12}\cite[\S 1.1]{MonacoPanati16}):

\vspace{-2mm}
\begin{equation}
  \label{BlochFloquetTransform}
  \hspace{-3cm}
  \begin{tikzcd}[column sep=10pt, row sep=2pt]
    \mbox{
      \tiny
      \color{darkblue}
      \bf
      \def\arraystretch{.9}
      \begin{tabular}{c}
        Quantum system of
        \\
        single electron/positron
        \\
        in crystal lattice
      \end{tabular}
    }
    \phantom{AAA}
    &&
    \mbox{
      \tiny
      \color{darkblue}
      \bf
      \def\arraystretch{.9}
      \begin{tabular}{c}
        Direct integral of
        \\
        Bloch wave systems
      \end{tabular}
    }
    \\[-5pt]
    \mathllap{
      \mbox{
        \tiny
        \color{darkblue}
        \bf
        \def\arraystretch{.9}
        \begin{tabular}{c}
          Hilbert space of
          \\
          quantum states
        \end{tabular}
      }
    }
    \mathscr{B}
    \quad
    \ar[
      rr,
      "{\sim}",
      "{
        \mbox{
          \tiny
          \color{greenii}
          \bf
          Bloch-Floquet transform
        }
      }"{swap}
    ]
    &&
    \;\;\;
    \underset{
      \mathclap{
        k \in \DualTorus{2}
      }
    }{\int}
   \;\;\;
   (
    \!
    \mathscr{H}
    \!\oplus\!
    \mathscr{H}
    \!
   )
    \;
    \Differential^d k
    \mathrlap{
      \;\;
      \simeq
      \;
      L^2
      \big(
        \DualTorus{2}
        ;\,
       (
        \!
        \mathscr{H}
        \!\oplus\!
        \mathscr{H}
        \!
       )
      \big)
      \mbox{
        \tiny
        \color{darkblue}
        \bf
        \def\arraystretch{.9}
        \begin{tabular}{c}
          Square integrable sections
          \\
          of relativistic Bloch bundle
        \end{tabular}
      }
    }
    \\[-10pt]
    \rotatebox{+90}{$\acts$}
    \quad
    &&
    \;\;\;
    \rotatebox{+90}{$\acts$}
    \\[-5pt]
    \mathllap{
      \mbox{
        \tiny
        \color{darkblue}
        \bf
        \def\arraystretch{.9}
        \begin{tabular}{c}
          Hamiltonian
          \\
          operator
        \end{tabular}
      }
    }
   \scalebox{0.8}{$H$}
   \quad
    &\longmapsto&
    \scalebox{0.8}{$
    \underset{
      \mathclap{
    k\in \DualTorus{2}
      }
    }{\int}
    \;\;\; H_k
    \,
    \Differential^d k
    \mathrlap{
      \mbox{
        \tiny
        \color{darkblue}
        \bf
        \def\arraystretch{.9}
        \begin{tabular}{c}
          Direct integral of
          \\
          Bloch Hamiltonians
        \end{tabular}
      }
    }
    $}
  \end{tikzcd}
\end{equation}
\end{fact}

\vspace{-2mm}
The collections of available energies at fixed $k$ (hence the eigenvalue spectra of the Bloch Hamiltonians $H_k$)
form graphs over the Brillouin torus, called the {\it energy bands} (e.g. \cite[\S 2]{Seeger04}\cite[\S 1]{Li06}, see \hyperlink{FigureBandStructure}{\it Figure 3}).

\noindent In its ground state under given ambient conditions (strain, temperature, etc.), the material is
in the Fock state which inhabits all the electron quantum states of energy $\leq \mu_F$,
where $\mu_F \in \RealNumbers$ is the chemical potential (or Fermi energy).
If this energy lies in a {\it band gap} then
the band right below the gap energy is called the {\it valence bundle}, while the band right
above is called the {\it conduction bundle}. Accordingly, we will say that
the sub-bundle of Bloch states of energy below the gap (e.g. \cite[Prop. D.13]{FreedMoore12}) is the {\it valence bundle}: \footnote{
  Sometimes (e.g. in \cite{Panati06} and its followups (\cite{DeNittisLein11}\cite{MonacoPanati16}),
  the
  ``valence bundle''
  \eqref{TheValenceBundle}
  is called the ``Bloch bundle''. This seems unnecessarily confusing
  (cf. \cite[p. 5]{FruchartCapentier13}), since ``Bloch bundle'' would instead seem to be the canonical name for the full bundle $\mathscr{B}$
  of all Bloch states.
}
\vspace{-3mm}
\begin{equation}
  \label{TheValenceBundle}
     \mathclap{
     \raisebox{0pt}{
      \tiny
      \color{darkblue}
      \bf
      \def\arraystretch{.9}
      \begin{tabular}{c}
        Valence
        \\
        bundle
      \end{tabular}
    }
    } \qquad
  {
    \mathscr{V}
  }
  \;=\;
  \Big\{
    k \,\in\, \DualTorus{d}
    ,\,
    \vert \psi\rangle
    \in
    \mathscr{B}_k
    \;\,\Big\vert\;\,
    \big\vert
    \langle \psi \vert H_k \vert \psi \rangle
    \big\vert
    \;\leq\;
    \mu_F
  \Big\}
  \;
  \subset
  \;
{
    \mathscr{B}
  }
\qquad
    \mathclap{
    \raisebox{0pt}{
      \tiny
      \color{darkblue}
      \bf
      \begin{tabular}{c}
        Bundle of all
        \\
        relativistic
        \\
        Bloch states
      \end{tabular}
    }
    }
\end{equation}

\vspace{-2mm}
Depending on the position of the lowest energy bands in relation to the maximal energy at which modes are excited (``occupied'')
in the material (the {\it Fermi energy} or ``chemical potential'' $\mu_F$) the material is (e.g. \cite[p. 314]{ReedSimon78}):

\vspace{1mm}
{\bf 1.} a {\it metal}/{\it conductor} if the chemical potential is inside an energy band;

{\bf 2.} a {\it semi-conductor} if the chemical potential is inside a small gap between two bands;

{\bf 3.} an {\it insulator} if the chemical potential is inside a substantial gap between bands.

\vspace{.2cm}

\hspace{+1.6cm}
\hypertarget{FigureBandStructure}{}
\begin{tikzpicture}
\begin{scope}[shift={(0,-.1)}]

\draw[line width=5pt, darkgray]
  (-1, +1.3)
    .. controls
      (-.8,+1.3)
      and
      (-.5, +1.3-.1)
    ..
  (0, +1.3-.1)
    .. controls
      (+.5, +1.3-.1)
      and
      (+.8, +1.3)
    ..
  (1, +1.3);

\draw[line width=9pt, darkgray]
  (-1, +1.0)
    .. controls
      (-.8,+1.0)
      and
      (-.5, +1.0+.1)
    ..
  (0, +1.0+.0)
    .. controls
      (+.5, +1.0-.1)
      and
      (+.8, +1.0)
    ..
  (1, +1.0);
\draw[line width=8pt, lightgray]
  (-1, +1.0)
    .. controls
      (-.8,+1.0)
      and
      (-.5, +1.0+.1)
    ..
  (0, +1.0+.0)
    .. controls
      (+.5, +1.0-.1)
      and
      (+.8, +1.0)
    ..
  (1, +1.0);

\draw[line width=4pt, lightgray]
  (-1, +1.3)
    .. controls
      (-.8,+1.3)
      and
      (-.5, +1.3-.1)
    ..
  (0, +1.3-.1)
    .. controls
      (+.5, +1.3-.1)
      and
      (+.8, +1.3)
    ..
  (1, +1.3);

\begin{scope}[shift={(0,-.5pt)}]
\draw[line width=3pt, darkgray]
  (-1, +1.6)
    .. controls
      (-.8,+1.6)
      and
      (-.5, +1.6-.03)
    ..
  (0, +1.6+.0)
    .. controls
      (+.5, +1.6+.03)
      and
      (+.8, +1.6)
    ..
  (1, +1.6);
\end{scope}

\draw[line width=3pt, lightgray]
  (-1, +1.6)
    .. controls
      (-.8,+1.6)
      and
      (-.5, +1.6-.03)
    ..
  (0, +1.6+.0)
    .. controls
      (+.5, +1.6+.03)
      and
      (+.8, +1.6)
    ..
  (1, +1.6);

\draw[line width=2.5pt, lightgray]
  (-1, +1.75)
    .. controls
      (-.8,+1.75)
      and
      (-.5, +1.75-.02)
    ..
  (0, +1.75+.0)
    .. controls
      (+.5, +1.75+.02)
      and
      (+.8, +1.75)
    ..
  (1, +1.75);

\draw[line width=2.5pt, lightgray]
  (-1, +1.9)
    .. controls
      (-.8,+1.9)
      and
      (-.5, +1.9-.01)
    ..
  (0, +1.9+.0)
    .. controls
      (+.5, +1.9+.01)
      and
      (+.8, +1.9)
    ..
  (1, +1.9);

\end{scope}

\begin{scope}[shift={(0,-.08)}]
\draw[line width=9pt, darkgray]
  (-1, +.5)
    .. controls
      (-.8,+.5)
      and
      (-.5, +.5+.08)
    ..
  (0, +.5+.08)
    .. controls
      (+.5, +.5+.08)
      and
      (+.8, +.5)
    ..
  (1, +.5);
\draw[line width=8pt, lightgray]
  (-1, +.5)
    .. controls
      (-.8,+.5)
      and
      (-.5, +.5+.08)
    ..
  (0, +.5+.08)
    .. controls
      (+.5, +.5+.08)
      and
      (+.8, +.5)
    ..
  (1, +.5);
\end{scope}

\draw (-1.2, +.42) node
 {$
   \mathllap{
   \scalebox{.7}{
     \color{darkblue}
     \bf
     Conduction band
   }
   }
 $};

\draw (-1.2, +1.2) node
 {$
   \mathllap{
   \scalebox{.7}{
     \color{darkblue}
     \bf
     \begin{tabular}{l}
     ever
     \\
     higher bands
     \end{tabular}
   }
   }
 $};

\draw[orangeii, dashed]
  (-1,0) to (1,0);

\draw[latex-latex]
  (-.2,+.32) to
  node {\scalebox{.6}{\colorbox{white}{gap}}}
  (-.2,-.37);

\draw[->]
  (-1.12,-1.6) to (-1.12,2.1);
\draw
 (-1.23, 1.65)
 node{
  \scalebox{.75}{$
    \mathllap{
      E \in \RealNumbers
    }
   $}
 };

\draw (-1.2, 0) node
  {$
    \mathllap{
      \scalebox{.7}{
        \color{orangeii}
        Chemical potential
      }
    }
  $};

\draw (-1.2+.15, 0) node
  {$
    \mathclap{
      \scalebox{.6}{
        \colorbox{white}{
          \hspace{-4pt}
          $\mu_F$
          \hspace{-4pt}
        }
      }
    }
  $};

\draw[->]
  (-1.3,-1.4) to (1.1, -1.4);
\draw (.65, -1.64) node
{
  \scalebox{.75}{$
    k \in \DualTorus{d}
  $}
};

\draw[line width=9pt]
  (-1,-.5)
    .. controls
      (-.8,-.5)
      and
      (-.5, -.5-.15)
    ..
  (0, -.5)
    .. controls
      (+.5,-.5+.15)
      and
      (+.8, -.5)
    ..
  (1, -.5);

\draw (-1.2, -.5) node
 {$
   \mathllap{
   \scalebox{.7}{
     \color{darkblue}
     \bf
     Valence band
   }
   }
 $};

\begin{scope}[shift={(0,-.4)}]
\draw[line width=5pt]
  (-1,-.5)
    .. controls
      (-.8,-.5)
      and
      (-.5, -.5-.05)
    ..
  (0, -.5-.05)
    .. controls
      (+.5,-.5-.05)
      and
      (+.8, -.5)
    ..
  (1, -.5);

\begin{scope}[shift={(0,-.3)}]
\draw[line width=3pt]
  (-1,-.5)
    .. controls
      (-.8,-.5)
      and
      (-.5, -.5-.01)
    ..
  (0, -.5-.01)
    .. controls
      (+.5,-.5-.01)
      and
      (+.8, -.5)
    ..
  (1, -.5);
\end{scope}

\draw (-1.2, -.65) node
 {$
   \mathllap{
   \scalebox{.7}{
     \color{darkblue}
     \bf
     lowest bands
   }
   }
 $};

\end{scope}

\draw (.55,-.7) node
{$
  \mathrlap{
  \left.
    \begin{array}{l}
      \phantom{-}
      \\
      \phantom{-}
      \\
      \phantom{-}
    \end{array}
  \right\}
    \!\!\!\!\!\!
    \scalebox{.7}{
      \color{darkblue}
      \bf
      \begin{tabular}{c}
        bands of
        \\
        valence
        bundle
      \end{tabular}
    }
  }
$};

\draw (3.1,0) node {};

\begin{scope}[xshift=4.6cm]

\begin{scope}[yshift=.5cm]

\begin{scope}[shift={(0,-.54)}]
\draw[line width=9pt, darkgray]
  (-1, +.5)
    .. controls
      (-.8,+.5)
      and
      (-.5, +.5+.08)
    ..
  (0, +.5+.08)
    .. controls
      (+.5, +.5+.08)
      and
      (+.8, +.5)
    ..
  (1, +.5);
\draw[line width=8pt, lightgray]
  (-1, +.5)
    .. controls
      (-.8,+.5)
      and
      (-.5, +.5+.08)
    ..
  (0, +.5+.08)
    .. controls
      (+.5, +.5+.08)
      and
      (+.8, +.5)
    ..
  (1, +.5);
\end{scope}

\draw[line width=9pt, darkgray]
  (-1,-.5)
    .. controls
      (-.8,-.5)
      and
      (-.5, -.5-.15)
    ..
  (0, -.5)
    .. controls
      (+.5,-.5+.15)
      and
      (+.8, -.5)
    ..
  (1, -.5);
\draw[line width=8pt, lightgray]
  (-1,-.5)
    .. controls
      (-.8,-.5)
      and
      (-.5, -.5-.15)
    ..
  (0, -.5)
    .. controls
      (+.5,-.5+.15)
      and
      (+.8, -.5)
    ..
  (1, -.5);

\draw[orangeii, dashed]
  (-1,-.5) to (1,-.5);

\draw (0,-1.5) node {
  \scalebox{.7}{
  \color{darkblue}
  \bf
  metal/conductor
  }
};

\clip (-1,-.52) rectangle (1,-1.5);

\draw[line width=9pt]
  (-1,-.5)
    .. controls
      (-.8,-.5)
      and
      (-.5, -.5-.15)
    ..
  (0, -.5)
    .. controls
      (+.5,-.5+.15)
      and
      (+.8, -.5)
    ..
  (1, -.5);

\end{scope}

\begin{scope}[shift={(2.3,.25)}]

\begin{scope}[shift={(0,-.54)}]
\draw[line width=9pt, darkgray]
  (-1, +.5)
    .. controls
      (-.8,+.5)
      and
      (-.5, +.5+.08)
    ..
  (0, +.5+.08)
    .. controls
      (+.5, +.5+.08)
      and
      (+.8, +.5)
    ..
  (1, +.5);
\draw[line width=8pt, lightgray]
  (-1, +.5)
    .. controls
      (-.8,+.5)
      and
      (-.5, +.5+.08)
    ..
  (0, +.5+.08)
    .. controls
      (+.5, +.5+.08)
      and
      (+.8, +.5)
    ..
  (1, +.5);
\end{scope}

\draw[orangeii, dashed]
  (-1,-.25) to (1,-.25);

\draw (0,-1.25) node {
  \scalebox{.7}{
  \color{darkblue}
  \bf
  semi-conductor
  }
};

\draw[line width=9pt]
  (-1,-.5)
    .. controls
      (-.8,-.5)
      and
      (-.5, -.5-.15)
    ..
  (0, -.5)
    .. controls
      (+.5,-.5+.15)
      and
      (+.8, -.5)
    ..
  (1, -.5);

\end{scope}

\begin{scope}[shift={(4.6,0)}]

\begin{scope}[shift={(0,0)}]
\draw[line width=9pt, darkgray]
  (-1, +.5)
    .. controls
      (-.8,+.5)
      and
      (-.5, +.5+.08)
    ..
  (0, +.5+.08)
    .. controls
      (+.5, +.5+.08)
      and
      (+.8, +.5)
    ..
  (1, +.5);
\draw[line width=8pt, lightgray]
  (-1, +.5)
    .. controls
      (-.8,+.5)
      and
      (-.5, +.5+.08)
    ..
  (0, +.5+.08)
    .. controls
      (+.5, +.5+.08)
      and
      (+.8, +.5)
    ..
  (1, +.5);
\end{scope}

\draw[orangeii, dashed]
  (-1,0) to (1,0);

\draw[line width=9pt]
  (-1,-.5)
    .. controls
      (-.8,-.5)
      and
      (-.5, -.5-.15)
    ..
  (0, -.5)
    .. controls
      (+.5,-.5+.15)
      and
      (+.8, -.5)
    ..
  (1, -.5);

\draw (0,-1) node {
  \scalebox{.7}{
  \color{darkblue}
  \bf
  insulator
  }
};

\draw (1.12, +.5)
  node
  {$
    \mathrlap{
      \scalebox{.7}{
        un-occupied
      }
    }
  $};

\draw (1.12, -.5)
  node
  {$
    \mathrlap{
      \scalebox{.7}{
        occupied
      }
    }
  $};

\end{scope}

\end{scope}

\end{tikzpicture}

\vspace{.0cm}

\begin{minipage}{12cm}
\footnotesize
{\bf Figure 3 -- Electron band structure in crystals.}
The remaining case of {\it semi-metals} is shown in \hyperlink{BandStructureOfSemiMetals}{\it Figure 6}.
\end{minipage}

\vspace{.2cm}

We are to be concerned with {\it topological insulators}: those insulators whose valence bundle has a non-trivial K-class.

\newpage
\noindent
{\bf The electron/positron field and Fredholm operators.}
Fact \ref{BlochFloquetTheory} serves to determine the energy levels of single electrons in the crystal (\hyperlink{FigureBandStructure}{\it Figure 3}). But even in the approximation of non-interacting\footnotemark[1]
electrons in a background of fixed (possibly screened) nuclei in the crystal, the proper description of the electronic ground state -- hence of the valence bundle -- requires treating them as excitations of the ``second quantized'' free relativistic Dirac field
(e.g. \cite{Thaller92}\cite{Strange98})
subject to the classical electromagnetic background sourced by the comparatively heavy nuclei (see also  \cite{Bongaarts18}\cite{HLS05}).

\smallskip
Indeed, the electron's spin-orbit coupling (e.g. \cite[(10)]{Mansuripur19}) -- a key phenomenon underlying the existence of topological insulators outside of an external magnetic field ({\it quantum spin Hall} materials \cite{MHZ11}) -- is a relativistic effect invisible in the non-relativistic approximation.
But in this relativistic description, the single electron Hilbert space $\mathscr{H}$ is necessarily accompanied by a copy of the single {\it positron}\footnote{
  Here we speak of the fundamental but dressed electron/positron field (\cite[(3.2)]{KS77})
  propagating in the electromagnetic background field
  (\cite[(1.1) (2.12)]{KS77})
  that is sourced by the crystal's nuclei and their tightly bound electrons.
  On the other hand, in solid state physics it is tradition to speak of an effective ``electron/hole'' field, which practically refers to the creation/annihilation operators in ad-hoc Fock space Hamiltonians (such as lattice hopping models) which are imagined to provide a tractable effective description of the
  complicated real physics embodied by the fundamental electron/hole field.
} Hilbert space,
to form a $\ZTwo$-graded Hilbert space $\mathscr{H} \oplus \mathscr{H}$ of the single Dirac particle.

\medskip
We now observe (Fact \ref{VacuaOfTheRelativisticElectronPositronFieldInBackground} below) that these charged electron/positron ground states are naturally encoded by {Fredholm operators} and, as such, naturally classified by K-theory (Fact \ref{KTheoryClassificationOfTopologicalPhasesOfMatter} below).
To appreciate this, recall that a {\it Fredholm operator}
(\cite[App.]{AtiyahAnderson67}, review in \cite[\S 33]{Arveson02})
is a bounded linear map between Hilbert spaces, whose kernel and cokernel are of finite dimension:
\vspace{-4mm}
\begin{equation}
  \label{FredholmOperator}
  \hspace{-4mm}
  \begin{tikzcd}
    \Big\{
      \psi \in \mathscr{H}
      \;\big\vert\;
      \forall_\phi
      \;
      \langle \phi \vert
      F
      \vert \psi \rangle = 0
    \Big\}
    \,\simeq\,
    \underset{
      \mathclap{
      \raisebox{-4pt}{
        \tiny
        \color{darkblue}
        \begin{tabular}{c}
        \bf   finite-dimensional kernel
        \end{tabular}
      }
      }
    }{
      \mathrm{ker}(F)
    }
   \;\;\; \ar[r, hook]
    &[-6pt]
    \mathscr{H}
    \ar[
      rr,
      "{
        \overset{
          \raisebox{3pt}{
            \tiny
            \color{greenii}
            \bf
            Fredholm operator
          }
        }{
          F
        }
       }",
      "{ \scalebox{.6}{bounded linear}  }"{swap}
    ]
    &&
    \mathscr{H}
    \ar[r, ->>]
    &[-6pt]
    \;\;\;
    \underset{
      \mathclap{
      \raisebox{-4pt}{
        \tiny
        \color{darkblue}
        \bf
        \begin{tabular}{c}
          finite-dimensional cokernel
        \end{tabular}
      }
      }
    }{
      \mathrm{coker}(F)
    }
    \,\simeq\,
    \Big\{
      \psi \in \mathscr{H}
      \;\vert\;
      \forall_{\phi}
      \;
      \langle
        \psi
      \vert
        F
      \vert
      \phi
      \rangle
      \,=\,0
    \Big\}
    \,.
  \end{tikzcd}
\end{equation}

\vspace{-3mm}
\noindent The {\it index} of a Fredholm operator is the difference of these dimensions:
\vspace{-2mm}
\begin{equation}
  \label{FredholmIndex}
  \overset{
    \mathclap{
      \raisebox{2pt}{
      \tiny
      \color{darkblue}
      \bf
      \begin{tabular}{c}
        Fredholm index
      \end{tabular}
    }
    }
  }{
    \mathrm{ind}(F)
  }
  \;\;:=\;
  \mathrm{dim}
  \big(
    \mathrm{ker}(F)
  \big)
  \,-\,
  \mathrm{dim}
  \big(
    \mathrm{coker}(F)
  \big)
  \;\;\;
  \in
  \;
  \NaturalNumbers
  \,.
\end{equation}

\begin{example}[Raising and lowering operators]
To get an intuition for Fredholm operators, choose an ordered Hilbert space basis $\big\{  \vert E_n \rangle \,\vert\, n \in \NaturalNumbers \big\}$ for $\mathscr{H}$ \eqref{BlochFloquetTransform} with linear shift operators
$c, a$
given by
$c \vert E_{n} \rangle \,:=\, \vert E_{n+1} \rangle$,
$a \vert E_{n+1} \rangle := \vert E_n \rangle$,
$a \vert E_0 \rangle = 0$.
Then, for all $n,m \in \NaturalNumbers$,
the following composite raising/lowering operator is Fredholm, with its index measuring a
{\it net} shift in ``energy levels'':
\vspace{-3mm}
$$
  F
  \;:=\;
  c^n a^m
  \;\;\;\;\;
  \Rightarrow
  \;\;\;\;\;\;\;
  \mathrm{ker}(F)
  \,\simeq\,
  \mathrm{span}
  \big\{
    E_k
    \;\vert\;
    k < m
  \big\}
  \,,
  \;\;\;\;
  \mathrm{coker}(F)
  \,\simeq\,
  \mathrm{span}
  \big\{
    E_k
    \;\vert\;
    k < n
  \big\}
  \,,
  \;\;\;\;
  \mathrm{ind}(F)
  \;=\;
  m - n
  \,.
$$
\end{example}

The nature of this simple example may help to conceptualize the following profound example:

\begin{fact}[\bf Fermi sea ground states of the relativistic free electron/positron field in the Coulomb potential of a crystal lattice]
  \label{VacuaOfTheRelativisticElectronPositronFieldInBackground}
  The Fermi sea ground states of the free electron/positron field in a classical Coulomb potential (``dressed vacua'', see \hyperlink{ElectronWaveFunctionInPotential}{\it Figure 2})
  are characterized by Fredholm operators
  \eqref{FredholmOperator} of the form

  \vspace{-3mm}
  \begin{equation}
    \label{TheDressedVacuumFredholmOperator}
    F
    \;=\;
    \gamma
    \circ
    P_+ U P_+
    \,,
  \end{equation}
  \vspace{-.6cm}

  \noindent
  where

\begin{itemize}

\item[\bf (i)]
  $P_+$ denotes the projector onto the single dressed electron Hilbert space $\mathscr{H}$;

\item[\bf (ii)]
  $U$ denotes the unitary operator
  \cite[(3.9)]{KS77}
  on the single particle Hilbert space
  $\mathscr{H} \oplus \mathscr{H}$ which {\it implements}
  (e.g. \cite[\S 10.2.1]{Thaller92})
  on the corresponding fermionic Fock space the transformation from the dressed to the bare field vacuum;

\item[\bf (iii)]
 $\gamma$ denotes an isomorphism exchanging the
 electron/positron summands:
 $\gamma \circ P_{\pm} \,=\,  P_{\mp} \circ \gamma$;
\end{itemize}
such that the total charge is given by the Fredholm index \eqref{FredholmIndex}:
\vspace{-3mm}
  \begin{equation}
    \label{FredholmOperatorBetweenElectronPositronHilbertSpaces}
    \begin{tikzcd}[row sep=-2pt, column sep=40pt]
      \mathllap{
        \mbox{
          \tiny
          \color{darkblue}
          \bf
          \def\arraystretch{.9}
          \begin{tabular}{c}
            Electron states in
            \\
            dressed vacuum
          \end{tabular}
        }
      }
      \mathrm{ker}(F)
      \ar[r, hook]
      &[-10pt]
      \overset{
        \mathclap{
        \raisebox{5pt}{
          \tiny
          \color{darkblue}
          \bf
          \def\arraystretch{.9}
          \begin{tabular}{c}
            single electron
            \\
            Hilbert space
          \end{tabular}
        }
        }
      }{
        \mathscr{H}
      }
      &&
      \mathscr{H}
      \\
      &
      \oplus
      &&
      \oplus
      \\
      &
      \underset{
        \raisebox{-4pt}{
          \tiny
          \color{darkblue}
          \bf
          \begin{tabular}{c}
            single positron
            \\
            Hilbert space
          \end{tabular}
        }
      }{
        \mathscr{H}
      }
     \ar[
        uurr,
        shorten <=-20pt,
        "{F^\dagger}"{pos=.8}
      ]
      &&
      \mathscr{H}
      \ar[r, ->>]
       \ar[
        from=uull,
        crossing over,
        "{
          F
        }"{pos=.2},
        "{
          \colorbox{white}{
            \tiny
            \color{greenii}
            \bf
            \begin{tabular}{c}
              dressed vacuum
              \\
              Fredholm operator
            \end{tabular}
          }
        }"{swap, yshift=1pt, sloped}
      ]
      &[-10pt]
      \mathrm{coker}(F)
      \mathrlap{
        \raisebox{3pt}{
          \tiny
          \color{darkblue}
          \bf
          \def\arraystretch{.9}
          \begin{tabular}{c}
            Positron states in
            \\
            dressed vacuum
          \end{tabular}
        }
      }
    \end{tikzcd}
  \end{equation}
  \vspace{-2mm}
  \def\arraystretch{1.4}
  \begin{equation}
    \hspace{-1.6cm}
    \begin{array}{ccccc}
    \overset{
      \raisebox{+4pt}{
        \tiny
        \color{darkblue}
        \bf
        \def\arraystretch{.9}
        \begin{tabular}{c}
          Total charge in
          \\
          dressed vacuum
        \end{tabular}
      }
    }{
      \mathrm{ind}(F)
    }
    &=&
    \overset{
      \raisebox{+4pt}{
        \tiny
        \color{darkblue}
        \bf
        \def\arraystretch{.9}
        \begin{tabular}{c}
          number of electrons in
          \\
          dressed vacuum state
        \end{tabular}
      }
    }{
      \mathrm{dim}
      \big(
        \mathrm{ker}(F)
      \big)
    }
    &-&
    \overset{
      \raisebox{+4pt}{
        \tiny
        \color{darkblue}
        \bf
        \def\arraystretch{.9}
        \begin{tabular}{c}
          number of positrons in
          \\
          dressed vacuum state
        \end{tabular}
      }
    }{
    \mathrm{dim}
    \big(
      \mathrm{coker}(F)
    \big)
    }
    \\
    &=&
    \mathrm{dim}
    \big(
      \mathrm{coker}(F^\dagger)
    \big)
    &-&
    \mathrm{dim}
    \big(
      \mathrm{ker}(F^\dagger)
    \big).
    \end{array}
  \end{equation}
\end{fact}
\begin{proof}
  The first statement is due to \cite{KS77}
  (the dressed vacuum appears in equation (3.48) there), while the second is made explicit in \cite{CareyHurstOBrien82}.
  In presenting this, there is freedom in choosing various isomorphisms;
  notably, in \eqref{TheDressedVacuumFredholmOperator}
  we have used the freedom to postcompose the expression
  $P_+ U P_+$ considered in \cite[p. 364]{CareyHurstOBrien82}
  with the isomorphism $\gamma$ so that as
  to make
  not only its
  kernel but also its cokernel be subspaces of the dressed electron- and positron-Hilbert space, respectively
  \eqref{FredholmOperatorBetweenElectronPositronHilbertSpaces}.
\end{proof}

\newpage

\noindent
{\bf Relativistic Bloch Hamiltonians and topological K-theory.}
The (anti-)particle physics of Fact \ref{VacuaOfTheRelativisticElectronPositronFieldInBackground} clearly suggests
to regard the Fredholm operator $F$
as odd-graded
and
as a summand in a self-adjoint operator
$\widehat{F} \;:=\; F + F^\dagger$,
as shown above.
Strikingly, this also turns out, much less obviously so (\cite{AtiyahSinger69}\cite{Karoubi70}), to be the most fruitful mathematical definition of Fredholm operators;  so that we take the space of all Fredholm operators (suitably topologized, see \cite[Def. 3.2]{AtiyahSegal04}\cite[Def. A.39]{FreedHopkinsTelement07TwistedKTheoryI}) to be:
\vspace{-4mm}
\begin{equation}
  \label{SpaceOfFredholmOperators}
  \FredholmOperators^0_{\ComplexNumbers}
  \;:=\;
  \left\{
  \hspace{-6pt}
  \mbox{\small
  \begin{tabular}{c}
   Bounded
    self-adjoint
    operators
    \\
    of odd degree on $\mathscr{H} \oplus \mathscr{H}$
    \\
    with finite-dimensional (co-)kernels
  \end{tabular}
  }
  \hspace{-7pt}
  \right\}
  \;\;
  =
  \;\;
  \left\{\!\!
  \def\arraystretch{1.4}
  \begin{array}{ll}
    \scalebox{.8}{Bounded oper.}
    &
    \widehat{F}
      :
    \!\!
    \begin{tikzcd}
      \mathscr{H}^2
      \ar[
        r,
        shorten=-1pt,
        "{\scalebox{.65}{bounded}}"{pos=.4},
        "{\scalebox{.65}{$\ComplexNumbers$-linear}}"
          {swap, yshift=1pt, pos=.4}
      ]
      &
      \mathscr{H}^2
    \end{tikzcd}
    \\
    \scalebox{.8}{Odd degree}
    &
    \widehat{F} \circ \beta
    \,=\,
    -
    \beta \circ \widehat{F}
    \\
    \scalebox{.8}{Self-adjoint}
    &
    \widehat{F}^\ast
    \,=\,
    \widehat{F}
    \;:=\;
    F + F^\ast
    \\
    \scalebox{.8}{Fredholm}
    &
    \mathrm{dim}
    \big(
      \mathrm{ker}\big(
        \widehat{F}
      \big)
    \big)
    \;<\;
    \infty
  \end{array}
 \!\!\! \right\}
  \,.
\end{equation}

\vspace{-2mm}
\noindent Here
$
  \beta
  \;:=\;
  \Big(
  \hspace{-2pt}
  \scalebox{.7}{$
    \arraycolsep=2pt
    \begin{array}{cc}
      1 & 0
      \\
      0 & -1
    \end{array}
  $}
  \hspace{-2pt}
  \Big)
$
denotes the grading operator
as usual in Dirac theory (e.g. \cite[(1.9)]{Thaller92}).

In combining Fact \ref{BlochFloquetTheory}
with Fact \ref{VacuaOfTheRelativisticElectronPositronFieldInBackground}, the following statement seems rather
plausible and its proof should be a fairly straightforward variation on the analysis in \cite{KS77}. But,
since here is not the place to enter into detailed solid state physics, we state this as a conjecture (cf. Rem. \ref{ComparisonToExistingLiteratureForKClassification} below):

\begin{conjecture}[Families of relativistic Bloch vacua]
  \label{FamiliesOfRelativisticBlochVacua}
  Under passage through the Bloch-Floquet transform
  of Fact \ref{BlochFloquetTheory}:

\begin{itemize}[leftmargin=.6cm]

\item[\bf (i)]
   The analogous statement of Fact \ref{VacuaOfTheRelativisticElectronPositronFieldInBackground} applies to
   each of the Bloch Hamiltonians \eqref{BlochFloquetTransform}
   of relevance for topological insulators.

\item[\bf (ii)]
  The resulting family
  $F_{{}_{\DualTorus{2}}}$
  of Fredholm operators over the Brillouin torus is continuous.

\item[\bf (iii)]
  Its homotopy class corresponds to the physical deformation class of the crystalline material.

\end{itemize}

\end{conjecture}
The first two points of Conjecture \ref{FamiliesOfRelativisticBlochVacua} mean that the Fredholm operator \eqref{FredholmOperatorBetweenElectronPositronHilbertSpaces} decomposes into a family of what we might call ``Bloch Fredholm operators'',
being a continuous function
of momenta in the Brillouin torus with values in the space
$\FredholmOperators$
\eqref{SpaceOfFredholmOperators}:
\begin{equation}
    \label{BLochFredholmOperator}
    \mathclap{
    \raisebox{0pt}{
      \tiny
      \color{greenii}
      \bf
      \def\arraystretch{.9}
      \begin{tabular}{c}
      Bloch family of
      \\
      Fredholm operators
      \end{tabular}
    }
    }
    \qquad \quad
  {
    F_{{}_{\DualTorus{2}}}
  }
  \;\; :\;\;
  \DualTorus{2}
  \xrightarrow[\continuous]{\phantom{-----}}
  \;\;
  {
    \FredholmOperators^0_{\ComplexNumbers}
  }
\qquad \quad
    \mathclap{
    \raisebox{0pt}{
      \tiny
      \color{darkblue}
      \bf
      \def\arraystretch{.9}
      \begin{tabular}{c}
        Topological space of
        \\
        all Fredholm operators.
      \end{tabular}
    }
    }
\end{equation}
This is, equivalently, a morphism of Bloch state bundles
\eqref{BlochFloquetTransform} which is fiberwise Fredholm, so that its kernel and cokernel are finite-rank sub-bundles,
to be thought of here as the valence bundle of electron and of positron states, respectively:
\vspace{-2mm}
\begin{equation}
  \label{BTParametrizedFredholmOperators}
  \begin{tikzcd}[column sep=50pt]
    \overset{
      \mathclap{
      \raisebox{3pt}{
        \tiny
        \color{darkblue}
        \bf
        Bloch bundle of
        electron states
      }
      }
    }{
    \DualTorus{2}
    \times
    \mathscr{B}
    }
    &&
    \DualTorus{2}
    \times
    \mathscr{B}
    \ar[r, ->>]
    &[-20pt]
    \mathrm{coker}
    \big(
      F^\dagger_{{}_{\DualTorus{2}}}
    \big)
    \mathrlap{
      \hspace{-4pt}
      \mbox{
        \tiny
        \color{orangeii}
        \bf
        \def\arraystretch{.9}
        \begin{tabular}{c}
          valence bundle
          \\
          of electron states
        \end{tabular}
      }
    }
    \\
    \DualTorus{2}
    \times
    \mathscr{B}
    \ar[
      urr,
      "{
        F^\dagger_{{}_{\DualTorus{2}}}
      }"{pos=.8}
    ]
    &&
    \underset{
      \mathclap{
      \raisebox{-3pt}{
        \tiny
        \color{darkblue}
        \bf
        Bloch bundle of positron states
      }
      }
    }{
    \DualTorus{2}
    \times
    \mathscr{B}
    }
    \ar[
      from=ull,
      crossing over,
      "{
        F_{{}_{\DualTorus{2}}}
      }"{pos=.2},
      "{
        \scalebox{.5}{ \bf
          \colorbox{white}{
            \hspace{-3pt}
            \color{greenii}
            \def\arraystretch{.9}
            \begin{tabular}{c}
              Bloch family of dressed
              \\
              vacuum Fredholm operators
            \end{tabular}
            \hspace{-3pt}
          }
        }
      }"{swap, sloped}
    ]
    \ar[r, ->>]
    &
    \mathrm{coker}
    \big(
      F_{{}_{\DualTorus{2}}}
    \big)
    \mathrlap{
      \hspace{-4pt}
      \mbox{
        \tiny
        \color{orangeii}
        \bf
        \begin{tabular}{c}
          anti-valence
          bundle
          \\
          of positron states
        \end{tabular}
      }
    }
  \end{tikzcd}
\end{equation}

\vspace{-2mm}
\noindent The third point of Conjecture \ref{FamiliesOfRelativisticBlochVacua} means that
{\it homotopies} of such Fredholm families, namely continuous 1-parameter deformations of them
\vspace{-2mm}
$$
  F_{{}_{\DualTorus{2}}}
\;  \underset
    {\homotopy}
    {\sim}
    \;
  F'_{{}_{\DualTorus{2}}}
  \;\;\;\;
  :
  \;\;\;\;
  \begin{tikzcd}[row sep=5pt, column sep=large]
    \DualTorus{2} \times \{0\}
    \ar[
      d,
      hook,
      shorten=-2pt
    ]
    \ar[
      drr,
      "{ F_{{}_{\DualTorus{2}}} }"
    ]
    \\
    \DualTorus{2}
    \times
    [0,1]
    \ar[rr, dashed, shorten >=-3pt, "{\exists}"{description}, pos=.3]
    &&
    \FredholmOperators^0_{\ComplexNumbers}
    \,,
    \\
    \DualTorus{2} \times \{0\}
    \ar[
      u,
      hook',
      shorten=-2pt
    ]
    \ar[
      urr,
      "{ F_{{}_{\DualTorus{2}}} }"{swap}
    ]
  \end{tikzcd}
$$

\vspace{-2mm}
\noindent
correspond to sufficiently gentle (e.g. adiabatic, Rem. \ref{QuantumAdiabaticTheorem}) deformations of the physical material and its properties.

\medskip
Now it is a famous fact
about topological K-theory
(\cite{Karoubi70}\cite[p. 14]{AtiyahSegal04}\cite[\S A.5]{FreedHopkinsTelement07TwistedKTheoryI},
which the non-expert reader may take as the definition) that for any compact domain space $X$, such homotopy classes of Fredholm operators are in natural bijection to the K-cohomology group of $X$ (the {\it Atiyah-J\"anich Theorem} \cite{Janich65}\cite{AtiyahAnderson67}, see \cite{BleekerBooss}\cite{Karoubi78}), namely to the
group completion
$\big\{
 [\mathscr{V}] - [\mathscr{W}]
 \,\big\vert\,
 [\mathscr{V}], [\mathscr{W}]
\big\}$
of the semigroup (under fiberwise direct sum) of stable equivalence classes
$[\mathscr{V}]$
of complex vector bundles
$\mathscr{V}$
on $X$:

\vspace{-4mm}
\begin{equation}
\label{KTheoryClassificationByFredholmOperators}
\begin{tikzcd}[row sep=-2pt, column sep=6pt]
     \mathclap{
    \raisebox{0pt}{
      \tiny
      \color{darkblue}
      \bf
      \def\arraystretch{.9}
      \begin{tabular}{c}
        Homotopy classes of continuous
        \\
        families of Fredholm operators
      \end{tabular}
    }
    }
  \qquad \qquad
   {
  \Big\{
      X
      \xrightarrow[\continuous]{}
      \FredholmOperators^0_{\ComplexNumbers}
  \Big\}_{\!\!\big/\sim_{\homotopy}}
  }
  \ar[
    rr,
    "{\sim}"
  ]
  &&
 {
  \mathrm{KU}^0
  (
    X
  )
  }
  \qquad \quad
    \mathclap{
    \raisebox{0pt}{
      \tiny
      \color{darkblue}
      \bf
      \begin{tabular}{c}
        Complex K-cohomology
      \end{tabular}
    }
    }
  \\[-3pt]
  \phantom{AAAAA} \scalebox{0.8}{$ F_{{}_{X}} $}
  &\longmapsto&
 \scalebox{0.75}{$ \big[
    \mathrm{ker}(F_{{}_X})
  \big]
  -
  \big[
    \mathrm{coker}(F_{{}_X})
  \big].
  $}
\end{tikzcd}
\end{equation}
In more generality, there is an analogous statement -- discussed around \eqref{TEKTheory} below -- for the case that the crystal's dynamics, and hence the families of Fredholm operators that characterize its ground state valence bundle, obey certain ``quantum symmetries'', in which case the ``K-theory'' appearing here is to be understood as referring {\it twisted equivariant K-theory}.

\smallskip
Therefore:
\begin{fact}[\bf K-Theory classification of topological insulators]
  \label{KTheoryClassificationOfTopologicalPhasesOfMatter}
  Assuming Conjecture \ref{FamiliesOfRelativisticBlochVacua}, the deformation classes of non-interacting crystalline topological insulators are classified by the topological
  K-theory class of the formal difference between the electronic valence bundle and the positronic anti-valence bundle over the material's Brillouin torus.
\end{fact}
\begin{proof}
  By Fact \ref{VacuaOfTheRelativisticElectronPositronFieldInBackground} and Conjecture \ref{FamiliesOfRelativisticBlochVacua}, the topological
  phase is characterized by the homotopy class of the family of Fredholm operators over the Brillouin torus which characterizes the dressed
  electron/positron field vacua of the Bloch Hamiltonians. By \eqref{KTheoryClassificationByFredholmOperators}, this is the K-theory class of
  the virtual difference of the kernel and cokernel of these Fredholm operators.
\end{proof}
\begin{remark}[\bf The role of relativistic band theory for topological phases]
\label{RoleOfRelativisticBandTheory}
Concretely, the crux of the above discussion is the following:
{\it
In a proper relativistic treatment, the valence electron bundle of a topological insulator is accompanied by an anti-valence bundle of positrons which -- while locally annihilating with valence electrons -- may globally annihilate only up to a residue class in K-theory, thus reflecting the topological phase of the material.}
\end{remark}

Compare this to the arguments traditionally provided in the literature:

\begin{remark}[\bf State of the K-theory classification of topological phases in the previous literature]
\label{ComparisonToExistingLiteratureForKClassification}
$\,$

\noindent {\bf (i)} The proposal of \cite{Kitaev09}, that topological insulators are classified by some form of K-theory has become folklore (e.g. \cite{Thiang14}\cite{ShiozakiSatoGomi17}\cite{SdBKP18}); but the available justification offered in existing literature seems not to have
gone beyond the original motivation from \cite[p. 4]{Kitaev09}, which -- in its entirety -- reads:

\vspace{-2mm}
\begin{quote}
  \it
  We generally augment by a trivial system, i.e., a set of local, disjoint modes, like inner atomic shells. \newline
  This corresponds to adding
  an extra flat band on an insulator.
\end{quote}

\vspace{-2mm}
\noindent {\bf (ii)}
It seems to us that:

\begin{itemize}

\item[\bf (a)]
The quoted statement has not actually been supported by Bloch theory, in general.  (Compare \cite{Panati06}\cite{DeNittisLein11}\cite{FiorenzaMonacoPanati14a} \cite{FiorenzaMonacoPanati14b}\cite{MonacoPanati16} for the kind of strong assumptions and hard analysis that is required to prove similar statements).

\item[\bf (b)]

Even if this could be justified, it would motivate only the passage to the semigroup of {\it stable-} or {\it Murray-von Neumann-}
equivalence classes (e.g. \cite[Def. 2.2.1, Prop. 2.2.2, 2.2.7]{RordamLarsenLaustsen09}), but not passage to their K-theory class,
which involves further {\it group completion} (e.g. \cite[Def. 2.3.3, Def. 3.1.4]{RordamLarsenLaustsen09}; see \cite[\S 1.7]{Blackadar86} for translation).
\end{itemize}

\noindent
Indeed, in \cite{FreedMoore12}, which has practically become the other main reference for the claim that K-theory classifies topological phases,
the authors openly say that no such claim is made (\href{https://arxiv.org/pdf/1208.5055.pdf#page=57}{p. 57}):

\vspace{-2mm}
\begin{quote}
\it
Although [K-theory] is used in the condensed matter literature, it is not clear to us that it is well motivated.
\end{quote}

\vspace{-2mm}
\noindent {\bf (iii)}
Accordingly, it had remained unclear
(cf. \cite[\S 3]{Thiang15Iso}, \href{https://arxiv.org/pdf/1412.4191.pdf#page=7}{p. 7})
even which K-theory class is supposed to characterize the topological phase:
The first proposal of \cite[p. 57]{FreedMoore12} (``Type F'', \href{https://arxiv.org/pdf/1208.5055.pdf#page=56}{p. 56}) was to use the class of the formal difference between the valence bundle and some finite-rank conduction  bundle, in an apparent attempt to identify the physics role of the group completion.
The second proposal (``Type I'', \href{https://arxiv.org/pdf/1208.5055.pdf#page=57}{p. 57}), which is tacitly followed by later authors (notably \cite[\S III.A]{ShiozakiSatoGomi17}\cite{SdBKP18}) chooses to classify just the valence bundle, without commenting on how to find then a physics interpretation of the virtual bundles appearing in the K-theory classification, and highlighting that there is then no mathematical place for the physically important Hamiltonian inversion (i.e., the case $c=-1$ in \cite[Def. 3.7]{FreedMoore12}
then plays no apparent role, as highlighted in
\href{https://arxiv.org/pdf/1208.5055.pdf#page=56}{lines 6-8 on p. 56}, leading to the truncated classification statement of \href{https://arxiv.org/pdf/1208.5055.pdf#page=61}{Cor. 10.28} there).

\vspace{1mm}
\noindent {\bf (iv)}
It seems to us that all these issues are resolved
(as stated in Fact \ref{KTheoryClassificationOfTopologicalPhasesOfMatter})
by
Fact \ref{VacuaOfTheRelativisticElectronPositronFieldInBackground}
about proper relativistic band theory, where it is the admixture of positronic contributions that accurately brings out all the structure seen in the K-theory.
While this still relies on a conjecture about details of the solid state physics
(Conjecture \ref{FamiliesOfRelativisticBlochVacua}), this conjecture makes no surprising claims and is concrete enough to allow fairly straightforward (if possibly laborious) verification.
\end{remark}

\newpage

\begin{remark}[\bf Analogy with D-branes in string theory and CMT/ST duality]
\label{AnalogyWithDBraneCharge} $\,$

\noindent {\bf (i)}
The discussion culminating in Remark \ref{RoleOfRelativisticBandTheory} reveals a close analogy (shown in \hyperlink{RosettaStone}{\it Table 1})
between the  {\bf K-theory classification of}:
\begin{itemize}[leftmargin=.7cm]
\item[(1)] {\bf topological phases in solid state physics}, as discussed here in \cref{TEDKClassifiesTopologicalPhasesOfMatter}, and specifically of anyonic topological order as discussed below in \cref{TEDKDescribesRealisticAnyonSpecies};

\item[(2)] {\bf stable D-branes in string theory} (\cite{Witten01}, see \cite{Fredenhagen08}\cite{GS-tAHSS}\cite{GS-RR}) and here specifically of anyonic defect branes as discussed in \cite{SS22AnyonicDefectBranes}.
\end{itemize}

\noindent {\bf (ii)}
Indeed,
the {\it AdS/CMT correspondence} predicts
(e.g. \cite{ZaanenLiuSunSchalm15}\cite{HartnollLucasSachdev18}\cite{Zaanen21})
that solid state systems as in the left column
of \hyperlink{RosettaStone}{\it Table 1}
may effectively approximate the worldvoume field theory on $N \gg 1$ coincident {\it classical} (``black'') branes in string theory.

\smallskip
\noindent {\bf (iii)}
\hyperlink{RosettaStone}{\it Table 1} suggests an identification between topological phases and D-brane vacua which remains valid for individual D-branes ($N \sim 1$, and here notably for those that are stable but non-supersymmetric) whose {\it quantum} ground states are elements of K-theory groups, here specifically of finite torsion subgroups (such as those classifying $T I$-semimetals in Exp. \ref{ClassificationOfTISymmetricSemiMetals}) for which a large $N$-limit does not even make sense.

\smallskip
\noindent {\bf (iv)}
But for such small value of $N$ (i.e. in the large $1/N$-limit) that is relevant for many topological phases of matter (cf. all the finite torsion subgroups in \eqref{KOTheoryAsTWistedKRTheory}),
the string theory-side of the would-be AdS/CMT correspondence is known to be {\it strongly coupled}, and hence requires a formulation of the widely expected but elusive ``M-theory''. It is in the context of proposing a (partial) mathematical solution (``Hypothesis H'', see pointers in \cite{SS21MF} and in \cite[Rem. 4.1]{SS22AnyonicDefectBranes}) to the problem of formulation M-theory that the reflection of anyonic topological order withing TED-K-theory was discovered in \cite{SS22AnyonicDefectBranes}.

\smallskip
\noindent {\bf (v)}
In particular, it is the passage from the K-theory of the plain domain space (here: of the Brillouin torus \eqref{BrillouinTorus}) to that of its {\it configuration space of points} (see \cref{InteractingPhasesAndTEDKOfConfigurationSpaces}) which ``Hypothesis H'' predicts (\cite{SS22ConfigurationSpaces}) as the non-perturbative -- meaning: strongly-interacting -- completion of brane charge quantization laws.
\end{remark}

\medskip

\begin{center}
\begin{tikzpicture}
\draw (0,0) node {
\begin{tikzcd}[
  column sep=90pt,
  row sep=40pt
]
      \scalebox{.8}{
      \def\arraystretch{.9}
      \begin{tabular}{c}
        Quantum
        \\
        field theory
        \\
        on coincident
        \\
        $N$ branes
      \end{tabular}
      }
  \ar[
    <->,
    rr,
    dashed,
    shorten <=-5pt,
    bend left=12,
    "{
      \scalebox{.8}{
        \hspace{-.3cm}
        \begin{tabular}{c}
          \bf
          \color{greenii} traditional
          \\
          \bf
          \color{olive} AdS/CMT duality
          \\
          large $N$, large $\lambda$
        \end{tabular}
        \hspace{-.3cm}
      }
    }"{description}
  ]
  \ar[
    <->,
    dd,
    shorten <=5pt,
    dashed,
    bend right=20,
    "{
      \scalebox{.8}{
        \def\arraystretch{.9}
        \begin{tabular}{c}
          \bf
          \color{purple}full
          \\
          \bf
          \color{olive}
          AdS/CMT duality
          \\
          any $N$, any $\lambda$
        \end{tabular}
      }
    }"{swap, xshift=3pt}
  ]
  \ar[
    <->,
    ddrr,
    shorten <=-6pt,
    dashed,
    "{
      \hspace{1.52cm}
      \mathclap{
      \scalebox{.8}{
        \begin{tabular}{c}
          \bf
          \color{orangeii} stringy
          \\
          \bf
          \color{olive} AdS/CMT duality
          \\
          large $N$, any $\lambda$
        \end{tabular}
      }
      \hspace{1.6cm}
      }
    }"{description}
  ]
  &&
    \fbox{
      \hspace{-.4cm}
      \def\arraystretch{1}
      \begin{tabular}{c}
        Classical
        \\
        (super-)
        \color{darkblue} gravity
      \end{tabular}
      \hspace{-.3cm}
    }
  \ar[
    dd,
    "{
      \scalebox{.8}{
      \begin{tabular}{c}
        small 't Hooft coupling
        $\lambda := g_s N \,\ll\, \infty$
        \\
        \rotatebox{90}{$\Leftrightarrow$}
        \\
        large curvature, string scale effects
      \end{tabular}
      }
    }"
  ]
  \\
  \\
    \fbox{
      \hspace{-.4cm}
      \def\arraystretch{1}
      \begin{tabular}{c}
        Full
        \\
        \color{darkblue} M-theory
      \end{tabular}
      \hspace{-.3cm}
    }
  \ar[
    <-,
    rr,
    "{
      \scalebox{.8}{
        small $N \,\ll\, \infty$
        $\;\;\;\Leftrightarrow\;\;\;$
        non-perturbative effects
      }
    }"{swap}
  ]
  &&
    \fbox{
      \hspace{-.4cm}
      \def\arraystretch{1}
      \begin{tabular}{c}
        Perturbative
        \\
        \color{darkblue} string theory
      \end{tabular}
      \hspace{-.3cm}
    }
\end{tikzcd}
};
\draw (-5.12,1.82) circle (.9);
\end{tikzpicture}

\vspace{.2cm}

\begin{minipage}{14cm}
  \footnotesize
  {\bf Figure 4 -- Regimes ofs ST/CMT-duality.}
  Holographic duality (Rem. \ref{AnalogyWithDBraneCharge})
  has been widely discussed in the regime of weakly coupled stringy bulk theory, and here mostly in the weakly curved regime of (super-)gravity. This is because the full duality --  subsuming the case of small numbers $N$ of branes --  involves
  (see pointers in \cite{SS21MF})
  the expected but elusive non-perturbative completion of string theory to ``M-theory''.
  In \cite{SS22ConfigurationSpaces}\cite{SS22AnyonicDefectBranes}
  we explain how our ``Hypothesis H'' about M-theory (see pointers in \cite{SS21MF}) predicts that the full duality involves brane charges in the K-cohomology of configuration spaces of points.  In \cref{AnyonicTopologicalOrderAndInnerLocalSysyemTEDK} we find this to correctly match the dual phenomenon of anyonic defects.
\end{minipage}

\end{center}

\newpage

\begin{center}
\hypertarget{RosettaStone}{}
\footnotesize
\def\arraystretch{2.5}
\begin{tabular}{|c||c||c|l}
  \hhline{---}
  \rowcolor{olive}
  \bf
  \def\arraystretch{.9}
  \begin{tabular}{c}
    Topological phases
  \end{tabular}
  &
  \bf
  \def\arraystretch{.9}
  \begin{tabular}{c}
    Topological K theory
  \end{tabular}
  &
  \bf
  \def\arraystretch{.9}
  \begin{tabular}{c}
    String/M theory
  \end{tabular}
  \\
  \hhline{===}
  \rowcolor{lightgray}
  \def\arraystretch{.9}
  \begin{tabular}{c}
    Single-electron state
    \\
    in $d$-dim crystal
  \end{tabular}
  &
  \def\arraystretch{.9}
  \begin{tabular}{c}
    Line bundle over
    \\
    Brillouin $d$-torus
  \end{tabular}
  &
  \def\arraystretch{.9}
  \begin{tabular}{c}
    Single probe D-brane
    \\
    of codimension $d$
  \end{tabular}
  &
  \cellcolor{white}
  \multirow{6}{*}{
    \cref{TEDKClassifiesTopologicalPhasesOfMatter}
  }
  \\
  \def\arraystretch{.9}
  \begin{tabular}{c}
    Single positron state
  \end{tabular}
  &
  \def\arraystretch{.9}
  \begin{tabular}{c}
    Virtual line bundle
    \\
    over Brillouin torus
  \end{tabular}
  &
  \def\arraystretch{.9}
  \begin{tabular}{c}
    Single anti $\overline{\mathrm{D}}$-brane
    \\
    of codimension $d$
  \end{tabular}
  \\
  \rowcolor{lightgray}
  \def\arraystretch{.9}
  \begin{tabular}{c}
    Bloch-Floquet transform
  \end{tabular}
  &
  \def\arraystretch{.9}
  \begin{tabular}{c}
    Hilbert space bundle
    \\
    over Brillouin $d$-torus
  \end{tabular}
  &
  \def\arraystretch{.9}
  \begin{tabular}{c}
    Unstable (tachyonic)
    \\
    $\mathrm{D9}/\overline{\mathrm{D9}}$-brane
    state
  \end{tabular}
  \\
  \hhline{---}
  \def\arraystretch{.9}
  \begin{tabular}{c}
    Dressed Dirac
    \\
    vacuum operator
  \end{tabular}
  &
  \def\arraystretch{.9}
  \begin{tabular}{c}
    Family of
    \\
    Fredholm operators
  \end{tabular}
  &
  \def\arraystretch{.9}
  \begin{tabular}{c}
    Tachyon field
  \end{tabular}
  \\
  \rowcolor{lightgray}
  \def\arraystretch{.9}
  \begin{tabular}{c}
    Valence bundle of
    \\
    electron/positron
    states
  \end{tabular}
  &
  \def\arraystretch{.9}
  \begin{tabular}{c}
    Virtual bundle of
    their
    \\
    kernels and cokernels
  \end{tabular}
  &
  \def\arraystretch{.9}
  \begin{tabular}{c}
    stable D-brane state
    \\
    after tachyon condensation
  \end{tabular}
  \\
  \def\arraystretch{.9}
  \begin{tabular}{c}
    Topological phase
  \end{tabular}
  &
  \def\arraystretch{.9}
  \begin{tabular}{c}
    K-theory class
  \end{tabular}
  &
  \def\arraystretch{.9}
  \begin{tabular}{c}
    Stable D-brane charge
  \end{tabular}
  \\
  \hhline{---}
   \multicolumn{3}{c}{}
  \\[-15pt]
  \hhline{---}
  \rowcolor{olive}
  \bf
  Symmetry protection
  &
  \bf
  Twisted equivariance
  &
  \bf
  Global symmetries
  \\
  \hhline{===}
  \def\arraystretch{.9}
  \begin{tabular}{c}
    CPT symmetry
  \end{tabular}
  &
  \def\arraystretch{.9}
  \begin{tabular}{c}
   $\mathrm{KR}/\mathrm{KU}/\mathrm{KO}$-theory
  \end{tabular}
  &
  \def\arraystretch{.9}
  \begin{tabular}{c}
    Type I/IIA/IIB
  \end{tabular}
  &
  \cellcolor{white}
  \multirow{1}{*}{
  \cref{CrystallographicSymmetriesAndOrbifoldKTheory}
  }
  \\
  \rowcolor{lightgray}
  \def\arraystretch{.9}
  \begin{tabular}{c}
    Crystallographic symmetry
  \end{tabular}
  &
  \def\arraystretch{.9}
  \begin{tabular}{c}
    Orbifold K-theory
  \end{tabular}
  &
  \def\arraystretch{.9}
  \begin{tabular}{c}
    Spacetime orbifolding
  \end{tabular}
  &
  \cellcolor{white}
  \multirow{1}{*}{
  \cref{CrystallographicSymmetriesAndOrbifoldKTheory}
  }
  \\
  \def\arraystretch{.9}
  \begin{tabular}{c}
    Gauged internal symmetry
  \end{tabular}
  &
  \def\arraystretch{.9}
  \begin{tabular}{c}
    Inner local system-twist
  \end{tabular}
  &
  \def\arraystretch{.9}
  \begin{tabular}{c}
    Inside of orbi-singularity
  \end{tabular}
  &
  \cellcolor{white}
  \multirow{1}{*}{
  \cref{InternalSymmetriesAndInnerLocalSystemTwists}}
  \\
  \hhline{---}
  \multicolumn{3}{c}{}
  \\[-15pt]
  \hhline{---}
  \rowcolor{olive}
  \bf Topological order
  &
  \bf Twisted differentiality
  &
  \bf
  Gauge symmetries
  &
  \cellcolor{white}
  \multirow{1}{*}{
  \cref{TEDKDescribesRealisticAnyonSpecies}
  }
  \\
  \hhline{===}
  \def\arraystretch{.9}
  \begin{tabular}{c}
    Berry connection
  \end{tabular}
  &
  \def\arraystretch{.9}
  \begin{tabular}{c}
    Differential K-theory
  \end{tabular}
  &
  \def\arraystretch{.9}
  \begin{tabular}{c}
    Chan-Paton gauge field
  \end{tabular}
  &
  \cellcolor{white}
  \multirow{1}{*}{
  \cref{BerryPhasesAndDifferentialKTheory}
  }
  \\
  \rowcolor{lightgray}
  \def\arraystretch{.9}
  \begin{tabular}{c}
    Mass terms
  \end{tabular}
  &
  \def\arraystretch{.9}
  \begin{tabular}{c}
    Differential K-LES
  \end{tabular}
  &
  \def\arraystretch{.9}
  \begin{tabular}{c}
    Axio-Dilaton RR-field
  \end{tabular}
  &
  \cellcolor{white}
  \multirow{2}{*}{
  \cref{InteractingPhasesAndTEDKOfConfigurationSpaces}
  }
  \\
  \def\arraystretch{.9}
  \begin{tabular}{c}
    Nodal point charge
  \end{tabular}
  &
  \def\arraystretch{.9}
  \begin{tabular}{c}
    Flat K-theory
  \end{tabular}
  &
  \def\arraystretch{.9}
  \begin{tabular}{c}
    Defect brane charge
  \end{tabular}
  \\
  \hhline{---}
  \multicolumn{3}{c}{}
  \\[-15pt]
  \hhline{---}
  \rowcolor{olive}
  \bf Anyonic defects
  &
  \bf TED-K of Configurations
  &
  \bf
  Defect branes
  &
  \cellcolor{white}
  \multirow{1}{*}{
  \cref{AnyonicTopologicalOrderAndInnerLocalSysyemTEDK}
  }
  \\
  \def\arraystretch{.5}
  \begin{tabular}{c}
    $N$ band nodes
  \end{tabular}
  &
  \def\arraystretch{.9}
  \begin{tabular}{c}
    $N$-punctured
    \\
    Brillouin torus
  \end{tabular}
  &
  \def\arraystretch{.1}
  \begin{tabular}{c}
    $N$ defect branes
  \end{tabular}
  \\
  \rowcolor{lightgray}
  \def\arraystretch{.9}
  \begin{tabular}{c}
    Interacting $n$-electron states
    \\
    around $N$ band nodes
  \end{tabular}
  &
  \def\arraystretch{.9}
  \begin{tabular}{c}
    Vector bundle over
    \\
    $n$-point configuration space in
    \\
    $N$-punctured Brillouin torus
  \end{tabular}
  &
  \def\arraystretch{.9}
  \begin{tabular}{c}
    Interacting $n$ probe branes
    \\
    around $N$ defect branes
  \end{tabular}
  \\
  \def\arraystretch{.5}
  \begin{tabular}{c}
    $\suTwo$-anyon species
  \end{tabular}
  &
  \def\arraystretch{.5}
  \begin{tabular}{c}
    Holonomy of
    \\
    inner local system
  \end{tabular}
  &
  \def\arraystretch{.4}
  \begin{tabular}{c}
    $\SLTwoZ$-charges
    \\
    of defect branes
  \end{tabular}
  \\
  \hhline{---}
  \multicolumn{3}{c}{}
  \\[-15pt]
  \hhline{---}
  \rowcolor{olive}
  \bf Anyon braiding
  &
  \bf TED-K Gauss-Manin connections
  &
  \bf
  Defect brane monodromy
  &
  \cellcolor{white}
  \multirow{1}{*}{
   \cite{SS22TQC}
  }
\end{tabular}
\hspace{.3cm}
    \adjustbox{
      rotate=-90,
      raise=5cm
    }{
      \begin{tikzcd}[column sep=5.0cm]
        {}
        \ar[
          rr,
          <->,
          dashed,
          "{
            \scalebox{1.7}{
              \tiny
              \color{darkblue}
              \bf
              folklore
            }
          }"{pos=.23, yshift=+4pt},
          "{
            \scalebox{1.7}{
              \tiny
              \color{orangeii}
              \bf
              novel
            }
          }"{pos=.7, yshift=4pt},
        ]
        &&
        {}
      \end{tikzcd}
    }

\vspace{.3cm}

\begin{minipage}{17cm}
  {\bf Table 1 -- Rosetta stone CMT \, $\leftrightarrow$ \, TED-K \, $\leftrightarrow$ \, ST.}\\
  The mathematics of
  TE(D)-K-theory is widely conjectured (and partially known) to equivalently describe non-perturbative vacua in two rather different looking but supposedly ``dual'' situations in physics:

  \begin{itemize}[leftmargin=.3cm]

  \item[--]  On the left, the topological ground states of crystalline materials in condensed matter theory (observed in experiment on mesoscopic atomic scales);

  \item[--]
  On the right, the stable quantum ground states of D-branes in string theory
  (hypothetical physics at truly microscopic sub-nuclear scale).

  \end{itemize}

  The upper parts of this dictionary are more widely appreciated in existing literature, while less has previously been known or even just conjectured about the items further down in the table.

\end{minipage}

\end{center}

\vspace{.2cm}

\newpage

\subsection{Quantum symmetries and twisted equivariant K-theory}
\label{CrystallographicSymmetriesAndOrbifoldKTheory}
\label{}

\noindent
{\bf Quantum symmetries and twisted equivariant KR-theory.}
More precisely, in general the dynamics inside a given crystalline material (its {\it Hamiltonian}) obeys certain {\it quantum symmetries}, and only those Bloch modes suitably respecting these symmetries can actually be excited. We discuss now (following \cite{FreedMoore12} and \cite{SS21EPB})  how taking this into account means to understand the K-theory in Fact \ref{KTheoryClassificationOfTopologicalPhasesOfMatter} as {\it twisted equivariant KR-theory}.

\medskip

First, by the classical {\it Wigner theorem} (reviewed and amplified in \cite{Freed12}\cite[\S 1]{FreedMoore12}), adapted to electron/positron systems:
\vspace{-1mm}
$$
  \mbox{
    A {\it quantum symmetry} is
  }
  \mbox{\!
    a projective
  }
  \left<\!\!
    \begin{array}{l}
      \mbox{unitary or}
      \\
      \mbox{anti-unitary}
    \end{array}
 \!\!\! \right>
  \,  \mbox{operator} \,
  \left<\!\!
    \begin{array}{l}
      \mbox{preserving or}
      \\
      \mbox{exchanging}
    \end{array}
 \!\!\! \right>
 \, \mbox{the electron/positron Hilbert spaces.}
$$
Concretely,
the quantum symmetries form the
semidirect product of the projective unitary group on the $\ZTwo$-graded electron/positron Hilbert space $\mathscr{H} \oplus \mathscr{H}$
\eqref{FredholmOperator}
with
{\bf (a)} grading involution $P$,
{\bf (b)} complex conjugation $T$,
{\bf (c)} the combination $C := P T$:
\vspace{-1mm}
\begin{equation}
  \label{GroupOfQuantumSymmetries}
  \underset{
    \raisebox{-2pt}{
      \tiny
      \color{orangeii}
      \bf
      group of quantum symmetries
    }
  }{
  \overset{
    \mathclap{
    \raisebox{3pt}{
      \tiny
      \color{darkblue}
      \bf
      \def\arraystretch{.9}
      \begin{tabular}{c}
        Even projective
        unitary group
      \end{tabular}
    }
    }
  }{
  \frac{
    \UH
    \times
    \UH
    }
    { \CircleGroup }
  }
  \rtimes
  \big(\;
    \overset{
      \mathclap{
        \rotatebox{25}{
          \scalebox{.55}{$
            \mathrlap{
              \hspace{-16pt}
              \mbox{
                \color{darkblue}
                \bf
                \def\arraystretch{.9}
                \begin{tabular}{c}
                  grading
                  \\
                  involution
                \end{tabular}
              }
            }
          $}
        }
      }
    }{
    \underset{
      \{
        \NeutralElement,
        P
      \}
    }{
      \underbrace{
        \ZTwo
      }
    }
    }
    \times
    \overset{
      \mathclap{
        \rotatebox{25}{
          \scalebox{.55}{$
            \mathrlap{
              \hspace{-16pt}
              \mbox{
                \color{darkblue}
                \bf
                \def\arraystretch{.9}
                \begin{tabular}{c}
                  complex
                  \\
                  conjugation
                \end{tabular}
              }
            }
          $}
        }
      }
    }{
    \underset{
      \{
        \NeutralElement,
        T
      \}
    }{
    \underbrace{
      \ZTwo
    }
    }
    }
  \;\big)
  }
  \,,
  \;\;\;\;\;\;\;\;\;\;\;\;\;\;\;\;
  \def\arraystretch{1.4}
  \begin{array}{cl}
  P
  \cdot
  \big[
    \UnitaryOperator_{{}_{+}}
    ,\,
    \UnitaryOperator_{{}_{-}}
  \big]
  \;:=\;
  \big[
    \UnitaryOperator_{{}_-}
    ,\,
    \UnitaryOperator_{{}_+}
  \big]
  \cdot
  P
  \,,
  &
  \;\;\;\;\;\;\;\;\;\;\;\;\;
    C := P \cdot T \;.
    \\
 T
  \cdot
  \big[
    \UnitaryOperator_{{}_{+}}
    ,\,
    \UnitaryOperator_{{}_{-}}
  \big]
  \;:=\;
  \big[
    \overline{
    \UnitaryOperator
    }_{{}_+}
    ,\,
    \overline{\UnitaryOperator}_{{}_-}
  \big]
  \cdot
  T \;,
   &
  \end{array}
\end{equation}
These quantum symmetries naturally act  by conjugation on
the above Fredholm operators
\eqref{FredholmOperator}:
\begin{equation}
  \label{ActionOfQuantumSymmetriesOnFredholmOperators}
  \hspace{-3cm}
   \left(
    \frac{\UH \times \UH}{\CircleGroup}
    \rtimes
    \big(
      \{\NeutralElement, P\}
      \times
      \{\NeutralElement, T\}
    \big)
  \!\right)
  \times
  \FredholmOperators^0_{\ComplexNumbers}
  \xrightarrow{\quad (-)\cdot (-) \quad}
  \FredholmOperators^0_{\ComplexNumbers}
\end{equation}
\vspace{-3mm}
{\footnotesize
$$
  \def\arraystretch{1.6}
  \begin{array}{cccrcccc}
 &\phantom{A} & \phantom{A} &   {[\UnitaryOperator_+, \UnitaryOperator_-]}
    &:&
    \FredholmOperator
    &\qquad  \longmapsto&
\quad     \UnitaryOperator_+^{-1}
      \circ
    \FredholmOperator
      \circ
    \UnitaryOperator_-
    \\
 &\phantom{A}&&   C
    \cdot
    {[\UnitaryOperator_+, \UnitaryOperator_-]}
    &:&
    \FredholmOperator
    &\qquad \longmapsto&
 \quad    U_-^{-1}
      \circ
    F^t
      \circ
    U_+
    \\
&\phantom{A}&&    P
    \cdot
    {[\UnitaryOperator_+, \UnitaryOperator_-]}
    &:&
    \FredholmOperator
    &\qquad \longmapsto&
 \quad    U_-^{-1}
      \circ
    F^\ast
      \circ
    U_+
    \\
&\phantom{A}&&    T
    \cdot
    {[\UnitaryOperator_+, \UnitaryOperator_-]}
    &:&
    \FredholmOperator
    &\qquad \longmapsto&
  \quad   \UnitaryOperator_+^{-1}
      \circ
    \overline{\FredholmOperator}
      \circ
    \UnitaryOperator_- \;.
  \end{array}
$$
}
Here for $P \cdot F$ we used
\eqref{SpaceOfFredholmOperators}
that the full Fredholm operator on the graded space is of the form \scalebox{.8}{$\left(\!\!\!\begin{array}{cc} 0 & F^\dagger \\ F & 0\end{array}\!\!\!\right)$}.

 In conclusion, a {\it group $G$ of quantum symmetries} is a (finite) group $G$ equipped with a group homomorphism to \eqref{GroupOfQuantumSymmetries}:
\vspace{-2mm}
\begin{equation}
  \label{QuantumSymmetries}
  \begin{tikzcd}
  \mathllap{
    \mbox{
      \tiny
      \color{darkblue}
      \bf
      \def\arraystretch{.9}
      \begin{tabular}{c}
        group acting by
        \\
        quantum symmetries
      \end{tabular}
    }
  }
  G
  \;\;
 \ar[
   r,
   "{
     {\;\;\widehat{(-)}\;\;}
   }"
   ]
 &
  \frac{\UH \times \UH}{\CircleGroup}
  \rtimes
  \big(
    \{\NeutralElement, P\}
    \times
    \{\NeutralElement, T\}
  \big)
  \,.
  \end{tikzcd}
\end{equation}

\vspace{-2mm}
\noindent Through the defining action \eqref{GroupOfQuantumSymmetries}, this
induces an action of $G$ on the space of Fredholm operators.
If $G$ also acts on the Brillouin torus $\DualTorus{2}$,
by crystallographic point group symmetries \eqref{EquivariantBrillouinTorus},
then the map
from the orbifold quotient
(see \cite{SS20OrbifoldCohomology})
to the moduli stack of the
quantum symmetry group
(see \cite[Ex. 1.2.7]{SS21EPB})
encodes the crystal's quantum symmetries:

\vspace{-.2cm}
\begin{equation}
  \label{DeloopedQuantumSymmetries}
  \begin{tikzcd}
  \mathllap{
    \mbox{
      \tiny
      \bf
      \def\arraystretch{.9}
      \begin{tabular}{c}
        \color{greenii}
        twist of equivariant K-theory
        \\
        encoding
        \\
        \color{greenii}
         crystal quantum symmetries
      \end{tabular}
    }
  }
  \tau
  \;:\;\;
  \underset{
    \mathclap{
    \raisebox{-2pt}{
      \tiny
      \color{darkblue}
      \bf
      \def\arraystretch{.9}
      \begin{tabular}{c}
        crystallographic
        \\
        orbi-orientifold
      \end{tabular}
    }
    }
  }{
  \HomotopyQuotient
    { \DualTorus{2} }
    { G }
  }
    \ar[r]
    &
  \underset{
    \mathrlap{
    \raisebox{-3pt}{
      \tiny
      \color{darkblue}
      \bf
      \def\arraystretch{.9}
      \begin{tabular}{r}
        moduli stack of
        \\
        quantum symmetries
      \end{tabular}
    }
    }
  }{
    \mathbf{B}G
  }
  \ar[
    r,
    "{
      \mathbf{B}\widehat{(-)}
    }"
  ]
  &
  \mathbf{B}
  \bigg(
    \frac{\UH \times \UH}{\CircleGroup}
    \rtimes
    \big(
      \{\NeutralElement, P\}
      \times
      \{\NeutralElement, T\}
    \big)
   \!\!
   \bigg)
   \,.
 \end{tikzcd}
\end{equation}
\vspace{-.3cm}

Under the correspondingly refined Conjecture \ref{FamiliesOfRelativisticBlochVacua}, the statement of
Fact \ref{KTheoryClassificationOfTopologicalPhasesOfMatter} is that (Fact \ref{ClassificationOfExternalSPTPhases}): Topological phases {\it protected} (Rem. \ref{OnSymmetryProtection})
by  quantum $G$-symmetries
are labeled by homotopy classes of {\it $\tau$-equivariant}
maps from $\DualTorus{2}$ to Fredholm operators.
In generalization of \eqref{KTheoryClassificationByFredholmOperators},
these form the following
{\it twisted equivariant K-theory} group
(from \cite[(4.128)]{SS21EPB}):
\begin{equation}
  \label{TEKTheory}
  \hspace{-4mm}
  \overset{
    \raisebox{3pt}{
      \tiny
      \color{orangeii}
      \bf
      \begin{tabular}{c}
        $\tau$-twisted
        $G$-equivariant
        \\
        KR-cohomology
      \end{tabular}
    }
  }{
  \mathrm{K}^{\tau}
  \big(
    \HomotopyQuotient
      { X }
      { G }
  \big)
  }
    \!\!\!\!\!  :=
  \;
  \left\{
  \begin{tikzcd}
   [column sep={between origins, 76pt}]
    &[-40pt]&
    \HomotopyQuotient
      { \FredholmOperators^0_{\ComplexNumbers} }
      { G }
    \ar[d]
    \ar[rr]
    \ar[
      drr,
      phantom,
      "{
        \mbox{\tiny(pb)}
      }"
    ]
    &&
    \overset{
      \mathclap{
      \hspace{-38pt}
      \raisebox{3pt}{
        \tiny
        \color{darkblue}
        \bf
        \begin{tabular}{l}
          Quotient stack of
          odd self-adjoint Fredholm operators
          \eqref{SpaceOfFredholmOperators}
          \\
          by conjugation action
          of quantum symmetries
          \eqref{ActionOfQuantumSymmetriesOnFredholmOperators}
        \end{tabular}
      }
      }
    }{
    \HomotopyQuotient
      { \FredholmOperators^0_{\ComplexNumbers} }
      {
        \Big(
          \frac{\UH \times \UH}{\CircleGroup}
            \rtimes
          \big(
            \{\NeutralElement,
              P
            \}
            \times
            \{\NeutralElement, T\}
          \big)
        \!\Big)
      }
    }
    \ar[
      d,
      shorten >=-12pt,
      shorten <=-1pt
    ]
    \\
    \HomotopyQuotient
      { \DualTorus{2} }
      { G }
    \ar[
      urr,
      dashed,
      shorten=-3pt,
      "{
        \mbox{
          \tiny
          \color{orangeii}
          \bf
          \def\arraystretch{.9}
          \begin{tabular}{c}
            equivariant family of
            \\
            Fredholm operators
          \end{tabular}
        }
      }"{sloped}
    ]
    \ar[rr, ->>]
    &&
    \mathbf{B} G
    \ar[
      rr,
      "{\mathbf{B}\widehat{(-)}}"
    ]
    &&
    \mathbf{B}
    \Big(
      \frac{\UH \times \UH}{\CircleGroup}
      \rtimes
      \big(
        \{
          \NeutralElement,
          \overset{
              \mathclap{
                \rotatebox{36}{
                \rlap{ \bf
                \tiny
                \color{darkblue}
                \hspace{-12pt}
                \def\arraystretch{.9}
                \begin{tabular}{c}
                  grading
                  \\
                  involution
                \end{tabular}
                }
                }
              }
          }{
            P
          }
        \}
        \times
        \{
          \NeutralElement,
          \overset{
              \mathclap{
                \rotatebox{36}{
                \rlap{\bf
                \tiny
                \color{darkblue}
                \hspace{-12pt}
                \def\arraystretch{.9}
                \begin{tabular}{c}
                  complex
                  \\
                  conjugation
                \end{tabular}
                }
                }
              }
          }{
            T
          }
        \}
      \big)
    \!\Big)
    \\[-26pt]
    \mathclap{
      \mbox{
      \tiny
      \color{darkblue}
      \bf
      \def\arraystretch{.9}
      \begin{tabular}{c}
        crystallographic
        \\
        orbi-orientifold
      \end{tabular}
      }
    }
    \ar[
      rrrr,
      bend right=2.5,
      shift left=9pt,
      shorten <=13pt,
      shorten >=64pt,
      "{
        \underset{
          \scalebox{.7}{
            \eqref{DeloopedQuantumSymmetries}
          }
        }{
          \tau
        }
      }"{swap, pos=.2}
    ]
    &&
      \mbox{
        \tiny
        \color{darkblue}
        \bf
        \def\arraystretch{.9}
        \begin{tabular}{c}
          crystal symmetries
        \end{tabular}
      }
    \ar[
      rr,
      phantom,
      "{
        \mbox{
          \tiny
          \color{greenii}
          \bf
          realized as
        }
      }"{yshift=-1pt, pos=.13}
    ]
    \ar[
      dr,
      shorten <=-4pt,
      "{
        \mbox{
          \tiny
          \def\arraystretch{.9}
          \begin{tabular}{c}
            \color{greenii}
            \bf
            underlying
            \\
            \color{greenii}
            \bf
            CPT-
            symmetry
            \\
            \eqref{GroupOfQuantumSymmetries}
          \end{tabular}
        }
      }"{swap, sloped, pos=.22}
    ]
    &&
    \mathllap{
    \mbox{
      \tiny
      \color{darkblue}
      \bf
      quantum symmetries
      \hspace{18pt}
    }
    }
    \ar[
      dl,
      shorten <=-7pt,
      "{
        C \,=\ P T
      }"{swap, sloped, pos=.3}
    ]
    \\
    &&
    &
    \mathbf{B}
    \big(
    \{
      \NeutralElement, C
    \}
    \times
    \{
      \NeutralElement, T
    \}
    \big)
  \end{tikzcd}
 \!\!\!\!\! \right\}_{\!\!\big/\sim_{\homotopy}}
\end{equation}

\noindent
We refer the reader to \cite{SS20OrbifoldCohomology}\cite{SS21EPB} for full details on such diagrams of stacks (see \cite{TalkNotes} for a gentle exposition), but the following special cases should serve to illustrate the natural inner working of these diagrams.

The twisted equivariance in \eqref{TEKTheory} captures diverse phenomenon that in condensed matter theory
are known as ``symmetry protection'' and more generally as ``symmetry enhancement''  of topological phases of matter (Rem. \ref{OnSymmetryProtection}), as shown  in\hyperlink{TableOfTwistedEquivariances}{\it Table 2}.

\begin{center}
\hypertarget{TableOfTwistedEquivariances}{}
\footnotesize
\def\arraystretch{1.8}
\begin{tabular}{|c|c|c|c@{\hskip -5pt}c}
  \hhline{---~}
  \cellcolor{olive}
  \bf Twisted equivariance
  &
  \bf
  \cellcolor{olive}
  Sector of TED-K
  &
  \bf
  \cellcolor{olive}
  Type of symmetry protection
  &
  \\
  \hhline{===~}
  \def\arraystretch{.9}
  \begin{tabular}{c}
    Projective involutions
  \end{tabular}
  &
  $\mathrm{KR}$/$\mathrm{KU}$/$\mathrm{KO}$-theory
  &
  \begin{tabular}{c}
    Quantum CPT-symmetries
  \end{tabular}
  &
  \cellcolor{white}
  \multirow{1}{*}{
    Fact \ref{ListOfCPTTwistings}
  }
  &
  \multirow{2}{*}{
    \def\arraystretch{.9}
    \begin{tabular}{c}
      \cite{Kitaev09}
      \\
      \cite{FreedMoore12}
    \end{tabular}
  }
  \\
  \hhline{---~}
  \cellcolor{lightgray}
  Orbifolding
  &
  \cellcolor{lightgray}
  Orbifold K-theory
  &
  \cellcolor{lightgray}
  Crystallographic symmetries
  &
  \multirow{1}{*}{
    Fact \ref{ClassificationOfExternalSPTPhases}
  }
  &
  \\
  \hhline{===~}
  \def\arraystretch{.9}
  \begin{tabular}{c}
    Orbi-singularity
  \end{tabular}
  &
  \def\arraystretch{.9}
  \begin{tabular}{c}
    Fixed point theory
  \end{tabular}
  &
  Internal symmetries
  &
  \cref{InternalSymmetriesAndInnerLocalSystemTwists}
  &
  \\
  \hhline{---~}
  \cellcolor{lightgray}
  \def\arraystretch{.9}
  \begin{tabular}{c}
    $\mathclap{\phantom{\vert^{\vert^{\vert}}}}$
    Orbi-singularity
    \\
    with Inner local system
    $\mathclap{\phantom{\vert_{\vert_{\vert}}}}$
  \end{tabular}
  &
  \cellcolor{lightgray}
  \def\arraystretch{.9}
  \begin{tabular}{c}
    $\mathclap{\phantom{\vert^{\vert^{\vert}}}}$
    Twisted differential
    \\
    Fixed point theory
  \end{tabular}
  &
  \cellcolor{lightgray}
  \def\arraystretch{.9}
  \begin{tabular}{c}
    $\mathclap{\phantom{\vert^{\vert^{\vert}}}}$
    ``fictitious'' gauge symmetries
    \\
    (anyonic braiding phases)
  \end{tabular}
  &
  \cref{AnyonicTopologicalOrderAndInnerLocalSysyemTEDK}
  \\
  \hhline{---~}
\end{tabular}

\vspace{.2cm}

\begin{minipage}{16cm}
  \footnotesize
  {\bf Table 2 -- SPT/SET phases in TED K-theory}. That twisted equivariant K-theory captures symmetry protection/enhancement
  (Rem. \ref{OnSymmetryProtection})
  by CPT-symmetries and by crystallographic point symmetries is essentially the proposal of \cite{FreedMoore12} (but cf. Rem. \ref{ComparisonToExistingLiteratureForKClassification}). The crucial case of {\it internal} SET/SPT seems not to have been discussed in terms of equivariant K-theory before; we show in \cref{InternalSymmetriesAndInnerLocalSystemTwists} how this  connects to the widely expected description in terms of higher group cohomology. Moreover, we find in \cref{AnyonicTopologicalOrderAndInnerLocalSysyemTEDK} that the gauging of internal SPT/SET via inner local system-twists of TED-K accounts for anyonic topological order.
\end{minipage}

\end{center}

\begin{remark}[\bf On symmetry protectected/enhanced topological phases of matter]
  \label{OnSymmetryProtection}
  $\,$

  \noindent {\bf (i)}
  By a phase of matter which is ``symmetry protected''
  (``SPT'', \cite{PBTO09}\cite{GuWen09})
  or
  ``symmetry enhanced''
  (``SET'', \cite[p. 3]{CGLW13}\cite[p 2]{CLM12})
  one means a  topological phase of matter which is $G$-equivariantly non-trivial, in that it cannot be adiabatically deformed
 (Rem. \ref{QuantumAdiabaticTheorem}) {\it while respecting given $G$-quantum symmetry}, to a topologically trivial phase. In the case of SPT one in addition requires that the underlying topological phase (i.e. when forgetting the quantum symmetry) {\it is} trivial.

  \noindent {\bf (ii)}
  Mathematically, this
  evidently corresponds to the basic phenomenon of {\it equivariant homotopy classes} \eqref{TEKTheory}
  being finer than plain homotopy classes (see \cite[\S B]{SS20OrbifoldCohomology}).

 \noindent {\bf (iii)}
   Therefore we may translate
  ``{\it symmetry protected/enhanced} phase'' to ``{\it equivariant homotopy class}'' and hence
  via \eqref{TEKTheory} to ``{\it equivariant K-theory class}'' (cf. \hyperlink{RosettaStone}{\it Table 1}).
  For CPT- and crystallographic symmetries this is essentially the point that was made in \cite{FreedMoore12}, which we review now.
More generally, for internal symmetries (see \hyperlink{SymmetryTypes}{\it Table 4})
and for ``fictitious gauge symmetries''
we further develop this in \cref{InternalSymmetriesAndInnerLocalSystemTwists}
and in \cref{AnyonicTopologicalOrderAndInnerLocalSysyemTEDK} below.

\end{remark}

\medskip

\noindent
{\bf CPT-protected topological phases and KR/KU/KO-theory.}
Specifically, a realization of the group element $T$
(from \eqref{GroupOfQuantumSymmetries})
as a quantum symmetry $\widehat{T}$
(see \eqref{QuantumSymmetries})
is a map (of smooth stacks \cite{SS21EPB}, for exposition see \cite{TalkNotes})
of this form:
\vspace{-2mm}
\begin{equation}
  \label{TimeTerversalQuantumSymmetry}
  \begin{tikzcd}[row sep=-1pt, column sep={between origins, 70pt}]
    \mathbf{B}
    \big(
      \{\NeutralElement, T\}
    \big)
    \ar[
      rrr,
      dashed,
      "{
        T
          \,\longmapsto\,
        \widehat{T}
      }"
    ]
    \ar[
      dr,
      hook,
      shorten >=-6pt
    ]
    &&&
    \mathbf{B}
    \bigg(\!\!\!
    \def\arraystretch{1}
    \begin{array}{c}
      \underline{\underline{ \UH \times \UH }}
      \\
      \CircleGroup
    \end{array}
    \!\!\!
    \rtimes
    \{\NeutralElement,\, T\}
  \!  \bigg)
    \ar[
      dll,
      shorten <=-6pt
    ]
    \ar[r]
    &[80pt]
    \mathbf{B}
    \big(
    \mathbf{B}\CircleGroup
    \rtimes
    \{\NeutralElement,\, T\}
    \big).
    \\
    &
   \qquad  \quad \; \mathbf{B}
    \big(
      \{\NeutralElement, \, P\}
      \times
      \{\NeutralElement, \, T\}
    \big)
  \end{tikzcd}
\end{equation}

\vspace{-3mm}
\noindent In components, this means
(see \cite[Rem. 3.3.38, Ex. 3.3.27]{SS21EPB}) that $\widehat{T}$ is a simplicial map of this form:
\vspace{-1mm}
$$
  \def\arraystretch{2}
  \begin{array}{c}
  \begin{tikzcd}[column sep=large]
    \bullet
    \ar[
      rr,
      "{T}"{description},
      "\ "{swap, name=s1, pos=.1}
    ]
    &&
    \bullet
    \ar[
      dd,
      "{T}"{description}
    ]
    \\
    \\
    \bullet
    \ar[
      uu,
      "{
        T
      }"{description}
    ]
    \ar[
      uurr,
      "{\NeutralElement}"{description},
      "{\ }"{name=t1},
      "{\ }"{swap, name=s2},
    ]
    \ar[
      rr,
      "{T}"{description},
      "\ "{name=t2, pos=.9}
    ]
    &&
    \bullet
    \ar[
      from=s1,
      to=t1,
      Rightarrow,
      -
    ]
    \ar[
      from=s2,
      to=t2,
      Rightarrow,
      -
    ]
  \end{tikzcd}
  \\
  \rotatebox{90}{$\equiv$}
  \\
  \begin{tikzcd}[column sep=large]
    \bullet
    \ar[
      rr,
      "{T}"{description},
      "{\ }"{swap, name=s1, pos=.9}
    ]
    \ar[
      ddrr,
      "{\NeutralElement}"{description},
      "{\ }"{name=t1},
      "{\ }"{swap, name=s2}
    ]
    &&
    \bullet
    \ar[
      dd,
      "{T}"{description}
    ]
    \\
    \\
    \bullet
    \ar[
      uu,
      "{
        T
      }"{description}
    ]
    \ar[
      rr,
      "{T}"{description},
      "{\ }"{name=t2, pos=.1}
    ]
    &&
    \bullet
    \ar[
      from=s1,
      to=t1,
      Rightarrow,
      -
    ]
    \ar[
      from=s2,
      to=t2,
      Rightarrow,
      -
    ]
  \end{tikzcd}
  \end{array}
  \;\;\;\;\;\;\;\;\;\;
  \begin{tikzcd}[column sep=large]
    \ar[r, |->]
    &
    {}
  \end{tikzcd}
  \;\;\;\;\;\;\;\;\;\;
  \def\arraystretch{2}
  \begin{array}{c}
  \begin{tikzcd}[column sep=large]
    \bullet
    \ar[
      rr,
      "{\widehat{T}}"{description},
      "\ "{swap, name=s1, pos=.1}
    ]
    &&
    \bullet
    \ar[
      dd,
      "{\widehat{T}}"{description}
    ]
    \\
    \\
    \bullet
    \ar[
      uu,
      "{
        \widehat{T}
      }"{description}
    ]
    \ar[
      uurr,
      "{\NeutralElement}"{description},
      "{\ }"{name=t1},
      "{\ }"{swap, name=s2},
    ]
    \ar[
      rr,
      "{\widehat{T}}"{description},
      "\ "{name=t2, pos=.9}
    ]
    &&
    \bullet
    \ar[
      from=s1,
      to=t1,
      Rightarrow,
      "{\color{purple}c}"{description}
    ]
    \ar[
      from=s2,
      to=t2,
      Rightarrow,
      -
    ]
  \end{tikzcd}
  \\
  \rotatebox{90}{$\equiv$}
  \\
  \begin{tikzcd}[column sep=large]
    \bullet
    \ar[
      rr,
      "{\widehat{T}}"{description},
      "{\ }"{swap, name=s1, pos=.9}
    ]
    \ar[
      ddrr,
      "{\NeutralElement}"{description},
      "{\ }"{name=t1},
      "{\ }"{swap, name=s2}
    ]
    &&
    \bullet
    \ar[
      dd,
      "{\widehat{T}}"{description}
    ]
    \\
    \\
    \bullet
    \ar[
      uu,
      "{
        \widehat{T}
      }"{description}
    ]
    \ar[
      rr,
      "{\widehat{T}}"{description},
      "{\ }"{name=t2, pos=.1}
    ]
    &&
    \bullet
    \ar[
      from=s1,
      to=t1,
      Rightarrow,
      "{\color{purple}c}"{description}
    ]
    \ar[
      from=s2,
      to=t2,
      Rightarrow,
      -
    ]
  \end{tikzcd}
  \end{array}
  \def\arraystretch{2}
  \begin{array}{c}
  \equiv
  \;\;\;
  \begin{tikzcd}[column sep=large]
    \bullet
    \ar[
      rr,
      "{\widehat{T}}"{description},
      "\ "{swap, name=s1, pos=.1}
    ]
    &&
    \bullet
    \ar[
      dd,
      "{\widehat{T}}"{description}
    ]
    \\
    \\
    \bullet
    \ar[
      uu,
      "{
        \widehat{T}
      }"{description}
    ]
    \ar[
      rr,
      rounded corners,
      to path={
        ([xshift=+4pt]\tikztostart.north)
        --
        ([xshift=+10pt, yshift=+30pt]\tikztostart.north)
        --
        node[name=s2]
          {
            \scalebox{.6}{
            \colorbox{white}{
              $\widehat{T}\widehat{T}\widehat{T}$
            }
            }
          }
        ([xshift=-10pt, yshift=+30pt]\tikztotarget.north)
        --
        ([xshift=-4pt, yshift=+2pt]\tikztotarget.north)
      }
    ]
    \ar[
      rr,
      "{\widehat{T}}"{description},
      "\ "{name=t2, pos=.5}
    ]
    &&
    \bullet
    \ar[
      from=s2,
      to=t2,
      shorten >=2pt,
      Rightarrow,
      "{ \color{purple}c }"{description}
    ]
  \end{tikzcd}
  \\
  \phantom{\equiv}\;\;\;
  \rotatebox{90}{$\equiv$}
  \\
  \equiv
  \;\;\;
  \begin{tikzcd}[column sep=large]
    \bullet
    \ar[
      rr,
      "{\widehat{T}}"{description},
      "\ "{swap, name=s1, pos=.1}
    ]
    &&
    \bullet
    \ar[
      dd,
      "{\widehat{T}}"{description}
    ]
    \\
    \\
    \bullet
    \ar[
      uu,
      "{
        \widehat{T}
      }"{description}
    ]
    \ar[
      rr,
      rounded corners,
      to path={
        ([xshift=+4pt]\tikztostart.north)
        --
        ([xshift=+10pt, yshift=+30pt]\tikztostart.north)
        --
        node[name=s2]
          {
            \scalebox{.6}{
            \colorbox{white}{
              $\widehat{T}\widehat{T}\widehat{T}$
            }
            }
          }
        ([xshift=-10pt, yshift=+30pt]\tikztotarget.north)
        --
        ([xshift=-4pt, yshift=+2pt]\tikztotarget.north)
      }
    ]
    \ar[
      rr,
      "{\widehat{T}}"{description},
      "\ "{name=t2, pos=.5}
    ]
    &&
    \bullet
    \ar[
      from=s2,
      to=t2,
      shorten >=2pt,
      Rightarrow,
      "{ \color{purple}\overline{c} }"{description}
    ]
  \end{tikzcd}
  \end{array}
$$
Here the identification on the bottom right uses that
$T$ acts by complex conjugation on unitary operators
\eqref{GroupOfQuantumSymmetries}.
With this,
the identification on the far right shows that
the phase must be real, $c = \overline{c}$, which means that either
\vspace{-1mm}
\begin{equation}
  \label{SignChoicesForTimeReversalQuantumSymmetry}
  \widehat{T} \cdot \widehat{T}
  \;=\;
  \pm
  1
  \,.
\end{equation}

\vspace{-2mm}
\noindent The same holds for coboundaries of these cocycles, so that these two choices exhaust
the gauge-equivalence classes of choices in group cohomology $H^2_{\mathrm{Grp}}\big( \{\NeutralElement, T\}
;\, \ZTwo\big) \,\simeq\, \ZTwo$.
Moreover, the directly analogous argument shows that the available choices for $\widehat{C} = \widehat{PT}$
are classified by
\begin{equation}
  \label{SignChoicesForChargeConjugationSymmetry}
  \widehat{C}
  \cdot
  \widehat{C}
  \;=\;
  \pm
  1
  \,.
\end{equation}
On the other hand, for $\widehat P$ alone the phase could be any $\CircleGroup$-valued cocycle on
$\ZTwo \,=\, \{\NeutralElement, P\}$.
But since
$H^2_{\mathrm{Grp}}\big( \{\NeutralElement, P\}
;\, \CircleGroup \big) \,\simeq\, 0$
these are all trivializable, so that there is only one gauge equivalence class of choices $\widehat{P}$ for realizing $P$ as a quantum symmetry:
\begin{equation}
  \label{SignChoicesForParitySymmetry}
  \widehat{P}
  \cdot
  \widehat{P}
  \;=\;
  +
  1
  \,.
\end{equation}

\begin{remark}[\bf CPT-symmetry]
As the notation indicates, the above analysis reproduces the structure of the Charge/Parity/ Time-symmetry (CPT) theory of the
relativistic Dirac equation (e.g. \cite[\S 3.4]{Thaller92}), identifying (e.g. \cite[Ex. 3.11]{Thaller92}):

\vspace{1mm}
 -- $P$ with the unitary operation of ``parity reversal'' (\cite[(2.147)]{Thaller92}),

\vspace{1mm}
 -- $T$ with the anti-unitary operation of time-reversal (e.g. \cite[Thm. 3.10, cf (1.26)]{Thaller92}),

 \vspace{1mm}
 -- $C := PT$ with the operation of ``charge conjugation'' (e.g. \cite[Thm. 3.10, cf (1.81)]{Thaller92}).

 \vspace{1mm}
 \noindent
 Notice how this definition of $C$ expresses the fundamental fact
 (e.g. \cite{SW01}\cite{Lehnert16})
 of both theoretical and experimental physics
 that the combination $C \cdot P \cdot T$ is an exact symmetry of observed fundamental processes (since with this definition
 it is given by multiplication with $C \cdot P \cdot T = (P \cdot T) \cdot P \cdot T = P^2 \cdot T^2 = \NeutralElement \cdot \NeutralElement = \NeutralElement$).

 It is somewhat remarkable that this situation of CPT-symmetry in fundamental physics finds such an accurate and natural reflection
 in the structure of twisted equivariant K-theory.
\end{remark}

\begin{example}[\bf Time-reversal symmetry and KR-theory]
\label{TimeReversalSymmetryAndKRTheory}
If $X = \DualTorus{d} \simeq \RealNumbers^d / \Integers^d$ is thought of as a Brillouin torus of quasi-momenta $k$ in a crystal
(as in Fact \ref{BlochFloquetTheory}), then in condensed matter theory it is understood by default that the element
$T \in \{\NeutralElement, T\}$ in \eqref{TimeTerversalQuantumSymmetry}
acts on this torus  by ``inversion'' (point reflection) of quasi-momenta (this in addition to its action
\eqref{ActionOfQuantumSymmetriesOnFredholmOperators}
on Fredholm operators by complex conjugation):
\vspace{-2mm}
\begin{equation}
  \label{TimeReversalActionOnBrillouinTorus}
  \begin{tikzcd}[column sep=10pt, row sep=-2pt]
    \DualTorus{d}
    \ar[rr, "{T}"]
    &&
    \DualTorus{d}
    \\
  \scalebox{0.8}{$    {[k]}$} &\longmapsto&  \scalebox{0.8}{$  {[- k]} $}
    \mathrlap{\,,}
  \end{tikzcd}
  \hspace{1cm}
\begin{tikzcd}[column sep=10pt, row sep=-2pt]
    \FredholmOperators
    \ar[rr, "{T}"]
    &&
    \FredholmOperators
    \\
  \scalebox{0.8}{$  F$} &\longmapsto&  \scalebox{0.8}{$  \overline{F}$}\;.
  \end{tikzcd}
\end{equation}

\vspace{-2mm}
\noindent We always mean this inversion action when we write the orbifold $\HomotopyQuotient{\DualTorus{d}}{\{\NeutralElement, T\}}$ (to be contrasted
with the quotient by the trivial action, which we denote by $\DualTorus{d} \times \HomotopyQuotient{\ast}{\{\NeutralElement, T I\}}$, see Ex. \ref{JointTimeReversalAndInversionSymmetry}).
With this inversion action
\eqref{TimeReversalActionOnBrillouinTorus}
of $T$ on $\DualTorus{d}$ understood, the
$\{\NeutralElement, \widehat{T}\}$-equivariant $k\in \DualTorus{d}$-parameterized families of Fredholm operators \eqref{TEKTheory}
are those that satisfy
\vspace{-1mm}
\begin{equation}
  \label{BosonicTimeReversalInvarianceOnFredholmOperators}
  F_{[k]} = \overline{F}_{[-k]}
  \,.
\end{equation}

\vspace{-1mm}
\noindent If, furthermore, the quantum symmetry action of $T$ on
$\FredholmOperators$ is given by $\widehat{T}^2 =  +1$, then the
homotopy classes of these ``time-reversal equivariant'' Fredholm families
constitute the cohomology group of $\DualTorus{d}$ known as
``Atiyah's Real K-theory'' (with a capital ``R'', to be distinguished
from the ``real K-theory'' groups denoted $\mathrm{KO}$), or {\it KR-theory},
for short:
\begin{equation}
  \mathrm{KR}
  \big(
    \DualTorus{d}
  \big)
  \;\;=\;\;
  \left\{\!\!\!
  \begin{tikzcd}[row sep=small]
    &&
    \HomotopyQuotient
      {
        \FredholmOperators^0_{\ComplexNumbers}
      }
      { \{\NeutralElement, T\} }
    \ar[d]
    \\
    \HomotopyQuotient
      { \DualTorus{d} }
      { \{\NeutralElement, T\} }
      \ar[rr, ->>]
      \ar[
        urr,
        dashed
      ]
      &&
      \mathbf{B}
      \{\NeutralElement, T\}
  \end{tikzcd}
   \!\!\! \right\}_{\!\!\big/\sim_{\homotopy}}
  \,.
\end{equation}
In the string theory literature, under the dictionary of \hyperlink{RosettaStone}{\it Table 1}, these are known as {\it orientifolds}
\cite{Gukov}\cite{SS19TadpoleCancellation}.

Under the identification of Fact \ref{KTheoryClassificationOfTopologicalPhasesOfMatter} this classifies topological phases of {\it time-reversal invariant} insulators
(e.g. \cite[\S 2.1.6]{Vanderbilt18})
whose valence quasi-particles behave like {\it bosons}.
\end{example}

Noticd that the inversion action \eqref{TimeReversalActionOnBrillouinTorus} has $2^d$ fixed points, forming the set
\vspace{-3mm}
$$
  \begin{tikzcd}
  \{ 0, \sfrac{1}{2} \}^d
    \subset
  \RealNumbers^d
  \ar[r, ->>]
  &
  \RealNumbers^d / \Integers^d
  \simeq
  \DualTorus{d}
  \,.
  \end{tikzcd}
$$

\vspace{-2mm}
\noindent These are also called the {\it time-reversal invariant momenta}
(TRIM, e.g. \cite[p. 51]{Vanderbilt18}).
Over these fixed points, the equivariance condition
\eqref{BosonicTimeReversalInvarianceOnFredholmOperators} becomes an {\it invariance}-condition (see also Ex. \ref{JointTimeReversalAndInversionSymmetry} below). We now consider the spaces of invariant Fredholm operators under all possible quantum CPT-symmetries.

\begin{fact}[\bf The 10-fold way of twistings of CPT-equivariant K-theory]
\label{ListOfCPTTwistings}
There are evidently 10 possibilities for realizing the quantum CPT-symmetries \eqref{SignChoicesForTimeReversalQuantumSymmetry}
-
\eqref{SignChoicesForParitySymmetry},
shown in the top part of \hyperlink{TableCliffordActions}{\it Table 3}.
\end{fact}
\vspace{-1mm}
The remainder of \hyperlink{TableCliffordActions}{\it Table 3} indicates the corresponding twisted equivariant K-theory groups \eqref{TEKTheory}, following via Karoubi's theorem (Prop. \ref{KaroubiTheorem} with Lem. \ref{AdjointnessOfQuantumCPT} below).
In a different but closely related form
(\cite[Table 2]{Kitaev09}, reviewed e.g. in \cite[Table IV]{CTSR16}),
this has become famous as the {\it 10-fold way}; our \hyperlink{TableCliffordActions}{\it Table 3} expresses instead the
Fredholm-operator version of the statement in \cite[Prop. 6.4, B.4]{FreedMoore12}. In actuality, this ``10-fold way''  (see Fact \ref{ListOfCPTTwistings})
is another incarnation of {\it Bott periodicity} (e.g. \cite[\S 15]{HJJS08}) being a 2-fold periodicity over the
complex numbers and an 8-fold periodicity over the real numbers, for a total of $2+8 = 10$ distinct possibilities:

\begin{center}
\hypertarget{TableCliffordActions}{}
\small
\def\arraystretch{1.8}
\begin{tabular}{|cr|||c|c||c|c|c|c|c|c|c|c|}
  \hline
  \small \bf
  \def\arraystretch{.9}
  \begin{tabular}{c}
    Equivariance
    group
  \end{tabular}
  &
  $G \; =$
  &
  $\;\{ \NeutralElement\}\;$
  &
  $\{ \NeutralElement, P\}$
  &
  \multicolumn{2}{c|}{
    $\{ \NeutralElement, T\}$
  }
  &
  \multicolumn{2}{c|}{
    $\{ \NeutralElement, C\}$
  }
  &
  \multicolumn{4}{c|}{
  $
    \{ \NeutralElement, T\}
    \times
    \{ \NeutralElement, C\}
  $
  }
  \\
  \hline
  \hline
  \multirow{2}{*}{
    \colorbox{white}{
    \small
    \def\arraystretch{.9}
    \begin{tabular}{c}
      Realization as
      \\
      quantum symmetry
    \end{tabular}
       $\mathrlap{\tau:}$
  }
  }
  &
  $\widehat{T}^2 =$
  &
  &
  &
  $+1$
  &
  $-1$
  &
  &
  &
  $+1$
  &
  $-1$
  &
  $-1$
  &
  $+1$
  \\
  \cline{2-12}
  &
  $\widehat{C}^2 =$
  &
  &
  &
  &
  &
  $+1$
  &
  $-1$
  &
  $+1$
  &
  $+1$
  &
  $-1$
  &
  $-1$
  \\
  \hline
  \hline
  \multirow{8}{*}{
    \small
    \def\arraystretch{.9}
    \begin{tabular}{c}
      Maximal induced
      \\
      Clifford action
      \\
      anticommuting with
      \\
      all $G$-invariant odd
      \\
      Fredholm operators
    \end{tabular}
  }
  &
  $E_{-3}=$
  &
  &
  &&&&&&
  $\ImaginaryUnit\widehat{T}\widehat{C}\beta$
  &&
  \\
  \cline{2-12}
  &
  $E_{-2}=$
  &
  &
  &&&
  $\ImaginaryUnit \widehat{C}\beta$
  &
  &&
  $\ImaginaryUnit\widehat{C}\beta$
  &&
  \\
  \cline{2-12}
  &
  $E_{-1}=$
  &
  &
  $\widehat{P}\beta$
  &&
  &
  $\widehat{C}\beta$
  &
  &
  $\widehat{C}\beta$
  &
  $\widehat{C}\beta$
  &&
  \\
  \cline{2-12}
  &
  $E_{+0}=$
  &
  $\beta$
  &
  $\beta$
  &
  $\beta$
  &
  \Big(
     \scalebox{.8}{$
     \hspace{-3.4pt}
     \arraycolsep=2.2pt
     \def\arraystretch{.91}
     \begin{array}{cc}
       \beta & 0
       \\
       0 & -\beta
     \end{array}
     \hspace{-3.4pt}
     $}
   \Big)
  &
  $\beta$
  &
  $\beta$
  &
  $\beta$
  &
  $\beta$
  &
  $\beta$
  &
  $\beta$
  \\
  \cline{2-12}
  &
  $E_{+1}=$
  &
  &
  &
  &
  \Big(
     \scalebox{.8}{$
     \hspace{-3.4pt}
     \arraycolsep=2.2pt
     \def\arraystretch{.91}
     \begin{array}{cc}
       0 & 1
       \\
       1 & 0
     \end{array}
     \hspace{-3.4pt}
     $}
   \Big)
  &
  &
  $\widehat{C}\beta$
  &&
  &
  $\widehat{C}\beta$
  &
  $\widehat{C}\beta$
  \\
  \cline{2-12}
  &
  $E_{+2}=$
  &
  &&
  &
  \Big(
     \scalebox{.8}{$
     \hspace{-3.4pt}
     \arraycolsep=2.2pt
     \def\arraystretch{.91}
     \begin{array}{cc}
       0 & \ImaginaryUnit
       \\
       \ImaginaryUnit & 0
     \end{array}
     \hspace{-3.4pt}
     $}
   \Big)
  &
  &
  $\ImaginaryUnit\widehat{C}\beta$
  &&
  &
  $\ImaginaryUnit\widehat{C}\beta$
  &
  \\
  \cline{2-12}
  &
  $E_{+3}=$
  &
  &&&
  \Big(
     \scalebox{.8}{$
     \hspace{-3.4pt}
     \arraycolsep=2.2pt
     \def\arraystretch{.91}
     \begin{array}{cc}
       0 & -\widehat{T}
       \\
       \widehat{T} & 0
     \end{array}
     \hspace{-3.4pt}
     $}
   \Big)
  &&&
  &
  &
  $\ImaginaryUnit\widehat{T}\widehat{C}\beta$
  &
  \\
  \cline{2-12}
  &
  $E_{+4}=$
  &
  &&&
  \Big(
     \scalebox{.8}{$
     \hspace{-3.4pt}
     \arraycolsep=2.2pt
     \def\arraystretch{.91}
     \begin{array}{cc}
       0 & \ImaginaryUnit\widehat{T}
       \\
       \ImaginaryUnit\widehat{T} & 0
     \end{array}
     \hspace{-3.4pt}
     $}
   \Big)
  &&&&&&
  \\
  \hline
  \hline
  \small
  \def\arraystretch{.9}
  \begin{tabular}{c}
    $\tau$-twisted $G$-equivariant
    \\
    K-theory of fixed loci
  \end{tabular}
  &
  $
  \mathrm{K}^{\tau}
  =
  $
  &
  $\mathrm{KU}^0$
  &
  $\mathrm{KU}^{1}$
  &
  $\mathrm{KO}^{0}$
  &
  $\mathrm{KO}^{4}$
  &
  $\mathrm{KO}^{2}$
  &
  $\mathrm{KO}^{6}$
  &
  $\mathrm{KO}^{1}$
  &
  $\mathrm{KO}^{3}$
  &
  $\mathrm{KO}^{5}$
  &
  $\mathrm{KO}^{7}$
 \\
  \hline
\end{tabular}

\vspace{.2cm}

\begin{minipage}{17cm}
\footnotesize

{\bf Table 3 -- CPT Quantum symmetries as twists of equivariant KR-theory.} The table indicates how the K-theoretic ``10-fold way''-classification proposed in \cite{Kitaev09} and elaborated on in \cite{FreedMoore12} comes about, via Thm. \ref{KaroubiTheorem}, in terms of fixed loci of spaces of Fredholm operators
\eqref{FredholmOperatorsAntiCommutingWithACliffordAction} whose homotopy quotients by quantum symmetries
\eqref{ActionOfQuantumSymmetriesOnFredholmOperators}
classify twisted equivariant K-theory \eqref{TEKTheory}.

\end{minipage}

\end{center}

\vspace{-1mm}
\begin{itemize}[leftmargin=.4cm]
\item The top part shows the 10 different ways to choose quantum symmetry lifts $\widehat{(-)}$  \eqref{QuantumSymmetries}
of subgroups of the CPT group $\{\NeutralElement, P\} \times \{\NeutralElement, T\}$ \eqref{TimeTerversalQuantumSymmetry}
according to the analysis around
\eqref{SignChoicesForTimeReversalQuantumSymmetry},
\eqref{SignChoicesForChargeConjugationSymmetry},
\eqref{SignChoicesForParitySymmetry}.

\item The middle part shows how, under the action \eqref{ActionOfQuantumSymmetriesOnFredholmOperators}
of quantum symmetries on odd Fredholm operators $\widehat{F}$, $\widehat{F} \circ \beta = - \beta \circ \widehat{F}$, these quantum symmetries equivalently constitute a collection of Clifford generators anti-commuting with all the Fredholm operators that are fixed by these quantum symmetries.
For example, the second column shows that when parity symmetry $\{\NeutralElement, P\}$ acts, with its essentially unique lift
\eqref{SignChoicesForParitySymmetry}
to a quantum symmetry $\widehat{P}$, then invariance of $\widehat{F}$ under the conjugation action, $\widehat{F} \circ \widehat{P} = \widehat{P} \circ \widehat{F}$,
equivalently means that
that $\widehat{F}$,
$\beta$ and $\widehat{P}\beta$ anti-commute:
\vspace{-2mm}
$$
  \widehat{F}
  \,\in\,
  \big(
    \FredholmOperators^0_{\ComplexNumbers}
  \big)^{\widehat{P}}
  \hspace{1.0cm}
  \Leftrightarrow
  \hspace{1.0cm}
  \widehat{F} \,\in\,
  \FredholmOperators^0_{\ComplexNumbers}
  \;\;\;\;
  \mbox{and}
  \;\;\;\;
  \widehat{F}
    \circ
  (\widehat{P}\beta)
  \;=\;
  -
  (\widehat{P}\beta)
  \circ
  \widehat{F}
  \,.
$$

\vspace{-2mm}
\item The bottom part shows how these compatible Clifford actions exhibit the spaces of Fredholm operators fixed
by the quantum symmetries \eqref{FredholmOperatorsAntiCommutingWithACliffordAction}
as classifying spaces for $\mathrm{KU}$- and $\mathrm{KO}$-theories (by Karoubi's theorem, Prop. \ref{KaroubiTheorem} and using Lemma \ref{AdjointnessOfQuantumCPT}).

\end{itemize}



\vspace{.1cm}

The analysis indicated in \hyperlink{TableCliffordActions}{\it Table 3} shows that subspaces of odd Fredholm operators \eqref{SpaceOfFredholmOperators} which are fixed by one of the 10 possible quantum CPT symmetries
are of the following form,
for $p \,\in\, \NaturalNumbers$
and
$\mathbb{K} \,=\, \mathbb{C}$ or $\mathbb{K} = \mathbb{R}$ (the latter if $T$ and/or $C = PT$ are contained in the given subgroup):

\vspace{-7mm}
\begin{equation}
  \label{FredholmOperatorsAntiCommutingWithACliffordAction}
  \adjustbox{raise=10pt}{$
  \def\arraystretch{2}
  \begin{array}{l}
      \hspace{1.2cm}
   \scalebox{.8}{
     \color{darkblue}
     \bf
     Space of self-adjoint Fredholm operators
     graded-commuting with $p+1$
     Clifford generators
  }
  \\
  \hspace{-5mm} \FredholmOperators^{-p}_{\mathbb{K}}
  \! :=
 \!\! \left\{\!\!\!\!\!\!
  \def\arraystretch{1.4}
  \begin{array}{ll}
    \scalebox{.8}{Bounded opers.}
    &
    \widehat{F}
      :
    \!\!
    \begin{tikzcd}
      \mathscr{H}^2
      \ar[
        r,
        shorten=-1pt,
        "{\scalebox{.65}{bounded}}"{pos=.4},
        "{\scalebox{.65}{$\mathbb{K}-$linear}}"
          {swap, yshift=1pt, pos=.4}
      ]
      &
      \mathscr{H}^2
    \end{tikzcd}
    \\
    \scalebox{.8}{self-adjoint}
    &
    \widehat{F}^\ast
    \,=\,
    \widehat{F}
    \;:=\;
    F + F^\ast
    \\
    \scalebox{.8}{Fredholm}
    &
    \mathrm{dim}
    \big(
      \mathrm{ker}\big(
        \widehat{F}\,
      \big)
    \big)
    \;<\;
    \infty
  \end{array}
  \!\!\!
  \middle\vert
  \!\!\!
  \def\arraystretch{.9}
  \begin{array}{l}
  \scalebox{.8}{graded comm.}
  \\
  E_i \circ \widehat{F}
    = -
  \widehat{F} \circ E_i
  \end{array}
  \scalebox{.8}{with}
  \def\arraystretch{1.4}
  \begin{array}{ll}
    \scalebox{.8}{Bounded oper.}
    &
    E_0, \cdots, E_p
      :
    \!\!
    \begin{tikzcd}
      \mathscr{H}^2
      \ar[
        r,
        shorten=-1pt,
        "{\scalebox{.65}{bounded}}"{pos=.4},
        "{\scalebox{.65}{$\mathbb{K}-$linear}}"
          {swap, yshift=1pt, pos=.4}
      ]
      &
      \mathscr{H}^2
    \end{tikzcd}
    \\
    \scalebox{.8}{(anti-)self-adjoint}
    &
    (
      E_{i}
    )^\ast
    =
    \mathrm{sgn}_i \cdot  E_i
    \\
    \scalebox{.8}{Clifford gen.}
    &
    E_i \circ E_j + E_j \circ E_i
    =
    2 \mathrm{sgn}_i \cdot \delta_{i j}
  \end{array}
 \!\!\!\!\!\!\!\! \right\}
  \end{array}
  $}
\end{equation}
where:

\begin{itemize}[leftmargin=.4cm]
\item we have abbreviated
$\footnotesize
  \mathrm{sgn}_i
  :=
  \left\{\!\!\!\!
  \def\arraystretch{1.2}
  \begin{array}{c}
    +1 \;\; \vert \; i \geq 0
    \,,
    \\
    -1 \;\; \vert \; i < 0
    \,;
  \end{array}
  \right.
$

\item
 it is understood that the $\pm 1$-eigenspaces of the Clifford generators $E_i$ are both infinite-dimensional.
\end{itemize}
With this, the reader may take the following Prop. \ref{KaroubiTheorem},
generalizing \eqref{KTheoryClassificationByFredholmOperators},
to be the {\it definition} of the $\mathrm{KU}$- and $\mathrm{KO}$-theory groups:
\begin{proposition}[ASK-theorem \cite{Karoubi70}\cite{AtiyahSinger69}]
\label{KaroubiTheorem}
The Fredholm operator spaces \eqref{FredholmOperatorsAntiCommutingWithACliffordAction} classify the K-theory groups that go by the following names:
\vspace{-2mm}
\begin{equation}
  \Big\{
    X
    \xrightarrow[\continuous]{\phantom{---}}
    \FredholmOperators^p_{\mathbb{K}}
  \Big\}_{\big/\sim_{\homotopy}}
  \;\;
  =
  \;
  \left\{
  \def\arraystretch{2}
  \begin{array}{lll}
    \mathrm{KU}^{p}(X)
    \underset{\scalebox{.7}{\color{greenii} \rm Bott per.}}{\simeq}
    \mathrm{KU}^{p+2}(X)
    &\vert&
    \mathbb{K} = \ComplexNumbers \;,
    \\
    \mathrm{KO}^{p}(X)
    \underset{\scalebox{.7}{\color{greenii} \rm Bott per.}}{\simeq}
    \mathrm{KO}^{p+8}(X)
    &\vert&
    \mathbb{K} = \RealNumbers \;.
  \end{array}
  \right.
\end{equation}
\end{proposition}

\noindent This is the result indicated in the last row of \hyperlink{TableCliffordActions}{\it Table 3}.

Before we turn to discussing examples, notice that in order to be able to apply Karoubi's theorem (Prop. \ref{KaroubiTheorem}) to deduce the last row in \hyperlink{TableCliffordActions}{\it Table 3}, one needs the following observation, to confirm that the operators shown in the middle part of \hyperlink{TableCliffordActions}{\it Table 3} are (anti-)self-adjoined as required in \eqref{FredholmOperatorsAntiCommutingWithACliffordAction}:
\begin{lemma}[\bf Adjointness of quantum CPT]
\label{AdjointnessOfQuantumCPT}
If the quantum symmetry operator $\widehat{T}$ \eqref{TimeTerversalQuantumSymmetry} squares to $+1$ or $-1$  \eqref{SignChoicesForTimeReversalQuantumSymmetry} then it is self-adjoint or anti-self-adjoint, respectively; similarly for the operator $\widehat{C}$ \eqref{SignChoicesForChargeConjugationSymmetry}:
\vspace{-2mm}
$$
  \def\arraystretch{1.2}
  \begin{array}{ccc}
  (\widehat{T})^2 = \pm 1
  &
  \;\;\;\;\;\;\;\;
  \Rightarrow
  \;\;\;\;\;\;\;\;
  &
  (\widehat{T})^\ast
  \;=\;
  \pm \widehat{C}^\ast
  ,\,
  \\
  (\widehat{C})^2 = \pm 1
  &
  \;\;\;\;\;\;\;\;
  \Rightarrow
  \;\;\;\;\;\;\;\;
  &
  (\widehat{C})^\ast
  \;=\;
  \pm \widehat{C}^\ast
  \,.
  \end{array}
$$
\end{lemma}
\begin{proof}
We write $(-)^t$ for the transpose operation on complex-linear operators over the complex numbers, and $\overline{(-)}$ for complex conjugation.
The star-operation $A^\ast := \overline{A}^{\, t}$
on complex linear operators
agrees with the canonical star operation on their underlying real-linear operators $A_{{}_{\RealNumbers}}$, in that
$
  (A^\ast)_{{}_{\RealNumbers}}
  \;=\;
  (A_{{}_{\RealNumbers}})^\ast
  \,,
$
since $\ImaginaryUnit_{{}_{\RealNumbers}}$ is a skew-symmetric real operator. In particular the parity operator satisfies
$
  P^\ast \,=\, P
$,
regarded either over the complex or over the real numbers.
With complex conjugation itself regarded as a real-linear operator, $T := \overline{(-)}$, we have in addition:
\vspace{-2mm}
$$
  T \circ A_{{}_{\RealNumbers}}
  =
  (\,\overline{A}\,)_{{}_{\RealNumbers}}
  \circ
  T
  \,,
  \;\;\;\;\;\;\;\;
  T^\ast  = T
  \,.
$$

\vspace{-2mm}
\noindent Moreover, the lifts in \eqref{TimeTerversalQuantumSymmetry}
are of the form
$\widehat{T} = U_{{}_{\RealNumbers}} T$
and
$\widehat{C} = U_{{}_{\RealNumbers}} P T$
for $U$ a unitary operator
(hence complex-linear with  $U^{-1} = U^\ast = \overline{U}^{\, t}$).
Using all this we have (suppressing now the $(-)_{{}_{\RealNumbers}}$-subscript on $U$, for readability):
\vspace{-2mm}
$$
  U^{-1} U^t
  =
  \overline{U}^{\, t} U^t
  =
  \big(U \overline{U}\big)^t
  =
  \big(
    U (T U T)
  \big)^t
  =
  \big(
    (U T) (U T)
  \big)^t
  =
  \big(
    \widehat{T}^2
  \big)^t
  =
  (\pm 1)^t
  =
  \pm 1
  \,,
$$

\vspace{-2mm}
\noindent which means
\vspace{-1mm}
$$
  U^t = \pm U
  \,,
$$

\vspace{-1mm}
\noindent and hence
$$
  (\widehat{T})^\ast
  =
  (U T)^\ast
  =
  T^\ast U^\ast
  =
  T \overline{U}^{\, t}
  =
  U^t T
  =
  \pm U T
  =
  \pm
  \widehat{T}
  \,.
$$

\vspace{-6mm}
\end{proof}

\begin{example}[\bf Joint time-reversal and inversion symmetry (space-time inversion symmetry)]
\label{JointTimeReversalAndInversionSymmetry}
The situation where time reversal acts as a quantum symmetry \eqref{SignChoicesForTimeReversalQuantumSymmetry}
on Fredholm operators but {\it trivially} on the Brillouin torus may be understood as the combination $T I$ of:

\begin{itemize}[leftmargin=.4cm]
\item an {\it inversion symmetry} $I$ (point reflection in) acting by $[k] \mapsto [-k]$ on quasi-momenta \eqref{InversionActionOnBrillouinTorus} but trivially on Bloch observables:
\vspace{-2mm}
\begin{equation}
  \label{InversionActionOnBrillouinTorus}
  \begin{tikzcd}[column sep=10pt, row sep=-1pt]
    \DualTorus{d}
    \ar[rr, "{I}"]
    &&
    \DualTorus{d}
    \\
  \scalebox{0.8}{$   {[k]} $} &\longmapsto& \scalebox{0.8}{$ {[- k]}$}
    \mathrlap{\,,}
  \end{tikzcd}
  \hspace{1cm}
\begin{tikzcd}[column sep=10pt, row sep=-1pt]
    \FredholmOperators
    \ar[rr, "{I}"]
    &&
    \FredholmOperators
    \\
\scalebox{0.8}{$     F $} &\longmapsto& \scalebox{0.8}{$  F $}
  \end{tikzcd}
\end{equation}

\item With the {\it time-reversal symmetry} $T$
\eqref{TimeReversalActionOnBrillouinTorus} of Ex. \ref{TimeReversalSymmetryAndKRTheory}, acting non-trivially both on momenta and on Bloch observables:
\end{itemize}

Such $T I$ symmetry is exhibited, notably, by graphene (e.g. \cite[pp. 41]{CayssolFuchs20}).

\noindent
Under this combined {\it time-reversal- and inversion-symmetry} (e.g. \cite[\S II.B]{FWDZ16}\cite[\S 2.1.6, 5.5.2]{Vanderbilt18})
\vspace{-2mm}
\begin{equation}
  \label{JointTimeReversalAndInversionActionOnBrillouinTorus}
  \begin{tikzcd}[column sep=10pt, row sep=-1pt]
    \DualTorus{d}
    \ar[rr, "{T I}"]
    &&
    \DualTorus{d}
    \\
  \scalebox{0.8}{$   {[k]}$} &\longmapsto&  \scalebox{0.8}{$ {[k]}$}
    \mathrlap{\,,}
  \end{tikzcd}
  \hspace{1cm}
\begin{tikzcd}[column sep=10pt, row sep=-1pt]
    \FredholmOperators
    \ar[rr, "{T I}"]
    &&
    \FredholmOperators
    \\
  \scalebox{0.8}{$   F$} &\longmapsto& \scalebox{0.8}{$  \overline{F}$}
  \end{tikzcd}
\end{equation}

\vspace{-2mm}
\noindent the equivariantly indexed Fredholm operators
\eqref{TEKTheory}
are those satisfying $F_k = \overline{F}_k$, hence are the {\it real} Fredholm operators.
In this case, \hyperlink{TableCliffordActions}{\it Table 3} shows that the topological insulating phases of time-reversal- {\it and} inversion-symmetric crystals are classified by quaternionic K-theory:
\vspace{-2mm}
\begin{equation}
  \label{KOTheoryAsTWistedKRTheory}
  \def\arraystretch{1.4}
 \begin{array}{ll}
  \mathrm{K}^{
    \scalebox{.6}{$
      (\widehat{T}^2 = +1)
    $}
  }
  \big(
    X
    \times
    \HomotopyQuotient
      { \ast }
      { \{\NeutralElement, T\} }
  \big)
  &
  \;\simeq\;
  \mathrm{KO}^{0}(X)
  \\
  \mathrm{K}^{
    \scalebox{.6}{$
      (\widehat{T}^2 = -1)
    $}
  }
  \big(
    X
    \times
    \HomotopyQuotient
      { \ast }
      { \{\NeutralElement, T\} }
  \big)
  &
  \;\simeq\;
  \mathrm{KO}^4
  (X)
  \,,
  \end{array}
  \hspace{-.2cm}
\raisebox{10pt}{
\def\arraystretch{1.4}
\begin{tabular}{|c||c|c|c|c|c|c|c|c||c|c|c|}
  \hline
  $n = $
  &
  $0$
  &
  $1$
  &
  $2$
  &
  $3$
  &
  $4$
  &
  $5$
  &
  $6$
  &
  $7$
  &
  $8$
  &
  $9$
  &
  $\cdots$
  \\
  \hline
  \hline
\rowcolor{lightgray}
  $\mathrm{KO}^0\big(S^n_\ast\big) =$
  &
  $\Integers$
  &
  $\ZTwo$
  &
  $\ZTwo$
  &
  $0$
  &
  $\Integers$
  &
  $0$
  &
  $0$
  &
  $0$
  &
  $\Integers$
  &
  $\ZTwo$
  &
  $\cdots$
  \\
  \hline
  $\mathrm{KO}^4\big(S^n_\ast\big) =$
  &
  $\Integers$
  &
  $0$
  &
  $0$
  &
  $0$
  &
  $\Integers$
  &
  $\ZTwo$
  &
  $\ZTwo$
  &
  $0$
  &
  $\Integers$
  &
  $0$
  &
  $\cdots$
  \\
  \hline
\end{tabular}
}
\end{equation}
\vspace{.2cm}

\noindent
We come back to this example of $T I$-symmetric topological materials below in Ex. \ref{ClassificationOfTISymmetricSemiMetals} in the discussion of classification of topological semi-metals.
\end{example}

\begin{example}[\bf No quantum symmetry and Chern insulators]
  \label{NoQuantumSymmetryAndSpinOrbitCOupling}
  Mathematically, the simplest example is certainly that of {\it no} non-trivial quantum symmetry.
  In this case,  \hyperlink{TableCliffordActions}{\it Table 3} shows that the corresponding
  topological insulator phases are classified by plain complex K-theory:
  \vspace{-2mm}
  $$
    \mathrm{K}^0
    (
      X
    )
    \;\simeq\;
    \mathrm{KU}^0
    (
      X
    )
    \hspace{1cm}
\raisebox{10pt}{
\def\arraystretch{1.4}
\begin{tabular}{|c||c|c|c|c|c|c|c|c||c|c|c|}
  \hline
  \rowcolor{lightgray}
  $n = $
  &
  $0$
  &
  $1$
  &
  $2$
  &
  $3$
  &
  $4$
  &
  $5$
  &
  $6$
  &
  $7$
  &
  $8$
  &
  $9$
  &
  $\cdots$
  \\
  \hline
  \hline
  $\mathrm{KU}^0\big(S^n_\ast\big) =$
  &
  $\Integers$
  &
  $0$
  &
  $\Integers$
  &
  $0$
  &
  $\Integers$
  &
  $0$
  &
  $\Integers$
  &
  $0$
  &
  $\Integers$
  &
  $0$
  &
  $\cdots$
  \\
  \hline
\end{tabular}
}
 $$
 For effectively 2-dimensional materials, this means
 (by comparison with \hyperlink{HomotopyTypeOfPuncturedTorus}{\it Figure 7})
 that these ``un-protected'' topological phases are classified by the integers.
From the theory of characteristic Chern classes $c_i$ (e.g. \cite{MilnorStasheff74}), this integer may be identified with
the {\it first Chern number} $c_1[V] \,=\, \int_{\DualTorus{2}} c_1(V)$ of the valence bundle $V$:
\vspace{-2mm}
 $$
   \begin{tikzcd}[column sep=-1pt, row sep=-2pt]
   \mathrm{KU}^0
   \big(
     \DualTorus{2}_\ast
   \big)
   &
   \simeq
   &
   \mathrm{KU}^0
   \big(
     S^2_\ast
   \big)
   &\simeq&
   \Integers
   \\
  \scalebox{0.8}{$ {[V]} $}
   \ar[rrrr, |->]
   &&
   &&
   \scalebox{0.8}{$ c_1[V] $}
   \,.
   \end{tikzcd}
 $$
 For this reason, the topological insulators which are ``un-protected'' by any quantum symmetry are also called {\it Chern insulators} (e.g. \cite[\S 5.1]{Vanderbilt18}).
  While this is mathematically the most immediate case, in solid state physics Chern insulators are typically  thought of as realized only with some
  effort by breaking symmetries in a material which by itself does enjoy time-reversal quantum symmetry and/or inversion symmetry
  (Ex. \ref{JointTimeReversalAndInversionSymmetry}).

  \noindent
  In closing this example, notice that:
  \begin{itemize}[leftmargin=.4cm]
  \item
  on the one hand, time-reversal symmetry is easily broken by placing a material into a strong magnetic field $B$, in which the Bloch state with
  momentum vector $k$ is subject to a {\it Lorentz force} $B \cdot k$ which is {\it opposite} that for the time-reversed momentum $- k$.
  This is the situation of quantum Hall materials;

  \item
  On the other hand, a $T$-breaking effect intrinsic to the material,  present also in the absence of an external magnetic field, can be induced from spin-orbit coupling
  but is much more subtle to analyze and realize. This was the achievement of the {\it Haldane model}
  (\cite{Haldane15}, \cite[\S 5.1.1]{Vanderbilt18}), see Ex. \ref{HaldaneModel}.
\end{itemize}
\end{example}

\medskip

\noindent
{\bf  Topological crystalline insulator phases and orbifold K-theory.}
Beware that the 10-fold way of Fact \ref{ListOfCPTTwistings} applies only to global CPT-symmetries which act trivially on the crystal. This case does (effectively) occur (Ex. \ref{JointTimeReversalAndInversionSymmetry}, Ex. \ref{HaldaneModel}), but in general, the C- and T-symmetries \eqref{TimeTerversalQuantumSymmetry} act by
reflection on the Brillouin torus (Ex. \ref{TimeReversalSymmetryAndKRTheory}) and
form semidirect products with the crystallographic point group symmetries \eqref{CrystallographicGroup}. These we will call the {\it external symmetry group} (\hyperlink{SymmetryTypes}{\it Table 4}),
and we call the {\it orbifold quotient}
(see pointers in \cite{SS20OrbifoldCohomology})
by this group action on a momentum space  the {\it crystallographic orbi-orientifold}:

\vspace{3mm}
\begin{equation}
  \label{CrystallographicOrbiOrientifold}
  \overset{
    \mathclap{
    \rotatebox{+45}{
      \rlap{
        \hspace{-15pt}
        \tiny
        \color{darkblue}
        \bf
        \def\arraystretch{.9}
        \begin{tabular}{c}
          external
          \\
          symmetry
        \end{tabular}
      }
    }
    }
  }{
    G_{\mathrm{ext}}
  }
  \;:=\;
  \overset{
    \mathclap{
    \rotatebox{+45}{
      \rlap{
        \hspace{-16pt}
        \tiny
        \color{darkblue}
        \bf
        \def\arraystretch{.9}
        \begin{tabular}{c}
          crystallographic
          \\
          point symmetry
        \end{tabular}
      }
    }
    }
  }{
    \mathrm{G}_{\mathrm{pt}}
  }
  \rtimes
  \big(
  \overset{
    \mathclap{
    \raisebox{3pt}{
      \tiny
      \color{darkblue}
      \bf
      CP-symmetries
    }
    }
  }{
    \{\NeutralElement, T\}
    \times
    \{\NeutralElement, C\}
  }
  \big)
  \;\;\;\;\;\;\;\;
  \vdash
  \;\;\;\;\;\;\;\;
  \overset{
    \raisebox{3pt}{
      \tiny
      \color{darkblue}
      \bf
      \begin{tabular}{c}
        crystallographic
        \\
        orbi-orientifold
      \end{tabular}
    }
  }
  {
  \HomotopyQuotient
    { \TopologicalSpace }
    { G_{\mathrm{ext}}\;. }
  }
\end{equation}
The topological insulator-phases which are protected/enriched (Rem. \ref{OnSymmetryProtection}) by such an external symmetry are known as {\it topological crystalline insulators}-phases: \cite{Fu11}\cite{SMJZ13}\cite{ShenCha14}\cite{AndoFu15}\cite{KdBvWKS17}\cite{LieWangQiuGao20}.

\medskip
For emphasis, we restate the classification statement so far, making this fully explicit:

\begin{fact}[\bf Classification of external-SPT/SET phases]
  \label{ClassificationOfExternalSPTPhases}
  To the extent that topological phases are
  classified by twisted equivariant K-theory (cf. Fact \ref{KTheoryClassificationOfTopologicalPhasesOfMatter}), the external-SPT/SET phases are classified specifically by the TED-K theory
  \eqref{TEKTheory}
  of the crystallographic orbi-orientifold
  \eqref{CrystallographicOrbiOrientifold}:
\vspace{-2mm}
 $$
   \bigg\{\!\!
   \mbox{
     \rm
     \hspace{-.3cm}
     \begin{tabular}{c}
       $G_{\mathrm{ext}}$-SPT/SET
       \\
        crystalline insulator phases
     \end{tabular}
     \hspace{-.3cm}   }
   \!\! \bigg\}
   \;\;
   =
   \;\;
   \underset{
     [\tau]
   }{\coprod}
   \;
    \mathrm{KR}^\tau
    \big(
      \HomotopyQuotient
        { \DualTorus{d} }
        { G_{\mathrm{ext}} }
    \big)
    \;\;
    \simeq
    \;\;
    \underset{[\tau]}{\coprod}
    \,
    \left\{
    \hspace{-4pt}
    \begin{tikzcd}[column sep=-2pt]
      &&
      \HomotopyQuotient
        { \FredholmOperators^0_{\ComplexNumbers} }
        {
          \left(
            \frac{
              \UnitaryGroup(\mathscr{H})
              \times
              \UnitaryGroup(\mathscr{H})
            }{
              \CircleGroup
            }
            \rtimes
            \{\NeutralElement, P\}
            \times
            \{\NeutralElement, T\}
        \!\!  \right)
        }
      \ar[
        d,
        shorten=-2pt
      ]
      \\
      \HomotopyQuotient
        { \DualTorus{d} }
        { G_{\mathrm{ext}} }
      \ar[
        rr,
        shorten=-2pt,
        "{\tau}"{swap, yshift=-2pt}
      ]
      \ar[
        urr,
        dashed,
        shorten <=-4pt
      ]
      &&
      \mathbf{B}
      \left(
        \frac{
          \UnitaryGroup(\mathscr{H})
          \times
          \UnitaryGroup(\mathscr{H})
        }{
          \CircleGroup
        }
        \rtimes
        \{\NeutralElement, P\}
        \times
        \{\NeutralElement, T\}
\!\!      \right)
    \end{tikzcd}
    \hspace{-4pt}
    \right\}_{\big/ \sim_{\homotopy}}.
  $$
\end{fact}

\subsection{Internal symmetry protection and K-valued group cohomology.}
\label{InternalSymmetriesAndInnerLocalSystemTwists}

\noindent
{\bf Internal symmetries in crystalline materials.}
In addition to CPT symmetries (\cref{CrystallographicSymmetriesAndOrbifoldKTheory})
and crystallographic symmetries $G_{\mathrm{pt}}$ (\cref{CrystallographicSymmetriesAndOrbifoldKTheory}), hence in addition to {\it external symmetries} \eqref{CrystallographicOrbiOrientifold}, a crystalline
material may exhibit {\it internal symmetries}, which typically arise from symmetries among electron degrees of freedom
located separately at each atomic site of the underlying crystal lattice,
whence they are also called {\it on-site symmetries}.

\begin{example}[Spin rotation symmetry] If external magnetic fields and spin-orbit coupling are both negligible,
then the energy of the electrons in the material
is typically independent of their spin, and the $\mathrm{Spin}(3) \simeq \mathrm{SU}(2)$-group of spin rotations
will be an internal
symmetry of the system, including in particular the spin flip operation
\vspace{-1mm}
\begin{equation}
\label{SpinFlip}
\vert \uparrow \rangle \; \xleftrightarrow{\;S\;} \; \vert \downarrow \rangle
\end{equation}

\vspace{-1mm}
\noindent with respect to any experimentally preferred spin-basis.
More generally, if a number $\kappa$ of (spinful) electron orbitals around any
one atomic site have negligible energy difference, then their permutation group $\mathrm{Sym}(\kappa)$
will act as an approximate internal symmetry group on the system.
\end{example}

Mathematically, this simply means that an internal symmetry group is a direct factor subgroup
\vspace{-2mm}
\begin{equation}
  \label{InternalSymmetryGroup}
  G_{\mathrm{int}}
  \;\subset\;
  G_{\mathrm{ext}} \times G_{\mathrm{in}}
  \;\simeq\;
  G
\end{equation}

\vspace{-2mm}
\noindent
(of the equivariance group $G$ that enters the relevant twisted $G$-equivariant K-theory
\eqref{TEKTheory}), whose action on the Brillouin torus domain is trivial:
\vspace{-1mm}
\begin{equation}
  \label{InternalSymmetryActsTriviallyOnBrillouinTorus}
  \HomotopyQuotient
    { \DualTorus{d} }
    { G }
  \;\simeq\;
  \HomotopyQuotient
    { \DualTorus{d} }
    { G_{\mathrm{ext}} }
  \times
  \HomotopyQuotient
    { \ast }
    { G_{\mathrm{int}} }
  \,.
\end{equation}
In this sense, also the parity symmetry $P$ \eqref{GroupOfQuantumSymmetries} is an internal symmetry, while time-reversal $T$ is not (Ex. \ref{TimeReversalSymmetryAndKRTheory}) and neither is spatial inversion \eqref{InversionActionOnBrillouinTorus}, but their
combination $T I$ again is an internal symmetry (Ex. \ref{JointTimeReversalAndInversionSymmetry}), as indicated in \hyperlink{SymmetryTypes}{\it Table 4}:

\begin{center}
\hyperlink{SymmetryTypes}{}

\hspace{-.9cm}
\footnotesize
\begin{tabular}{cc}
\raisebox{0pt}{\footnotesize
\def\arraystretch{2.35}
\begin{tabular}{|c|c|c|}
  \hline
  \multicolumn{3}{|c|}{
    \def\arraystretch{.6}
    \begin{tabular}{c}
      \bf Symmetries
      \\
      $G$
      $\mathclap{\phantom{\vert_{\vert_{\vert}}}}$
    \end{tabular}
  }
  \\
  \hline
  \multicolumn{2}{|c|}{
    \def\arraystretch{.6}
    \begin{tabular}{c}
      \bf External
      \\
      $G_{\mathrm{ext}}$
      $\mathclap{\phantom{\vert_{\vert_{\vert}}}}$
    \end{tabular}
  }
  &
  \def\arraystretch{.6}
  \begin{tabular}{c}
    \bf Internal
    \\
    $G_{\mathrm{int}}$
      $\mathclap{\phantom{\vert_{\vert_{\vert}}}}$
  \end{tabular}
  \\
  \hline
  \def\arraystretch{1}
  \begin{tabular}{c}
    \bf Crystallographic
    \\
    $G_{\mathrm{pt}}$,
    e.g. $\{\NeutralElement,I\}$
      $\mathclap{\phantom{\vert_{\vert_{\vert}}}}$
  \end{tabular}
  &
  \def\arraystretch{1}
  \begin{tabular}{c}
    {\bf Time-reversal}
    \\
    $\{\NeutralElement, T\}$,
    $\{\NeutralElement, C\}$
      $\mathclap{\phantom{\vert_{\vert_{\vert}}}}$
  \end{tabular}
  &
  \def\arraystretch{1}
  \begin{tabular}{c}
    \bf On-site
    \\
    e.g.
    $\{\NeutralElement, P\}$,
    $\{\NeutralElement,S\}$
      $\mathclap{\phantom{\vert_{\vert_{\vert}}}}$
  \end{tabular}
  \\
  \hline
\end{tabular}
}
&
\def\arraystretch{3.06}
\begin{tabular}{|c|}
\hline
$
  \overset{
    G
  }{
  \overbrace{
  \underset{
    G_{\mathrm{ext}}
  }{
  \underbrace{
  G_{\mathrm{pt}}
  \rtimes
  \{\NeutralElement, T/C\}
  }
  }
  \times
  \underset{
    G_{\mathrm{int}}
    \mathclap{\phantom{\vert_{\vert_{\vert_{\vert}}}}}
  }{
  \underbrace{
  \{ \NeutralElement, P \}
  \times
  \{\NeutralElement, S\}
  }
  }
  }
  }
$
\\
\hline
  $
    \HomotopyQuotient
      { \DualTorus{d} }
      { G }
    \;\simeq\;
    \overset{
      \mathclap{
      \raisebox{3pt}{
        \hspace{-18pt}
        \tiny
        \color{darkblue}
        \bf
        \def\arraystretch{.9}
        \begin{tabular}{c}
          orbi-orienti-folded
          Brillouin torus
        \end{tabular}
      }
      }
    }{
    \HomotopyQuotient
      { \DualTorus{d} }
      { G_{\mathrm{ext}} }
    }
    \times
    \HomotopyQuotient
      { \ast }
      { G_{\mathrm{int}} }
  $
 \\
 \hline
 \end{tabular}
 \\
 \phantom{-}
 \\
 \multicolumn{2}{c}{
\def\arraystretch{1.5}
\begin{tabular}{|c|l|l|l}
  \hhline{---}
  \multicolumn{2}{|c|}{
    \cellcolor{white}
    \bf Symmetry name
  }
  &
  \cellcolor{white}
  \bf Action
  &
  \\
  \hhline{===}
  \cellcolor{lightgray}
  $G_{\mathrm{pt}}$
  &
  \cellcolor{lightgray}
  {\it Crystallographic point transformation}
  &
  \cellcolor{lightgray}
  Orthogonal transformation on BT
  &
  \multirow{1}{*}{
  \eqref{CrystallographicGroup}
  }
  \\
  \cellcolor{white}
  $I$
  &
  \cellcolor{white}{\it Inversion}
  &
  \cellcolor{white}Point reflection on BT
  &
  \multirow{1}{*}{
  \eqref{InversionActionOnBrillouinTorus}
  }
  \\
  \cellcolor{lightgray}$T$
  &
  \cellcolor{lightgray}{\it Time reversal}
  &
  \cellcolor{lightgray}Point reflection on BT \& complex conj. on obs.
  &
  \multirow{1}{*}{
  \eqref{TimeReversalActionOnBrillouinTorus}
  }
  \\
  \cellcolor{white}$C$
  &
  \cellcolor{white}{\it Charge conjugation}
  &
  \cellcolor{white}Point reflection on BT \& complex conj. + deg. flip on obs.
  &
  \multirow{2}{*}{
  \eqref{GroupOfQuantumSymmetries}
  }
  \\
  \cellcolor{lightgray}$P$
  &
  \cellcolor{lightgray}{\it Parity reversal}
  &
  \cellcolor{lightgray}No action on BT
  \& degree flip on obs.
  &
  \\
  \cellcolor{white}$S$
  &
  \cellcolor{white}{\it Spin flip}
  &
  \cellcolor{white}No action on BT
  \&
  some projective action on obs.
  &
  \multirow{1}{*}{
  \eqref{SpinFlip}
  }
  \\
  \hhline{---}
\end{tabular}
 }
\end{tabular}

\vspace{.2cm}

\begin{minipage}{17cm}
  \footnotesize
  {\bf Table 4 --
  Possible symmetry groups
  of electron dynamics in
  a crystalline material.}
    Here ``BT'' refers to the {\it Brillouin torus} $\DualTorus{d}$ \eqref{BrillouinTorus},
   while ``obs'' refers to the ground state {\it observables}, i.e., the Fredholm operators \eqref{BTParametrizedFredholmOperators}.
  \end{minipage}
\end{center}

  Any of these symmetries may be respected in a given quantum material (e.g. none of them, which is the case of plain Chern insulators, Ex. \ref{NoQuantumSymmetryAndSpinOrbitCOupling}). In experimental practice, symmetries are often {\it approximately} respected, hence effectively respected up to some scale of resolution but broken at higher resolution. (Ultimately, at sub-atomic resolution the dynamics of the nuclei in the crystal lattice will become visible and the entire picture of electron dynamics in a fixed background field breaks down.)

\medskip
  The list of possible external symmetries displayed is meant to be exhaustive,
  but there can be any number of further internal symmetries, reflecting the internal degrees of freedom at each crystal site (e.g. accidental orbital energy degeneracies that remain unresolved).


\begin{remark}[Internal-symmetry protection in the literature]
Even with the conjectured classification of {\it external}-SPT phases in twisted equivariant K-theory
(as in \cref{CrystallographicSymmetriesAndOrbifoldKTheory}, \cref{CrystallographicSymmetriesAndOrbifoldKTheory})
becoming widely appreciated (beginning with \cite{Kitaev09})
it was felt that K-theory could apply only to systems well-approximated by (fermionic and) {\it free} dynamics, and that internal-symmetry protection specifically of interacting phases needed an approach different from K-theory. (We address the issue of interacting phases in K-theory below in \cref{InteractingPhasesAndTEDKOfConfigurationSpaces}).

\vspace{1mm}
\noindent {\bf (i)}
A first influential proposal asserted
\cite{CGLW13}\cite{CGLW12}
that ``bosonic'' and interacting $G_{\mathrm{int}}$-SPT phases of dimension $d$ are classified by the group cohomology of $G_{\mathrm{int}}$ in degree $d + 1$ with coefficients in $U(1)$ (and that for fermionic and interacting such systems an analogous statement holds for a suitable notion of group super-cohomology was claimed in \cite{GuWen12}). In our language of stacks, such group cohomology is given by homotopy classes of maps as shown on the right here (e.g. \cite[Ex. 2.4]{FSS20Character}):
\vspace{-2mm}
$$
  \overset{
    \raisebox{2pt}{
      \tiny
      \color{darkblue}
      \bf
      group cohomology
    }
  }{
  H^{d+1}_{\mathrm{grp}}
  \big(
    G_{\mathrm{int}}
    \, ;\,
    \CircleGroup
  \big)
  }
  \;\;
    \simeq
  \;\;
  \Big\{\!\!\!
    \begin{tikzcd}
      \mathbf{B}G_{\mathrm{int}}
      \ar[r, dashed]
      &
      \mathbf{B}^{d+1} \CircleGroup
    \end{tikzcd}
  \!\!\!\!\Big\}_{\!\!\big/\sim_{\homotopy}}.
$$

\vspace{-2mm}
\noindent A physics motivation for this proposal is offered in \cite[\S V]{CLW11}.
However, this proposal is now ``known not to be complete'' \cite[p. 2]{WangSenthil14}; in fact, on the general question of classification of SPT/SETs:
``a completely general understanding is lacking, and many questions remain. $[\cdots]$ the current understanding of fractionalization of quantum numbers, along with the classification and characterization of SETs is incomplete.'' \cite[p. 3]{BPCW14}.

\vspace{1mm}
\noindent {\bf (ii)} Another proposal has been put forward in \cite{BPCW14}\cite[\S 2.2]{Wang17}  based (somewhat tacitly) on the reasonable idea that
$G_{\mathrm{int}}$-symmetry protection should mean that all relevant structures found in the underlying topological phase/order
acquire {\it $G_{\mathrm{int}}$-equivariant enhancements}. Concretely, using the common (also conjectural, but see Rem. \ref{TopologicallyOrderedGroundStatesAndModularTensorCategories}) assumption that an
anyonic topological order is characterized by a {\it unitary fusion category} $\mathcal{C}$, the proposal of \cite[(1)]{BPCW14} says that the
$G_{\mathrm{int}}$-symmetry enrichments of the given topological phase are classified by (we paraphrase slightly, see \cite[Rem. 2.9]{FSS20Character}) the
{\it 2-group cohomology} of $G_{\mathrm{int}}$ with coefficients in the {\it automorphism 2-group} $\mathrm{Aut}(\mathcal{C})$ of
$\mathcal{C}$ (i.e., the 2-group whose objects are braided monoidal endofunctors $\mathcal{C} \to \mathcal{C}$ which are equivalences
of braided monoidal categories, and whose morphisms are compatible natural isomorphisms between these, this is implicit in \cite[(81)]{BPCW14}):
\vspace{-3mm}
\begin{equation}
  \label{TwoGroupCohohomologyWithCoefficientsInAutOfFusionCategory}
  \overset{
    \mathclap{
    \raisebox{2pt}{
      \tiny
      \color{darkblue}
      \bf
      \def\arraystretch{.9}
      \begin{tabular}{c}
        2-group
        \\
        cohomology
      \end{tabular}
    }
    }
  }{
    H^{1}_{\mathrm{grp}}
  }
  \big(
    G_{\mathrm{int}}
    \, ;\,\;
    \overset{
      \mathclap{
      \raisebox{+5pt}{
        \tiny
        \color{darkblue}
        \bf
        \begin{tabular}{c}
          Automorphism 2-group
          \\
          of unit. fusion category
          \\
          of anyon species
        \end{tabular}
      }
      }
    }{
      \mathrm{Aut}(\mathcal{C})
    }
    \;
  \big)
  \;\;
    \simeq
  \;\;
  \big\{
    \!\!\!
    \begin{tikzcd}
      \mathbf{B}G_{\mathrm{int}}
      \ar[
        r,
        dashed,
        "{
          \mbox{
            \tiny
            \color{darkblue}
            \bf
            \def\arraystretch{.9}
            \begin{tabular}{c}
              Equivalence classes of
              \\
              maps of moduli 2-stacks
            \end{tabular}
          }
        }"{yshift=6pt}
      ]
      &
      \mathbf{B}
      \mathrm{Aut}(\mathcal{C})
    \end{tikzcd}
    \!\!\!\!
  \big\}_{\!\!\big/\sim_{\homotopy}}
  \,.
\end{equation}

\vspace{-2mm}
\noindent Due to the 2-groupal nature of $\mathrm{Aut}(\mathcal{C})$ there are ordinary but higher degree group
cohomology classes in $H^3(G_{\mathrm{int}};\, \mathcal{A})$  underlying this 2-group cohomology, with coefficients in
(we again paraphrase slightly) the Picard group $\mathcal{A} := \mathrm{Pic}(\mathcal{C})_{\!/\sim}$
(of invertible anyons species).

\vspace{1mm}
\noindent {\bf (iii)}  The proposal \eqref{TwoGroupCohohomologyWithCoefficientsInAutOfFusionCategory} is conceptually
robust relative to its assumption: To the extent that unitary fusion categories indeed reflect aspects of topological
orders in solid state physics, their internal-symmetry protected incarnation essentially must be such a 2-group
cohomology class \eqref{TwoGroupCohohomologyWithCoefficientsInAutOfFusionCategory} if the notion of fusion categories
is at all the appropriate mathematical structure to speak about topological order.

\vspace{1mm}
\noindent {\bf (iv)}
Still, the question remains how to connect this proposal to the widely expected K-theory classification of free topological phases. This physics demands that the K-theory classification of (symmetry protected) topological phases ought to be recovered by a more general classification of (symmetry protected) topological order in the special case or limit of ``trivial order''. However, the fusion category $\mathcal{C}$ reflecting trivial order is itself trivial, and has trivial automorphism 2-group, so that \eqref{TwoGroupCohohomologyWithCoefficientsInAutOfFusionCategory} collapses in this case.
\end{remark}

\medskip

\noindent
{\bf Internal-SPT in TED-K-theory.}
We now observe that our stacky formulation of
twisted equivariant K-theory
\eqref{TEKTheory} allows us to {\it read off} which kind of (higher) group cohomology must classify $G_{\mathrm{int}}$-symmetry-protected
phases. (For the moment we discuss this for non-interacting phases to which \eqref{TEKTheory} is thought to apply, by Fact \ref{VacuaOfTheRelativisticElectronPositronFieldInBackground}, but the argument will generalize verbatim to the case of interacting phases,
as discussed below in \cref{InteractingPhasesAndTEDKOfConfigurationSpaces}).
All we need here is the evident assumption \eqref{InternalSymmetryActsTriviallyOnBrillouinTorus} on the nature of ``internal symmetry''
$G_{\mathrm{int}}$, together with the {\it mapping stack adjunction} (reviewed in \cite[Prop. 2.31]{SS20OrbifoldCohomology}\cite[Prop. 3.2.44]{SS21EPB}), which says that for any triple $\mathcal{X}$, $\mathcal{Y}$ and $\mathcal{Z}$ of stacks, there are {\it natural equivalences} between
the following types of maps between them:
\vspace{-3mm}
\begin{equation}
\label{TheMappingStackAdjunction}
\adjustbox{raise=-50pt}{
\begin{tikzpicture}

\draw (0,0)
node{$
  \Big\{\!\!\!\!
  \begin{tikzcd}[column sep=30pt]
    \mathcal{X}
    \times
    \mathcal{Y}
    \ar[
      rr,
      "{
        (x,y) \,\mapsto\, f(x,y)
      }"]
    &&
    \mathcal{Z}
  \end{tikzcd}
  \!\!\!\!\Big\}
$};

\draw[latex-latex]
  (-1+.2,-.5)
    to node[sloped, above, yshift=-2pt]{$\sim$}
  (-3+.2, -1.4);
\draw[latex-latex]
  (+1,-.5)
    to node[sloped, above, yshift=-2pt]{$\sim$}
  (+3, -1.4);

\draw (.1, -1.1)
node{
  \tiny
  \color{greenii}
  \bf
  mapping stack adjunction
};

\draw (-5.2,-1.8)
node{$
  \Big\{\!\!\!\!
  \begin{tikzcd}[column sep=30pt]
    \mathcal{X}
    \ar[
      rr,
      "{
        x
          \,\mapsto\,
        \scalebox{1.2}{$($}
          y \,\mapsto\, f(x,y)
        \scalebox{1.2}{$)$}
       }"]
    &&
    \Maps{\big}
      { \mathcal{Y} }
      { \mathcal{Z} }
  \end{tikzcd}
  \!\!\!\!\! \Big\}
$};

\draw (+5.5,-1.8)
node{$
  \Big\{\!\!\!\!
  \begin{tikzcd}[column sep=30pt]
    \mathcal{Y}
    \ar[
      rr,
      "{
        y
          \,\mapsto\,
        \scalebox{1.2}{$($}
          x \,\mapsto\, f(x,y)
        \scalebox{1.2}{$)$}
       }"]
    &&
    \Maps{\big}
      { \mathcal{X} }
      { \mathcal{Z} }
  \end{tikzcd}
 \!\!\!\!\! \Big\}.
$};

\end{tikzpicture}
}
\end{equation}

\vspace{-2mm}
\noindent
Hence we get the two identifications of the TED-K cohomology group in the presence of internal symmetries \eqref{InternalSymmetryActsTriviallyOnBrillouinTorus} shown in \hyperlink{InternalSymmetryMappingStackAdjunction}{Fig. 5}
\vspace{-6mm}
\begin{center}
\hyperlink{InternalSymmetryMappingStackAdjunction}{}

\begin{tikzpicture}

\draw (0,0) node {
$
  \mathllap{
    \overset{
      \mathclap{
      \raisebox{3pt}{
        \tiny
        \color{orangeii}
        \bf
        \begin{tabular}{c}
          TE-K cohomology group
        \end{tabular}
      }
      }
    }{
      \mathrm{KR}
        ^\tau
        _{
          G_{\mathrm{ext}}
            \times
          G_{\mathrm{int}}
        }
      (\TopologicalSpace)
    }
    \;=\;
    \;
  }
  \left\{
  \hspace{-4pt}
  \begin{tikzcd}[column sep=25pt]
    &
    \overset{
      \raisebox{2pt}{
        \tiny
        \color{darkblue}
        \bf
        \def\arraystretch{.9}
        \begin{tabular}{c}
          symmetry protected
          \\
          topological observables
        \end{tabular}
      }
    }{
    \HomotopyQuotient
      { \FredholmOperators^0_{\ComplexNumbers} }
      {
          \frac{
            \UnitaryGroup(\mathscr{H})^2
          }{\CircleGroup}
          \!\rtimes
          \ZTwo^2
      }
    }
    \ar[
      d
    ]
    \\
    \underset{
      \raisebox{10pt}{
        \tiny
        \color{darkblue}
        \bf
        \;\;\;\;\;\;symmetry
      }
    }{
    \HomotopyQuotient
      { \TopologicalSpace }
      {
        \underset{
         \mathclap{
          \raisebox{-0pt}{
            \tiny
            \color{darkblue}
            \bf
            external
          }
          }
        }{
          G_{\mathrm{ext}}
        }
      }
    \times
    \underset{
      \raisebox{-0pt}{
        \tiny
        \color{darkblue}
        \bf
        internal
      }
    }{
      \HomotopyQuotient
        { \ast }
        { G_{\mathrm{int}} }
    }
    }
    \ar[
      r,
      shorten=-2pt,
      "{
        \tau
      }"{},
      "{
        \mbox{
          \tiny
          \color{greenii}
          \bf
          \def\arraytretch{.9}
          \begin{tabular}{c}
            equivariant
            \\
            twist
          \end{tabular}
        }
      }"{swap, pos=.4}
    ]
    \ar[
      ur,
      dashed,
      shorten=-3pt,
      "{
        \mbox{
          \tiny
          \color{orangeii}
          \bf
          TE-K cocycle
        }
      }"{sloped, pos=.4}
    ]
    &
    \mathbf{B}
    \Big(
      \underset{
        \raisebox{+1pt}{
          \tiny
          \color{darkblue}
          \bf
          quantum symmetries
        }
      }{
        \frac{
          \UnitaryGroup(\mathscr{H})^2
        }{\CircleGroup}
        \!\rtimes
        \ZTwo^2
        }
    \Big)
  \end{tikzcd}
  \hspace{-4pt}
  \right\}_{\!\!\!\!\!\big/\sim_{\homotopy}}
$
};

\draw (-4.5,-5) node {
$
  \left\{
  \hspace{-4pt}
  \begin{tikzcd}[column sep=20pt]
    &
    \overset{
      \raisebox{2pt}{
        \tiny
        \color{darkblue}
        \bf
        moduli stack of
        {\color{purple}internal}
        -symmetry
        protected topological
        observables
      }
    }{
    \Maps{\Big}
      { \mathbf{B}G_{\mathrm{int}} }
    {
      \HomotopyQuotient
      {
        \FredholmOperators^0_{\ComplexNumbers}
      }
      {
          \frac{
            \UnitaryGroup(\mathscr{H})^2
          }{\CircleGroup}
          \!\rtimes
          \ZTwo^2
      }
    }
    }
    \ar[
      d
    ]
    \\
    \HomotopyQuotient
      { \TopologicalSpace }
      {
        G_{\mathrm{ext}}
      }
    \ar[
      r,
      shorten=-2.5pt,
      "{
        \mbox{
          \tiny
          \color{greenii}
          \bf
          \def\arraystretch{.9}
          \begin{tabular}{c}
            inner
            \\
            local system
          \end{tabular}
        }
      }"{swap, pos=.5}
    ]
    \ar[
      ur,
      dashed,
      shorten=-3pt,
      "{
        \mbox{
          \tiny
          \color{orangeii}
          \bf
          TE-$\mathrm{K}^G$ cocycle
        }
      }"{sloped, pos=.4}
    ]
    &
    \underset{
      \mathclap{
        \raisebox{-2pt}{
          \tiny
          \color{darkblue}
          \bf
          \begin{tabular}{c}
          moduli stack of
          {\color{purple}internal} quantum symmetries
          \end{tabular}
        }
      }
    }{
      \Maps{\bigg}
        {\! \mathbf{B}G_{\mathrm{int}} }
      {
      \mathbf{B}
      \Big(
          \frac{
            \UnitaryGroup(\mathscr{H})^2
          }{\CircleGroup}
          \!\rtimes
          \ZTwo^2
      \Big)
      \!\!}
    }
  \end{tikzcd}
  \hspace{-4pt}
  \right\}_{\!\!\!\!\big/\sim_{\homotopy}}
$
};

\draw[double]
  (-1.7,-2)
  to
    node[sloped, above]
    {
      \tiny
      mapping stack adjunction
    }
    node[sloped, below]
    {
      \tiny
      on {\color{purple}internal} factor
    }
  (-4,-3);
\draw[double]
  (+1.7-1,-2)
  to
    node[sloped, above]
    {
      \tiny
      mapping stack adjunction
    }
    node[sloped, below]
    {
      \tiny
      on {\color{purple}external} factor
    }
  (+4-1,-3);

\draw (+4.5,-5) node {
$
  \left\{
  \hspace{-4pt}
  \begin{tikzcd}[column sep=20pt]
    &
    \overset{
      \raisebox{2pt}{
        \tiny
        \color{darkblue}
        \bf
        moduli stack of
        {\color{purple}external}-symmetry
        protected
        topological
        phases
      }
    }{
    \Maps{\Big}
      {
        \HomotopyQuotient
          { \TopologicalSpace }
          { G_{\mathrm{ext}} }
      }
    {
      \HomotopyQuotient
      {
        \FredholmOperators^0_{\ComplexNumbers}
      }
      {
          \frac{
            \UnitaryGroup(\mathscr{H})^2
          }{\CircleGroup}
          \!\rtimes
          \ZTwo^2
      }
    }
    }
    \ar[
      d
    ]
    \\
    \mathbf{B}G_{\mathrm{int}}
    \ar[
      r,
      shorten=-2pt,
      "{ \tilde \tau }"
      "{
        \mbox{
          \tiny
          \color{greenii}
          \bf
          \def\arraystretch{.9}
          \begin{tabular}{c}
          \end{tabular}
        }
      }"{swap, pos=.3}
    ]
    \ar[
      ur,
      dashed,
      shorten=-3pt,
      "{
        \mbox{
          \tiny
          \color{orangeii}
          \bf
          $\infty$-group
          cohom. cocycle
        }
      }"{sloped, pos=.4}
    ]
    &
    \underset{
      \mathclap{
        \raisebox{-2pt}{
          \tiny
          \color{darkblue}
          \bf
          \begin{tabular}{c}
          moduli stack of
          {\color{purple}external}
          quantum symmetries
          \end{tabular}
        }
      }
    }{
      \Maps{\bigg}
      {
        \HomotopyQuotient
          { \TopologicalSpace }
          { G_{\mathrm{ext}} }
      }
      {
      \mathbf{B}
      \Big(
          \frac{
            \UnitaryGroup(\mathscr{H})^2
          }{\CircleGroup}
          \!\rtimes
          \ZTwo^2
      \Big)
      \!\!}
    }
  \end{tikzcd}
  \hspace{-4pt}
  \right\}_{\!\!\!\!\big/\sim_{\homotopy}}
$
};

\end{tikzpicture}
\begin{minipage}{17cm}
\footnotesize
 {\bf Figure 5.}
 An ``internal symmetry'' group $G_{\mathrm{int}}$ in TED-K theory is one whose action on the remaining domain space/orbifold is trivial,
 hence which exhibits the domain as being ``inside a $G_{\mathrm{int}}$-orbi-singularity'' \eqref{InternalSymmetryActsTriviallyOnBrillouinTorus}.

 \end{minipage}
\end{center}

 The above \hyperlink{InternalSymmetryMappingStackAdjunction}{Fig. 5} shows how the mapping stack adjunction
 \eqref{TheMappingStackAdjunction}
 identifies
 the TED-K cohomology groups with respct to such an
 internal symmetry $G_{\mathrm{int}}$ with both:

\begin{itemize}[leftmargin=.6cm]

\item[\bf 1.] inner local system-twisted $G_{\mathrm{int}}$-fixed
 TE-K-theory of the given domain, and

 \item[\bf 2.] $\infty$-group cohomology of $G_{\mathrm{int}}$ with coefficients in the symmetry $\infty$-group of external-SPT
 phases on the given domain.

\end{itemize}

\medskip
Here the adjunction on the right identifies the TED-K cohomology group with an $\infty$-group cohomology of a form analogous
to the above proposal \eqref{TwoGroupCohohomologyWithCoefficientsInAutOfFusionCategory}:
\vspace{-2mm}
\begin{equation}
  \label{InfinityGroupCohomologyOfInternalSymmetryGroup}
  H^{1 + \tilde \tau}_{\mathrm{grp}}
  \big(
    \mathbf{B}
    G_{\mathrm{int}}
    \, ;\,
    \mathrm{Aut}(\mathcal{V})
  \big)
  \;\simeq\;
  \Big\{\!\!\!\!
  \begin{tikzcd}
    G_{\mathrm{int}}
    \ar[
      r,
      dashed
    ]
    &
    \mathbf{B}
    \mathrm{Aut}\big(\mathcal{V}\big)
  \end{tikzcd}
  \!\!\!\!\Big\}_{\!\big/\sim_\homotopy}
  \,.
\end{equation}

\vspace{-2mm}
\noindent On the right we have the
{\it automorphism $\infty$-group} of the underlying
external-SPT phase
\vspace{-1mm}
$$
  [\mathcal{V}]
  \;\in\;
  \mathrm{KR}^{\tau(-,\ast)}
  \big(
    \HomotopyQuotient{\TopologicalSpace}{G_{\mathrm{ext}}}
  \big)
$$

\vspace{-2mm}
\noindent
in the twisted {\it external}-equivariant K-theory cocycle space:
\vspace{-2mm}
\begin{equation}
  \label{AutomorphismInfinityGrouOfExternalPlusCPTProtectedPhase}
  \mathrm{Aut}\big( \mathcal{V} \big)
  \;:=\;
    \Omega_{\mathcal{V}}
    \Bigg(
    \Maps{\bigg}
      { \HomotopyQuotient{ \TopologicalSpace }{G_{\mathrm{ext}}} }
      {
        \HomotopyQuotient
          { \FredholmOperators^0_{\ComplexNumbers} }
          {
            \frac{\UnitaryGroup(\mathscr{H})^2}{\CircleGroup}
            \rtimes \ZTwo^2
          }
      }
   \!\! \Bigg)
  \,,
\end{equation}

\vspace{-2mm}
\noindent
both regarded as sliced over the moduli stack of {\it external}-equivariant twists.
In conclusion we have:
\begin{fact}[\bf Classification of internal-SPT/SET phases]
  \label{ClassificationOfInternalSPTPhases}
  To the extent that topological phases are
  classified by twisted equivariant K-theory (cf. Fact \ref{KTheoryClassificationOfTopologicalPhasesOfMatter}),
  the internal-symmetry protected/enriched such phases
  with underlying external-symmetry protected phase $[\mathcal{V}]$ (as in Fact. \ref{ClassificationOfExternalSPTPhases})
  are classified
  by the $\infty$-group cohomology
  \eqref{InfinityGroupCohomologyOfInternalSymmetryGroup}
  of the given
  internal symmetry group $G_{\mathrm{int}}$
  \eqref{InternalSymmetryGroup}
  with coefficients in the
  automorphism $\infty$-group
  \eqref{AutomorphismInfinityGrouOfExternalPlusCPTProtectedPhase}
  of $\mathcal{V}$
  formed in the TED-K cocycle stack.
\end{fact}

We observe that this should {\it subsume} the proposal of \cite{BPCW14}
that SPT phases are reflected in 2-groupal automorphisms
\eqref{TwoGroupCohohomologyWithCoefficientsInAutOfFusionCategory}
of the fusion category which characterizes the corresponding topological order:
It remains to see how this internal/external-SPT classification in TED-K theory reflects any anyonic topological order.
This, too, turns out to be a question which the mathematics of TED-K theory answers for us  -- we discuss this in \cref{InteractingPhasesAndTEDKOfConfigurationSpaces} below.

\newpage

\section{TED-K classifies interacting topological order}
\label{TEDKDescribesRealisticAnyonSpecies}

We now take the mechanism of K-theory classification of topological phases ``beyond band insulators''
(in the words of \cite{TurnerVishwanath13}) in that we argue that TED-K theory naturally classifies
the generalization of topological insulator phases to:

\cref{BerryPhasesAndDifferentialKTheory}
--
topological 2d semimetal-phases;

\cref{InteractingPhasesAndTEDKOfConfigurationSpaces}
--
interacting topological phases;

\cref{AnyonicTopologicalOrderAndInnerLocalSysyemTEDK}
--
topologically ordered phases.

\subsection{Berry phases and Differential K-theory}
\label{BerryPhasesAndDifferentialKTheory}

\noindent
{\bf Semi-metals and flat K-theory.}
In generalization of the situation of topological insulators (\hyperlink{FigureBandStructure}{\it Figure 3}), it may happen that a quantum material
is not strictly gapped but that the gap closes only over a lower-dimensional sub-manifold of
``nodal points'' in the Brillouin torus (\hyperlink{BandStructureOfSemiMetals}{\it Figure 6}, cf. \cite[Fig. 5.20]{Vanderbilt18}).
In this case one speaks of a topological {\it semi-metal}. These are now often considered as 3-dimensional materials (e.g. \cite{BHB11}\cite{FWDZ16}\cite{AMV18}\cite{GVKR19}) but the concept applies notably also to effectively 2-dimensional materials (general review includes \cite[\S 5]{Vanderbilt18}\cite{FZWY21}), in fact the original and archetypical example of a semi-metal is the effectively 2-dimensional graphene \cite{WZLLJHD12}\footnote{A tiny spin-orbit coupling in graphene de facto opens the gap at the would-be nodal points, making graphene theoretically a topological insulator; but since the gap at the nodal points is too small to be visible in most experiments, graphene behaves like a semi-metal for most practical purposes.} (this was predicted already in \cite{Wallace47}, long before the modern terminology was coined, but observed only after graphene was synthesized by \cite{NGMJZDGF04}).

\vspace{.2cm}
\hypertarget{BandStructureOfSemiMetals}{}

\hspace{-.8cm}
\begin{tabular}{ll}
\begin{minipage}{9.6cm}
{\footnotesize
{\bf Figure 6. -- Band structure of 2d semi-metals.}

{\it Top row:} In a semi-metal there is a global gap between the valence band and the conduction band as for a topological insulator (\hyperlink{FigureBandStructure}{\it Figure 3}) {\it except} over a lower-dimensional submanifold of {\it nodal points} in the Brillouin torus, where the two valence band and the conduction band touch right at the Fermi sea level $\mu_F$.
(If the gap closes over an isolated point then one speaks of a ``Dirac point'' or a ``Weyl point'', depending on the order of degenracy, while if it closes over a 1-dimensional submanifold one speaks of a ``nodal line'', etc.).
This means that over the {\it complement} of the nodal points (shown on the right) the band structure of a semi-metal is like that of a topological insulator (cf. \hyperlink{FigureBandStructure}{\it Figure 3}), while over the nodal points the semi-metal band structure is singular.

\smallskip
\par
{\it Bottom row:}
The singularity at the nodal points in 2-dimensions is reflected in the fact that the Berry curvature is typically tightly concentrated (often spiked) close to the nodal points $k_{{}_I}$ (and un-defined right at the nodal points), such that the complement of a tubular neighborhood of the nodal points in the Brillouin torus carries a Berry connection which is effectively flat.
This is the case notably for the Haldane model (e.g. \cite[Fig. 2.7]{Atteia16}\cite{DTC}).
Moreover, in the special case of materials with time\&space-reflection-symmetric Bloch dynamics the Berry curvature strictly vanishes away from
nodal points, by symmetry reasons (see \cite[\S III.B]{XCN10}\cite[p. 105]{Vanderbilt18} and Ex. \ref{ClassificationOfTISymmetricSemiMetals}), hence is to be thought as spiked right at the nodal points in the sense of a Dirac delta-distribution.
This is the case for graphene (e.g. \cite[(2.72)]{Atteia16}\cite[p. 41]{CayssolFuchs20}).
The general phenomenon is highlighted in
\cite[p. 3]{SodemannFu15}\cite[p. 1]{ZNT21}, exam-

\par
}
\end{minipage}
&
\hspace{.5cm}
\raisebox{-3.2cm}{
\begin{tikzpicture}

\begin{scope}[xshift=-1.9cm]

 \begin{scope}[shift={(0,.535)}]
 \draw
 [line width=9pt]
 (-1,0)
 .. controls (-.3,0) and (-.3,-.3) ..
 (0,-.3)
 .. controls (+.3,-.3) and (+.3,0) ..
 (+1,0);
 \draw
 [line width=8pt, lightgray]
 (-1,0)
 .. controls (-.3,0) and (-.3,-.3) ..
 (0,-.3)
 .. controls (+.3,-.3) and (+.3,0) ..
 (+1,0);
 \end{scope}

\begin{scope}
\clip (-1,0) rectangle (+1,+.23);
\begin{scope}[yscale=.6]
\draw[rotate=45, line width=.5pt,fill=lightgray]
  (0,0) rectangle (1,1);
\end{scope}
\end{scope}

\begin{scope}[yscale=-1]
\clip (-1,0) rectangle (+1,+.23);
\begin{scope}[yscale=.6]
\draw[rotate=45, line width=.5pt,fill=black]
  (0,0) rectangle (1,1);
\end{scope}
\end{scope}

 \begin{scope}[shift={(0,-.535)}]
 \draw
 [line width=9pt]
 (-1,0)
 .. controls (-.3,0) and (-.3,.3) ..
 (0,.3)
 .. controls (+.3,.3) and (+.3,0) ..
 (+1,0);
 \end{scope}

\draw[<-]
  (-1.2,1.3)
  to
    node[very near start,xshift=-5pt]
      {\scalebox{.8}{$
        \mathllap{
          \raisebox{1.5pt}{
          \scalebox{.8}{
            \color{darkblue}
            \begin{tabular}{c}
              Energy
            \end{tabular}
          }
        }
        }
        \hspace{-2pt}
        E
      $}}
  (-1.2,-1.3);

\draw[dashed, orangeii]
  (-1.24, 0) to (1.1,0);

\draw[purple, dashed]
  (0,-1.07)
  to
  (0,1);

\draw (0,-1.5)
  node
  {
    \scalebox{.8}{$
      \underset{
        \raisebox{-1pt}{
          \color{darkblue}
          \scalebox{.7}{
          \def\arraystretch{.85}
          \begin{tabular}{c}
          Dirac/Weyl
          \\
          nodal point
          \end{tabular}
        }
        }
      }{
        k_{{}_I}
      }
    $}
  };

\draw (-1.4,0) node
  {\scalebox{.7}{$\mu_F$}};

\draw[->]
  (-1.4, -1)
    to
    node[very near end]
      {
        \raisebox{-38pt}{
          \scalebox{.8}{$
            \;\;\;\;
            \underset
            {
              \mathrlap{
                \scalebox{.8}{
                  \color{darkblue}
                  \def\arraystretch{.9}
                  \begin{tabular}{c}
                  Brillouin
                  \\
                  torus
                  \end{tabular}
                }
              }
            }
            {
              \DualTorus{2}
            }
          $}
          }
      }
  (1.2, -1);

\end{scope}

\begin{scope}[xshift=1.9cm]

 \begin{scope}[shift={(0,.535)}]
 \draw
 [line width=9pt]
 (-1,0)
 .. controls (-.3,0) and (-.3,-.3) ..
 (0,-.3)
 .. controls (+.3,-.3) and (+.3,0) ..
 (+1,0);
 \draw
 [line width=8pt, lightgray]
 (-1,0)
 .. controls (-.3,0) and (-.3,-.3) ..
 (0,-.3)
 .. controls (+.3,-.3) and (+.3,0) ..
 (+1,0);
 \end{scope}

 \begin{scope}[shift={(0,-.535)}]
 \draw
 [line width=9pt]
 (-1,0)
 .. controls (-.3,0) and (-.3,.3) ..
 (0,.3)
 .. controls (+.3,.3) and (+.3,0) ..
 (+1,0);
 \end{scope}

\draw[line width=5pt]
  (0,.39) to (0,+.09);

\draw[<-]
  (-1.2,1.3)
  to
    node[very near start,xshift=-5pt]
      {\scalebox{.8}{$E$}}
  (-1.2,-1.3);

\draw[dashed, orangeii]
  (-1.24, 0) to (1.1,0);

\draw (-1.4,0) node
  {\scalebox{.7}{$\mu_F$}};

\draw[->]
  (-1.4, -1)
    to
    node[very near end]
      {
        \raisebox{-45pt}{
          \hspace{-13pt}
          \scalebox{.8}{$
            \underset{
              \mathclap{
                \raisebox{-2pt}{
                  \scalebox{.8}{
                    \color{darkblue}
                    \def\arraystretch{.9}
                    \begin{tabular}{c}
                      Complement of
                      \\
                      small neighborhood
                      \\
                      of nodal point
                    \end{tabular}
                  }
                }
              }
            }{
              \DualTorus{2} \setminus \{k_{}{_I}\}
            }
          $}
          }
      }
  (1.2, -1);

\draw[line width=4pt, white]
  (0,1) to (0,-1.1);

\draw[purple, dashed]
  (0,-1.3)
  to
  (0,1);

\draw (-.091,-1)
 node
 {\scalebox{.7}{]}};
\draw (+.091,-1)
 node
 {\scalebox{.7}{[}};

\end{scope}

\draw (0,-2.05) node
  {
    \scalebox{.7}{
      \color{darkblue}
      \bf
      of semi-metal
    }
  };

\begin{scope}[xshift=-1.9cm, yshift=-3.6cm]

\draw[<-]
  (-1.25,1.3)
  to
    node[very near start,xshift=-5pt]
      {\scalebox{.8}{$
        \mathllap{
          \scalebox{.8}{
            \color{darkblue}
            \def\arraystretch{.9}
            \begin{tabular}{c}
              Berry
              \\
              curvature
            \end{tabular}
          }
        }
        \hspace{-4pt}
        \Omega
     $}}
  (-1.2,-1.3);

\draw (0,1.2)
  node
  {
    \scalebox{.6}{
      \color{orangeii}
      \def\arraystretch{.8}
      \begin{tabular}{c}
        curvature
        \\
        spike
      \end{tabular}
    }
  };

\draw[purple, dashed]
  (0,-1.07)
  to
  (0,1)
  ;

\draw (0,-1.3)
  node
  {
    \scalebox{.8}{$
        k_{{}_I}
    $}
  };

\draw[->]
  (-1.4, -1)
    to
    node[very near end]
      {
        \raisebox{-19pt}{
          \scalebox{.8}{
            $\DualTorus{2}$}
          }
      }
  (1.2, -1);

\draw (-1.4, -1)
  node
  {\scalebox{.73}{\colorbox{white}{$0$}}};

\draw
  (-.1, -1)
  .. controls (-.05,-1) and (-.05,1)  ..
  (0,.8);
\begin{scope}[xscale=-1]
\draw
  (-.1, -1)
  .. controls (-.05,-1) and (-.05,1)  ..
  (0,.8);
\draw[draw=gray,fill=white]
  (0,.8) circle (.05);
\end{scope}

\end{scope}

\begin{scope}[xshift=+1.9cm, yshift=-3.6cm]

\draw[<-]
  (-1.25,1.3)
  to
    node[very near start,xshift=-5pt]
      {\scalebox{.8}{$
        \mathllap{
          \scalebox{.8}{
            \color{darkblue}
            \bf
            \def\arraystretch{.9}
            \begin{tabular}{c}
            \end{tabular}
          }
        }
        \hspace{-4pt}
        \Omega
     $}}
  (-1.2,-1.3);

\draw[->]
  (-1.4, -1)
    to
    node[very near end]
      {
        \raisebox{-21pt}{
          \hspace{-14pt}
          \scalebox{.8}{
            $\DualTorus{2} \setminus \{k_{}{_I}\}$}
          }
      }
  (1.2, -1);

\draw (-1.4, -1)
  node
  {\scalebox{.73}{\colorbox{white}{$0$}}};

\draw
  (-.1, -1)
  .. controls (-.05,-1) and (-.05,1)  ..
  (0,.8);
\begin{scope}[xscale=-1]
\draw
  (-.1, -1)
  .. controls (-.05,-1) and (-.05,1)  ..
  (0,.8);
\end{scope}

\draw[line width=5pt, white]
  (0,1)
  to
  (0,-1);

\draw[purple, dashed]
  (0,-1.07)
  to
  (0,1)
  ;

  \draw (-.1,-1)
 node
 {\scalebox{.7}{]}};
\draw (+.1,-1)
 node
 {\scalebox{.7}{[}};

\draw (.75,.4)
  node
  {
    \scalebox{.8}{
      \rotatebox{-30}{
        \scalebox{.9}{
        \color{red}
        \def\arraystretch{.8}
        \begin{tabular}{c}
          \scalebox{.9}{$\sim$}flat
          \\
          Berry
          \\
          connection
        \end{tabular}
      }
      }
    }
  };

\end{scope}

\end{tikzpicture}
}
\end{tabular}
\vspace{-.15cm}

{\footnotesize
\noindent
ples are in
\cite[Fig. 1]{FPGM10}\cite[Fig. 1]{YXL14}\cite[Fig. 3]{SSL15}\cite[Fig. 1]{PRFM16}\cite[\S 3.4]{Atteia16}\cite[Fig. 3g]{WangXuLai17}
\cite{KMM20}\cite[Fig. 3c]{JinEtAl20}\cite[Fig. 4]{JRN21}\footnote{
  This essential vanishing of the Berry curvature away from a neighborhood of the nodal points
  in concrete examples of 2d semi-metals
  seems to not have found a systematic theoretical discussion yet.
  It is a plausible consequence in situations where small effects, such as spin-orbit coupling, perturb the system away from a point
  in parameter space where time\&space-reflection-symmetries enforce the strict vanishing of the Berry curvature away from nodal points.
}.
}

\smallskip
\par

\medskip
Generally, the (hypothetical) adiabatic transport (Rem. \ref{QuantumAdiabaticTheorem}) of gapped valence states \eqref{TheValenceBundle} along closed curves in the Brillouin torus \eqref{BrillouinTorus}
picks up a relative quantum phase factor (\cite{Zak89}\cite[\S 3.4]{Vanderbilt18}, see also \cite{ChangNiu08}\cite[\S I.D]{XCN10}) known as a {\it Berry phase} (\cite{Berry84}, review includes\cite[IV.C]{CayssolFuchs20}). This is the holonomy of a canonical connection $\nabla$ on the valence bundle (\cite{Simon83}\cite[\S 10]{Nakahara03}).
The holonomy of this connection around the non-trivial 1-cycles of the Brillouin torus
(denoted $S^1_a$, $S^1_b$ in \hyperlink{HomotopyTypeOfPuncturedTorus}{\it Figure 7})
is a special case of Berry phase known as the {\it Zak phase} (\cite{Zak89}\cite[pp. 106]{Vanderbilt18}); see also \hyperlink{FlatConnectionOnNodalPuncturedBrillouinTorus}{\it Figure 8}.

\smallskip
  Thinking of the 2-torus as the Brillouin torus of an
  effectively 2-dimensional semi-metal with nodal points at the given punctures
  (see \hyperlink{BandStructureOfSemiMetals}{\it Figure 6}),
  the holonomy of the flat Berry connection (see also \hyperlink{FlatConnectionOnNodalPuncturedBrillouinTorus}{\it Figure 8}) around $S^1_a$ and $S^1_b$
  (\hyperlink{HomotopyTypeOfPuncturedTorus}{\it Figure 7})
  separately gives the {\it Zak phases} while holonomy along non-trivial composites of edges gives the Berry phases associated with the nodal points.

  \smallskip
  For example,
  if both Zak phases happen to be trivial, then,
  in the case of two punctures, the Berry phase around one nodal point is the holonomy along the diagonal edge, and that around the other nodal point is necessarily the inverse of that (since the diagonal edge goes clockwise around one puncture and anti-clockwise around the other).

\vspace{-.7cm}
\begin{center}
\hypertarget{HomotopyTypeOfPuncturedTorus}{}
\begin{equation}
\label{TheHomotopyTypeOfThePuncturedTorus}
\hspace{-.2cm}
\raisebox{-20pt}{
  \begin{tikzpicture}

   \begin{scope}[xshift=-3.3cm]

    \draw (-1.18,-1.18) node{
      \scalebox{.7}{$\mathrm{pt}$}
    };

    \draw (+1.18,-1.18) node{
      \scalebox{.7}{$\mathrm{pt}$}
    };

    \draw (+1.18,+1.18) node{
      \scalebox{.7}{$\mathrm{pt}$}
    };

    \draw (-1.18,+1.18) node{
      \scalebox{.7}{$\mathrm{pt}$}
    };

    \draw[
      fill=lightgray,
      fill opacity=.9,
      draw opacity=0
    ]
      (-1,-1) rectangle (+1,+1);

    \draw[-Latex, orangeii, line width=.6pt]
      (-1,-1) to (-1,0+.1);
    \draw[orangeii, line width=.6pt]
      (-1,0) to (-1,1);
    \draw
      (-1.3,0) node
      {
        \scalebox{.8}{$
          S^1_a
        $}
      };
    \draw
      (+1.3,0) node
      {
        \scalebox{.8}{$
          S^1_a
        $}
      };

    \draw[-Latex, greenii, line width=.6pt]
      (-1,1) to (0+.1,1);
    \draw[greenii, line width=.6pt]
      (0,1) to (1,1);

    \draw
      (0,1.3) node
      {
        \scalebox{.8}{$
          S^1_b
        $}
      };
    \draw
      (-.6,-1.3) node
      {
        \scalebox{.8}{$
          S^1_b
        $}
      };

    \draw[-Latex, orangeii, line width=.6pt]
      (+1,-1) to (+1,0+.1);
    \draw[orangeii, line width=.6pt]
      (+1,0) to (+1,1);

    \draw[-Latex, greenii, line width=.6pt]
      (-1,-1) to (0+.1,-1);
    \draw[greenii, line width=.6pt]
      (0,-1) to (1,-1);

   \draw[fill=darkblue] (-1,-1) circle (.04);
   \draw[fill=darkblue] (-1,+1) circle (.04);
   \draw[fill=darkblue] (+1,-1) circle (.04);
   \draw[fill=darkblue] (+1,+1) circle (.04);

   \draw
     (0,-1.9) node
     {
       \scalebox{.9}{$
         \def\arraystretch{1.2}
         \begin{array}{c}
           \DualTorus{2}
           \\
           \underset{
             \stable
           }{\simeq}
           S^1_a \vee S^1_b \vee S^2_{\mathrm{bulk}}
         \end{array}
       $}
     };

  \end{scope}

   \begin{scope}

    \draw[
      fill=lightgray,
      fill opacity=.9,
      draw opacity=0
    ]
      (-1,-1) rectangle (+1,+1);

    \draw[-Latex, orangeii, line width=.6pt]
      (-1,-1) to (-1,0+.1);
    \draw[orangeii, line width=.6pt]
      (-1,0) to (-1,1);

    \draw[-Latex, greenii, line width=.6pt]
      (-1,1) to (0+.1,1);
    \draw[greenii, line width=.6pt]
      (0,1) to (1,1);

    \draw[-Latex, orangeii, line width=.6pt]
      (+1,-1) to (+1,0+.1);
    \draw[orangeii, line width=.6pt]
      (+1,0) to (+1,1);

    \draw[-Latex, greenii, line width=.6pt]
      (-1,-1) to (0+.1,-1);
    \draw[greenii, line width=.6pt]
      (0,-1) to (1,-1);

   \draw[
     fill=white,
     draw=darkgray,
     draw opacity=.9
   ]
     (0,0) circle (.3);

   \draw[fill=black] (-1,-1) circle (.04);
   \draw[fill=black] (-1,+1) circle (.04);
   \draw[fill=black] (+1,-1) circle (.04);
   \draw[fill=black] (+1,+1) circle (.04);

   \draw
     (0,0)
     node
     {
       \scalebox{.7}{$
         k_1
       $}
     };

   \draw
     (0,-1.9) node
     {
       \scalebox{.9}{$
         \def\arraystretch{1.2}
         \begin{array}{c}
           \DualTorus{2}
           \setminus \{k_1\}
           \\
           \underset{
             \homotopy
           }{\simeq}
           S^1_a \vee S^1_b
         \end{array}
       $}
     };

  \end{scope}

  \begin{scope}[xshift=3.3cm]

    \draw[
      fill=lightgray,
      fill opacity=.9,
      draw opacity=0
    ]
      (-1,-1) rectangle (+1,+1);

    \draw[-Latex, orangeii, line width=.6pt]
      (-1,-1) to (-1,0+.1);
    \draw[orangeii, line width=.6pt]
      (-1,0) to (-1,1);

    \draw[-Latex, greenii, line width=.6pt]
      (-1,1) to (0+.1,1);
    \draw[greenii, line width=.6pt]
      (0,1) to (1,1);

    \draw[-Latex, orangeii, line width=.6pt]
      (+1,-1) to (+1,0+.1);
    \draw[orangeii, line width=.6pt]
      (+1,0) to (+1,1);

    \draw[-Latex, greenii, line width=.6pt]
      (-1,-1) to (0+.1,-1);
    \draw[greenii, line width=.6pt]
      (0,-1) to (1,-1);

    \draw[-Latex, line width=.6pt]
      (-1,-1) to (0+.05,+.05);
    \draw[line width=.6pt]
      (0,0) to (1,1);

   \draw[
     fill=white,
     draw=darkgray,
     draw opacity=.9
   ]
     (0-.4,0+.4) circle (.3);

   \draw[
     fill=white,
     draw=darkgray,
     draw opacity=.8
   ]
     (0+.4,0-.4) circle (.3);

   \draw[fill=black] (-1,-1) circle (.04);
   \draw[fill=black] (-1,+1) circle (.04);
   \draw[fill=black] (+1,-1) circle (.04);
   \draw[fill=black] (+1,+1) circle (.04);

    \draw
     (-.4,+.4)
     node
     {
       \scalebox{.7}{$
         k_1
       $}
     };
    \draw
     (+.4,-.4)
     node
     {
       \scalebox{.7}{$
         k_2
       $}
     };

   \draw
     (0,-1.9) node
     {
       \scalebox{.9}{$
         \def\arraystretch{1.2}
         \begin{array}{c}
           \DualTorus{2}
           \setminus \{k_1, k_2\}
           \\
           \underset{
             \homotopy
           }{\simeq}
           S^1_a \vee S^1_b \vee S^1
         \end{array}
       $}
     };

  \end{scope}

   \begin{scope}[xshift={2*3.3cm}]

    \draw[
      fill=lightgray,
      fill opacity=.9,
      draw opacity=0
    ]
      (-1,-1) rectangle (+1,+1);

    \draw[-Latex, orangeii, line width=.6pt]
      (-1,-1) to (-1,0+.1);
    \draw[orangeii, line width=.6pt]
      (-1,0) to (-1,1);

    \draw[-Latex, greenii, line width=.6pt]
      (-1,1) to (0+.1,1);
    \draw[greenii, line width=.6pt]
      (0,1) to (1,1);

    \draw[-Latex, orangeii, line width=.6pt]
      (+1,-1) to (+1,0+.1);
    \draw[orangeii, line width=.6pt]
      (+1,0) to (+1,1);

    \draw[-Latex, greenii, line width=.6pt]
      (-1,-1) to (0+.1,-1);
    \draw[greenii, line width=.6pt]
      (0,-1) to (1,-1);

    \draw[
      line width=.6pt,
      bend left=34
    ]
      (-1,-1)
      to
      (1,1);
    \draw[
      line width=.6pt,
      bend right=34
    ]
      (-1,-1)
      to
      (1,1);

    \draw[
      -Latex,
      line width=.6pt,
      shift={(-.32,+.32)}
    ]
      (.09,.09) to (.1,.1);
    \draw[
      -Latex,
      line width=.6pt,
      shift={(+.32,-.32)}
    ]
      (.09,.09) to (.1,.1);

   \draw[
     fill=white,
     draw=darkgray,
     draw opacity=.9
   ]
     (0-.6,0+.6) circle (.25);

   \draw[
     fill=white,
     draw=darkgray,
     draw opacity=.8
   ]
     (0+.6,0-.6) circle (.25);

   \draw[
     fill=white,
     draw=darkgray,
     draw opacity=.8
   ]
     (0,0) circle (.25);

   \draw[fill=black] (-1,-1) circle (.04);
   \draw[fill=black] (-1,+1) circle (.04);
   \draw[fill=black] (+1,-1) circle (.04);
   \draw[fill=black] (+1,+1) circle (.04);

    \draw
     (-.6,+.6)
     node
     {
       \scalebox{.7}{$
         k_1
       $}
     };
    \draw
     (+.6,-.6)
     node
     {
       \scalebox{.7}{$
         k_3
       $}
     };
    \draw
     (0,0)
     node
     {
       \scalebox{.7}{$
         k_2
       $}
     };

   \draw
     (0,-1.9) node
     {
       \scalebox{.9}{$
         \def\arraystretch{1.2}
         \begin{array}{c}
           \DualTorus{2}
           \setminus \{k_1, k_2, k_3\}
           \\
           \underset{
             \homotopy
           }{\simeq}
           S^1_a \vee S^1_b \vee S^1 \vee S^1
         \end{array}
       $}
     };

  \end{scope}

   \begin{scope}[xshift={3*3.3cm}]

    \draw[
      fill=lightgray,
      fill opacity=.9,
      draw opacity=0
    ]
      (-1,-1) rectangle (+1,+1);

    \draw[-Latex, orangeii, line width=.6pt]
      (-1,-1) to (-1,0+.1);
    \draw[orangeii, line width=.6pt]
      (-1,0) to (-1,1);

    \draw[-Latex, greenii, line width=.6pt]
      (-1,1) to (0+.1,1);
    \draw[greenii, line width=.6pt]
      (0,1) to (1,1);

    \draw[-Latex, orangeii, line width=.6pt]
      (+1,-1) to (+1,0+.1);
    \draw[orangeii, line width=.6pt]
      (+1,0) to (+1,1);

    \draw[-Latex, greenii, line width=.6pt]
      (-1,-1) to (0+.1,-1);
    \draw[greenii, line width=.6pt]
      (0,-1) to (1,-1);

    \draw[
      line width=.6pt,
      bend left=47
    ]
      (-1,-1)
      to
      (1,1);
    \draw[
      line width=.6pt,
      bend left=3
    ]
      (-1,-1)
      to
      (1,1);

    \draw[
      line width=.6pt,
      bend right=47
    ]
      (-1,-1)
      to
      (1,1);

    \draw[
      -Latex,
      line width=.6pt,
      shift={(-.43,+.43)}
    ]
      (.09,.09) to (.1,.1);
    \draw[
      -Latex,
      line width=.6pt,
      shift={(+.43,-.43)}
    ]
      (.09,.09) to (.1,.1);
    \draw[
      -Latex,
      line width=.6pt,
      shift={(-.03,+.03)}
    ]
      (.09,.09) to (.1,.1);

   \draw[
     fill=white,
     draw=darkgray,
     draw opacity=.9
   ]
     (0-.7,0+.7) circle (.15);

   \draw[
     fill=white,
     draw=darkgray,
     draw opacity=.8
   ]
     (0+.7,0-.7) circle (.15);

   \draw[
     fill=white,
     draw=darkgray,
     draw opacity=.9
   ]
     (0-.23,0+.23) circle (.15);

  \begin{scope}[shift={(-.00,+.00)}]
   \draw[
     fill=white,
     draw=darkgray,
     draw opacity=.9
   ]
     (0+.1,0-.1) circle (.05);

   \draw[
     fill=white,
     draw=darkgray,
     draw opacity=.9
   ]
     (0+.2,0-.2) circle (.05);

   \draw[
     fill=white,
     draw=darkgray,
     draw opacity=.9
   ]
     (0+.3,0-.3) circle (.05);
  \end{scope}

   \draw[fill=black] (-1,-1) circle (.04);
   \draw[fill=black] (-1,+1) circle (.04);
   \draw[fill=black] (+1,-1) circle (.04);
   \draw[fill=black] (+1,+1) circle (.04);

    \draw
     (-.7,+.7)
     node
     {
       \scalebox{.6}{$
         k_1
       $}
     };
    \draw
     (-.23,+.23)
     node
     {
       \scalebox{.6}{$
         k_2
       $}
     };
    \draw
     (+.7,-.7)
     node
     {
       \scalebox{.6}{$
         k_n
       $}
     };

   \draw
     (0,-1.9) node
     {
       \scalebox{.9}{$
         \def\arraystretch{1.2}
         \begin{array}{c}
           \DualTorus{2}
           \setminus \{k_1, \cdots, k_n\}
           \\
           \underset{
             \homotopy
           }{\simeq}
           \bigvee_{1 + n} \; S^1
         \end{array}
       $}
     };

  \end{scope}

  \end{tikzpicture}
}
\end{equation}

\vspace{-.2cm}

\begin{minipage}{17cm}
{\footnotesize
  {\bf Figure 7 -- Homotopy type of the $N$-punctured 2-torus.}
  The 2-torus with $N \geq 1$ punctures is homotopy equivalent to the wedge sum
  of $N + 1$ circles
  (i.e., the gluing at their basepoint, denoted ``$\mathrm{pt}$'' in the graphics,
  see e.g. \cite[\S 71]{Munkres00}\cite[p. 31]{tomDieck08}).
  Two of these circles, denoted $S^1_a$ and $S^1_b$, represent the two non-trivial 1-cycles of the torus.
  (The usual identifications apply, of vertices and edges at the boundary of the square, as shown.)
  The un-punctured torus is still {\it stably} homotopy
  equivalent to
  the wedge sum
  of these two circles
  with a 2-sphere that represents the un-punctured torus bulk (e.g. \cite[Thm. 11.8]{FreedMoore12}, shown on the bottom left above).


}
\end{minipage}

\end{center}

\medskip
While it is intuitively plausible that the Berry phases around nodal points should be indicators of the topological phases of 2d-semimetals,
a full mathematical statement along these lines seems not to have received substantial attention before -- certainly not in a way that would connect to the K-theoretic classification of topological insulators according to the widely accepted Fact \ref{KTheoryClassificationOfTopologicalPhasesOfMatter}, which however ought to be subsumed as a degenrate special case of the classification of semi-metals.
We now observe that a phenomenon well-known in the condensed matter literature strongly suggests that topological phases of 2d semi-metals are controlled by {\it flat} K-theory (we formulate this as Conjecture \ref{FlatKTheoryClassificationOfSemiMetals} below):

\vspace{1mm}
\begin{itemize}[leftmargin=.4cm]

\item  For 2-dimensional semi-metals the Bianchi identity on the curvature 2-form is vacuous by degree-reasons\footnote{
In contrast, for 3-dimensional semi-metals the Bianchi identity $\Differential F_\nabla = 0$ satisified by the curvature 2-form of the Berry connection implies that nodal points behave in momentum space much as Dirac's hypothetical magnetic monopoles would behave in position space.
}, but in practice one observes that the
2-dimensional Berry curvature is still strongly constrained, namely it tends to be concentrated (spiked) on a tubular neighborhood of the nodal points and to be practically vanishing away from the nodal points (e.g. \cite{KMM20}, see \hyperlink{BandStructureOfSemiMetals}{\it Figure 6} for illustration and further references).

\item When the Berry curvature is concentrated at or around the nodal points in this way, it means that on the complement of a tubular neighborhood of the
nodal points the Berry curvature {\it vanishes} for practical purposes, hence that the Berry connection is practically {\it flat} on the complement.

\item  If a complex vector bundle admits a flat connection then all its Chern classes vanish. Now for 3d semi-metals the
first Chern class has often been argued to classify the topological charge carried by the nodal points (e.g. \cite[(3)]{LieEtAl20}\cite{MathaiThiang15Semimetals} \cite{MathaiThiang16}), but for 2d semi-metals Chern classes cannot carry any information,
in that the ordinary cohomology of the punctured torus vanishes in degree 2 and higher (see \hyperlink{HomotopyTypeOfPuncturedTorus}{\it Figure 7}).

\item However, the datum of a flat connection itself entails ``secondary'' characteristic classes (e.g. \cite[\S 4.3]{FSS20Character}), which for 2-dimensional
semi-metals are the holonomies of the flat connection along any loop around nodal points. Due to the flatness of the connection
these {\it Berry phases around nodal points} are independent of the shape of the loop and must be addressed as the topological charge carried by the nodal points, reflecting the obstruction to the semi-metal phase turning into an insulating phase.

\end{itemize}


\vspace{-1mm}
\begin{center}
\hypertarget{FlatConnectionOnNodalPuncturedBrillouinTorus}{}
\begin{tabular}{ll}
\begin{minipage}{8.2cm}
  \footnotesize
  {\bf Figure 8.}
  The flat Berry connection on the valence bundle over the complement of (a tubular neighbourhood of) the nodal points in the Brillouin torus
  of a 2d semi-metal
  (cf. \hyperlink{BandStructureOfSemiMetals}{\it Figure 6})
  has two contributions:

\begin{itemize}[leftmargin=.5cm]

\item[{\bf (1)}]
  The holonomy around 1-cycles in the full Brillouin torus
  (denoted $S^1_a$, $S^1_b$ in \hyperlink{HomotopyTypeOfPuncturedTorus}{\it Figure 7})
  gives the Zak phases that may be present also in topological insulators (i.e. in the absence of nodal points).

\item[{\bf (2)}]
  The holonomy around nodal points reflects their ``momentum-space charge'', the obstruction
  to opening up the gap closure,
  measuring how far the topological semi-metal phase is from being a topological insulator.
\end{itemize}

\end{minipage}
&
$
  \begin{tikzcd}[row sep=small, column sep=large]
    \mbox{
      \tiny
      \color{darkblue}
      \bf
      \begin{tabular}{c}
        Full
        \\
        Brillouin torus
      \end{tabular}
    }
    &[-34pt]
    \DualTorus{2}
    \ar[
      drr,
      "{
        \mbox{
          \tiny
          \color{greenii}
          \bf
          Zak phases
        }
      }"{sloped}
    ]
    &[+10pt]
    \\
    \mbox{
      \tiny
      \color{darkblue}
      \bf
      \begin{tabular}{c}
        Node-punctured
        \\
        Brillouin torus
      \end{tabular}
    }
    &
    \DualTorus{2} \setminus \{\vec k\}
    \ar[u, hook]
    \ar[
      rr,
      "{
        \mbox{
          \tiny
          \color{greenii}
          \bf
          flat
          Berry phases
        }
      }"{pos=.4}
    ]
    &&
    \underset{
      \mathclap{
      \raisebox{-2pt}{
        \tiny
        \color{darkblue}
        \bf
        \def\arraystretch{.9}
        \begin{tabular}{c}
          coefficients for
          \\
          flat Berry connections.
        \end{tabular}
      }
      }
    }{
      \mathbf{B} \flat \UnitaryGroup(n)
    }
    \;.
    \\
    \mbox{
      \tiny
      \color{darkblue}
      \bf
      \begin{tabular}{c}
        Open neighborhood of
        \\
        the nodal punctures
      \end{tabular}
    }
    &
    \ComplexPlane \setminus \{\vec k\}
    \ar[
      u,
      hook
    ]
    \ar[
      urr,
      shorten <=-4pt,
      "{
        \mbox{
         \tiny
         \color{greenii}
         \bf
         \def\arraystretch{.9}
         \begin{tabular}{c}
           Berry phases around
           \\
           nodal points
         \end{tabular}
        }
      }"{swap, sloped, pos=.4}
    ]
  \end{tikzcd}
$
\end{tabular}
\end{center}

In view of Fact \ref{KTheoryClassificationOfTopologicalPhasesOfMatter}, this suggests that phases of 2-dimensional topological semi-metals
ought to be classified by a version of K-theory (evaluated on the complement of the nodal points) which is appropriate for flat
vector bundles, or more generally for vector bundles equipped with a trivialization of their Chern classes.
The evident candidate cohomology theory is the version of differential K-theory (e.g. \cite{BunkeSchick11}\cite[Ex. 4.41]{FSS20Character})
known as {\it flat K-theory}, whose cohomology groups are characterized as
arranging into a
hexagonal commuting diagram (\cite[\S 2]{SimonsSullivan08}\cite[\S 6]{BNV13}), the lower part of which looks as follows,
where the bottom sequence of groups is long exact (\cite[\S 7.21]{Karoubi87}\cite[Ex. 3]{Karoubi90}\cite[(16)]{Lott94}):

\vspace{0mm}

\begin{equation}
\label{TheDifferentialCohomologyHexagonForKTheory}
\adjustbox{raise=-2cm}{
\begin{tikzpicture}
  \clip (-6.5,-1.9) rectangle (6.1,1.1);
  \draw (0,0) node {
  \begin{tikzcd}[column sep=15pt]
    &
    {}
    \ar[rr, lightgray]
    \ar[
      dr,
      gray,
      shorten >=-0
      ]
    {}
    &&
    {}
    \ar[
      dr,
      gray,
      shorten >=-0pt
    ]
    &
    \\
    \mathrm{KU}
      ^{\NumberOfProbeBranes -2}
      (\SmoothManifold;\, \ComplexNumbers)
  \ar[
    dr,
    shorten <=-0pt,
    "{
        \mbox{
          \tiny
          \color{greenii}
          \bf
          \def\arraystretch{.9}
          \begin{tabular}{c}
            secondary
            \\
            Chern character
          \end{tabular}
        }
    }"{swap, sloped, pos=.35}
  ]
  \ar[
    ur,
    gray
  ]
  \ar[
    rr,
    "{
      \mbox{
        \tiny
        \color{greenii}
        \bf
        \def\arraystretch{.9}
        \begin{tabular}{c}
        \end{tabular}
      }
    }"{pos=.4}
  ]
  &&
    \overset{
      {
      \mathclap{
      \raisebox{1pt}{
        \tiny
        \color{darkblue}
        \bf
        \def\arraystratch{.9}
        \begin{tabular}{c}
        Differential K-theory
        \end{tabular}
      }
      }
      }
    }{
      \mathrm{KU}^{
        \NumberOfProbeBranes
          -
        1
      }
      _{\mathrm{diff}}
      \big(
        \SmoothManifold
      \big)
    }
    \ar[
      rr,
      "{
        \mbox{
          \tiny
          \color{greenii}
          \bf
          \def\arraystretch{.9}
          \begin{tabular}{c}
          \end{tabular}
        }
      }"{pos=.6}
    ]
    \ar[
      ur,
      gray,
      shorten <=-0
    ]
    \ar[
      dr,
      "{
        \mbox{
          \tiny
          \color{greenii}
          \bf
          forget
        }
      }"{swap, sloped},
      shorten <=-0
    ]
    &&
    \mathrm{KU}^{\NumberOfProbeBranes - 1}
    (\SmoothManifold;\, \ComplexNumbers)\;.
    \\
    &
    \underset{
      \raisebox{-1pt}{
        \tiny
        \color{darkblue}
        \bf
        Flat K-theory
      }
    }{
    \mathrm{KU}
      ^{
        \NumberOfProbeBranes
          -
        1
      }
      _{\flat}
      (
        \SmoothManifold
      )
    }
    \ar[
      ur,
      shorten >=-0pt,
      "{
        \mbox{
          \tiny
          \color{greenii}
          \bf
          include
        }
      }"{swap, sloped}
    ]
    \ar[
      rr,
      "{
        \mbox{
          \tiny
          \color{greenii}
          \bf
          \def\arraystretch{.9}
          \begin{tabular}{c}
          \end{tabular}
        }
      }"{swap, yshift=-0pt, pos=.4}
    ]
    &&
    \underset{
      \mathclap{
      \raisebox{-2pt}{
        \tiny
        \color{darkblue}
        \bf
        \begin{tabular}{c}
          Underlying
          plain K-theory
        \end{tabular}
      }
      }
    }{
    \mathrm{KU}
      ^{
        \NumberOfProbeBranes
          -
        1
      }
    (
      \SmoothManifold
    )
    }
        \ar[
      ur,
      shorten >=-0pt,
      "{
        \mbox{
          \tiny
          \color{greenii}
          \bf
          \def\arraystretch{.9}
          \begin{tabular}{c}
            Chern
            \\
            character
          \end{tabular}
        }
      }"{sloped, swap, pos=.5}
    ]
  \end{tikzcd}
  };
  \end{tikzpicture}
  }
\end{equation}

\medskip


In analogy with Fact \ref{KTheoryClassificationOfTopologicalPhasesOfMatter}, this
leads us to state:
\begin{conjecture}[\bf Flat TED-K-theory classification of semi-metals]
  \label{FlatKTheoryClassificationOfSemiMetals}
  The deformation classes of
  2-dimensional crystalline semi-metals are classified by the
  (twisted equivariant)
  flat
  K-cohomology class of their Berry-flat valence bundle over the complement of (a tubular neighborhood of) the nodal points in the material's (orbifolded) Brillouin torus.
  Here the restriction of the (twisted equivariant) flat K-cohomology to small circles $S^1_{I}$ surrounding the $I$-th nodal points is identified with the class of the Berry phase around that circle, hence with the topological charge carried by that nodal point:
  \vspace{-2mm}
  \begin{equation}
    \label{RestrictingTEDKTOCirclesAroundNodalPoints}
    \begin{tikzcd}[row sep=2pt]
      \HomotopyQuotient
        {
          \big(
            \DualTorus{2}
            \setminus
            \{\vec k\}
          \big)
        }
        { G }
    \ar[
      from=rr,
      hook',
      "{ \iota_{{}_I} }"{swap}
    ]
    &&
      \HomotopyQuotient
        {
          S^1
        }
        { G }
    \mathrlap{
      \mbox{
        \tiny
        \color{darkblue}
        \bf
        \def\arraystretch{.9}
        \begin{tabular}{c}
          Neighborhood around
          \\
          $I$-th nodal point
        \end{tabular}
      }
    }
    \\
    \mathllap{
      \mbox{
        \tiny
        \color{darkblue}
        \bf
        \begin{tabular}{c}
          Twisted equivariant {\it flat} K-theory
          \\
          of nodal-punctured Brillouin torus
        \end{tabular}
      }
    }
    \mathrm{K}^\tau_{\flat}
    \big(
      \HomotopyQuotient
        {
         (
            \DualTorus{2}
            \setminus
            \{\vec k\}
          )
        }
        { G }
    \big)
    \ar[
      rr,
      "{
        \iota_{{}_I}^\ast
      }"
    ]
    &&
    \mathrm{K}^{i^\ast\tau}_{\flat}
    \big(
      \HomotopyQuotient
        {
          S^1
        }
        { G }
    \big)
    \mathrlap{
      \mbox{
        \tiny
        \color{darkblue}
        \bf
        \def\arraystretch{.9}
        \begin{tabular}{c}
          Group of
          \\
          Berry phases around/
          \\
          topological charges of
          \\
          $I$th nodal point
        \end{tabular}
      }
    }
    \end{tikzcd}
  \end{equation}

  \vspace{-3mm}
\noindent  and similarly the restriction to $S^1_a$ and $S^1_b$ is identified with the deformation classes of the corresponding Zak phase (\hyperlink{FlatConnectionOnNodalPuncturedBrillouinTorus}{\it Figure 8}).
\end{conjecture}

\begin{remark}[Comparison to the literature]
  A comprehensive proposal for the classification of topological semi-metals aligned with the K-theoretic classification of topological insulators (Fact \ref{KTheoryClassificationOfTopologicalPhasesOfMatter}) has not been available in the literature. We are aware of the following partial suggestions:

\begin{itemize}

\item[\bf (i)]
\cite{YangNagaosa14} means to classify the quantum symmetries which may fix nodal points -- this is subsumed in Conjecture \ref{FlatKTheoryClassificationOfSemiMetals} by employing {\it equivariant} K-theory.

\item[{\bf (ii) }]
\cite{Schnyder18}\cite{Schnyder20} suggests
(following \cite{MorimotoFurusaki13}\cite{CTSR16}) that topologically protected band crossings are those for which the linearized Hamiltonian in the vicinity of a nodal point does {\it not} admit a mass term (whose addition would open a mass gap), hence for
which the Clifford module structure (as in \hyperlink{TableCliffordActions}{\it Table 3}.) does not extend to one further Clifford generator. It seems to be left open how this local argument is meant to classify the semi-metal globally. But since addition of
a Clifford generator has the effect of shifting the K-theory degree (\hyperlink{TableCliffordActions}{\it Table 3}), one may think of the left term
$\mathrm{K}^{\bullet -1 }(\cdots)$
in \eqref{SemiMetalExactSequence} as implementing the quotient by globally defined mass terms, in this sense.

\item[\bf (iii)]
\cite{MathaiThiang15Semimetals}\cite{MathaiThiang16}  propose global classification of semi-metals by the {\it ordinary cohomology} of the
complement of the nodal points, following tradition
for 3d semi-metals in solid state physics (e.g. \cite[(3)]{LieEtAl20}) and suggesting that this might also appy to 2d semi-metals.
It seems clear that this proposal needs to be refined form ordinary cohomology to K-theory in order to connect to the classification of topological insulators by Fact \ref{KTheoryClassificationOfTopologicalPhasesOfMatter}, and we may think of Conjecture \ref{FlatKTheoryClassificationOfSemiMetals} as providing this refinement.

\end{itemize}

\end{remark}

\begin{example}[{\bf Classification of $TI$-symmetric 2d semi-metals}]
  \label{ClassificationOfTISymmetricSemiMetals}
  We check that the statement of Conjecture \ref{FlatKTheoryClassificationOfSemiMetals}  reproduces the expectation in the literature (e.g. \cite{ZZLXZZ16}\cite{FWDZ16} \cite[\S 5]{Vanderbilt18}),
  for the case of 2d semi-metals subject to joint time and inversion symmetry $T I$ (from Ex. \ref{JointTimeReversalAndInversionSymmetry}):
    By the discussion around
  \eqref{KOTheoryAsTWistedKRTheory}, the relevant twisted equivariant KR-theory in this case is $\mathrm{KO}^0$
  (assuming that the system's excitations are bosonic, namely that $\widehat{T}^2 = + 1$, see
  \eqref{SignChoicesForTimeReversalQuantumSymmetry}), and the plain $\mathrm{KO}^0$ groups of the $N$-punctured
  torus ($N \geq 1$) are a direct sum of copies of $\ZTwo$ and hence pure torsion:
  \vspace{-1mm}
  $$
    \mathrm{KO}^0
    \Big(
      \DualTorus{2}
      \setminus
      \{k_1, \cdots, k_N\}
    \Big)
    \;\underset{
      \raisebox{-2pt}{
      \scalebox{.7}{
        \eqref{TheHomotopyTypeOfThePuncturedTorus}
      }
      }
    }{\simeq}\;
    \mathrm{KO}^0
    \Big(
      {\bigvee}_{N + 1}
      \;
      S^1
    \Big)
    \;\simeq\;
    \underset{
      N + 1
    }{\bigoplus}
    \;
    \mathrm{KO}^0
    \big(
      S^1
    \big)
    \;\underset{
      \raisebox{-2pt}{
      \scalebox{.7}{
        \eqref{KOTheoryAsTWistedKRTheory}
      }
      }
    }{\simeq}\;
    \underset{
      N + 1
    }{\bigoplus}
    \;
    \ZTwo
    \,.
  $$

  \vspace{-2mm}
\noindent
This implies that the Chern(-Pontrjagin) character in this case is trivial (in fact, even its (co-)domain is trivial), so that the
  relevant long exact sequence \eqref{TheDifferentialCohomologyHexagonForKTheory}
  implies that the flat $\mathrm{KO}^0$-cohomology of the $N$-punctured torus canonically coincides with the plain $\mathrm{KO}^0$-cohomology:
  \vspace{-3mm}
  \begin{equation}
    \label{FlatKOZeroOfPuncturedTorus}
    \begin{tikzcd}[
      column sep=14pt,
      row sep=-1pt
    ]
    \mathrm{KO}^{-1}
      \big(
        \DualTorus{2}
        \setminus \{\vec k\}
        ;\,
        \RealNumbers
      \big)
    \ar[rr, "{\mathrm{ch}^{-1}_{{}_{\mathrm{KO}}}}"]
    &&
      \mathrm{KO}^0_\flat
      \big(
        \DualTorus{2}
        \setminus \{\vec k\}
      \big)
      \ar[rr]
      &&
      \mathrm{KO}^0
      \big(
        \DualTorus{2}
        \setminus \{\vec k\}
      \big)
      \ar[rr, "{\mathrm{ch}_{{}_{\mathrm{KO}}}}"]
      &&
      \mathrm{KO}^0
      \big(
        \DualTorus{2}
        \setminus \{\vec k\}
        ;\,
        \RealNumbers
      \big)
      \\
      0
      \ar[rr]
      &&
      \color{darkblue}
      \underset{
        \{\vec k\}
      }{\bigoplus}
      \,
      \ZTwo
      \ar[rr, "{\sim}"]
      &&
      \underset{
        \{\vec k\}
      }{\bigoplus}
      \,
      \ZTwo
      \ar[rr]
      &&
      0
      \,.
    \end{tikzcd}
  \end{equation}

  \vspace{-4mm}
   \noindent
  This may be understood as saying that gapped and $T I$-equivariant (namely real) valence bundles on $\DualTorus{2}\setminus \{\vec k\}$ carry, up to adiabatic and $T I$-equivariant deformation equivalence, a {\it unique} flat Berry connection, whose holonomy -- hence whose two Zak phases and $N$ Berry phases around the nodal points according to \eqref{RestrictingTEDKTOCirclesAroundNodalPoints} -- all take values in $\ZTwo$.
    Hence, in this case, Conjecture \ref{FlatKTheoryClassificationOfSemiMetals}
   asserts that these values are the topological charges carried by the nodal points, whose joint image in the flat TED-K-theory group of the full punctured Brillouin torus is the class of the given topological phase of the TI-symmetric semimetal.
  But this is just the statement commonly expected in the literature, see \cite[\S II.A]{ZZLXZZ16}\cite[\S II.B]{FWDZ16} \cite[5.5.2]{Vanderbilt18}.

  \medskip
  Specifically, notice that
  \cite[(12)-(13)]{FWDZ16} identify
  (by arguing informally about the Bloch Hamiltonians of the semi-metals, cf. Rem. \ref{ComparisonToExistingLiteratureForKClassification})
  the group $\ZTwo$ appearing here with the fundamental group of the orthogonal Grassmannian, hence with the fundamental group of the stable classifying space $B \mathrm{O}$ (e.g. \cite[\S 1.3]{Kochman96}). Since this is equivalently the classifying space for (reduced) $\mathrm{KO}^0$ (e.g. \cite[p. 77, 86]{Kochman96}, cf. \cite[(6)]{ZZLXZZ16}), the $\ZTwo$-charges predicted by our
  Conjecture \eqref{FlatKTheoryClassificationOfSemiMetals} coincide with those traditionally expected not just abstractly (coincidentally), but {\it naturally} (operationally):
    \begin{align*}
    \mbox{
      \tiny
      \color{darkblue}
      \bf
      \bf
      \begin{tabular}{c}
        Nodal charge group
        \\
        according to \cite[\S II.B]{FWDZ16}
      \end{tabular}
    }
    \ZTwo
    \;\;\;\; &\simeq   \;\;\;
    \pi_1( B \OrthogonalGroup)
    \\
    \;& \simeq  \;\;\;
    \PointedMaps{}
      { S^1 }
      { B \OrthogonalGroup }
      \\
    \; & \simeq   \;\;\;
    \PointedMaps{}
      { S^1 }
      { \mathrm{KO}_0 }
      \\
    \; & \simeq  \;\;\;
    \mathrm{KO}^0
    (
      S^1
    )
    \;
    \\
    & \!\!\!\!
     \underset{
      \raisebox{-1pt}{
        \scalebox{.7}{
          \eqref{FlatKOZeroOfPuncturedTorus}
        }
      }
    }{
      \simeq
    }
             \mathrm{KO}_{\flat}^0
    \big(
      S^1
    \big)
    \;
    \simeq
    \;
    \ZTwo
    \mbox{
      \tiny
      \color{darkblue}
      \bf
      \bf
      \begin{tabular}{c}
        Nodal charge group
        \\
        according to Conj. \ref{FlatKTheoryClassificationOfSemiMetals}
      \end{tabular}
    }
    \end{align*}

  \vspace{-1mm}
\noindent  This proves our Conjecture \ref{FlatKTheoryClassificationOfSemiMetals} for the case of $TI$-symmetric 2d materials,
relative to the established understanding of the relevant solid state physics. Of course, this $TI$-symmetric example is a
comparatively simple special case of the conjecture, due to the isomorphism \eqref{FlatKOZeroOfPuncturedTorus}; the general
case will be richer. On the other hand, the $T I$-symmetric case seems to be the only case for which classification of
codimension=2 nodal loci in semi-metals has been discussed in the literature (as just recalled). With this case verified,
for all other cases our conjecture is now a prediction about general 2d semi-metals, which deserves to be checked in
theory and experiment.
\end{example}

\noindent
{\bf Mass terms at Dirac/Weyl points.}
In the absence of any protecting/enriching symmetry, a topological semi-metal-phase is called a {\it Chern semi-metal}-phase (Ex. \ref{ClassificationOfTwoDChernSemiMetals} below), for the same reason as for {\it Chern insulators} (Ex. \ref{NoQuantumSymmetryAndSpinOrbitCOupling}).

\begin{itemize}[leftmargin=.6cm]
    \item[{\bf (1)}] In the case of 3d semi-metals with nodal points, the first Chern class of such a Chern insulator may have non-trivial evaluation on small spheres surrounding these points, which is then naturally interpreted as the topological charge carried by these nodal points. This situation of codimension=3 nodal points in 3d Chern semi-metals has found much attention in the literature (\cite{MathaiThiang15Semimetals}\cite{MathaiThiang16}).

\item[{\bf(2)}] However, for nodal points in 2d semi-metals (and nodal lines in 3d semi-metals) the Chern classes do not provide non-trivial invariants (in fact the entire 2-cohomology of the punctured Brillouin torus vanishes, as shown in \hyperlink{HomotopyTypeOfPuncturedTorus}{Figure 7}).
In this case, an alternative proposal
(see Rem. \ref{LiteratureOnMassTermGapOpenings} below)
for how to classify the topological stability of codimension=2 nodal points is by the classification of {\it mass terms}.
Namely, by the assumption that nodal points $k_I$ lie right at the chemical potential, hence at the reference null value of the energy, $E(k_I) - \mu_{F} = 0$, and using that the dispersion relation tends to be non-vanishing at the nodal point, $\frac{d E}{d k}(k_I) \neq 0$ (see \hyperlink{MassTermsOpeningUpBandNodes}{\it Figure 9}), it is traditionally argued (\cite[\S II.A]{Schnyder18}\cite[\S 2.1]{Schnyder20}) that:

\begin{itemize}
[leftmargin=.4cm]

\item[{$\bullet$}] The dynamics of quasi-particle excitations around $\mu_F$ at momenta $k = k_I + \Delta k$ around the nodal points is to first order in $\Delta k$ given by an effective {\it massless Dirac equation}. Depending on whether this equation describes Dirac fermions (4 components) or Weyl fermions (2 components), one calls the nodal point a {\it Dirac point} or {\it Weyl point}, respectively.

\item[{$\bullet$}] The obstruction to opening the gap at the nodal crossing is the im-possibility to deform this equation by a {\it mass term}, which for a Dirac equation means
(cf. \cite[Lem. 9.55]{FreedHopkins16})
to find a further Clifford generator that skew-commutes with those already involved.
\end{itemize}
\end{itemize}

\vspace{-.4cm}

\begin{center}
\hypertarget{MassTermsOpeningUpBandNodes}{}
\begin{tabular}{ll}
\hspace{-2mm}
\begin{minipage}{4cm}
\footnotesize
{\bf Figure 9.}
Near nodal points $k_{{}_I}$ in the Brillouin torus of a semi-metal, the dispersion
relation $k \mapsto E(k)$ exhibited by the energy bands is thought to be approximated, to first order in $k - k_I$, by that of a massless Dirac/Weyl equation, whence the terminology {\it Dirac point} or {\it Weyl point}.



\end{minipage}

&
\hspace{.5cm}
\adjustbox{raise=-70pt}{
\begin{tikzpicture}

\begin{scope}[xshift=-1.9cm]

\draw (-.4,2)
 node
 {
   \scalebox{.7}{
     \def\arraystretch{1.2}
     \begin{tabular}{c}
     Massless Dirac-type
     \\
     dispersion relation
     \\
          $E(k) - \mu_F
     \;=\;
     \pm
     \;
     \sqrt{
       \Big\langle
         (k\!\!\!\!/ - k\!\!\!\!/_I)^2
       \Big\rangle
     }$
     \\
     \phantom{.}
     \end{tabular}
   }
 };

 \begin{scope}[shift={(0,.535)}]
 \draw
 [line width=9pt, white]
 (-1,0)
 .. controls (-.3,0) and (-.3,-.3) ..
 (0,-.3)
 .. controls (+.3,-.3) and (+.3,0) ..
 (+1,0);
 \draw
 [line width=8pt, white]
 (-1,0)
 .. controls (-.3,0) and (-.3,-.3) ..
 (0,-.3)
 .. controls (+.3,-.3) and (+.3,0) ..
 (+1,0);
 \end{scope}

 \begin{scope}[shift={(0,-.535)}]
 \draw
 [line width=9pt, white]
 (-1,0)
 .. controls (-.3,0) and (-.3,.3) ..
 (0,.3)
 .. controls (+.3,.3) and (+.3,0) ..
 (+1,0);
 \end{scope}

\draw[<-]
  (-1.2,1.3)
  to
    node[very near start,xshift=-5pt]
      {\scalebox{.8}{$
        \mathllap{
          \raisebox{1.5pt}{
          \scalebox{.8}{
            \color{darkblue} \bf
            \begin{tabular}{c}
              Energy
            \end{tabular}
          }
        }
        }
        \hspace{-2pt}
        E
      $}}
  (-1.2,-1.3);

\draw[dashed, orangeii]
  (-1.24, 0) to (1.1,0);

\draw (0,-1.5)
  node
  {
    \scalebox{.8}{$
      \underset{
        \raisebox{-1pt}{
          \color{darkblue} \bf
          \scalebox{.7}{
          \def\arraystretch{.85}
          \begin{tabular}{c}
          Dirac/Weyl
          \\
          nodal point
          \end{tabular}
        }
        }
      }{
        k_{{}_I}
      }
    $}
  };

\draw (-1.45,0) node
  {\scalebox{.7}{$\mu_F$}};

\draw[->]
  (-1.4, -1)
    to
    node[very near end]
      {
        \raisebox{-38pt}{
          \scalebox{.8}{$
            \;\;\;\;
            \underset
            {
              \mathrlap{
                \scalebox{.8}{
                  \color{darkblue} \bf
                  \def\arraystretch{.9}
                  \begin{tabular}{c}
                  Brillouin
                  \\
                  torus
                  \end{tabular}
                }
              }
            }
            {
              \DualTorus{2}
            }
          $}
          }
      }
  (1.2, -1);

\begin{scope}
\clip (-1,0) rectangle (+1,+.43);
\begin{scope}[yscale=.6]
\draw[rotate=45, line width=.5pt,fill=white]
  (0,0) rectangle (1,1);
\end{scope}
\end{scope}

\begin{scope}[yscale=-1]
\clip (-1,0) rectangle (+1,+.43);
\begin{scope}[yscale=.6]
\draw[rotate=45, line width=.5pt,fill=white]
  (0,0) rectangle (1,1);
\end{scope}
\end{scope}

\draw[purple, dashed]
  (0,-1.07)
  to
  (0,1);

\end{scope}

\begin{scope}[xshift=+2.8cm]

\draw (-.4,2.1)
 node
 {
   \scalebox{.7}{
     \def\arraystretch{1.2}
     \begin{tabular}{c}
     Massive Dirac-type \\
     dispersion relation
     \\
          $E(k) - \mu_F
     \;=\;
     \pm
     \;
     \sqrt{
       \Big\langle
         (
           k\!\!\!\!/
           -
           k\!\!\!\!/_I
           +
           m \gamma_{{}_0}
         )^2
       \Big\rangle
     }$
     \\
     \phantom{.}
     \end{tabular}
   }
 };

 \draw
  (.9,1.7) node{
  \scalebox{.7}{
    $
      \underset{
        \mathclap{
          \scalebox{1}{
            \color{purple}
            mass term
          }
        }
      }{
      \underbrace{
        \phantom{m \gamma_{{}_0}}
      }
      }
    $
  }
  };

 \begin{scope}[shift={(0,.535)}]
 \draw
 [line width=9pt, white]
 (-1,0)
 .. controls (-.3,0) and (-.3,-.3) ..
 (0,-.3)
 .. controls (+.3,-.3) and (+.3,0) ..
 (+1,0);
 \draw
 [line width=8pt, white]
 (-1,0)
 .. controls (-.3,0) and (-.3,-.3) ..
 (0,-.3)
 .. controls (+.3,-.3) and (+.3,0) ..
 (+1,0);
 \end{scope}

 \begin{scope}[shift={(-.5,-.535)}]
 \draw
 [line width=9pt, white]
 (-1,0)
 .. controls (-.3,0) and (-.3,.3) ..
 (0,.3)
 .. controls (+.3,.3) and (+.3,0) ..
 (+1,0);
 \end{scope}

\draw[<-]
  (-1.2,1.3)
  to
    node[very near start,xshift=-5pt]
      {\scalebox{.8}{$
        \mathllap{
          \raisebox{1.5pt}{
          \scalebox{.8}{
            \color{darkblue} \bf
            \begin{tabular}{c}
              Energy
            \end{tabular}
          }
        }
        }
        \hspace{-2pt}
        E
      $}}
  (-1.2,-1.3);

\draw (0,-1.5)
  node
  {
    \scalebox{.8}{$
      \underset{
        \raisebox{-1pt}{
          \color{darkblue} \bf
          \scalebox{.7}{
          \def\arraystretch{.85}
          \begin{tabular}{c}
            Resolved
            \\
            nodal point
          \end{tabular}
        }
        }
      }{
        k_{{}_I}
      }
    $}
  };

\draw (-1.7,.15) node
  {\scalebox{.7}{$\mu_F + m$}};
\draw (-1.7,-.15) node
  {\scalebox{.7}{$\mu_F - m$}};

\draw[purple, dashed]
  (0,-1.07)
  to
  (0,1);

\draw (0,0) node
{\scalebox{.6}{
  \colorbox{white}{
  \color{purple}
  gap opening
  }
}};

\draw[dashed, orangeii]
  (-1.24, +.14) to (1.1,+.14);
\draw[dashed, orangeii]
  (-1.24, -.14) to (1.1,-.14);

\draw[->]
  (-1.4, -1)
    to
    node[very near end]
      {
        \raisebox{-38pt}{
          \scalebox{.8}{$
            \;\;\;\;
            \underset
            {
              \mathrlap{
                \scalebox{.8}{
                  \color{darkblue} \bf
                  \def\arraystretch{.9}
                  \begin{tabular}{c}
                  Brillouin
                  \\
                  torus
                  \end{tabular}
                }
              }
            }
            {
              \DualTorus{2}
            }
          $}
          }
      }
  (1.2, -1);

\begin{scope}[yshift=+1pt]
\clip (-1,0) rectangle (+1,+.43);
\begin{scope}[yscale=.6]
\draw[rotate=45, line width=.5pt]
  (0,1)
    .. controls
    (0,0) and (0,0) ..
  (1,0);
\end{scope}
\end{scope}

\begin{scope}[yscale=-1, yshift=+1pt]
\clip (-1,0) rectangle (+1,+.43);
\begin{scope}[yscale=.6]
\draw[rotate=45, line width=.5pt]
  (0,1)
    .. controls
    (0,0) and (0,0) ..
  (1,0);
\end{scope}
\end{scope}

\end{scope}

\end{tikzpicture}
}

\end{tabular}

\end{center}

When the material's parameters can be and are adiabatically tuned (Rem. \ref{QuantumAdiabaticTheorem}) such that this dispersion relation turns into a {\it massive}
Dirac equation, then the band gap at the former nodal crossing will ``open up'' in proportion to the coefficient $m$ of the effective
{\it mass term}
$m \gamma_{{}_0}$ in the Dirac equation. If this happens to all nodal points (while keeping the band gap open everywhere else),
it means that the topological semi-metal phase decays into a topological insulator-phase.
A necessary condition for such a {\it mass term} to exist at all is that a further Clifford generator $\gamma_{{}_0}$ is represented on the Bloch-Hilbert space of electrons such that it skew-commutes with all Clifford momenta $k\!\!\!\!/$.

\smallskip
By Karoubi/Atiyah-Singer-type theorems (Prop. \ref{KaroubiTheorem}), such mass terms are typically again classified by topological K-theory groups.
Note that, interestingly, this means that the topological {\it semi-metal}-phase must be classified in a group {\it modulo}
the K-group of mass terms.

\begin{remark}[Literature on mass term gap openings]
\label{LiteratureOnMassTermGapOpenings}
In attempts towards the classification of topological semi-metal phases,
it has been argued
(\cite[\S A.2]{ChiuSchnyder14},
reviewed in \cite[\S II.A]{Schnyder18}\cite[\S 2.1]{Schnyder20} and
following \cite[\S V]{MorimotoFurusaki13}\cite[\S III.C]{CTSR16})
by appeal to the ASK-Theorem (Prop. \ref{KaroubiTheorem}) or related statements, that possible choices of such mass-terms -- hence of ``nodal gap openings'' (\hyperlink{MassTermsOpeningUpBandNodes}{\it Figure 9}) -- are again classified by topological K-theory.
However, we need to beware of the following subtleties, which may not always have received due attention in the literature:

\begin{itemize}[leftmargin=.6cm]

\item[\bf (i)] Mass terms indicate the {\it absence} of a semi-metal phase. Hence topological semi-metal-phases ought to be classified by {\it quotienting out} a group of mass terms (cf. \cite[Thm. 9.63]{FreedHopkins16}) from a group of potential charges at nodal points.

\item[\bf (ii)] Since in general there are multiple nodal points which jointly satisfy constraints on their total topological charge (see Ex. \ref{HaldaneModel}),
the group of mass terms cannot be a K-theory group of a point, as considered in the above references, but must somehow be identified with the K-theory group of the whole punctured Brillouin torus.

\end{itemize}
\end{remark}

Both of these effects can be seen explicitly in the celebrated Haldane model:

\begin{example}[\bf The Haldane model]
  \label{HaldaneModel}
  The archetypical example of transitioning between
  2-dimensional topological semi-metal- and insulator-phases by switching on a mass term (\hyperlink{MassTermsOpeningUpBandNodes}{\it Figure 9}) is the {\it Haldane model} (\cite{Haldane15}, good review is in \cite{Atteia16}\cite{DTC}), originally motivated as a theoretical model for the non-trivial Chern insulator phase (Ex. \ref{NoQuantumSymmetryAndSpinOrbitCOupling}) of graphene with  its spin-orbit coupling not neglected (which, however, has remained elusive to experimental detection).
    In fact, the Haldane model is obtained (from a simple model for the semi-metal phase of graphene) by adding {\it two} summands:

  \begin{itemize}[leftmargin=.6cm]
\item[{\bf  1.}] an actual mass term of the form $m \gamma_{{}_0}$
  (\hyperlink{MassTermsOpeningUpBandNodes}{\it Figure 9}), ie. for a {\it constant} $m \,\in\, \RealNumbers$;

\item[{\bf   2.}] an interaction or ``background field'' term $t \cdot I(k) \gamma_{{}_0}$ for a {\it non-constant} function $I$ on the Brillouin torus, scaled by another parameter $t \in \RealNumbers$.
 \end{itemize}
  The interest in the Haldane model draws from the fact
  (see \hyperlink{HaldaneModelPhaseDiagram}{\it Figure 10})
  that at $m \neq 0$ and beyond some interaction strength $\vert t \vert > t_{\mathrm{crit}}(m)$ it realizes a non-trivial Chern-insulator phase (Ex. \ref{NoQuantumSymmetryAndSpinOrbitCOupling}). But in fact, for $m \neq 0$ but $\vert t \vert \leq t_{\mathrm{crit}}$ -- and hence in particular for the case of a pure mass term deformation $m \neq 0$, $t = 0$ -- the Haldane model is in a topologically {\it trivial} Chern insulator phase (ie. the valence bundle is gapped but has vanishing first Chern class, $c_1 = 0$, and hence is isomorphic to a trivial complex vector bundle).

  This trivial insulator phase (with $c_1 = 0$) of the non-interacting Haldane model at $m \neq 0$ but $t = 0$ is traditionally perceived as the problem which the seminal model building by Haldane did overcome, but for the purpose of uncovering the mathematical classification of 2d semi-metal phases (Conj. \ref{FlatKTheoryClassificationOfSemiMetals}) we highlight this as a most interesting datapoint:

  \medskip

  {\it To the extent that the properties of the Haldane model are generic for 2d semi-metal phases with effectively flat Berry connection away from the nodal points  (\hyperlink{BandStructureOfSemiMetals}{\rm Figure 6}), it shows that the deformation by an actual (namely constant) mass term (\hyperlink{MassTermsOpeningUpBandNodes}{\rm Figure 9}) leads to a  Chern-insulator phase which, while non-vanishing Chern numbers are still associated with the vicinity of each nodal point, is {\it globally topologically trivial} in that the total global Chern number vanishes.}

  \begin{center}
  \hypertarget{HaldaneModelPhaseDiagram}{}
  \begin{tabular}{cc}

\begin{minipage}{5cm}

\footnotesize
{\bf Figure 10.}  Indicated is the phase diagram of the Haldane model in dependence of the pure mass term $m$ and the interaction strength $t$.



\end{minipage}

&
\quad
\adjustbox{raise=-50pt}{
\begin{tikzpicture}

\draw[fill=lightgray, draw opacity=0]
  (0,.1) rectangle (2.8,2.8);

\draw[fill=lightgray, draw opacity=0]
  (0,.1)
  --
  (0,2.8)
  --
  (2,2.8)
  --
  (0,.1);

\draw[orangeii, line width=2pt]
(0,0) to (2,2.8);

\draw[->, line width=1pt]
  (-.5,0) to (3.6,0);
\draw[->, line width=1pt]
  (0, -.5) to (0,3.6);

  \draw[fill=orangeii]
    (0,0)
    circle (3pt);

\draw[line width=2.2pt, white]
  (1pt,2.165) to (3,2.165);
\draw[line width=2.2pt, white]
  (1pt,2.165+.22) to (3,2.165+.22);
\draw[line width=2.2pt, white]
  (1pt,2.165+.44) to (3,2.165+.44);
\draw[line width=2.2pt, white]
  (1pt,2.165+.66) to (3,2.165+.66);

\begin{scope}[rotate=-90, xscale=-1]
\draw[line width=2.2pt, white]
  (1pt,2.165) to (3,2.165);
\draw[line width=2.2pt, white]
  (1pt,2.165+.22) to (3,2.165+.22);
\draw[line width=2.2pt, white]
  (1pt,2.165+.44) to (3,2.165+.44);
\draw[line width=2.2pt, white]
  (1pt,2.165+.66) to (3,2.165+.66);
\end{scope}

\draw (-.6, 3.2)
 node
 {\scalebox{.7}{
   \color{darkblue} \bf
   mass
 }
 $m$
 };

\draw (3.8, -.17)
 node
 {
   $
   \underset{
     \scalebox{.7}{
       \color{darkblue} \bf
       interaction
     }
   }{
     t
   }
   $
 };

\draw
  (3,3.2)
  node[rotate=0]
  {
    \llap{
      \scalebox{.8}{
        \color{orangeii}
        \bf
        semi-metal\,
      }
    }
  };

\draw[draw opacity=.6]
  (.6,1.5)
  node[rotate=53]
  {
    \scalebox{.7}{
      \colorbox{lightgray}{
      $c_1 = 0$
      }
    }
  };

\draw[draw opacity=.6]
  (.6+.64,1.5-.44)
  node[rotate=54]
  {
    \scalebox{.7}{
      \colorbox{lightgray}{
      $c_1 \neq 0$
      }
    }
  };

\draw
  (1.8,.4)
  node{
    \scalebox{.7}{
      \color{greenii}
      \bf
      top. insulator
    }
  };

\end{tikzpicture}
}
    \end{tabular}
  \end{center}
  \end{example}

\vspace{0mm}
The topological semi-metal phase right at the origin ($m = 0$, $t = 0$) is the graphene-like phase with two Dirac points. The semi-metal phase on the slope $t = t_{\mathrm{crit}}(m)$ for $m > 0$ has a single Dirac point and interpolates between a trivial and a non-trivial topological insulator phase.
In particular, when the interaction vanishes, $t = 0$, and only a pure mass term $m > 0$ is turned on, then the graphene-like semi-metal phase is gapped into a {\it trivial} insulator phase: The Berry curvature is still concentrated at the locations $k_I$ where the Dirac points used to be -- and hence reflects a local topological charge $c_1\big[\mathcal{V}\vert_{S^2_I}\big]$ in the compactly supported K-theory around these points--, but
the sum of these contributions, being the integral of the Berry curvature over the whole Brillouin torus and hence equal to the first Chern class of the valence bundle, vanishes:
$$
  c_1\big[
    \mathcal{V}
  \big]
  \;=\;
  \sum_I
  c_1\!\big[
    \mathcal{V}\vert_{ S^2_I}
  \big]
  \;=\;
  0\;.
$$

We now observe that this peculiar property of the Haldane model is accurately reflected by Conj. \ref{FlatKTheoryClassificationOfSemiMetals}:

\begin{example}[\bf Classification of 2d Chern semi-metals]
\label{ClassificationOfTwoDChernSemiMetals}
We may combine the flat K-theory exact sequence \eqref{TheDifferentialCohomologyHexagonForKTheory} with the exact sequences induced from the homotopy cofiber sequence
\vspace{-2mm}
$$
  \begin{tikzcd}
    \overset{
      \raisebox{3pt}{
        \tiny
        \color{darkblue}
        \bf
        \def\arraystretch{.9}
        \begin{tabular}{c}
          Complement of
          \\
          nodal points in
          \\
          Brillouin torus
        \end{tabular}
      }
    }{
      \DualTorus{2}
      \setminus \{\vec k\}
    }
    \ar[r, hook]
    &
    \quad
    \overset{
      \mathclap{
      \raisebox{3pt}{
        \tiny
        \color{darkblue}
        \bf
        Brillouin torus
      }
      }
    }{
      \DualTorus{2}
    }
    \quad
    \ar[r]
    &
    \overset{
      \raisebox{3pt}{
        \tiny
        \color{darkblue}
        \bf
        \begin{tabular}{c}
          Tubulur nbhds
          \\
          of nodal points
        \end{tabular}
      }
    }{
    \underset{
      k_I
    }{\bigvee}
    \,
    S^2_I
    }
  \end{tikzcd}
$$

\vspace{-3mm}
\noindent
(where $S^2_{I}$ denotes the white disk labeled $k_I$ in \hyperlink{HomotopyTypeOfPuncturedTorus}{\it Figure 7}, where the grey boundary area is all identified with a single basepoint)
to obtain the following exact sequence of exact sequences:
\vspace{-5mm}
\begin{equation}
  \label{SemiMetalExactSequence}
  \hspace{-2cm}
  \begin{tikzcd}[
    column sep=19pt,
    row sep=5pt,
    decoration=snake
  ]
  \\
  \scalebox{.6}{
    \color{orangeii}
    \bf
    Zak phases
  }
  \ar[
    -,
    dddddd,
    shorten=-21pt,
    shift left=30pt
  ]
    &[-22pt]
    &[-8pt]
    \overset{
      \mathclap{
      \raisebox{0pt}{
        \begin{tabular}{c}
          \tiny
          \color{darkblue}
          \bf
          K-theory
          \\
          $\overbrace{\phantom{--------}}$
        \end{tabular}
      }
      }
    }{
      \mathrm{K}^{-1}
      \big(
        \DualTorus{2}
      \big)
    }
    \ar[r]
    \ar[d, ->>, "{\sim}"{sloped, pos=.35}]
    &[-4pt]
    \overset{
      \mathclap{
      \raisebox{0pt}{
        \begin{tabular}{c}
          \tiny
          \color{darkblue}
          \bf
          Ordinary cohomology
          \\
          $\overbrace{\phantom{---------}}$
        \end{tabular}
      }
      }
    }{
      \mathrm{K}^{-1}
      \big(
        \DualTorus{2}
        ;\,
        \ComplexNumbers
      \big)
    }
    \ar[r]
    \ar[d, ->>, "{\sim}"{sloped, pos=.35}]
    &[-7pt]
    \overset{
      \mathclap{
      \raisebox{0pt}{
        \begin{tabular}{c}
          \tiny
          \color{purple}
          \bf
          Flat K-theory
          \\
          $\overbrace{\phantom{------}}$
        \end{tabular}
      }
      }
    }{
      \mathrm{K}^0_{\flat}
      \big(
        \DualTorus{2}
      \big)
    }
    \ar[d, ->>, "{\sim}"{sloped, pos=.35}]
    &[-8pt]
    \\
  \scalebox{.6}{
    \color{orangeii}
    \bf
    \begin{tabular}{c}
      Zak phases among
      \\
      Berry-Zak phases
    \end{tabular}
  }
    &
    &
    \ar[d, hook]
    \ar[r, greenii]
    &
    \ar[d, hook]
    \ar[r, greenii]
    \ar[
      dddd,
      phantom,
      "{}"{name=C}
    ]
    &
    \ar[d, hook]
    \ar[
      ddddll,
      greenii,
      rounded corners,
      to path={
        -- ([xshift=2pt]\tikztostart.east)
        -- ([xshift=31pt]\tikztostart.east)
        |- ([yshift=2pt]C.center) \tikztonodes
        -| ([xshift=-33pt]\tikztotarget.west)
        -- ([xshift=-2pt]\tikztotarget.west)
      }
    ]
    \\
  \scalebox{.6}{
    \color{orangeii}
    \bf
    \def\arraystretch{.9}
    \begin{tabular}{c}
      Berry-Zak phases /
      \\
      nodal point
      charges
    \end{tabular}
  }
    &
    0
    \;\simeq\;
    \mathrm{ker}(\mathrm{ch}^{-1})
    \ar[r]
    &
    \underset{
      \mathclap{
      \raisebox{0pt}{
        \tiny
        \color{darkblue} \bf
        valence bundles with mass term
      }
      }
    }{
    \mathrm{K}^{-1}
    \big(
      \DualTorus{2}
      \setminus
      \{\vec k\}
    \big)
    }
    \ar[
      r,
      hook,
      shorten=-2pt,
      "{\mathrm{ch}^{-1}}"{pos=.6}
    ]
    \ar[
      dd,
      crossing over,
      shorten <= -2pt,
      shorten >= -2pt
    ]
    &
    \underset{
      \mathclap{
      \raisebox{-0pt}{
        \tiny
        \color{darkblue}
        \bf
        na{\"i}ve nodal point charges
      }
      }
    }{
    \mathrm{K}^{-1}
    \big(
      \DualTorus{2}
      \setminus
      \{\vec k\}
      ;\,
      \ComplexNumbers
    \big)
    }
    \ar[
      ddl,
      dashed,
      crossing over,
      shorten <=-9pt,
      shorten >=-6pt,
      decorate
    ]
    \ar[r, ->>]
    \ar[
      dd,
      shorten <=-2pt,
      shorten >=-4pt,
      crossing over, "{\Sigma'}"{description}
    ]
    &
    \underset{
      \mathclap{
      \raisebox{-0pt}{
        \tiny
        \color{purple}
        \bf
        top. semi-metal phases
      }
      }
    }{
    \mathrm{K}^0_{\flat}
    \big(
      \DualTorus{2}
      \setminus
      \{\vec k\}
    \big)
    }
    \ar[r, crossing over]
    \ar[
      dd,
      shorten >=-1pt,
      shorten <=-4pt,
      crossing over
    ]
    &
    \mathrm{K}^0
    \Big(
      \DualTorus{2}
      \setminus
      \{\vec k\}
    \Big)
    \mathrlap{
      \;\simeq\;
    0
    }
    \\
    \\
  \scalebox{.6}{
    \color{orangeii}
    \bf
    \def\arraystretch{.9}
    \begin{tabular}{c}
      Local Berry curv. /
      \\
      local top. charges
      \\
      at gapped nodes
    \end{tabular}
  }
    &
    0
    \;\simeq\;
    \mathrm{ker}
    \big(
      \mathrm{ch}^{0}
    \big)
    \ar[r, crossing over]
    &
    \mathrm{K}^{0}
    \Big(
        {\bigvee}_{\vec k}
        \,
        S^2
    \Big)
    \ar[r, hook, "{\mathrm{ch}^{0}}"]
    \ar[d, ->>]
    &
    \mathrm{K}^{0}
    \Big(
        {\bigvee}_{\vec k}
        \,
        S^2
      ;\,
      \ComplexNumbers
    \Big)
    \ar[r, ->>]
    \ar[d, ->>]
    &
    \mathrm{K}^1_{\flat}
    \Big(
        {\bigvee}_{\vec k}
        \,
        S^2
    \Big)
    \ar[r]
    \ar[d, ->>]
    &
    \mathrm{K}^1
    \Big(
     {\bigvee}_{\vec k}
     \,
     S^2
    \Big)
    \mathrlap{
      \;\simeq\;
      0
    }
    \\
  \scalebox{.6}{
    \color{orangeii}
    \bf
    \def\arraystretch{.9}
    \begin{tabular}{c}
      Local top. charges
      \\
      modulo
      gapped \\ node charges
    \end{tabular}
  }
    &
    &
    \ar[d, hook, "{\sim}"{sloped}]
    \ar[r, greenii]
    &
    \ar[d, hook, "{\sim}"{sloped}]
    \ar[r, greenii]
    &
    \ar[d, hook, "{\sim}"{sloped}]
    \\
  \scalebox{.6}{
    \color{orangeii}
    \bf
    \def\arraystretch{.9}
    \begin{tabular}{c}
      Global
      \\
      topological phases
    \end{tabular}
  }
    &
    &
    \underset{
      \mathclap{
      \raisebox{-0pt}{
        \tiny
        \color{darkblue}
        \bf
        top. insulator phases
      }
      }
    }{
    \mathrm{K}^0
    \big(
      \DualTorus{2}
    \big)
    }
    \ar[r]
    &
    \mathrm{K}^0
    \big(
      \DualTorus{2}
      ;\,
      \ComplexNumbers
    \big)
    \ar[r]
    &
    \mathrm{K}^1_{\flat}
    \big(
      \DualTorus{2}
    \big)
  \end{tikzcd}
\end{equation}
\vspace{-3mm}

\noindent
Here:

\begin{itemize}

\item[1.]
The column labeled ``ordinary cohomology'' essentially coincides with the Mayer-Vietoris sequence considered in \cite[(2.3)]{MathaiThiang15Semimetals} \cite[(2)]{MathaiThiang16} (there thought of as applying to integer coefficients, but all pertinent arguments hold verbatim also for complex coefficients). The following argument is akin to that in these articles, but the shift in Brillouin torus dimension from 3 (there) to 2 (here) makes a real difference for the physical interpretation: In 2d the Chern class $c_1$ and hence the $\mathrm{K}^0$-groups cannot reflect nodal point charges (which instead is accomplished now by $\mathrm{K}^0_{\flat}$\,) but remain indicative of the global gapped topological phases, in accord with Fact \ref{KTheoryClassificationOfTopologicalPhasesOfMatter}.

\item[2]
  At the boundaries of the diagram we have used the nature of the homotopy type of the punctured torus (\hyperlink{HomotopyTypeOfPuncturedTorus}{\it Figure 7}) to evaluate the given cohomology groups. For example (recalling that all cohomology groups we display are {\it reduced}) we see that would-be topological insulator phases necessarily trivialize on the complement of some points:
   \vspace{-1mm}
  $$
    \mathrm{K}^0
    \big(
      \DualTorus{2}
      \setminus
      \{k_1, \cdots, k_N\}
    \big)
    \;\simeq\;
    \mathrm{K}^0
    \big(
      {\bigvee}_{N+1}
      \,
      S^1
    \big)
    \;\simeq\;
    \underset{n+1}{\bigoplus}
    \,
    \mathrm{K}^0\big( S^1 \big)
    \;\simeq\;
    0
    \,,
  $$

   \vspace{-2mm}
\noindent
  while it is now the shifted
  K-group which reflects potential charges associated with the nodal points, both rationally
   \vspace{-2mm}
  $$
    \mathrm{K}^{-1}
    \big(
      \DualTorus{2}
      \setminus
      \{k_1, \cdots, k_N\}
      ;\,
      \ComplexNumbers
    \big)
    \;\simeq\;
    \underset{k}{\bigoplus}
    \,
    \mathrm{H}^{2k+1}
    \big(
      {\bigvee}_{N+1}\,S^1
      ;\,
      \ComplexNumbers
    \big)
    \;\simeq\;
    \underset{N+1}{\bigoplus}
    \mathrm{H}^1
    \big(
      S^1
    \big)
    \;\simeq\;
    \ComplexNumbers^{N+1}
  $$

  \vspace{-3mm}
\noindent
  as well as integrally:
   \vspace{-3mm}
  $$
    \mathrm{K}^{-1}
    \big(
      \DualTorus{2}
      \setminus
      \{k_1,  \cdots, k_N\}
    \big)
    \;\simeq\;
    \Integers^{N+1}
    \xhookrightarrow{\phantom{--}\mathrm{ch}^{-1}\phantom{--}}
    \ComplexNumbers^{N+1}
    \;\simeq\;
    \mathrm{K}^{-1}
    \big(
      \DualTorus{2}
      \setminus
      \{k_1,  \cdots, k_N\}
      ;\,
      \ComplexNumbers
    \big)
    \,.
  $$

  \vspace{-2mm}
\noindent  These boundary identifications are special to the symmetry un-protected case of Chern phases considered here,
but the structure of the diagram
  generalizes to (twisted) {\it equivariant} K-theory groups describing SPT/SET semi-metal phases if the boundary groups are instead replaced by their appropriate cohomological truncation.
  For example, when $\mathrm{K}^{\tau}_G\big( \DualTorus{2} \setminus \{\vec k\}\big)$ does not vanish, then the analogous diagram still characterizes those SPT semi-metal phases whose image in this group vanishes.

\item[3.] In this vein, the green arrows indicate the long exact sequence implied by the Snake lemma, which yields no further information in the case of Chern phases at hand, but may be non-trivial for SPT/SET phases.

\end{itemize}

\vspace{.2cm}

Finally, the squiggly dashed arrows indicates where the ``gapping'' of band nodes by mass terms is reflected in this diagram:
Given a set of naive nodal point charges which happen to be ``accidental'' or ``spurious'' in that there exists a mass term which opens the node crossings, then

\noindent {\bf (i)}  exactness of the middle horizontal exact sequence implies that the underlying semi-metal phase is trivial;

\noindent {\bf (ii)}  exactness of the left vertical exact sequence shows that the resulting local topological charges $c_1[S^2_I]$
(may each be non-trivial but)
have vanishing total topological charge $\sum_I c_1\big[S^2_I\big] = 0$.

\begin{equation}
  \begin{tikzcd}[column sep=large,
    decoration=snake
  ]
    \mathrm{K}^{-1}
    \big(
      \DualTorus{2}
      \setminus
      \{\vec k\}
    \big)
    \ar[dd]
    \ar[rr]
    \ar[
      dr,
      phantom,
      "{\ni}"{sloped, pos=1}
    ]
    &[-35pt]
    &
    \mathrm{K}^{-1}
    \big(
      \DualTorus{2} \setminus \{\vec k\}
      ;\,
      \ComplexNumbers
    \big)
    \ar[r]
    \ar[
      d,
      phantom,
      "{\ni}"{sloped}
    ]
    &
    \mathrm{K}^0_{\flat}
    \big(
      \DualTorus{2}
      \setminus
      \{\vec k\}
    \big)
    \ar[
      d,
      phantom,
      "{\ni}"{sloped}
    ]
    \\[-12pt]
    &
    \scalebox{.6}{
      \color{darkblue} \bf
      \begin{tabular}{c}
        choice of
        \\
        mass term
      \end{tabular}
    }
    \ar[
      r,
      <-|,
      shorten=6pt
    ]
    \ar[
      d,
      |->,
      shorten=4pt
    ]
    &
    \scalebox{.6}{
      \color{darkblue} \bf
      \def\arraystretch{.9}
      \begin{tabular}{c}
            nodal point
            \\
            charges
          \end{tabular}
      }
    \ar[
      r,
      |->,
      shorten=7pt,
      "{
        \scalebox{.6}{
          \color{greenii} \bf
            but trivial
        }
      }",
      "{
        \scalebox{.6}{
          \color{greenii} \bf
            semimetal phase
        }
      }"{swap}
    ]
    \ar[
      dl,
      |->,
      decorate,
      shorten >=-1pt,
      "{
        \scalebox{.6}{
          \color{greenii} \bf
          gap out band nodes
        }
      }"{swap, yshift=-2pt, sloped, pos=.4}
    ]
    &
    0
    \\
    \mathrm{K}^0
    \big(
      {\bigvee}_{N}
      \,
      S^2
    \big)
    \ar[
      d,
      "{
        \sum_{I=1}^N
      }"
    ]
    \ar[
      r,
      phantom,
      "{\ni}"{pos=.2}
    ]
    &
    \scalebox{.6}{
      \color{darkblue} \bf
      \def\arraystretch{.9}
      \begin{tabular}{c}
        local Berry
        \\
        curvatures
      \end{tabular}
    }
    \ar[
      d,
      |->,
      shorten=3pt,
      "{
        \scalebox{.6}{
          \color{orangeii} \bf
          \def\arraystretch{.9}
          \begin{tabular}{c}
            necessarily trivial
            \\
            insulator phase
          \end{tabular}
        }
      }"
    ]
    \\
    \mathrm{K}^0
    \big(
      \DualTorus{2}
    \big)
    \ar[
      r,
      phantom,
      "{\ni}"{pos=.3}
    ]
    &
    0
  \end{tikzcd}
\end{equation}

This is exactly the phenomenon seen in the Haldane model (\hyperlink{HaldaneModelPhaseDiagram}{\it Figure 10}) for constant mass terms.

\end{example}

In summary, Examples \ref{ClassificationOfTISymmetricSemiMetals} and \ref{ClassificationOfTwoDChernSemiMetals} seem to provide decent evidence for Conjecture \ref{FlatKTheoryClassificationOfSemiMetals}. The following discussion of topological interacting phases does not strictly depend on this conjecture, but together the two give a coherent picture.

\subsection{Interacting phases and TED-K of configurations}
\label{InteractingPhasesAndTEDKOfConfigurationSpaces}

So far we considered Bloch band theory, which
is based on the assumption that the (screened) electrons in the crystal may be regarded as {\it free}, namely as not interacting
with each other, but only with the effective background Coulomb field (this entered in Fact \ref{VacuaOfTheRelativisticElectronPositronFieldInBackground}).
It is remarkable that this assumption works so well (when it does, such as for many topological phases of matter) in that there
is a good match between electron band theory and experimental observation: Apparently the strong mutual Coulomb interaction
between real pairs of electrons in a crystal averages out in these cases (a fully satisfactory theoretical derivation of this
phenomenon from first principles seems not to be available).

\medskip
\noindent
{\bf Topological order.}
However, it is thought that the free-electron approximation certainly does break down in other crystalline materials and
specifically in some topological phases of matter, a now famous phenomenon which goes by a variety of technical terms
that we may schematically organize into the following list of implications (which may be gleaned from reviews such as \cite{ZCZW19}):

\vspace{-1cm}

\begin{center}
\hypertarget{FigureAnyonsFromInteractions}{}
$$
  \hspace{-5pt}
  {\small
  \begin{tikzcd}[
    column sep=12pt,
    row sep=-2pt
  ]
    \rlap{
      \mbox{
      \hspace{-1.6cm}
      Topological phase with...
      }
    }
    \\
    \fbox{
      \def\arraystretch{.9}
      \tabcolsep=-1pt
      \hspace{-5pt}
      \begin{tabular}{c}
      Strong/long-range
      \\
      \bf interaction
      \end{tabular}
    }
    \ar[r, Rightarrow]
    \ar[
      rrrrr,
      Rightarrow,
      rounded corners,
      to path={
        ([yshift=-0pt]\tikztostart.south)
        --
        ([yshift=-12pt]\tikztostart.south)
        --
        ([yshift=-18pt]\tikztotarget.south)
        --
        ([yshift=-0pt]\tikztotarget.south)      }
    ]
    &
    \fbox{
      \def\arraystretch{.9}
      \tabcolsep=-1pt
      \hspace{-5pt}
      \begin{tabular}{c}
      Strong
      \\
      \bf quantum
      \\
      \bf correlation
      \end{tabular}
    }
    \ar[r, Leftrightarrow]
    &
    \fbox{
      \def\arraystretch{.9}
      \tabcolsep=-1pt
      \hspace{-5pt}
      \begin{tabular}{c}
        Long-range
        \\
        \bf entanglement
      \end{tabular}
    }
    \ar[r, Rightarrow]
    &
    \fbox{
      \def\arraystretch{.9}
      \tabcolsep=-1pt
      \hspace{-5pt}
      \begin{tabular}{c}
        Topological
        \\
        \bf entanglement
        \\
        \bf entropy
      \end{tabular}
    }
    \ar[r, Leftrightarrow]
    &
    \fbox{
      \def\arraystretch{.9}
      \tabcolsep=-1pt
      \hspace{-5pt}
      \begin{tabular}{c}
      Topological
      \\
      \bf order
      \end{tabular}
    }
    \ar[r, Leftrightarrow]
    &
    \fbox{
      \def\arraystretch{.9}
      \tabcolsep=-1pt
      \hspace{-5pt}
      \begin{tabular}{c}
      \bf Anyons
      \end{tabular}
    }
  \end{tikzcd}
  }
$$
\begin{minipage}{16cm}
  \footnotesize
  {\bf
    Figure 11 --
    Topological order from
    non-negligible electron-interactions.
  }

\end{minipage}
\end{center}
\footnotetext{\cite[p. 527]{ZaanenLiuSunSchalm15}: ``In a way it appears obvious that the strongly interacting bosonic quantum critical state is subject to long-range entanglement. Nonetheless, the status of this claim is conjectural.
It is at present impossible to arrive at more solid conclusions that are based on rigorous mathematical procedures.'' }

 Specifically, by {\it topological order}
  (\cite{Wen91Review}\cite[p. 2]{GuWen09}\cite[\S III]{ZCZW19}),
  one means particularly rich topological phases of 2-dimensional quantum materials which may host defects (or ``non-local quasi-particle excitations'') whose adiabatic braiding (Rem. \ref{QuantumAdiabaticTheorem}) around each other has the effect of transforming the Hilbert space of ground states according to a non-trivial braid group representation, see around
  \hyperlink{NotionsOfAnyons}{\it Table 5}
  and
  \hyperlink{DefectAnyons}{\it Figure 12} below. (If this braid representation is non-abelian -- hence if there are {\it non-abelian anyons} --  then
  the ground state energy must be degenerate in that the Hilbert space of joint ground states must be higher dimensional -- this property was the
  original definition of ``topological order'' \cite{Wen89}\cite{WenNiu90}\cite{Wen91}\cite{Wen93}\cite{Wen95}).

  \medskip
  It is thought that {\it topological order} is characterized by a non-vanishing constant contribution to the entanglement entropy of the ground
  state -- called {\it topological entanglement entropy}
  (\cite{KitaevPreskill06}\cite{LevinWen06}, review in \cite{Grover13}\cite[\S 5]{ZCZW19})
  -- which signifies the presence of {\it long-range entanglement}  (\cite[\S V]{ChenGuWen10})  in the ground state (and the {\it absence} of
  {\it short-range entanglement} \cite{Kitaev11}\cite{Kitaev13}\cite{Freed14}).

 \medskip
  Beware that this is often referred to as {\it strong correlation} (e.g. \cite{Wen91Review}) which, however, is meant as ``quantum correlation''
  and as such synonymous with ``quantum entanglement'' (cf. \cite[\S 1.5]{ZCZW19} and generally \cite[p. 2]{LuoLuo03}). In contrast,
  {\it classical} long-range correlation is indicative of Landau-theory phases and hence orthogonal to topological order.

  \medskip
  Last but not least, the typical source of long-range entanglement in the ground state are expected to be non-negligible long-range electron-electron
  {\it interactions} (e.g. \cite[p. 527]{ZaanenLiuSunSchalm15}\cite[p. 1]{LuVijay22}).

\medskip
\noindent
{\bf The open problem of classifying topological order.}
  Little to no aspect of the schematic sequence of implications
  in \hyperlink{FigureAnyonsFromInteractions}{\it Figure 11}
  has previously found a mathematical formulation akin to the K-theory classification of non-interacting topological phases (Fact \ref{KTheoryClassificationOfTopologicalPhasesOfMatter}). In fact, the success of K-theory in capturing non-interaction topological phases seems to have been tacitly understood as implying that K-theory cannot play a role in the classification of interacting topological phases.

We highlight now that this is not actually a problem of K-theory as such, but of the {\it domain space} on which it is evaluated: The $n$-electron interactions are instead captured by the K-cohomology of the {\it configuration space of $n$ points} \eqref{ConfigurationSpaceOfPoints}
in the Brillouin torus. This observation combined with the main theorem of \cite{SS22AnyonicDefectBranes}  suggests a mathematical formulation and classification of {\it anyonic topological order} which captures the outermost part of the diagram in \hyperlink{FigureAnyonsFromInteractions}{\it Figure 11}.

\medskip

\noindent
{\bf K-theory of $n$-electron states over the configuration space of $n$ points.}
As one considers non-negligible $n$-electron interaction, the relevant wavefunctions are superpositions
(linear combinations)
of {\it Slater determinant} states $\vert \Psi_{i_1, \cdots, i_n}\rangle$,
\vspace{-1mm}
\begin{equation}
  \label{SlateDeterminant}
  \Psi_{i_1, \cdots, i_n}
  \Big(\!\!
    \big(k^1, s^1\big),
    \cdots,
    \big(k^n, s^n\big)
  \!\!\Big)
  \;:=\;
  \underset{
    \sigma \in \mathrm{Sym}(n)
  }{\sum}
  \;
  (-1)^{\mathrm{sgn}(\sigma)}
  \;
  \psi_{i_1}
      \big(
      k^{\sigma(1)},
      s^{\sigma(1)}
  \big)
  \cdot
  \psi_{i_2}
      \big(
      k^{\sigma(2)},
      s^{\sigma(2)}
    \big)
  \cdots
  \psi_{i_n}
  \big(
    k^{\sigma(n)},
    s^{\sigma(n)}
  \big)
\end{equation}

\vspace{-2mm}
\noindent (e.g. \cite[\S 2.2.3]{SzaboOstlund82}\cite[p. 196]{LanczosClarkDerrick86})
of $n$-tuples of single electron Bloch wavefunctions $\vert\psi_i\rangle$, regarded as functions of
momentum $k \in \mathbb{T}^d$ and spin polarization $s \,\in\, \big\{\uparrow, \downarrow\big\}$
(with respect to any fixed axis). To these Slater determinants, Bloch theory still applies
(``$n$-electron band theory'', e.g. \cite[p. 4]{GLG88}) and shows that the interacting energy bands
become functions on the product spaces of $n$-tuples of Bloch momenta:
\vspace{-2mm}
$$
  \big(\DualTorus{d}\big)^{n}
  \;=\;
  \big\{
    (k^1, \cdots, k^n)
    \;\in\;
    \DualTorus{d}
  \big\}
  \,.
$$

\vspace{-2mm}
\noindent
For $n = 2$ this is worked out for explicit examples in \cite{HGZ21}. For $n = 1$ this reduces to ordinary band theory
(\hyperlink{FigureBandStructure}{\it Figure 3}).

\medskip
These interacting bands over $\big(\DualTorus{d}\big)^{n}$ will be the eigenvalue bands of vector bundles spanned by the
underlying $n$-electron states, much as is the case for $n =1$, by standard Bloch theory (Fact \ref{BlochFloquetTheory}).
The only subtlety to beware of here is that all $n$-electron wave-functions (of the same spin) necessarily vanish where
any pair among the $n$ electrons has coinciding momenta -- due to the skew-symmetry enforced by the Slater determinants
\eqref{SlateDeterminant}, expressing the ``fermion statistics'' of electrons, and hence the ``Pauli exclusion principle''
by which no two electrons  inhabit the same single-particle state:
\vspace{-2mm}
$$
 \underset{i \neq j}{\exists}
 \;
 \big(
 k^i \,=\, k^j
 \;\;
 \mbox{and}
 \;\;
 s^i \,=\, s^j
 \big)
 \hspace{14pt}
 \;\Rightarrow\;
 \hspace{14pt}
  \Psi_{i_1, \cdots, i_n}
  \big(
    (k^1,s^1),
    \cdots,
    (k^n, s^n)
  \big)
  \;=\;
  0
  \,.
$$

\vspace{-2mm}
\noindent
Therefore, such a vector bundle of $n$-electron Bloch wavefunctions cannot exist over all of $\big(\DualTorus{d}\big)^{n}$,
since its fibers would degenerate on the ``fat diagonal''
\vspace{-3mm}
$$
  \boldDelta^n_{{}_{\TopologicalSpace}}
  \;:=\;
  \Big\{
    (k^1, \cdots, k^n)
    \in
    \TopologicalSpace^{\times^n}
    \;\Big\vert\;
    \underset{i \neq j}{\exists}
    \;
    k^i = k^j
  \Big\}
  \;\subset\;
  \TopologicalSpace^{\times^n}
  \,.
$$

\vspace{-2mm}
\noindent
But the $n$-electron Bloch vector bundle should exist over the complement of these problematic points (cf. \cite[p. 334]{FGR96}), which is the
{\it configuration space of $n$ ``probe'' points} (e.g. \cite[\S 2.2]{SS22ConfigurationSpaces}):
 \vspace{-2mm}
\begin{equation}
  \label{ConfigurationSpaceOfPoints}
  \ConfigurationSpace{n}
  \big(
    \TopologicalSpace
  \big)
  \;:=\;
  \TopologicalSpace^{\times^n}
  \setminus
  \boldDelta^n_{{}_{\TopologicalSpace}}
  \,.
\end{equation}

 \vspace{-2mm}
\noindent
Notice how this manifestly embodies the {\it Pauli Exclusion Principle}: the points where electron states would coincide
(in momentum space $\TopologicalSpace \,=\, \DualTorus{d}$) are {\it excluded} from the configuration space.

\medskip
More generally, if there are $N$ band nodes as in \hyperlink{BandStructureOfSemiMetals}{\it Figure 6},
then the valence bundle in {\it interacting $n$-electron approximation} should be a complex vector bundle over the
configuration space of $n$ (probe) points inside the complement of $N$ (nodal) points inside the Brillouin torus:
\vspace{-2mm}
\begin{equation}
  \label{SlaterBlochValenceBundle}
  \hspace{-11cm}
  \begin{tikzcd}[row sep=small, column sep=-12pt]
      \raisebox{0pt}{
        \tiny
        \color{orangeii}
        \bf
        \begin{tabular}{c}
          Slater-Bloch valence bundle of
          \\
          interacting $n$-electron states
        \end{tabular}
      }
    &
    \mathcal{V}_n
    \mathrlap{
      \;\;
      \subset
      \qquad
      \underset{
        \mathclap{
          (k^1, \cdots,\, k^n)
        }
      }{\coprod}
      \;\;\;\;
      \mathrm{Span}
      \bigg\{
          \overset{
            \mathclap{
            \hspace{-40pt}
            \raisebox{3pt}{
              \tiny
              \color{darkblue}
              \bf
              \begin{tabular}{c}
                Slater determinants
                of Bloch states
              \end{tabular}
            }
            }
          }{
            \Psi_{i_1, \cdots, i_n}
            \Big(
              \big(k^1, s^1),
              \cdots,
              \big(k^n, s^n\big)
            \Big)
          }
      \bigg\}_{
        \hspace{-4pt}
        \scalebox{.7}{$
          \begin{array}{c}
          (i_1, \cdots, i_n)
          \\
          (s^1, \cdots, s^n)
          \end{array}
        $}
      }
    }
    \ar[
      d,
      start anchor={[yshift=+9pt]},
      end anchor={[yshift=-4pt]},
    ]
    \\
    \mbox{
      \tiny
      \color{darkblue}
      \bf
      \begin{tabular}{c}
        configuration space of
        \\
        $n$ ``probe'' points
      \end{tabular}
    }
    &
    \ConfigurationSpace{n}
    \mathrlap{
    \Big(
      \underset{
        \mathclap{
        \raisebox{-4pt}{
          \tiny
          \color{darkblue}
          \bf
          \begin{tabular}{c}
            in complement of $N$
            ``nodal''
            \\
            points
            inside the Brillouin torus
          \end{tabular}
        }
        }
      }{
        \DualTorus{d}
        \setminus
        \{k_1, \cdots, k_N\}
      }
    \Big)
    \;=\;
    \bigg\{
      (k^1, \cdots, k^n)
      \,\in\,
      \big(
        \DualTorus{d}
      \big)^n
      \,\bigg\vert\,
      \underset{i \neq j}{\forall}
      \;
      \underset{
        \mathclap{
        \raisebox{-4pt}{
          \tiny
          \color{darkblue}
          \bf
          \def\arraystretch{.9}
          \begin{tabular}{c}
            Pauli
            \\
            exclusion
          \end{tabular}
          }
        }
      }{
        k^i \neq k^j
      }
      \;\;\;\mbox{and}\;\;\;
      \underset{i, I}{\forall}
      \;
      \underset{
        \mathclap{
        \raisebox{-3pt}{
          \tiny
          \color{darkblue}
          \bf
          \def\arraystretch{.9}
          \begin{tabular}{c}
            nodal
            \\
            singularities
          \end{tabular}
        }
        }
      }{
        k^i \neq k_I
      }
    \bigg\}.
    }
  \end{tikzcd}
\end{equation}

\noindent (Moreover, this should descend to a vector bundle over the {\it un}-ordered configuration space, but here we stick with the ordered configuration space.)

In conclusion:

\begin{conjecture}[\bf K-Theory classification of crystalline SPT order]
\label{KTheoryClassificationOfTopologicalOrder}
The adiabatic deformation classes of
symmetry protected (\hyperlink{TableOfTwistedEquivariances}{\textrm Figure 2})
topological order of interacting electrons
(according to \hyperlink{FigureAnyonsFromInteractions}{Figure 11}), for non-negligible $\leq n$-electron interactions,
is still classified by TED-K-theory (as in Fact \ref{KTheoryClassificationOfTopologicalPhasesOfMatter} and Conj. \ref{FlatKTheoryClassificationOfSemiMetals}), but now of the
\emph{configuration space of $n$ points}
\eqref{SlaterBlochValenceBundle}
in the complement of any given $N$ nodal points inside the Brillouin torus:

\vspace{.4cm}
\begin{equation}
  \label{KRGroupAsSPT}
  \mbox{
    \tiny
    \color{orangeii}
    \bf
    \begin{tabular}{c}
      symmetry protected
      \\
      topological order
    \end{tabular}
  }
  \big[
    \overset{
      \mathclap{
        \rotatebox{-45}{
          \llap{
            \tiny
            \color{darkblue}
            \bf
            \def\arraystretch{.9}
            \begin{tabular}{c}
              Slater-Bloch
              \\
              valence bundle
            \end{tabular}
            \hspace{-8pt}
          }
        }
      }
    }{
      \mathcal{V}_n
    }
  \big]
  \;\in\;\;
  \underset{
    \mathclap{
    \raisebox{-4pt}{
      \tiny
      \color{darkblue}
      \bf
      \def\arraystretch{.9}
      \begin{tabular}{c}
        SPT
        \\
        classes
      \end{tabular}
    }
    }
  }{
    \mathrm{KR}
  }^{
    \overset{
      \mathclap{
        \rotatebox{-50}{
          \llap{
            \tiny
            \color{greenii}
            \bf
            \def\arraystretch{}
            \begin{tabular}{r}
              (``fictitious'')
              \\
              gauging
            \end{tabular}
            \hspace{-9pt}
          }
        }
      }
    }{
      \tau
    }
  }
  \bigg(
    \HomotopyQuotient
    {
      \Big(
        \overset{
          \mathclap{
          \raisebox{+3pt}{
            \tiny
            \color{darkblue}
            \bf
            \def\arraystretch{.9}
            \begin{tabular}{c}
              $n$-electron
              \\
              interaction
            \end{tabular}
          }
          }
        }{
          \ConfigurationSpace{n}
        }
        \big(
          \DualTorus{d}
          \setminus
          \overset{
            \mathclap{
            \raisebox{3pt}{
              \tiny
              \color{darkblue}
              \bf
              $N$ band nodes
            }
            }
          }{
            \{ k_1, \cdots, k_N \}
          }
        \big)
      \Big)
    }{
      \overset{
        \mathclap{
          \rotatebox{-50}{
            \llap{
              \tiny
              \color{darkblue}
              \bf
              \def\arraystretch{.9}
              \begin{tabular}{c}
                symmetry\;\;\;\;
                \\
                protection
              \end{tabular}
            \hspace{-11pt}
            }
          }
        }
      }{
        G_{\mathrm{ext}}
      }
    }
  \bigg).
\end{equation}
\end{conjecture}
\begin{remark}[\bf Interacting topological phases subsume free topological phases]
  \label{InteractingTopologicalPhasesSubsumeFreeTopologicalPhases}
  In the degenerate case that $n = 1$, hence for the case of vanishing effective electron-electron interaction, we have
  \begin{equation}
    \label{ConfigurationSpaceOfASinglePoint}
    \ConfigurationSpace{1}(\TopologicalSpace)
    \,=\, \TopologicalSpace
  \end{equation}
  and the statement of Conjecture \ref{KTheoryClassificationOfTopologicalOrder} reduces to that of Fact \ref{ClassificationOfExternalSPTPhases}.
\end{remark}
Besides the clear physical motivation \eqref{SlaterBlochValenceBundle} for Conjecture \ref{KTheoryClassificationOfTopologicalOrder},
the key evidence, which we turn to next, comes from the observation (\cite{SS22AnyonicDefectBranes}) that (flat) TED-K-theory groups of the form \eqref{KRGroupAsSPT}
do reflect the presence of non-abelian anyons as expected/demanded for a topological order
(\hyperlink{FigureAnyonsFromInteractions}{\it Figure 11}), in that they
naturally contain the expected anyon ground state wave functions (namely {\it conformal blocks}) which do constitute braid group representations (``anyon statistics'') under movement of the nodal points $k_I$. This is the result of \cite{SS22AnyonicDefectBranes},
which we now review and further connect to condensed matter theory.

\medskip

\subsection{Anyonic topological order and Inner local system-TED-K}
\label{AnyonicTopologicalOrderAndInnerLocalSysyemTEDK}

\noindent
{\bf Notions of anyons -- quanta \& defects, in position- \& momentum-space.}  The idea of {\it anyon particles}
in $2+1$ dimensions (review in \cite{Wilczek90}\cite{Lerda92}\cite{Rao16}, following \cite{LeinaasMyrheim77} \cite{Wilczek82b}, cf. \hyperlink{FigureAnyonsFromInteractions}{\it Figure 11}) is that their wave-functions pick up {\it any} (whence the name) fixed unitary transformation (instead of just multiplication by -1, as for fermions) whenever one of them completes a full rotation around another (as made precise in a moment).
We may recognize {\it two distinct} conceptualizations of anyons implicit in the literature, which we will refer to as shown in \hyperlink{NotionsOfAnyons}{\it Table 5}:

\begin{center}
\hypertarget{NotionsOfAnyons}{}
\small
\def\arraystretch{3.4}
\begin{tabular}{|c|l|l|}
\hline
  \rowcolor{lightgray}
\def\arraystretch{1}
\begin{tabular}{c}
 {\bf Anyonic quanta}
 \\
 (abelian)
\end{tabular}
&
\begin{minipage}{8.5cm}

$\mathclap{\phantom{\vert^{\vert^{\vert^{\vert}}}}}$like fermionic quanta (such as electrons) but subject to {\it additional} abelian braiding phases, understood as  {\bf Aharonov-Bohm phases} due to a flat abelian  ``fictitious'' gauge field
\eqref{FlatConnectionOneForm}
which is sourced by and coupled to each of the quanta.
\\

\end{minipage}
&
\begin{minipage}{3.1cm}
(\cite{CWWH89} following \cite{ASWZ85}, reviewed in \cite[\S I.3]{Wilczek90}\cite{wilczek91}, see also \cite{IengoLechner92})
\end{minipage}
\\
\hline
\def\arraystretch{1}
\begin{tabular}{c}
  {\bf Anyonic defects}
  \\
  (possibly non-abelian)
\end{tabular}
&
\begin{minipage}{8.5cm}

$\mathclap{\phantom{\vert^{\vert^{\vert^{\vert}}}}}$like solitonic defects (such as vortices) whose position is
a classical parameter (boundary condition) to the quantum system and whose {\it adiabatic movement}
(Rem. \ref{QuantumAdiabaticTheorem})
acts on the quantum ground state by (non-abelian) {\bf Berry phases}.
\\
\end{minipage}
&
\def\arraystretch{1}
\hspace{-8pt}
\begin{tabular}{l}
$\mathclap{\phantom{\vert^{\vert^{\vert^{\vert}}}}}$(e.g. \cite[p. 1]{ArovasSchriefferWilczek84}
\\
\cite[pp. 6]{FKLW01}
\\
\cite[\S II.A.2]{NSSFS08}
\\
\cite{CGDS11}\cite{CLBFN15}
\\
\cite{BarlasProdan19}\cite[p. 321]{Stanescu20})
\\
$\mathclap{\phantom{-}}$
\end{tabular}
\\
\hline
\end{tabular}

\medskip

\begin{minipage}{14cm}
  \footnotesize
  {\bf Table 5 -- Notions of anyons.} -- Even though the term {\it anyon} (or {\it plekton}) is traditionally used indiscriminately, we highlight that {\it anyonic quanta} and {\it anyonic defects}  are on distinct conceptual footing. Below we formalize both notions and find them unified within the TED-K theory of configuration spaces of points (reflecting the anyonic quanta) inside surfaces with punctures (reflecting the anyonic defects).
\end{minipage}

\end{center}

While the common terminology of ``anyon statistics'' evokes the notion of anyonic {\it quanta} (quasi-particles), the early motivation of anyons as particles ``bound'' to practically infinite solenoids/magnetic flux tubes (\cite[\S III]{GMS81}\cite{Wilczek82a}\cite[p. 5]{Wilczek90}) refers to their incarnation as {\it defects}. In fact, the braiding of defects of co-dimension=2 had been discussed in detail (\cite{Mermin79}\cite{LoPreskill93}) before and while the notion of ``anyons'' became established terminology. More recently, this is gaining renewed attention:

\vspace{-2mm}
\begin{quote}
  {\it
  Anyonic particles are best viewed as a kind of topological defects that reveal nontrivial properties of the ground  state.}
  \cite[p. 4]{Kitaev06}.
\end{quote}
\noindent Specifically, anyonic {\it vortex defects} are realized in Bose-Einstein condensates \cite{MPSS19} and
other superfluids \cite{MasakiMizushimaNitta21}. Vortices with bound Majorana zero modes are among the most
studied anyon species for potential laboratory realization \cite{DasSarmaFreedmanNayak15}
(cf. \cite[Fig. 1]{MMBDRSC19}). This is expected to generalize to $\mathfrak{su}(2)$-anyons given by zero modes bound to solitonic defects in parafermion models \cite[p. 2]{Tsvelik14}\cite[pp. 3]{Borcherding18}.
Therefore, it is important to make explicit that, besides their incarnation as quanta or quasi-particles, anyons have {\it another} incarnation as defects:
\begin{quote}
{\it Anyons can arise in two ways: as localised excitations of an interacting quantum Hamiltonian
or as defects in an ordered system.} \cite[p. 1]{DasSarmaFreedmanNayak15}.
\end{quote}
We will present in the following a theory (TED-K theory of configuration spaces) which brings out and unifies both these distinct notions of anyons.
But first we note that a  notion of defects subject to adiabatic braiding (Rem. \ref{QuantumAdiabaticTheorem}) is clearly not constrained to defects in position space:

\begin{remark}[{\bf Momentum-space anyons.}]
\label{MomentumSpaceAnyons}

{\bf (i)} Solitonic defects also exist in ``momentum space'' (``reciprocal space''), namely in the guise of the familiar and ubiquitous band nodes in the Brillouin torus of a semi-metal (\hyperlink{BandStructureOfSemiMetals}{\it Figure 6}) -- an observation recently highlighted in \cite{AhnParkYang18}, whose authors refer to ``{\it
  band crossing points, henceforth called vortices}''.
This notion of semi-metal band nodes as anyon-like defects in momentum space is recently finding attention \cite{BWSWYB20}\cite{TiwariBzdusek20}\cite{JBL21}\cite{PBSM22}, notably for the case of (twisted bilayer) graphene (see \cite[Fig. 18]{KangVafek20}). In particular, the all-important adiabatic {\it movement} (braiding) of anyons is quite tractable for
band nodes in momentum space \cite{CBSM21} \cite{PBMS22}\cite{PGZO22}\footnote{
\cite{JBL21}: ``we consider an exotic type of topological phases beyond the above paradigms that, instead, depend on topological charge conversion processes when band nodes are braided with respect to each other in momentum space or recombined over the Brillouin zone. The braiding of band nodes is in some sense the reciprocal space analog of the non-Abelian braiding of particles in real space. $[\cdots]$ we experimentally observe non-Abelian topological semimetals and their evolutions using acoustic Bloch bands in kagome acoustic metamaterials. By tuning the geometry of the metamaterials, we experimentally confirm the creation, annihilation, moving, merging and splitting of the topological band nodes in multiple bandgaps and the associated non-Abelian topological phase transitions''

\cite{CBSM21}: ``Our work opens up routes to readily manipulate Weyl nodes using only slight external parameter changes, paving the way for the
practical realization of reciprocal space braiding.''.

\cite{PBSM22}: ``new opportunities for exploring non-Abelian braiding of band crossing points (nodes) in reciprocal space, providing an
alternative to the real space braiding exploited by other strategies.
Real space braiding is practically constrained to boundary states, which has made experimental observation and manipulation difficult;
instead, reciprocal space braiding occurs in the bulk states of the band structures and we demonstrate in this work that this provides
a straightforward platform for non-Abelian braiding.''.

\cite{PBMS22}: ``it is possible to controllably braid Kagome band nodes in monolayer $\mathrm{Si}_2 \mathrm{O}_3$ using strain and/or
an external electric field.''.
}, while it remains elusive for defect
anyons in position space (cf. \cite[p. 8]{Kitaev06}\cite[p. 7-8]{SarmaFreedmanNayak15} and \cite{KouwenhovenEtAl21}\cite{KouwenhovenEtAl22}).

\noindent {\bf (ii)}
At the same time, a theoretical underpinning for understanding band nodes as {\it anyons in momentum space} had been missing. We suggest that the results developed below \cref{AnyonicTopologicalOrderAndInnerLocalSysyemTEDK} go towards providing this theory.

\noindent {\bf (iii)}
In the following we show that TED-K-theory of configuration spaces provides a theory which formalizes and unifies both anyonic quanta and anyonic defects in a way that subsumes their existing mathematical models.
\end{remark}
The resulting picture is most coherent for anyons in momentum/reciprocal-space:
\begin{center}
\hypertarget{DefectAnyons}{}

\begin{tabular}{ll}

\begin{minipage}{7.4cm}
\footnotesize
{\bf Figure 12 -- Anyon braiding.}
Shown is a configuration of three points on a surface, evolving in time such that two of the points rotate around each other with
their worldlines forming a braid.

In the text we consider this general situation for the non-standard case where the points are in  {\it momentum space} (``reciprocal space''),
namely in the Brillouin torus of a topological semi-metal. Here, in the terminology of \hyperlink{NotionsOfAnyons}{\it Table 5}:

\vspace{1mm}
\begin{itemize}[leftmargin=.6cm]

\item[--]
the {\it anyonic quanta} are interacting Bloch electron states \eqref{SlateDeterminant} whose interaction involves an effective
``fictitious gauge field'' (in momentum space) inducing abelian Aharonov-Bohm phases under braiding;

\item[--]
the {\it anyonic defects} are nodal points, namely loci of band nodes (\hyperlink{BandStructureOfSemiMetals}{\it Figure 6}),
which are singular defects in momentum space, for instance in that they act as delta-sources for Berry curvature.

\end{itemize}
\end{minipage}

&

\hspace{2mm}
\adjustbox{raise=-3cm}{

\begin{tikzpicture}

  \shade[right color=lightgray, left color=white]
    (3,-3)
      --
      node[above, yshift=-1pt, sloped]{
        \scalebox{.7}{
          \color{darkblue}
          \bf
          Brillouin torus
        }
      }
    (-1,-1)
      --
    (-1.21,1)
      --
    (2.3,3);

  \draw[]
    (3,-3)
      --
    (-1,-1)
      --
    (-1.21,1)
      --
    (2.3,3)
      --
    (3,-3);

\draw[-Latex]
  ({-1 + (3+1)*.3},{-1+(-3+1)*.3})
    to
  ({-1 + (3+1)*.29},{-1+(-3+1)*.29});

\draw[-Latex]
    ({-1.21 + (2.3+1.21)*.3},{1+(3-1)*.3})
      --
    ({-1.21 + (2.3+1.21)*.29},{1+(3-1)*.29});

\draw[-Latex]
    ({2.3 + (3-2.3)*.5},{3+(-3-3)*.5})
      --
    ({2.3 + (3-2.3)*.49},{3+(-3-3)*.49});

\draw[-latex]
    ({-1 + (-1.21+1)*.53},{-1 + (1+1)*.53})
      --
    ({-1 + (-1.21+1)*.54},{-1 + (1+1)*.54});

  \begin{scope}[rotate=(+8)]
   \draw[dashed]
     (1.5,-1)
     ellipse
     ({.2*1.85} and {.37*1.85});
   \begin{scope}[
     shift={(1.5-.2,{-1+.37*1.85-.1})}
   ]
     \draw[->, -Latex]
       (0,0)
       to
       (180+37:0.01);
   \end{scope}
   \begin{scope}[
     shift={(1.5+.2,{-1-.37*1.85+.1})}
   ]
     \draw[->, -Latex]
       (0,0)
       to
       (+37:0.01);
   \end{scope}
   \begin{scope}[shift={(1.5,-1)}]
     \draw (.43,.65) node
     { \scalebox{.8}{$
       \sfrac{\weight_I}{\ShiftedLevel}
     $} };
  \end{scope}
  \draw[fill=white, draw=gray]
    (1.5,-1)
    ellipse
    ({.2*.3} and {.37*.3});
  \draw[line width=3.5, white]
   (1.5,-1)
   to
   (-2.2,-1);
  \draw[line width=1.1]
   (1.5,-1)
   to node[above, yshift=-3pt, pos=.85]{
     \;\;\;\;\;\;\;\;\;\;\;\;\;
     \rotatebox[origin=c]{7}
     {
     \scalebox{.7}{
     \color{orangeii}
     \bf
     \colorbox{white}{anyon} position
     }
     }
   }
   (-2.2,-1);
  \draw[
    line width=1.1
  ]
   (1.5+1.2,-1)
   to
   (3.5,-1);
  \draw[
    line width=1.1,
    densely dashed
  ]
   (3.5,-1)
   to
   (4,-1);

  \draw[line width=3, white]
   (-2,-1.3)
   to
   (0,-1.3);
  \draw[-latex]
   (-2,-1.3)
   to
   node[below, yshift=+3pt]{
     \scalebox{.7}{
       \rotatebox{+7}{
       \color{darkblue}
       \bf
       time
       }
     }
   }
   (0,-1.3);
  \draw[dashed]
   (-2.7,-1.3)
   to
   (-2,-1.3);

 \draw
   (-3.15,-.8)
   node{
     \scalebox{.7}{
       \rotatebox{+7}{
       \color{greenii}
       \bf
       braiding
       }
     }
   };

  \end{scope}

  \begin{scope}[shift={(-.2,1.4)}, scale=(.96)]
  \begin{scope}[rotate=(+8)]
  \draw[dashed]
    (1.5,-1)
    ellipse
    (.2 and .37);
  \draw[fill=white, draw=gray]
    (1.5,-1)
    ellipse
    ({.2*.3} and {.37*.3});
  \draw[line width=3.1, white]
   (1.5,-1)
   to
   (-2.3,-1);
  \draw[line width=1.1]
   (1.5,-1)
   to
   (-2.3,-1);
  \draw[line width=1.1]
   (1.5+1.35,-1)
   to
   (3.6,-1);
  \draw[
    line width=1.1,
    densely dashed
  ]
   (3.6,-1)
   to
   (4.1,-1);
  \end{scope}
  \end{scope}

  \begin{scope}[shift={(-1,.5)}, scale=(.7)]
  \begin{scope}[rotate=(+8)]
  \draw[dashed]
    (1.5,-1)
    ellipse
    (.2 and .32);
  \draw[fill=white, draw=gray]
    (1.5,-1)
    ellipse
    ({.2*.3} and {.32*.3});
  \draw[line width=3.1, white]
   (1.5,-1)
   to
   (-1.8,-1);
\draw
   (1.5,-1)
   to
   (-1.8,-1);
  \draw
    (5.23,-1)
    to
    (6.4-.6,-1);
  \draw[densely dashed]
    (6.4-.6,-1)
    to
    (6.4,-1);
  \end{scope}
  \end{scope}

\draw (2.41,-2.3) node
 {
  \scalebox{1}{
    $\DualTorus{2}$
  }
 };

\draw (1.73,-1.06) node
 {
  \scalebox{.8}{
    $k_{{}_{I}}$
  }
 };

\begin{scope}
[ shift={(-2,-.55)}, rotate=-82.2  ]

 \begin{scope}[shift={(0,-.15)}]

  \draw[]
    (-.2,.4)
    to
    (-.2,-2);

  \draw[
    white,
    line width=1.1+1.9
  ]
    (-.73,0)
    .. controls (-.73,-.5) and (+.73-.4,-.5) ..
    (+.73-.4,-1);
  \draw[
    line width=1.1
  ]
    (-.73+.01,0)
    .. controls (-.73+.01,-.5) and (+.73-.4,-.5) ..
    (+.73-.4,-1);

  \draw[
    white,
    line width=1.1+1.9
  ]
    (+.73-.1,0)
    .. controls (+.73,-.5) and (-.73+.4,-.5) ..
    (-.73+.4,-1);
  \draw[
    line width=1.1
  ]
    (+.73,0+.03)
    .. controls (+.73,-.5) and (-.73+.4,-.5) ..
    (-.73+.4,-1);

  \draw[
    line width=1.1+1.9,
    white
  ]
    (-.73+.4,-1)
    .. controls (-.73+.4,-1.5) and (+.73,-1.5) ..
    (+.73,-2);
  \draw[
    line width=1.1
  ]
    (-.73+.4,-1)
    .. controls (-.73+.4,-1.5) and (+.73,-1.5) ..
    (+.73,-2);

  \draw[
    white,
    line width=1.1+1.9
  ]
    (+.73-.4,-1)
    .. controls (+.73-.4,-1.5) and (-.73,-1.5) ..
    (-.73,-2);
  \draw[
    line width=1.1
  ]
    (+.73-.4,-1)
    .. controls (+.73-.4,-1.5) and (-.73,-1.5) ..
    (-.73,-2);

 \draw[
   densely dashed
 ]
   (-.2,-2)
   to
   (-.2,-2.5);
 \draw[
   line width=1.1,
   densely dashed
 ]
   (-.73,-2)
   to
   (-.73,-2.5);
 \draw[
   line width=1.1,
   densely dashed
 ]
   (+.73,-2)
   to
   (+.73,-2.5);

  \end{scope}
\end{scope}

\end{tikzpicture}

}
\end{tabular}
\end{center}

\newpage

\begin{remark}[\bf Natural toroidal geometry for momentum-space topological order]
\label{NaturalToroidalGeometryForMomentumSpaceTopologicalOrder}
$\,$

\noindent {\bf (i)}
Crystalline momentum/reciprocal space naturally provides the toroidal geometry -- in form of the Brillouin torus \eqref{BrillouinTorus} -- which is thought  (starting with \cite{Wen89}) to be so important for realizing non-trivial topological order.

\noindent {\bf (ii)}
In contrast, it was never clear
(cf. \cite[p. 1]{Lan19})
how anyons inside toroidal position-space geometry
(envisioned by so many authors, e.g. \cite{Einarsson90}\cite{HosotaniHo92}\cite{GreiterWilczek92}\cite{PuJain21})
would be realized in solid state physics. Even with the crystal lattice being periodic in position space, the position of anyon defects would hardly be. But for Bloch features in momentum space (such as band nodes), all this is natural and automatic.

\noindent {\bf (iii)}
Last but not least,
topological order formulated in momentum space naturally connects
(as shown in the following)
to the established understanding of topological gapped  phases (Fact \ref{KTheoryClassificationOfTopologicalPhasesOfMatter}) which is all concerned with phenomena in momentum space.
\end{remark}

Therefore, we regard the following analysis (culminating in Facts \ref{TopologicallyOrderedGroundStatesViaChernForms} and \ref{NonAbelianAnyonStatisticsOfBandNodes}) as a prediction from TED-K-theory of a good momentum-space anyon phenomenology (see also Rem. \ref{ThePhysicalOriginOfNonAbelioanAnyonicBraiding} below).

\vspace{.2cm}

\medskip
\noindent
{\bf Anyon quanta and Equivariant valence bundles.}
In mathematical detail, the nature of {\it anyonic quanta} (in the sense of \hyperlink{NotionsOfAnyons}{\it Table 5})
propagating on some (position- or momentum-) space $\TopologicalSpace$ is
(e.g. \cite[(1.2)]{BCMS93}\cite[(1.3)]{DFT97})
that their $n$-particle wave-functions $\Psi$
(see \eqref{SlaterBlochValenceBundle})
are complex functions on the configuration space
\eqref{ConfigurationSpaceOfPoints} which are
``multi-valued'' according to the anyon braiding phases. That is, these are actual complex functions $\widehat{\Psi}$
on the universal covering space
\vspace{-2mm}
\begin{equation}
  \label{EquivariantFunctionOnUniversalCover}
  \begin{tikzcd}
  \widehat{\ConfigurationSpace{n}}(\TopologicalSpace)
  \ar[out=180-60, in=60,
    looseness=4.6,
    "\scalebox{1}{${
        \hspace{6pt}
        \mathclap{\mathrm{Br}_X(n)}
        \hspace{10pt}
      }$}"{pos=.41, description},
      shift right=1,
      start anchor={[xshift=-4pt]},
      end anchor={[xshift=-4pt]},
    ]
    \ar[
    d,
    "{
      \mbox{
        \tiny
        \color{darkblue}
        \bf
        \def\arraystretch{.9}
        \begin{tabular}{c}
          universal
          \\
          cover
        \end{tabular}
      }
    }"{swap}
  ]
  \ar[
    rr,
    "{
      \widehat{\Psi}
    }",
    "{
      \mbox{
        \tiny
        \color{greenii}
        \bf
        \def\arraystretch{.9}
        \begin{tabular}{c}
          anyonic $n$-quanta
          \\
          wave-function
        \end{tabular}
      }
    }"{swap}
  ]
  &&
  \ComplexNumbers
  \ar[out=180-56, in=56,
    looseness=5,
    "\scalebox{1}{${
        \hspace{4pt}
        \mathclap{\CyclicGroup{\ShiftedLevel}}
        \hspace{5pt}
      }$}"{pos=.41, description},
      shift right=0,
      start anchor={[xshift=-0pt]},
      end anchor={[xshift=-0pt]},
    ]
    \\
  \ConfigurationSpace{n}(X)
  \end{tikzcd}
\end{equation}

\vspace{-2mm}
\noindent
that  are {\it equivariant} with respect to the action of {\it braids}, namely of loops in configuration space:
\vspace{-2mm}
\begin{equation}
  \label{LoopsInConfigurationSpace}
  \overset{
    \mathclap{
    \raisebox{3pt}{
      \tiny
      \color{darkblue}
      \bf
      braid
    }
    }
  }{
    [\gamma\,]
  }
  \;\in\;
  \overset{
    \mathclap{
    \raisebox{3pt}{
      \tiny
      \color{darkblue}
      \bf
      braid group
    }
    }
  }{
    \mathrm{Br}_{{}_{\TopologicalSpace}}(n)
  }
  \;:=\;
  \overset{
    \raisebox{4pt}{
      \tiny
      \color{darkblue}
      \bf
      \def\arraystretch{.9}
      \begin{tabular}{c}
        fundamental group of
        \\
        configuration space
      \end{tabular}
    }
  }{
  \pi_1
  \big(
    \ConfigurationSpace{n}(\TopologicalSpace)
  \big)
  }.
\end{equation}

\vspace{-2mm}
\noindent The latter condition means that they satisfy the following constraint
\vspace{-1mm}
\begin{equation}
  \label{EquivarianceOfAnyonWaveFunctions}
  \widehat{\psi}
  \big(
    [\gamma\,]\cdot
    (k^1, \cdots, k^n)
  \big)
  \;=\;
  \phi(\gamma)
  \cdot
  \widehat{\psi}
  \big(
    k^1, \cdots, k^n
  \big)
\end{equation}

\vspace{-1mm}
\noindent
for all braids $[\gamma\,]$ and $n$-tuples of positions $k^1, \cdots, k^n \,\in\, \TopologicalSpace$, and
for a given choice of {\it braiding phases}, given by a group homomorphism
\begin{equation}
  \label{BraidingPhases}
  \phi
    \;:\;
  \mathrm{Br}_{{}_\TopologicalSpace}(n)
  \xrightarrow{
    \mbox{
      \tiny
      \color{greenii}
      \bf
        braiding
        phases
    }
  }
  \CyclicGroup{\ShiftedLevel}
  \xhookrightarrow{\quad}
  \CircleGroup
  \,.
\end{equation}
This \eqref{EquivarianceOfAnyonWaveFunctions} was essentially understood in \cite{Wu84}\cite{IIS90},
a clear account is in \cite{MundSchrader95}, see also \cite[p. 20]{FGM90}\cite[\S 1]{BCMS93}\cite[\S 1]{DFT97};
and for references specifically in the context of solid state physics see also \cite{DMV03}\cite{MurthyShankar09}.
Here we highlight that \eqref{EquivarianceOfAnyonWaveFunctions} means that the vector bundle $\widehat{\mathcal{V}_n}$
of which the wave-functions $\widetilde{\Psi}$ are the  sections -- notably: the anyonic generalization of the $n$-electron valence bundle $\mathcal{V}_n$ \eqref{SlaterBlochValenceBundle} -- is equipped with
{\it equivariant bundle structure} relative to $\phi$ (see \cite{SS21EPB} for exposition and pointers):
\vspace{-3mm}
$$
  \begin{tikzcd}[column sep=large]
    &
    \widehat{\mathcal{V}}_n
    \big(
      [\gamma_{\, 1}]
      (k^1, \cdots, k^n)
    \big)
    \ar[
      dr,
      "{
        \phi
        \left(
          [\gamma_{\, 2}]
        \right)
        \cdot
      }"{sloped},
      "{
        \sim
      }"{swap, sloped}
    ]
    \\
    \mathllap{
      \mbox{
        \tiny
        \color{darkblue}
        \bf
        \def\arraystretch{.9}
        \begin{tabular}{c}
          Fiber of anyonic $n$-particle
          \\
          valence bundle at some
          \\
          $n$-tuple of momenta
        \end{tabular}
      }
    }
    \widehat{\mathcal{V}}_n
    \big(
      (k^1, \cdots, k^n)
    \big)
    \ar[
      ur,
      "{
        \phi\left([\gamma_1]\right)
        \cdot
      }"{sloped},
      "{\sim}"{swap, sloped}
    ]
    \ar[
      rr,
      "{
        \phi
        \left(
          [\gamma_{\, 2 }\cdot \gamma_{\, 1}]
        \right)
        \cdot
      }",
      "{
        \mbox{
          \tiny
          \color{greenii} \bf
          multiplication by braiding phases
        }
      }"{swap}
    ]
    &&
    \widehat{\mathcal{V}}_n
    \big(
      [\gamma_{\,2} \cdot \gamma_{\, 1}]
      \cdot (k^1, \cdots, k^n)
    \big)
    \mathrlap{
      \mbox{
        \tiny
        \color{darkblue}
        \bf
        \def\arraystretch{.9}
        \begin{tabular}{c}
          Fiber of anyonic $n$-particle
          \\
          valence bundle at
          {\color{orangeii}braided}
          \\
          $n$-tuple of momenta
        \end{tabular}
      }
    }  \end{tikzcd}
$$

\vspace{-2mm}
\noindent In the convenient language of stacks (as laid out in
\cite{SS20OrbifoldCohomology}\cite{SS21EPB}),
with $\mathrm{Vect}_{\ComplexNumbers}$ denoting the moduli stack of complex vector bundles,
we may sum this up as saying that the $n$-particle valence bundle of {\it anyonic quanta}
is given by a system of horizontal maps making the following diagram homotopy-commute:
\vspace{-2mm}
\begin{equation}
  \label{AnyonicValenceBundles}
  \begin{tikzcd}[column sep=huge]
      \scalebox{.5}{
        \color{darkblue}
        \bf
        \def\arraystretch{.9}
        \begin{tabular}{c}
          Universal
          \\
          cover
        \end{tabular}
      }
    &[-30pt]
    \widehat{
    \ConfigurationSpace{n}
    }
    (
      \TopologicalSpace
    )
    \ar[out=180-60, in=60, looseness=4.6, "\scalebox{1}{${
        \hspace{6pt}
        \mathclap{\mathrm{Br}_X(n)}
        \hspace{10pt}
      }$}"{pos=.41, description},
      shift right=1,
      start anchor={[xshift=-4pt]},
      end anchor={[xshift=-4pt]},
    ]
    \ar[d]
    \ar[
      rr,
      "{
        \scalebox{.5}{
          \color{greenii}
          \bf
          anyonic braiding equivariance
        }
      }",
      "{
        \vdash \; \widehat{\mathcal{V}_n}
      }"{swap}
    ]
    &&
    \mathrm{Vect}_{\ComplexNumbers}
    \ar[out=180-56, in=56, looseness=5, "\scalebox{1}{${
        \hspace{4pt}
        \mathclap{\CyclicGroup{\ShiftedLevel}}
        \hspace{5pt}
      }$}"{pos=.41, description},
      shift right=0,
      start anchor={[xshift=-0pt]},
      end anchor={[xshift=-0pt]},
    ]
    \ar[d]
    &[-33pt]
    \scalebox{.5}{
      \color{darkblue}
      \bf
      \def\arraystretch{.9}
      \begin{tabular}{c}
        Moduli stack of
        \\
        vector bundles
      \end{tabular}
    }
    \\
      \scalebox{.5}{
        \color{darkblue}
        \bf
        \begin{tabular}{c}
         Configuration space
         \\
         of $n$ anyonic quanta
        \end{tabular}
      }
    &
    \ConfigurationSpace{n}
    (
      \TopologicalSpace
    )
    \ar[
      drr,
      shorten <=-5pt,
      "{
        \scalebox{.5}{
          \color{greenii}
          \bf
          ``fictitious gauge field''
        }
      }"{sloped}
    ]
    \ar[
      d,
      shorten <=-0pt
    ]
    \ar[
      rr,
      "{
        \scalebox{1}{ \tiny
          \color{orangeii}
          \bf
          \begin{tabular}{c}
            interacting $n$-particle
            \\
            valence bundle
          \end{tabular}
        }
      }",
      "{
        \vdash {\mathcal{V}_n}
      }"{swap, pos=.7}
    ]
    &&
    \HomotopyQuotient
      { \mathrm{Vect}_\ComplexNumbers }
      { \CyclicGroup{\ShiftedLevel} }
    \ar[d]
    &
    \scalebox{.5}{
      \color{darkblue}
      \bf
      \begin{tabular}{c}
        Homotopy quotient
        \\
        by anyon phases
      \end{tabular}
    }
    \\
    \scalebox{.5}{
      \def\arraystretch{.9}
      \color{darkblue}
      \bf
      \begin{tabular}{c}
        Classifying stack of
        \\
        fundamental group
      \end{tabular}
    }
    &
    \mathbf{B}
    \;
    \underset{
      =: \;
      \underset{
        \mathclap{
          \raisebox{-3pt}{
            \scalebox{.6}{
              \color{darkblue}
              \bf
              braid group
            }
          }
        }
      }{
        \mathrm{Br}_\TopologicalSpace
        (n)
      }
    }{
      \underbrace{
        \pi_1
        \!
        \Big(
          \ConfigurationSpace{n}
          (\TopologicalSpace)
        \Big)
      }
    }
    \ar[
      rr,
      shorten <=-2pt,
      "{
        \mathbf{B} \phi
      }"{swap},
      "{
        \scalebox{.5}{
          \color{greenii}
          \bf
          braiding phases
        }
      }"{pos=.4}
    ]
    &&
    \mathbf{B}
    \underset{
      \mathclap{
        \subset \, \CircleGroup
      }
    }{
    \underbrace{
      \CyclicGroup{\ShiftedLevel}
    }}
    &
    \scalebox{.5}{
      \color{darkblue}
      \bf
      \def\arraystretch{.9}
      \begin{tabular}{c}
        Moduli stack of
        \\
        anyon phases
      \end{tabular}
    }
  \end{tikzcd}
\end{equation}

\noindent
{\bf Anyon defects and Local system-twisted de Rham cohomology of configuration spaces.}
The presentation \eqref{AnyonicValenceBundles} makes it clear how the theory of Chern classes of valence bundles -- as familiar from the topological phases of Chern insulators (Exp. \ref{NoQuantumSymmetryAndSpinOrbitCOupling}) -- generalizes to the case of interacting and anyonic topological order (\hyperlink{FigureAnyonsFromInteractions}{\it Figure 11}):
 \vspace{-1mm}
 \begin{center}
 \fcolorbox{black}{lightbackgroundgray}{
 \begin{minipage}{13cm}
 The Chern classes of anyonic $n$-particle valence bundles are in the cohomology of the $n$-point configuration space {\it with local coefficients} given by the braiding phases.
\end{minipage}
}
 \end{center}

\noindent
In detail, assume that $\TopologicalSpace$ is a smooth manifold and let $\FlatConnectionForm$
be a closed differential 1-form, whose holonomy gives the prescribed braiding phases \eqref{BraidingPhases},
hence let  $\FlatConnectionForm$ be a  {\it vector potential of the ``fictitious gauge field''}
from \hyperlink{NotionsOfAnyons}{\it Table 5}:
\vspace{-2mm}
\begin{equation}
  \label{FlatConnectionOneForm}
  \overset{
    \mathrlap{
    \raisebox{6pt}{
      \tiny
      \color{darkblue}
      \bf
      \begin{tabular}{c}
        Vector potential of
        \\
        ``fictitious gauge field''
      \end{tabular}
    }
    }
  }{
  \quad \FlatConnectionForm
  }
  \quad \quad \in\;\;
  \Omega^1\Big(
    \ConfigurationSpace{n}
    (\TopologicalSpace)
    ;\,
    \ComplexNumbers
  \Big)\big\vert_{\DeRhamDifferential = 0}
  \end{equation}
  \mbox{such that}
  \begin{equation}
  \underset{[\gamma] \in \pi_1}{\forall}
  \;\;\;
  \overset{
    \mathclap{
    \raisebox{5pt}{
      \tiny
      \color{darkblue}
      \bf
      its Aharonov-Bohm phases
    }
    }
  }{
  \exp
  \Big(
    2\pi\ImaginaryUnit
    \, \int_{\gamma} \FlatConnectionForm
  \Big)
  }
  \;=\;
  \overset{
    \mathclap{
    \raisebox{8pt}{
      \tiny
      \color{darkblue}
      \bf
      \begin{tabular}{c}
        braiding
        \\
        phases
      \end{tabular}
    }
    }
  }{
  \phi\big(
    [\gamma]
  \big)
  }
  \;\;\;
  \in
  \CyclicGroup{\ShiftedLevel}
  \,\subset\,
  \CircleGroup
  \,\subset\,
  \ComplexNumbers^{\times}
  \,.
\end{equation}
Then the Chern forms of the $\phi$-anyonic valence bundles \eqref{AnyonicValenceBundles} are in the $\FlatConnectionForm$-{\it twisted}
complex-valued de Rham cohomology (\cite[\S 2, 6]{Deligne70}\cite{ESV92}, review in \cite[\S 9.2.1]{Voisin03I}\cite[\S 2.5]{Dimca04}, cf. \cite{GS-Deligne}):
\vspace{-2mm}
\begin{equation}
  \label{OneTwistedDeRhamCohomology}
  \underset{
    \raisebox{-4pt}{
      \tiny
      \color{darkblue}
      \bf
      \begin{tabular}{c}
        complex cohomology with
        \\
        {\color{purple}local system of coefficients}
      \end{tabular}
    }
  }{
  H^n
  \Big(
    \ConfigurationSpace{n}
    (\TopologicalSpace)
    ;\,
    {\color{purple}\phi}
  \Big)
  }
  \;\;
  \simeq
  \;\;
  \underset{
    \raisebox{-4pt}{
      \tiny
      \color{darkblue}
      \bf
      {\color{purple}$\FlatConnectionForm$-twisted}
      de Rham cohomology
    }
  }{
  H_{\mathrm{dR}}
  ^{
    q
    +
    {\color{purple}\FlatConnectionForm}
  }
  \Big(
    \ConfigurationSpace{n}
    (\TopologicalSpace)
    ;\,
    \ComplexNumbers
  \Big)
  }
  \;\;\;
  :=
  \;\;\;
  H^{q}
  \bigg(\,
    \underset{
      \raisebox{-4pt}{
        \tiny
        \color{darkblue}
        \bf
        {\color{purple}$\FlatConnectionForm$-twisted}
        de Rham complex
      }
    }{
    \underbrace{
     \DeRhamComplex{\Big}{
      \ConfigurationSpace{n}(\TopologicalSpace)
      ;\,
      \ComplexNumbers
    }
    ,
    \DeRhamDifferential +
    {
      \color{purple}
      \FlatConnectionForm\wedge
    }
    }
    }
  \,\bigg)
  \,.
\end{equation}
Hence we may equivalently re-formulate the previous statement as:
 \begin{center}
 \fcolorbox{black}{lightbackgroundgray}{
 \begin{minipage}{13cm}
 The Chern classes of anyonic $n$-particle valence bundles are in the complex de Rham cohomology of the $n$-point configuration space
 {\it twisted by} the ``fictitious'' gauge potential.
\end{minipage}
}
 \end{center}
Strikingly, this {\it implies} that defects (punctures) appear as {\it defect anyons} (according to \hyperlink{NotionsOfAnyons}{\it Table 5}), as we now explain.

\medskip
Consider the case of interest where $\TopologicalSpace$ is the Brillouin torus of a 2d semi-metal with nodal points $k_1, \cdots, k_N$ removed (\hyperlink{BandStructureOfSemiMetals}{\it Figure 6}); or rather: consider for the moment the {\it annulus} resulting from cutting this punctured Brillouin torus
along one of its non-trivial 1-cycles $S^1_a \subset \DualTorus{2}$ (as in \hyperlink{HomotopyTypeOfPuncturedTorus}{\it Figure 7},
hence assumed to be disjoint from the nodal points):
\begin{equation}
  \label{PuncturedSlicedTorus}
  \overset{
    \mathclap{
    \raisebox{5pt}{
      \tiny
      \color{darkblue}
      \bf
      \def\arraystretch{.9}
      \begin{tabular}{c}
        Brillouin torus
        cut along a 1-cycle
        \\
        with $N$ nodal punctures
      \end{tabular}
    }
    }
  }{
  \Big(
  \DualTorus{2}
  \setminus
  S^1_a
  \Big)
  \setminus
  \big\{
    k_1, \cdots, k_N
  \big\}
  }
  \;\;\;
  \simeq
  \;\;\;
  \overset{
    \mathclap{
    \raisebox{5pt}{
      \tiny
      \color{darkblue}
      \bf
        annulus with $N$ punctures
    }
    }
  }{
  \big(
  \mathbb{D}^2
  \setminus
  \{k_0\}
  \big)
  \setminus
  \big\{ k_1, \cdots, k_N \big\}
  }
  \;\;\;
  \simeq
  \;\;\;
  \overset{
    \mathclap{
    \raisebox{5pt}{
      \tiny
      \color{darkblue}
      \bf
      complex plane with $N + 1$ punctures
    }
    }
  }{
  \ComplexPlane
  \setminus
  \big\{ k_0, k_1, \cdots, k_N \big\}
  }
  \;\;\;
  \simeq
  \;\;\;
  \overset{
    \mathclap{
    \raisebox{5pt}{
      \tiny
      \color{darkblue}
      \bf
      Riemann sphere with $N+2$ punctures
    }
    }
  }{
  \ComplexNumbers P^1
  \setminus
  \big\{
    k_0, k_1, \cdots , k_N, k_{N+1}
  \big\}
  }
  \,.
\end{equation}
(On the far right of \eqref{PuncturedSlicedTorus}, $k_0, k_{N+1}$ are any further pairwise distinct points which, without restriction of generality,
we may think of as fixed to $k_0 = 0$ and $k_{N+1} = \infty$.)
The corresponding configuration space inherits the evident complex structure, whose canonical holomorphic coordinate functions we denote by
\begin{equation}
  \label{CanonicalCoordinatesOnConfigurationSpace}
  k^1, \cdots, k^n
  \;\;
  :
  \;\;
  \ConfigurationSpace{n}
  \Big(
    \ComplexPlane
    \setminus
    \big\{
      k_0, k_1, \cdots, k_N, k_{N+1}
    \big\}
  \Big)
  \xrightarrow{\phantom{---}}
  \ComplexNumbers
  \,.
\end{equation}
In terms of these holomorphic coordinates, the ``fictitious'' vector potential \eqref{FlatConnectionOneForm} may be chosen to be the following holomorphic differential form (cf. \cite[(41)]{SS22AnyonicDefectBranes} following \cite[(19)]{FeiginSchechtmanVarchenko94}, and compare
\hyperlink{AdiabaticBraiding}{\it Figure 1}):
\begin{equation}
  \label{CanonicalFlatConnectionOneForm}
  \underset{
    \raisebox{-6pt}{
      \tiny
      \color{darkblue}
      \bf
      \begin{tabular}{c}
        ``fictitious''
        \\
        gauge potential
      \end{tabular}
    }
  }{
  \FlatConnectionForm(\vec \weight, \ShiftedLevel)
  }
  \;\;
  :=
  \;\;
  \underset{
    \scalebox{.7}{$
      1 \leq i \neq j \leq n
    $}
  }{\sum}
  \;
  \underset{
    \mathrlap{
    \hspace{5pt}
    \raisebox{-6pt}{
      \tiny
      \color{darkblue}
      \bf
      \def\arraystretch{.9}
      \begin{tabular}{c}
        quanta-quanta (qq)
        \\
        braiding phases
      \end{tabular}
    }
    }
  }{
  \frac{2}{\ShiftedLevel}
  }
  \;
  \frac{
    \Differential k^i
  }{
    (k^i - k^j)
  }
  \;\;-
  \underset{
    \scalebox{.7}{$
      \begin{array}{c}
        0 \leq I \leq N
        \\
        1 \leq i \leq n
      \end{array}
    $}
  }{\sum}
  \underset{
    \mathrlap{
    \hspace{5pt}
    \raisebox{-6pt}{
      \tiny
      \color{darkblue}
      \bf
      \def\arraystretch{.9}
      \begin{tabular}{c}
        quanta-defect (qd)
        \\
        braiding phases
      \end{tabular}
    }
    }
  }{
    \frac{
      \weight_I
    }{
      \ShiftedLevel
    }
  }
  \;
  \frac{\Differential k^i}{
    (k^i - k_I)
  }\,.
\end{equation}
Here the first sum in \eqref{CanonicalFlatConnectionOneForm}  gives the constant braiding phases among the anyon quanta as considered in
\cite[(1.1)]{CWWH89}; concretely, our $\kappa$ equals 4 times the ``$n$'' used in \cite{CWWH89}. Then the second summand in
\eqref{CanonicalFlatConnectionOneForm} specifies the additional phases obtained when braiding an anyonic quantum around an anyonic
defect. Such a mixed quanta/defect-phase is  parameterized by
an integer $\weight_I$  modulo $\ShiftedLevel$, to be called  the {\it weight of/at the $I$-th defect}:
\vspace{-1mm}
$$
  \weight_I \,\in\, \Integers,
  \;\;
  \;\;\;\;\;\;
  [\weight_I]
  \;\;
  \in
  \;\;
  \CyclicGroup{\ShiftedLevel}
  \,,
  \;\;\;\;\;
  \vec \weight
  \,:=\,
  \overset{
    \mathclap{
    \raisebox{3pt}{
      \tiny
      \color{darkblue}
      \bf
      ``incoming'' weights
    }
    }
  }{
  (\weight_0, \weight_1, \cdots, \weight_N)
  }
  ,\,
  \;\;\;\;\;
  \overset{
    \mathclap{
    \raisebox{6pt}{
      \tiny
      \color{darkblue}
      \bf
      ``outgoing'' weight
    }
    }
  }{
  \weight_{N+1}
  }
  \;:=\;
  \Big(
    {\sum}_{I=0}^N
    \weight_I
  \Big)
  -
  {\color{orangeii}  n}
  \,.
$$

\vspace{-1mm}
\noindent Concretely, a differential form $\Psi$ on the configuration space which is $\FlatConnectionForm$-twisted
closed is equivalently an ordinary closed form $\widehat{\Psi}$ \eqref{EquivariantFunctionOnUniversalCover}
on the universal cover of the configuration space of the following form
(cf. \cite[(42)]{SS22AnyonicDefectBranes} following \cite[(20)]{FeiginSchechtmanVarchenko94}\cite[(2.1)]{SchechtmanVarchenko90},
called the ``master function'' in \cite[\S 2.1]{SlinkinVarchenko18}):
\begin{equation}
  \label{MasterFunctionLaughlinState}
  \overset{
    \mathclap{
    \raisebox{+6pt}{
      \tiny
      \color{darkblue}
      \bf
      \def\arraystretch{.9}
      \begin{tabular}{c}
        Twisted-closed
        wavefunction
        on configuration space
      \end{tabular}
    }
    }
  }{
  \Differential \Psi + \FlatConnectionForm(\vec \weight, \ShiftedLevel) \wedge\Psi  = 0
  }
  \hspace{30pt}
  \xleftrightarrow{\phantom{--}}
  \hspace{30pt}
  \overset{
    \mathclap{
    \raisebox{+5pt}{
      \tiny
      \color{darkblue}
      \bf
      \def\arraystretch{.9}
      \begin{tabular}{c}
        Equivariant
        closed wavefunction
        \\
        on universal cover
      \end{tabular}
    }
    }
  }{
\widehat{\Psi}\big(\,\widehat{k}^1, \cdots, \widehat{k}^n\big)
}
  =\;\;\;
  \underset{
    \mathclap{
    1 \leq i < j \leq n
    }
  }{\prod}
  \;\,
  \overset{
    \mathclap{
    \raisebox{3pt}{
      \tiny
      \color{darkblue}
      \bf
      \def\arraystretch{.9}
      \begin{tabular}{c}
        quanta-quanta (qq)
        \\
        braiding phases
      \end{tabular}
    }
    }
  }{\;\;
  \big(\,
    \widehat{k}^i - \widehat{k}^j
  \big)^{2/\ShiftedLevel}
  }
  \;\;\;\;
  \underset{
    \mathclap{
    \scalebox{.7}{$
      \begin{array}{c}
        0 \leq I \leq N
        \\
        1 \leq i \leq n
      \end{array}
    $}
    }
  }{\prod}
  \;\,
  \overset{
    \mathclap{
    \raisebox{4pt}{
      \tiny
      \color{darkblue}
      \bf
      \def\arraystretch{.9}
      \begin{tabular}{c}
        quanta-defect (qd)
        \\
        braiding phases
      \end{tabular}
    }
    }
  }{
  \big(\,
    \widehat{k}^i - k_I
  \big)^{ \weight_I/\ShiftedLevel }
  }
  \,\cdot\,
  {\Psi}
  \big(
    k^1, \cdots, k^i
  \big)
  \,,
\end{equation}
\vspace{-.4cm}

\noindent
where $\widehat{k}^i$ denote coordinates on the universal cover, while
$k^i$ denote the pullbacks of the corresponding coordinates
\eqref{CanonicalCoordinatesOnConfigurationSpace}
on the configuration space itself.

\begin{remark}[\bf Generalized Laughlin wavefunctions with mixed quanta-defect braiding phases]
 \label{MixedQuantaDefectBraidingPhasesInTheLiterature}
The form \eqref{MasterFunctionLaughlinState} is just that of {\it generalized Laughlin wavefunctions} for anyons considered in
  \cite[(11)]{Halperin84}\cite[(89), (93)]{NSSFS08}\cite[(3)]{Lan19},
  which generalize the original
{\it Laughlin wavefunctions} \cite{Laughlin83}\cite[\S 2.2]{MooreRead91} (review in \cite[\S 2.1]{Girvin04})
to a situation with mixed quanta-defect braiding phases.
\end{remark}

\vspace{-1mm}
\noindent Hence for given $\ShiftedLevel \in \NaturalNumbers_{+}$ -- determining the phase picked up by braiding any two anyonic quanta around
each other -- equation \eqref{CanonicalFlatConnectionOneForm} parameterizes general quanta-defect braiding phases, subject only to the constraint
that these come in integer multiples of {\it half} the quanta-quanta braiding phases. This curious constraint has its secret origin in the root
lattice geometry of the Lie algebra $\suTwo$ and guarantees that the following crucial fact holds
(\cite[Cor. 3.4.2, Rem. 3.4.3]{FeiginSchechtmanVarchenko94}\cite[Prop. 2.17]{SS22AnyonicDefectBranes}\footnote{
  The discussion in \cite{FeiginSchechtmanVarchenko94}\cite{SS22AnyonicDefectBranes} is in terms of ($\FlatConnectionForm$-twisted)
  {\it holomorphic} de Rham cohomology. This is still equivalent
  \eqref{OneTwistedDeRhamCohomology}
  to the complex cohomology (with local system $\phi$ of coefficients) of the configuration space, since
  (e.g. \cite[Thm. 2.5.11]{Dimca04})
  configuration spaces of
  points in punctured Riemann surfaces are complex Stein domains \cite[Rem. 2.2]{SS22AnyonicDefectBranes}.
}):

 \begin{center}
 \fcolorbox{black}{lightbackgroundgray}{
 \begin{minipage}{14.4cm}
 The complex de Rham cohomology of configuration space,
twisted \eqref{OneTwistedDeRhamCohomology}
by the ``fictious vector potential'' \eqref{CanonicalFlatConnectionOneForm},
naturally contains the space of
$\suTwo$-{\it conformal blocks}, identified with
the following
Laughlin state (Rem. \ref{MixedQuantaDefectBraidingPhasesInTheLiterature})
Slater determinants
\eqref{SlateDeterminant}
weighted by the canonical holomorphic volume form:
\end{minipage}
}
\end{center}
\begin{equation}
  \label{ConformalBlocksInsideTwistedCohomology}
  \hspace{-5mm}
  \begin{tikzcd}[column sep=1pt]  \overset{
    \mathclap{
    \raisebox{8pt}{
      \tiny
      \color{darkblue}
      \bf
      \begin{tabular}{c}
        $\suTwo$-affine
        \\
        conformal blocks
        \\
        at level $\kappa - 2$
      \end{tabular}
    }
    }
  }{
  \ConformalBlocks_{\suTwoAffine{\ShiftedLevel-2}}
  }
  \bigg(\!\!
    \overset{
      \mathrlap{
      \raisebox{9pt}{
        \tiny
        \color{darkblue}
        \bf
        \begin{tabular}{c}
          on the Riemann sphere
          \\
          with
          {\color{greenii}$N+1$ punctures}
        \end{tabular}
      }
      }
    }{
      {
        ({\color{greenii}k_I})_{I=0}^{{\color{greenii}N}+1}
      }
    }
    ,
    \overset{
      \raisebox{3pt}{
        \tiny
        \color{darkblue}
        \bf
        of given
        {\color{purple} weights}
      }
    }{
    \Big\{
      {
      (
      {
        \color{purple}
        \weight_I
      }
      )_{I =0}^{\color{greenii}N},
      }
      \weight_{N + 1}
        =
      \Big(
        \underoverset{I = 0}{N}{\sum} \weight_I
      \Big)
      -
      {\color{orangeii}n}
    \Big\}
    }
 \!\!\bigg)
  \ar[
    rr,
    hook
  ]
  &&
  \overset{
    \mathclap{
    \raisebox{6pt}{
      \tiny
      \color{darkblue}
      \bf
      \begin{tabular}{c}
      de Rham cohomology twisted
      \\
      by
      {\color{purple}``fictitious'' vector potential}
      \end{tabular}
    }
    }
  }{
  H^{
    {\color{orangeii}n}
    +
    \FlatConnectionForm
    ({\color{purple}\vec\weight}, \ShiftedLevel)
  }_{\mathrm{dR}}
  }
  \bigg(
    \overset{
      \mathclap{
      \raisebox{5pt}{
        \tiny
        \color{darkblue}
        \bf
        \begin{tabular}{c}
          configuration space of
          \\
          {\color{orangeii} $n$ quanta}
          among
          {\color{greenii} $N$ defects}
        \end{tabular}
      }
      }
    }{
    \ConfigurationSpace{{\color{orangeii}n}}
    \Big(
      \big(
        \DualTorus{2}
          \setminus
        S^1_a
      \big)
        \setminus
        \{
          {\color{greenii}
          k_I}
        \}_{I=1}^{{\color{greenii}N}}
    \Big)
    }
    ;\,
    \ComplexNumbers
  \bigg)
  \\[-15pt]
    \scalebox{1}{$
    \underset{
      \mathclap{
      \raisebox{-4pt}{
        \tiny
        \begin{tabular}{c}
          \color{darkblue}
          \bf
          conformal block for
          {\color{greenii}$N$} punctures
          and {\color{orangeii}$n$} insertions
          \\
          (cf.
          \cite[2.3.3, 2.3.6]{FeiginSchechtmanVarchenko94}
          \cite[Ex. 2.14]{SS22AnyonicDefectBranes})
        \end{tabular}
      }
      }
    }{
    f_{I_1}
    \cdots
    f_{I_{\color{orangeii}n}}
   \vert
     \HighestWeightVector_{1}
     \cdots,
     \HighestWeightVector
       _{{\color{greenii}\NumberOfPunctures}}
   \rangle
   }
   $}
   &\longmapsto&
   \scalebox{1}{$
   \underset{
     \mathclap{
     \raisebox{+1pt}{
       \tiny
       \color{darkblue}
       \bf
       Slater determinant
       Laughlin state
     }
     }
   }{
   \Bigg[
   \mathrm{det}
   \bigg(\!\!\!
   \Big(
     \frac{
       {
         {
         \color{purple}
         \weight
         }_{{\color{greenii}I_j}}
       }
     }{\ShiftedLevel}
     \frac{
       1
     }
     {
       \mathclap{\phantom{\vert^{\vert}}}
       k^{\color{orangeii}i}
       -
       k_{{\color{greenii}I_j}}
     }
   \Big)_{i,j = 1}^{{\color{orangeii}n}}
   \bigg)
   }
   \,
   \Differential k^1
     \wedge
     \cdots
     \wedge
   \Differential k^{\color{orangeii}n}
   \Bigg].
  $}
  \end{tikzcd}
\end{equation}

\begin{remark}[\bf Topologically ordered anyonic ground states in terms of modular tensor categories]
\label{TopologicallyOrderedGroundStatesAndModularTensorCategories}
$\,$

\noindent {\bf (i)} Chiral {\it conformal blocks} as appearing in \eqref{ConformalBlocksInsideTwistedCohomology}
are thought to be the Laughlin-type ground state wavefunctions of non-abelian defect anyons
(this is due to \cite{MooreRead91}\cite{ReadRezayi99},
reviewed in \cite[III.D.2]{NSSFS08}, further developments in \cite{GuHaghighatLiu21}\cite{ZWXT21}, review in \cite[\S 9]{Lerda92}\cite[\S 8.3]{Wang10}\cite{Su18}), specifically
(\cite{Ino98})
of ``$\suTwo$-anyons'' (i.e. described by an $\suTwoAffine{\ShiftedLevel-2}$ CS/WZW theory, as in \cite{FLW02})
such as Majorana/Ising-anyons for $\kappa - 2=2$  and Fibonacci-anyons for $\kappa-2 = 3$ (e.g. \cite{TTWL08} \cite{GATHLTW13}\cite[p. 11]{SarmaFreedmanNayak15}\cite[\S III]{JohansenSimula20}). In general this is the case for {\it fractional} shifted levels $\ShiftedLevel/\Denominator$, see further below around \hyperlink{RelationBetweenFractionalLevelAndLogarithmicCFT}{\it Table 12}.

\vspace{1mm}
\noindent {\bf (ii)} In fact, the
unitary  {\it modular tensor categories} (MTCs) which arise as representation categories of chiral 2d conformal field theories (CFTs)
such as of the $\suTwo$-affine  CFT above (the chiral $\suTwo$ WZW model), specifically of their {\it vertex operator algebras} (VOAs),
are expected to be the mathematical structure accurately encoding topological order and anyon species.
In particular, modular tensor categories are {\it braided fusion categories}, and their category-theoretic braiding is widely thought
to reflect the corresponding anyonic braiding.
The  origin of this idea may be \cite[\S8, \S E]{Kitaev06}, where it is
argued in a concrete model. The general statement has become folklore, traditionally re-iterated without proof or even attribution (e.g. in \cite[pp. 28]{NSSFS08}\cite[\S 6.3]{Wang10}\cite[\S 2.4]{RowellWang18}\cite{Bonderson21}) and claimed to be ``mature'' in \cite[p. 1]{Wang17}. That a proof had actually been missing was highlighted recently in \cite{Valera21} (which goes on to establish a list of sufficient conditions that need to be established for the statement to hold.)

Here  we find a derivation of MTC structure of anyon braiding from a first-principles definition of anyons as in \hyperlink{NotionsOfAnyons}{\it Table 5}:

\vspace{1mm}
\noindent {\bf (iii)}
We may observe that the braiding structure in an MTC arising as a representation category of 2d CFT is entirely determined by
the CFT's conformal blocks on the punctured Riemann sphere (this fact is highlighted in \cite[p. 266]{EGNO15}\cite[p. 36]{Runkel}; phrased in terms
of modular functors this is a result due to \cite{AndersenUeno12}). In this sense, the conformal blocks appearing in
\eqref{ConformalBlocksInsideTwistedCohomology} may be regarded as the missing link between the physics of anyons according
to \hyperlink{NotionsOfAnyons}{\it Table 5} and the expected classification of species of anyons (really: {\it defect anyons}) by MTCs:

\vspace{-.6cm}
$$
  \begin{tikzcd}
  [column sep=40pt]
    \underset{
      \raisebox{-3pt}{
        \scalebox{.8}{
          (\hyperlink{NotionsOfAnyons}{Table 5})
        }
      }
    }{
    \fbox{\!\!\!\!\!
      \def\arraystretch{.9}
      \begin{tabular}{c}
        Physical
        \\
        defect anyons
      \end{tabular}
    \!\!\!\!}
    }
    \ar[
      rr,
      Rightarrow,
      "{
        \scalebox{.8}{
          \color{orangeii}
          \eqref{ConformalBlocksInsideTwistedCohomology},
          Fact \ref{TopologicallyOrderedGroundStatesViaChernForms}
        }
      }"{swap, yshift=-2pt}
    ]
    &&
    \underset{
      \raisebox{-3pt}{
        \scalebox{.8}{
          \cite{MooreRead91}\cite{ReadRezayi99}
        }
      }
    }{
    \fbox{\!\!\!\!\!
      \begin{tabular}{c}
        Topologically ordered ground states
        \\
        as genus-zero
        conformal blocks
      \end{tabular}
    \!\!\!\!}
    }
    \ar[
      r,
      Rightarrow,
      "{
        \scalebox{.8}{
          \cite{AndersenUeno12}
        }
      }"{swap, yshift=-2pt,}
    ]
    &
    \underset{
      \raisebox{-3pt}{
        \scalebox{.8}{
          \cite[\S 8 \& \S E]{Kitaev06}
        }
      }
    }{
    \fbox{\!\!\!\!\!
      \def\arraystretch{.9}
      \begin{tabular}{c}
        Anyon species encoded in
        \\
        modular tensor
        categories
      \end{tabular}
    \!\!\!\!}
    }
  \end{tikzcd}
$$

\end{remark}

Hence, in view of Rem. \ref{TopologicallyOrderedGroundStatesAndModularTensorCategories}, the combination of the above boxed facts yields the following conclusion:

\begin{fact}[\bf Topologically ordered ground states via Chern forms of interacting valence bundles]
\label{TopologicallyOrderedGroundStatesViaChernForms}
The complex vector space of complex Chern-de Rham classes of interacting valence bundles of any number $n \geq 1$ of anyonic quanta among $N$ anyonic nodal defects in the cut Brillouin torus $\DualTorus{2}\setminus S^1_a$ naturally contains the Hilbert space of topologically ordered ground states of $\{0, \cdots, I, \cdots, N+1\}$ $\suTwo$-anyons, whose:

\begin{itemize}[leftmargin=.6cm]
\item[--] {\rm level} (i.e. species: Majorana, Ising,  Fibonacci, ...) is $k = 2/\phi^{qq} - 2$,
for $\phi^{qq} \,\in\, \RationalNumbers \twoheadrightarrow \RationalNumbers/\Integers \,\simeq \CircleGroup$ being the phase picked up by braiding a pair of the anyonic quanta;

\item[--] {\rm weight} (i.e. $\suTwoAffine{k}$-spin) is $\weight_I = 1/\phi_I^{qd}$,
for $\phi_I^{qd} \,\in\, \RationalNumbers \twoheadrightarrow  \RationalNumbers/\Integers \simeq \CircleGroup$ being the phase picked up by braiding an anyonic quantum around the $I$th anyonic defect.
\end{itemize}
\end{fact}

It remains to discuss how exactly this encodes the non-abelian braiding statistics expected in topologically orderered ground states:

\medskip
\noindent
{\bf Non-abelian topological order and hypergeometric KZ-solutions.}
It is familiar from the discussion 3d Chern semimetals (e.g. \cite[(2.3)]{MathaiThiang15Semimetals})  that the topological charge of a nodal point is the integral of the valence bundle's Chern form over a cycle in the Brillouin torus which encloses the nodal point. In evident variation of this principle, we must regard cycles with complex coefficients in the local system $\phi$ \eqref{AnyonicValenceBundles}
-- hence the homology-dual of the twisted cohomology \eqref{OneTwistedDeRhamCohomology} --
as reflecting the nodal configurations of the anyonic $n$-particle interacting system:
\begin{equation}
  \label{AnyonicNodalConfiguration}
  \begin{tikzcd}[
    column sep=60pt
  ]
  \overset{
    \mathclap{
    \raisebox{8pt}{
      \tiny
      \color{darkblue}
      \bf
      {\color{greenii}$N$}-nodal points
    }
    }
  }{
  \{k_I\}_{I =1}^{{\color{greenii}N}}
  }
  \ar[
    |->,
    r,
    shorten=3pt,
    "{
      \mbox{
        \tiny
        \color{greenii}
        \bf
        \def\arraystretch{.9}
        \begin{tabular}{c}
          flat section of
          \\
          GM connection
        \end{tabular}
      }
    }"{swap, yshift=-1pt}
  ]
  &
  \overset{
    \mathclap{
    \raisebox{8pt}{
      \tiny
      \color{darkblue}
      \bf
      \def\arraystretch{.9}
      \begin{tabular}{c}
        {\color{orangeii}$n$}-cycle around
        nodal points
      \end{tabular}
    }
    }
  }{
    \sigma
    \big(
      \{k_I\}_{I=1}^{{\color{greenii}N}}
    \big)
  }
  \;\;\;\in\;\;
  \overset{
    \mathclap{
    \raisebox{3pt}{
      \tiny
      \color{darkblue}
      \bf
      \def\arraystretch{.9}
      \begin{tabular}{c}
        homology of configuration space
        with
        {\color{purple}local system of coefficients}
      \end{tabular}
    }
    }
  }{
  H_{\color{orangeii}n}
  \bigg(
    \ConfigurationSpace{{\color{orangeii}n}}
    \Big(
      \big(
        \DualTorus{2}
        \setminus
        S^1_a
      \big)
      \setminus
      \big\{
        k_I
      \big\}_{I=1}^{{\color{greenii}N}}
    \Big)
    ;\,
    {\color{purple}\phi}
  \bigg).
  }
  \end{tikzcd}
 \end{equation}
In degree 1 such twisted cycles are {\it Pochhammer loops} (see around \cite[Fig. 1.1]{Varchenko95}\cite[Fig. 4.1]{EtingofFrenkelKirillov98});
for a construction in general degrees see \cite[\S 7.6]{EtingofFrenkelKirillov98}.
Here one wants to assume that the choice of the twisted cycle $\sigma$ in \eqref{AnyonicNodalConfiguration} is carried along with the
positions $k_I$ of the nodal defects. Technically, this makes sense since the twisted homology groups on the right
 \eqref{AnyonicNodalConfiguration} canonically form a {\it flat}
vector bundle over $\ConfigurationSpace{N}\big( \DualTorus{2}\setminus S^1_a\big)$,
known as the {\it Gauss-Manin connection} (GM, see e.g. \cite{Kulikov98}\cite[Def. 9.13]{Voisin03I} for the general concept and see
\cite[\S 7.5]{EtingofFrenkelKirillov98}\cite{SS22TQC}  for the case at hand);
and we agree that $\sigma$ in
\eqref{AnyonicNodalConfiguration}
denotes a parallel section with respect to this flat GM connection.

\medskip
Now Fact \ref{TopologicallyOrderedGroundStatesViaChernForms}
says that the
possible topological fluxes through such a cycle around the nodal points subsume those indexed by $\suTwoAffine{\ShiftedLevel-2}$-conformal blocks
and, as such, are given by evaluating their associated twisted cocycles
\eqref{ConformalBlocksInsideTwistedCohomology} on the twisted cycles \eqref{AnyonicNodalConfiguration}:
\vspace{-2mm}
\begin{equation}
  \label{SectionOfDualConformalBlockBundle}
  \hspace{8mm}
  \begin{tikzcd}[
    column sep=0pt,
    row sep=2pt
  ]
  \overset{
    \mathclap{
    \raisebox{6pt}{
      \tiny
      \color{darkblue}
      \bf
      \def\arraystretch{.9}
      \begin{tabular}{c}
        Possible topological charges
        (Chern numbers)
        \\
        of anyonic nodal configuration $\sigma$
      \end{tabular}
    }
    }
  }{
    c[\sigma]
  }
  &:&
  \ConfigurationSpace{{\color{greenii}N}}
  \Big(
    \DualTorus{2}
    \setminus
    S^1_a
  \Big)
  \ar[
    rr,
    "{
      \mbox{
        \tiny
        \color{greenii}
        \bf
        \begin{tabular}{c}
          section of
          dual conformal block bundle
        \end{tabular}
      }
    }"
  ]
  &&
  \Big(
  \ConformalBlocks_{\suTwoAffine{\ShiftedLevel-2}}
  \big(
    (\weight_I)_{I=0}^{N+1}
  \big)
  \Big)^\ast
  \\[-2pt]
  & &
  \underset{
    \mathclap{
    \raisebox{-4pt}{
      \tiny
      \color{darkblue}
      \bf
      \begin{tabular}{c}
        positions of anyonic
        \\
        nodal points
      \end{tabular}
    }
    }
  }{
    (k_1, \cdots, k_{\color{greenii}N})
  }
  &\longmapsto&
  \Bigg(\!
    \underset{
      \mathclap{
      \raisebox{-3pt}{
        \tiny
        \color{darkblue}
        \bf
        \begin{tabular}{c}
          possible Chern forms
          \eqref{ConformalBlocksInsideTwistedCohomology}
          of
          \\
          anyonic interacting
          $n$-particle valence bundle
        \end{tabular}
      }
      }
    }{
    \underbrace{
    \mathclap{
      \phantom{
        \displaystyle
          \underset{
            \sigma
            \big(\!\!
              \{k_I\}_{I=1}^{\color{greenii}N}
            \big)
          }{\int}
      }
    }
    f_{I_1}
    \cdots
    f_{I_{\color{orangeii}n}}
   \vert
     \HighestWeightVector_{1}
     \cdots,
     \HighestWeightVector
       _{{\color{greenii}\NumberOfPunctures}}
   \rangle
   }
   }
  \;\;
  \mapsto \quad
  \underset{
    \raisebox{.4pt}{
      \tiny
      \color{darkblue}
      \bf
      anyonic nodal charges
    }
  }{
  \underbrace{
  \displaystyle
  \underset{
    \mathclap{
    \sigma
    \big(\!\!
      \{k_I\}_{I=1}^{\color{greenii}N}
    \!\big)
    }
  }{\int} \quad
  \mathrm{det}
   \bigg(\!\!\!
   \Big(
     \frac{
       {
         {
         \color{purple}
         \weight
         }_{{\color{greenii}I_j}}
       }
     }{\ShiftedLevel}
     \frac{
       1
     }
     {
       \mathclap{\phantom{\vert^{\vert}}}
       k^{\color{orangeii}i}
       -
       k_{{\color{greenii}I_j}}
     }
   \Big)_{i,j = 1}^{{\color{orangeii}n}}
   \bigg)
   \,
   \Differential k^1
     \wedge \cdots \wedge
   \Differential k^{\color{orangeii}n}
   }
   }
   \!\Bigg).
  \end{tikzcd}
\end{equation}
Recognizing this expression
\eqref{SectionOfDualConformalBlockBundle}
as a  {\it hypergeometric KZ-solution} (due to \cite{DJMM90}\cite[Thm. 1]{SchechtmanVarchenko90}, here specifically \cite[Cor. 3.4.2]{FeiginSchechtmanVarchenko94}; for exposition see \cite[\S 4.3, 4.4]{EtingofFrenkelKirillov98}) we find that these systems of charges of anyonic defects satisfy -- in their dependence on the anyon defect positions $k_I$ --  the Knizhnik-Zamolodchikov equation
(e.g. \cite[\S 3,4]{EtingofFrenkelKirillov98}\cite[\S 1.5]{Kohno02})
and as such constitute a non-abelian {\it monodromy braid representation}
(e.g. \cite[\S 8]{EtingofFrenkelKirillov98}\cite[\S 2.1]{Kohno02}).

In conclusion, this means that we have derived the following fact -- a prediction of {\it momentum space anyon statistics} (Rem. \ref{MomentumSpaceAnyons}\footnote{The mathematics expressed in \eqref{ConformalBlocksInsideTwistedCohomology} and \eqref{ConformalBlocksInsideTwistedCohomology} is indifferent to whether the variables $k^i$, $k_I$ are thought of as momenta or positions, and an analogous conclusion
would hold for defect anyons in position space, to the extend that this concept makes good sense in itself (cf. Rem. \ref{NaturalToroidalGeometryForMomentumSpaceTopologicalOrder}).} obtained from the above re-analysis of the established notion of multi-valued anyon wavefunctions \eqref{EquivarianceOfAnyonWaveFunctions}):

\begin{fact}[\bf Non-abelian anyon statistics of band nodes]
  \label{NonAbelianAnyonStatisticsOfBandNodes}
    The systems
    \eqref{SectionOfDualConformalBlockBundle} of charges
    carried by anyonic nodal points
    exhibit the
    braiding statistics and hence the topological order of $\suTwoAffine{\ShiftedLevel-2}$-anyons.
\end{fact}

\begin{remark}[\bf Physical origin of non-abelian anyonic braiding]
 \label{ThePhysicalOriginOfNonAbelioanAnyonicBraiding}

Observe that
Facts \ref{TopologicallyOrderedGroundStatesViaChernForms}
and \ref{NonAbelianAnyonStatisticsOfBandNodes}
explain,
via \eqref{ConformalBlocksInsideTwistedCohomology}
and \eqref{SectionOfDualConformalBlockBundle},
the origin (for $\ShiftedLevel \geq 3$) of {\it non-abelian} anyonic phases emerging from abelian braiding phases \eqref{BraidingPhases}
and their abelian Laughlin states (Rem. \ref{MixedQuantaDefectBraidingPhasesInTheLiterature}), an expected phenomenon whose explanation had previously remained at least unclear:

The non-abelian structure arises
via \eqref{SectionOfDualConformalBlockBundle}
not (just) from the global $n$-particle Laughlin wavefunctions \eqref{MasterFunctionLaughlinState}
which in themselves just see abelian braiding phases, but from the topology (the Chern classes) of the twisted/equivariant $n$-particle Bloch bundle \eqref{AnyonicValenceBundles}
which they span.
\end{remark}

\begin{remark}
[\bf Towards full classification]
While Facts \ref{TopologicallyOrderedGroundStatesViaChernForms} and \ref{NonAbelianAnyonStatisticsOfBandNodes}
seem remarkable (Rem. \ref{ThePhysicalOriginOfNonAbelioanAnyonicBraiding}), it is not yet the full answer to the classification
of anyonic topological order (to which we turn next):

\begin{itemize}[leftmargin=.6cm]
  \item
  It pertains to complex-linear combinations of twisted Chern forms,
  while the bare interacting valence bundles will contribute only a lattice of integral twisted Chern forms.
  Instead, the complex-linear combinations of Chern forms appear as {\it secondary} Chern forms of flat Berry connections on these
  interacting valence bundles, hence by passage to {\it flat differential} K-theory,
  according to Conjecture \ref{FlatKTheoryClassificationOfSemiMetals}.

  \item It pertains to unitary interacting valence bundles and hence to the special case of completely broken symmetry protection (reducing,
  in the degenerate case of trivial topological order, to the case of Chern insulators, Ex. \ref{NoQuantumSymmetryAndSpinOrbitCOupling}).
  After understanding the statement of Fact \ref{TopologicallyOrderedGroundStatesViaChernForms} in terms of flat differential K-theory,
  this will be generalized by passing to the full TED-K-theory of the orbi-orientifolded configuration space of points in the Brillouin
  torus, in generalization of Fact \ref{ClassificationOfExternalSPTPhases}.
\end{itemize}
\end{remark}

This is what we turn to next.

\medskip

\noindent
{\bf Logarithmic topological order and
Inner local systems of K-theory.}
In more detail, the chiral conformal blocks appearing as topologically ordered ground states
(Rem. \ref{TopologicallyOrderedGroundStatesAndModularTensorCategories}) are in general not just those of a rational 2dCFT like the $\suTwoAffine{\ShiftedLevel-2}$-WZW model at integral level $\ShiftedLevel - 2 \,\in\, \NaturalNumbers$, but those of a {\it logarithmic} conformal field theory
 \cite{GFN97}\cite[\S 5.4]{Flohr03}
(whose representation categories are still braided tensor categories \cite{CLR21} encoding anyonic braiding).
But in recent years it became understood
(see the references in \hyperlink{RelationBetweenFractionalLevelAndLogarithmicCFT}{\it Figure 13})
that prime examples of (chiral 2d) logarithmic CFTs are again $\suTwoAffine{\ShiftedLevel/\Denominator-2}$-CFTs, but now at {\it fractional} (meaning: rational) shifted level $\ShiftedLevel/\Denominator$ (for which $k := \ShiftedLevel/\Denominator -2 $ is called an {\it admissible fractional level}, see \cite[Rem. 2.22]{SS22AnyonicDefectBranes}).
 But the hypergeometric integral construction of braid statistics \eqref{SectionOfDualConformalBlockBundle} works verbatim at any fractional level (this was the original generality in which the construction was conceived); and it is expected (though a proof is not yet in the literature) that its relation to conformal blocks \eqref{ConformalBlocksInsideTwistedCohomology} remains valid for admissible fractional levels.

\begin{center}
\hypertarget{RelationBetweenFractionalLevelAndLogarithmicCFT}{}
\hspace{-.3cm}
\begin{tabular}{cc}
  \begin{minipage}{4.3cm}
    \footnotesize
    {\bf Figure 13 -- Expected relations between the chiral WZW model CFTs at admissible {\it fractional level}
    (as obtained from TED-K)
    to logarithmic CFTs.}

    In particular, the level $\ShiftedLevel - 2 = 0$
    (cf. \cite{PakuliakPerelomov94}\cite{Smirnov93})
    is admissible and essentially identified
    (\cite{Nichols02}\cite{NicholsThesis02})
    with the logarithmic triplet algebra of the ``$c = -2$'' model
    (\cite{GaberdielKausch96}\cite[\S 3]{Gaberdiel03})
    that is related to Laughlin wavefunctions
    in \cite{GFN97}\cite[\S 5.4]{Flohr03}.
  \end{minipage}
  &
 \begin{tikzcd}[
   column sep=-20pt,
   row sep=15pt
 ]
 &
 \fbox{
   \begin{tabular}{c}
     TED-K cohomology
     \\
     of configuration spaces
   \end{tabular}
 }
 \ar[
   d,
   -,
   shorten=-2pt,
   "{
     \scalebox{.7}{
       \color{greenii}
       reflects
     }
   }",
   "{
     \scalebox{.7}{
       \begin{tabular}{r}
       \eqref{DirectSumOfFractionalAnyons}
       \\
       \cite[Thm. 2.18]{SS22AnyonicDefectBranes}
       \end{tabular}
     }
   }"{swap}
 ]
 \\
 &
  \fbox{
  \def\arraystretch{.9}
  \hspace{-5pt}
  \begin{tabular}{c}
    Fractional-level
    \\
    $\suTwoAffine{\ShiftedLevel/\Denominator-2}$-CFT
  \end{tabular}
  \hspace{-5pt}
}
\ar[
  dr,
  -,
  "{
    \scalebox{.7}{
      \color{greenii}
      is
    }
  }"{swap, sloped},
  "{
    \scalebox{.7}{
      \begin{tabular}{l}
        \cite{Gaberdiel01}\cite[(3.1)]{CreutzigRidoutII13}\cite[(5.1)]{KawasetsuRidout19}\cite[(1.2)]{KawesetsuRidout22},
        \\
        review in \cite[\S 5]{Gaberdiel03}\cite{Ridout10}\cite{Ridout20}
      \end{tabular}
    }
  }"
]
\ar[
  dl,
  -,
  "{
    \scalebox{.7}{
      \color{greenii}
      yields
    }
  }"{swap, sloped, pos=.3},
  "{
    \scalebox{.7}{
      \begin{tabular}{r}
      \eqref{ConformalBlocksInsideTwistedCohomology}
      \eqref{SectionOfDualConformalBlockBundle}
      \\
      \cite[Conj. 2.21]{SS22AnyonicDefectBranes}
      \end{tabular}
    }
  }"{swap}
]
\\
 \fbox{
  \def\arrastretch{.9}
    \hspace{-5pt}
   \begin{tabular}{c}
    Topologically ordered
    \\
    ground states
  \end{tabular}
  \hspace{-5pt}
}
\ar[
  rr,
  -,
  "{
    \scalebox{.7}{
      \color{greenii}
      are generally given by
    }
  }",
  "{
    \scalebox{.7}{
      \cite{GFN97}\cite[\S 5.4]{Flohr03}
    }
  }"{swap}
]
&&
\fbox{
   \hspace{-5pt}
   \def\arraystretch{.9}
   \begin{tabular}{c}
    Chiral 2d
    \\
    logarithmic CFT
    \end{tabular}
    \hspace{-5pt}
 }
\end{tikzcd}
\end{tabular}
\end{center}

Hence assuming the situation in \hyperlink{RelationBetweenFractionalLevelAndLogarithmicCFT}{\it Figure 13}, we conclude that in general
the shifted level $\ShiftedLevel$ appearing in Facts \ref{TopologicallyOrderedGroundStatesViaChernForms} and \ref{NonAbelianAnyonStatisticsOfBandNodes}
must be understood as a rational number. Equivalently, for fixed $\ShiftedLevel \,\in\, \NaturalNumbers_{\geq 2}$ and fixed form $\FlatConnectionForm(\vec\weight,-)$ of the ``fictitious'' gauge potential   \eqref{CanonicalFlatConnectionOneForm} we find the full
space of topologically ordered ground states \eqref{ConformalBlocksInsideTwistedCohomology}
with braiding phases being $\ShiftedLevel$-th roots of unity
inside the direct sum of $\FlatConnectionForm(\vec \weight,\ShiftedLevel/\Denominator)$-twisted cohomology groups as the denominator
$\Denominator$ ranges between 1 and $\ShiftedLevel$:
\vspace{-2mm}
\begin{equation}
  \label{DirectSumOfFractionalAnyons}
  \hspace{-4mm}
  \begin{tikzcd}[column sep=1pt]
  \Bigg\{\!\!\!\!
  \scalebox{.8}{
    \hspace{-6pt}
    \begin{tabular}{c}
      Topologically ordered
      ground states of
      \\
      {\color{orangeii} $n$ anyonic quanta}
      among
      {\color{greenii} $N$ anyon defects}
      \\
      with braiding phases in
      $\CyclicGroup{\ShiftedLevel}
      \,\subset\, \CircleGroup$
    \end{tabular}
    \hspace{-6pt}
  }
 \!\!\! \Bigg\}
  \ar[
    r,
    hook,
    shorten >=-20pt
  ]
  &
  \qquad
  \overset{
    \mathclap{
    \raisebox{14pt}{
      \tiny
      \color{darkblue}
      \bf
      \begin{tabular}{c}
        weights determining
        \\
        form of ``fictitious''
        \\
        gauge potentials
      \end{tabular}
    }
    }
  }{
  \underset{
    \scalebox{.7}{$
      \begin{array}{c}
        \vec\weight
        \,\in\,
        \\
        \{0,\cdots, \ShiftedLevel-1\}^{
          {\color{greenii}N+1}
        }
    \end{array}
    $}
  }{\bigoplus}
  }
  \quad
  \underset{
    \mathclap{
    \scalebox{.8}{$
    \underset{
      \raisebox{-4pt}{
       \scalebox{.7}{
        \color{darkblue}
        \bf
        \begin{tabular}{c}
          $\CyclicGroup{\ShiftedLevel}$-equivariant
          complex K-theory
          with complex coefficients
          \\
          twisted by
          the
          {\it inner local system}
          given by
          $\FlatConnectionForm(\vec \weight,\ShiftedLevel)$
        \end{tabular}
       }
      }
    }{
    \mathrm{KU}^{
      {\color{orangeii}n}
      +
      \FlatConnectionForm(\vec\weight,\ShiftedLevel)
    }
    \Big(
      \big(
        (-)
        \times
        \HomotopyQuotient
          {\ast}
          {\CyclicGroup{\ShiftedLevel}}
      \big)
      ;\, \ComplexNumbers
    \Big)
    }
    $}
    }
  }{
  \underbrace{
  \mathclap{\phantom{\vert_{\vert_{\vert_{\vert_{\vert_{\vert_{\vert_{\vert_{\vert_{\vert}}}}}}}}}}}
  \overset{
    \mathclap{
    \raisebox{10pt}{
      \tiny
      \color{darkblue}
      \bf
      \begin{tabular}{c}
      de Rham cohomology twisted
      \\
      by
      {``fictitious'' vector potential}
      \\
      at some shifted {\color{purple} fractional level}
      \end{tabular}
    }
    }
  }{
  \underset{
    1 \leq {\color{purple}r} \leq \ShiftedLevel
  }{
    \bigoplus
  }
  \;
  H^{
    {\color{orangeii}n}
    +
    \FlatConnectionForm
    ({\vec\weight}, \ShiftedLevel/{\color{purple}\Denominator})
  }_{\mathrm{dR}}
  }
  }
  }
  \bigg(
    \overset{
      \mathclap{
      \raisebox{5pt}{
        \tiny`
        \color{darkblue}
        \bf
        \begin{tabular}{c}
          configuration space of
          \\
          {\color{orangeii} $n$ quanta}
          among
          {\color{greenii} $N$ defects}
        \end{tabular}
      }
      }
    }{
    \ConfigurationSpace{{\color{orangeii}n}}
    \Big(
      \big(
        \DualTorus{2}
          \setminus
        S^1_a
      \big)
        \setminus
        \{
          {\color{greenii}
          k_I}
        \}_{I=1}^{{\color{greenii}N}}
    \Big)
    }
  \!\!\bigg).
  \end{tikzcd}
\end{equation}

  \vspace{-2mm}
\noindent Remarkably, as shown under the brace, just this kind of direct sum of twisted de Rham cohomology groups is equal
(\cite[Prop. 2.1, Thm. 2.19]{SS22AnyonicDefectBranes})
to a certain TED-K-theory group, namely to the equivariant K-theory ``with complex coefficients'' (in the terminology of \cite{FreedHopkinsTeleman02ComplexCoefficients}) of the trivial $\CyclicGroup{\ShiftedLevel}$-action {\it twisted} by
the ``fictious'' vector potential regarded as an ``inner local system''.

\medskip
That this is the case may be extracted from \cite[Def. 3.10, Thm. 1.1]{TuXu06}\cite[Def. 3.6, Thm. 3.9]{FreedHopkinsTeleman02ComplexCoefficients}.
We briefly recall the detailed argument provided in \cite[\S 3]{SS22AnyonicDefectBranes}, also to highlight that this phenomenon
is neatly brought out by the stacky Fredholm formulation of K-theory \eqref{TEKTheory}:

\begin{remark}[Understanding the inner local system twist of equivariant K-theory {\cite[\S 3]{SS22AnyonicDefectBranes}}]
\label{UnderstandingInnerLocalSystemTwist}
$\,$

\begin{itemize}[leftmargin=.6cm]

\item[\bf (i)]
There is an essentially  unique ``stable'' group homomorphism (see \cite[Lem. 4.1.44]{SS21EPB}) from a finite cyclic group to the projective unitary group
\begin{equation}
  \label{StableGroupHomomorphism}
  \begin{tikzcd}
    \CyclicGroup{\ShiftedLevel}
    \ar[
      r,
      "{\mathrm{stable}}"
    ]
    &
    \frac{\UH}{\CircleGroup}
    \,,
  \end{tikzcd}
\end{equation}
hence a unique ``stable'' $\ShiftedLevel$-cyclic group of quantum symmetries \eqref{QuantumSymmetries}.

\item[\bf (ii)]
The space $\big(\FredholmOperators^0_{\ComplexNumbers}\big)^{\CyclicGroup{\ShiftedLevel}}$ of Fredholm operators which are fixed
(\cite[(56)]{SS22AnyonicDefectBranes})
by the induced $\CyclicGroup{\ShiftedLevel}$-action \eqref{ActionOfQuantumSymmetriesOnFredholmOperators}
is weakly equivalent to
the disjoint union of Fredholm operators
indexed by $\CyclicGroup{\ShiftedLevel}$ irreps
(\cite[(58)]{SS22AnyonicDefectBranes}):
\begin{equation}
  \label{CyclicGroupFixedLocusOfFredholmOperators}
 \big(
   \FredholmOperators^0_{\ComplexNumbers}
 \big)^{\CyclicGroup{\ShiftedLevel}}
 \;\underset{\mathrm{whe}}{\simeq}\;
 \underset{
   \rho \in \CyclicGroup{\ShiftedLevel}^\ast
 }{\coprod}
 \FredholmOperators^0_{\ComplexNumbers}
 \,.
\end{equation}

\item[\bf (iii)]
The group of automorphisms of the delooping
$\mathbf{B}\CyclicGroup{\ShiftedLevel} \xrightarrow{\;} \mathbf{B}\frac{\UnitaryGroup{\mathscr{H}}}{\CircleGroup}$
of the stable homomorphism \eqref{StableGroupHomomorphism} is
(\cite[(4.101)]{SS21EPB} \cite[(54)]{SS22AnyonicDefectBranes})
equivalently the Pontrjagin dual group of characters
$$
  \CyclicGroup{\ShiftedLevel}^\ast
    :=
  \mathrm{Hom}\big(\CyclicGroup{\ShiftedLevel},\, \CircleGroup \big)
\,\simeq\, \CyclicGroup{\ShiftedLevel}
  \,,
$$
and its induced action on the fixed locus \eqref{CyclicGroupFixedLocusOfFredholmOperators} is by multiplication of the irrep labels.

\item[\bf (iv)]
By the mapping stack adjunction for internal symmetries (\hyperlink{InternalSymmetryMappingStackAdjunction}{\it Figure 5}),
this means that inside a $\CyclicGroup{\ShiftedLevel}$-singularity the twists of equivariant K-theory subsume
flat-connections $\FlatConnectionForm$ on a $\CyclicGroup{\ShiftedLevel}^\ast \,\subset\, \CircleGroup$-principal bundle:
\vspace{-3mm}
$$
\hspace{-6mm}
  \begin{tikzcd}[column sep=large]
    &
    \HomotopyQuotient
     {\FredholmOperators^0_{\ComplexNumbers}}
     {\frac{\UH}{\CircleGroup}}
     \ar[d]
    \\
    \big(
      \DualTorus{2} \setminus \{\vec k\}
    \big)
    \times
    \HomotopyQuotient
      { \ast }
      { \CyclicGroup{\ShiftedLevel} }
    \ar[
      ur,
      dashed,
      "{
        \mathclap{
        \mbox{
          \tiny
          \color{orangeii}
          \bf
          \begin{tabular}{c}
            twisted equivariant
            \\
            K-cocycle
          \end{tabular}
        }
        }
      }"{sloped, pos=.35}
    ]
    \ar[
      r,
      "{\tau}",
      "{
        \mbox{
          \tiny \bf
          \color{greenii}
          \def\arraystretch{.9}
          \begin{tabular}{c}
            twist by
            \\
            inner local system
          \end{tabular}
        }
      }"{swap}
    ]
    &
    \mathbf{B}
    \frac{\UH}{\CircleGroup}
  \end{tikzcd}
  \hspace{3mm}
  \underset{
    \mathclap{
    \raisebox{-3pt}{
      \tiny \bf
      \color{darkblue}
      \begin{tabular}{c}
        mapping
        \\
        adjunction
      \end{tabular}
    }
    }
  }{
    \longleftrightarrow
  }
  \hspace{3mm}
  \begin{tikzcd}[column sep=25pt]
    &
    \HomotopyQuotient
    {
    \big(
      \FredholmOperators^0_{\ComplexNumbers}
    \big)^{\CyclicGroup{\ShiftedLevel}}
    }{ \CyclicGroup{\ShiftedLevel}^\ast }
    \ar[dr, phantom, "{\mbox{\tiny (pb)}}"]
    \ar[r]
    \ar[d]
    &
    \Maps{\Big}
      { \mathbf{B}\CyclicGroup{\ShiftedLevel} }
      {
        \HomotopyQuotient
         {\FredholmOperators^0_{\ComplexNumbers}}
         {\frac{\UH}{\CircleGroup}}
     }.
     \ar[d]
    \\
    \DualTorus{2} \setminus \{\vec k\}
    \ar[
      rr,
      rounded corners,
      to path={
        ([yshift=-0pt]\tikztostart.south)
        --
        ([yshift=-8pt]\tikztostart.south)
        --
        node[yshift=-5pt] {
        \mbox{
          \tiny
          \color{greenii}
          \bf
          adjoint twist
        }
        }
        ([yshift=-5pt]\tikztotarget.south)
        --
        ([yshift=+2pt]\tikztotarget.south)
      }
    ]
    \ar[
      r,
      "{\FlatConnectionForm}",
      "{
        \mbox{
          \tiny
          \color{greenii}
          \bf
          \def\arraystretch{.9}
          \begin{tabular}{c}
            local
            \\
            system
          \end{tabular}
        }
      }"{swap}
    ]
    \ar[
      ur,
      dashed,
      shorten <=-4pt,
      "{
        \mbox{
          \tiny
          \color{orangeii}
          \bf
          \begin{tabular}{c}
          \end{tabular}
        }
      }"{sloped, pos=.25}
    ]
    &
    \mathbf{B}\CyclicGroup{\ShiftedLevel}^\ast
    \ar[
      r,
      hook,
      "{
        \mbox{
          \tiny
          \color{greenii}
          \bf
          \def\arraystretch{.9}
          \begin{tabular}{c}
            full subgroupoid on the
            \\
            stable homomorphism
          \end{tabular}
        }
      }"{swap, yshift=-1pt}
    ]
    &
    \Maps{\Big}
      { \mathbf{B} \CyclicGroup{\ShiftedLevel} }
      {
        \mathbf{B}
        \frac{\UH}{\CircleGroup}
    }
  \end{tikzcd}
$$
which twists the K-cocycles by twisting the corresponding virtual vector bundles through the regular representation
of $\CyclicGroup{\ShiftedLevel}$.

\item[\bf (v)] Since the regular representation is equivalently the direct sum of all irreps, and since the irreps
of $\CyclicGroup{\ShiftedLevel}$ are the 1-dimensional complex reps generated by multiplication with $\exp(2 \pi \ImaginaryUnit {\color{purple}\Denominator}/\ShiftedLevel)$, the identification under the brace in \eqref{DirectSumOfFractionalAnyons} follows \cite[(70)]{SS22AnyonicDefectBranes}.
\end{itemize}
\end{remark}

\begin{remark}[\bf Exotic topological order]
  At this point it is natural to conjecture that the inclusion in \eqref{DirectSumOfFractionalAnyons} is in fact a bijection,
  hence a linear isomorphism.
 Settling this in any detail may require a deeper understanding of conformal blocks of logarithmic/rational-level CFTs than
 is currently available. Here we shall not further dwell on this point, but some observations in this direction may be found in
  \cite[Rem. 2.21]{SS22AnyonicDefectBranes}.
\end{remark}

\medskip

Comparison with \eqref{SemiMetalExactSequence} in Exp. \ref{ClassificationOfTwoDChernSemiMetals} now gives the following TED K-cohomology groups classifying anyonic topological order in the case of ``Chern phases'' without any symmetry protection:
\begin{equation}
  \begin{tikzcd}[
    column sep=-35pt
  ]
  &&
  \overset{
    \mathclap{
    \raisebox{4pt}{
      \scalebox{.7}{
      \color{darkblue}
      \bf
      \begin{tabular}{c}
        Topologically ordered
        ground states
        of interacting
        Chern semi-metal phase
      \end{tabular}
      }
    }
    }
  }{
  \mathrm{KU}^{
    {\color{orangeii}n}
    +
    \FlatConnectionForm(\vec \weight,{\color{purple}\ShiftedLevel})
  }
    \bigg(
      \ConfigurationSpace{{\color{orangeii}n}}
      \Big(
        \DualTorus{2}
        \setminus
        \{k_I\}_{I=1}^{{\color{greenii}N}}
      \Big)
      \times
      \HomotopyQuotient
        {\ast}
        {\CyclicGroup{{\color{purple}\ShiftedLevel}}}
        \,;\,
        \mathbb{C}
    \bigg)
    }
    \ar[
      ->>,
      drr,
      "{
        \mathrm{quotient}
      }"{sloped}
    ]
    \\
  \underset{
    \mathclap{
    \raisebox{-4pt}{
    \scalebox{.7}{
      \color{darkblue}
      \bf
      Compatible mass terms opening the gap
    }
    }
    }
  }{
  \mathrm{KU}^{
    {\color{orangeii}n}
    +
    \FlatConnectionForm(\vec \weight,{\color{purple}\ShiftedLevel})
  }
    \bigg(
      \ConfigurationSpace{{\color{orangeii}n}}
      \Big(
        \DualTorus{2}
        \setminus
        \{k_I\}_{I=1}^{{\color{greenii}N}}
      \Big)
      \times
      \HomotopyQuotient
        {\ast}
        {\CyclicGroup{{\color{purple}\ShiftedLevel}}}
    \bigg)
  }
  \ar[
    urr,
    "{
      \mathrm{ch}^{
        {\color{orangeii}n}
        +
        \FlatConnectionForm
      }
    }"{sloped}
  ]
    &&
    &&
  \underset{
    \mathclap{
    \raisebox{-4pt}{
    \scalebox{.7}{
      \color{darkblue}
      \bf
      Deformation classes of topologically ordered
      Chern semi-metal phases
    }
    }
    }
  }{
  \mathrm{KU}_{\flat}^{
    {\color{orangeii}n}-1
    +
    \FlatConnectionForm(\vec \weight,{\color{purple}\ShiftedLevel})
  }
    \bigg(
      \ConfigurationSpace{{\color{orangeii}n}}
      \Big(
        \DualTorus{2}
        \setminus
        \{k_I\}_{I=1}^{{\color{greenii}N}}
      \Big)
      \times
      \HomotopyQuotient
        {\ast}
        {\CyclicGroup{{\color{purple}\ShiftedLevel}}}
    \bigg)
  }
  \end{tikzcd}
\end{equation}

It is this situation which, seen under the CMT/ST-dictionary (\hyperlink{RosettaStone}{\it Figure 4}) we showed in \cite[\S 4]{SS22AnyonicDefectBranes} to accurately match the expectations for defect branes in string theory.

\medskip
Therefore, and in view of Fact \ref{ClassificationOfExternalSPTPhases}, the {\bf final conclusion} is, in refinement of Conjecture \ref{KTheoryClassificationOfTopologicalOrder}:

\begin{conjecture}[\bf Classification of SPT/SET order in TED-K]
\label{ClassificationOfSPTOrderInTEDK}
The $G_{\mathrm{ext}}$-SPT/SET phases of $\suTwoAffine{\Level-2}$-anyonic topological order are classified by the following flat twisted equivariant K-theory of configuration spaces of points in the complement of nodal points in the Brillouin torus:
\begin{equation}
\hspace{4mm}
  \begin{tikzcd}[
    column sep=-35pt
  ]
  \mathllap{
  \hspace{2cm}
\fbox{\!\!\!\!\!\!
  \rm \footnotesize
  \begin{tabular}{c}
    Topological order of
    {\color{greenii}N}
    {\color{purple}$\ShiftedLevel$}-anyonic
    band nodes
    \\
    in
    {\color{orangeii}n}-particle interacting
    semi-metal phase
  \end{tabular}
\!\!\!\!\!\!}
  }
  &&
  \overset{
    \mathclap{
    \raisebox{4pt}{
      \scalebox{.7}{
      \color{darkblue}
      \bf
      \begin{tabular}{c}
        $G_{\mathrm{ext}}$-protected/enhanced
        \\
        topologically ordered
        ground states
        of interacting
        semi-metal phase
      \end{tabular}
      }
    }
    }
  }{
  \mathrm{KR}^{
    {\color{orangeii}n}
    +
    [\widehat{T}^2, \widehat{P}^2= \pm 1]
    +
    \FlatConnectionForm(\vec \weight,{\color{purple}\ShiftedLevel})
  }
    \bigg(
      \HomotopyQuotient{
      \ConfigurationSpace{{\color{orangeii}n}}
      \Big(
        \DualTorus{2}
        \setminus
        \{k_I\}_{I=1}^{{\color{greenii}N}}
      \Big)
      }{G_{\mathrm{ext}}}
      \,\times\,
      \HomotopyQuotient
        {\ast}
        {\CyclicGroup{{\color{purple}\ShiftedLevel}}}
      \,;\,
      \mathbb{C}
    \bigg)
    }
    \ar[
      ->>,
      dd,
      "{
        \mathrm{quotient}
      }"{sloped}
    ]
    \\
  \underset{
    \mathclap{
    \raisebox{-4pt}{
    \scalebox{.7}{
      \color{darkblue}
      \bf
      Compatible mass terms opening the gap
    }
    }
    }
  }{
  \mathrm{KU}^{
    {\color{orangeii}n}
    +
    [\widehat{T}^2, \widehat{P}^2= \pm 1]
    +
    \FlatConnectionForm(\vec \weight,{\color{purple}\ShiftedLevel})
  }
    \bigg(
      \HomotopyQuotient{
      \ConfigurationSpace{{\color{orangeii}n}}
      \Big(
        \DualTorus{2}
        \setminus
        \{k_I\}_{I=1}^{{\color{greenii}N}}
      \Big)
      }{G_{\mathrm{ext}}}
      \,\times\,
      \HomotopyQuotient
        {\ast}
        {\CyclicGroup{{\color{purple}\ShiftedLevel}}}
    \bigg)
  }
  \ar[
    urr,
    shorten >=-10pt,
    "{
      \mathrm{ch}^{
        {\color{orangeii}n}
        +
        [\widehat{T}^2,\widehat{P}^2 = \pm 1]
        +
        \FlatConnectionForm
      }
    }"{sloped}
  ]
    &&
  \\
  &&
  \underset{
    \mathclap{
    \raisebox{-4pt}{
    \scalebox{.7}{
      \color{darkblue}
      \bf
      \begin{tabular}{c}
      Deformation classes of
      $G$-symmetry protected/enhanced
      \\
      topologically ordered
      interacting
      semi-metal phases
      \end{tabular}
    }
    }
    }
  }{
  \mathrm{KR}_{\flat}^{
    {\color{orangeii}n}
    -1
    +
    [\widehat{T}^2, \widehat{P}^2= \pm 1]
    +
    \FlatConnectionForm(\vec \weight,{\color{purple}\ShiftedLevel})
  }
    \bigg(
      \HomotopyQuotient{
      \ConfigurationSpace{{\color{orangeii}n}}
      \Big(
        \DualTorus{2}
        \setminus
        \{k_I\}_{I=1}^{{\color{greenii}N}}
      \Big)
      }{G_{\mathrm{ext}}}
      \,\times\,
      \HomotopyQuotient
        {\ast}
        {\CyclicGroup{{\color{purple}\ShiftedLevel}}}
    \bigg)
  }
  \end{tikzcd}
\end{equation}

\end{conjecture}

\begin{remark}
  In the case ${\color{orangeii}n} = 1$, hence
  for vanishing interaction
  (as in Rem. \ref{InteractingTopologicalPhasesSubsumeFreeTopologicalPhases}),
  Conjecture \ref{ClassificationOfSPTOrderInTEDK} reduces
  to Conjecture \ref{FlatKTheoryClassificationOfSemiMetals}.
\end{remark}

\medskip

\noindent
{\bf Conclusion.}
It remains to produce further checks and examples of Conjecture \ref{ClassificationOfSPTOrderInTEDK}; but the point here is that this is now,
to a large extent, a problem purely in (twisted equivariant differential) topological K-theory, for which a good supply of powerful tools exist.
In particular, TED K-theory is (by \cite{SS21EPB}\cite{SS22TED}) a natural construction in the foundational context of
{\it cohesive $\infty$-topos theory}. For such constructions there exists a novel {\it programming language} known as {\it cohesive homotopy type theory}
(cohesive HoTT, see \cite[p. 5-6]{SS20OrbifoldCohomology} for pointers). We see here through Conjecture \ref{ClassificationOfSPTOrderInTEDK}, and in
view of its tight relation to the hardware model of topological quantum computation (\hyperlink{AdiabaticBraiding}{\it Figure 1}),
that {\it cohesive HoTT may naturally implement topological quantum computation} in a way which is fully ``hardware aware'' of the fine detail
of topological q-bits and their braid quantum gates.
This point is further discussed in \cite{SS22TQC}.

$$
  \hspace{-6pt}
  \begin{tikzcd}[
    column sep={between origins, 74pt},
    row sep=-4pt]
    \mbox
    {
      \bf
      Programming platform:
    }
    &[+10pt]&
    \mbox
    {\bf Library/Module:}
    &&
    \mathclap{
    \mbox
    {\bf Hardware platform:}
    }
    &[-35pt]&
    \mbox{
      \bf Architecture:
    }
    \\
    \fbox{
      \hspace{-10pt}
      \small
      \begin{tabular}{c}
        Cohesive Homotopy
        \\
        Type Theory with
        \\
        dependent linear types
      \end{tabular}
      \hspace{-10pt}
    }
    \ar[
      rr,
      "{
        \color{greenii}
        \mbox{implements}
      }",
      "{
        \scalebox{.8}{
          \cite{SS22TQC}
        }
      }"{swap}
    ]
    &[+2pt]&
    \fbox{
      \hspace{-10pt}
      \small
      \begin{tabular}{c}
        \TED-K-cohomology of
        \\
        defect configurations in
        \\
        crystallographic orbifolds
      \end{tabular}
      \hspace{-10pt}
    }
    \ar[
      rr,
      "{
        \color{greenii}
        \mbox{emulates}
      }"{yshift=1pt},
      "{
        \scalebox{.8}{
          \cref{TEDKDescribesRealisticAnyonSpecies}
        }
      }"{swap}
    ]
    &[+2pt]&
    \fbox{
      \hspace{-10pt}
      \small
      \begin{tabular}{c}
        Anyonic orders in
        \\
        topological phases
        \\
        of quantum materials
      \end{tabular}
      \hspace{-10pt}
    }
    \ar[
      rr,
      "{
        \color{greenii}
        \mbox{runs}
      }"{yshift=1pt},
      "{
        \scalebox{.8}{
          \cite{FKLW01}
        }
      }"{swap}
    ]
    &&
    \fbox{
      \hspace{-10pt}
      \small
      \begin{tabular}{c}
        Topological
        \\
        quantum
        \\
        circuits
      \end{tabular}
      \hspace{-10pt}
    }
  \end{tikzcd}
$$

\bigskip
\noindent  Hisham Sati,
{\it Mathematics, Division of Science,
\\
and Center for Quantum and Topological Systems (CQTS), NYUAD Research Institute,
\\
New York University Abu Dhabi, UAE;
\\
The Courant Institute for Mathematical Science, NYU.
}
\\
{\tt hsati@nyu.edu}
\\
\\
\noindent  Urs Schreiber, {\it Mathematics, Division of Science,
\\
and Center for Quantum and Topological Systems (CQTS), NYUAD Research Institute,
\\
New York University Abu Dhabi, UAE;
\\
{\tt us13@nyu.edu}


\newpage

\begin{thebibliography}{100}



\vspace{-.3cm}
\bibitem[APY19]{AhnParkYang18}
J. Ahn, S. Park, and B.-J. Yang,
{\it Failure of Nielsen-Ninomiya theorem and fragile topology in two-dimensional systems with space-time inversion
symmetry: application to twisted bilayer graphene at magic angle},
Phys. Rev. X {\bf 9} (2019) 021013,
[\href{https://doi.org/10.1103/PhysRevX.9.021013}{\tt doi:10.1103/PhysRevX.9.021013}],
[\href{https://arxiv.org/abs/1808.05375}{\tt arXiv:1808.05375}].

\vspace{-.3cm}
\bibitem[AL16]{AlbashLidar16}
T. Albash and D. A. Lidar,
{\it Adiabatic Quantum Computing},
Rev. Mod. Phys. {\bf 90} (2018) 015002, \newline
[\href{https://doi.org/10.1103/RevModPhys.90.015002}{\tt doi:10.1103/RevModPhys.90.015002}],
[\href{https://arxiv.org/abs/1611.04471}{\tt arXiv:1611.04471}].


\vspace{-.3cm}
\bibitem[AU12]{AndersenUeno12}
J. E. Andersen and K. Ueno,
{\it Modular functors are determined by their genus zero data}, Quantum Topology {\bf 3}  (2012),
255-291,  [\href{https://doi.org/10.4171/qt/29}{\tt doi:10.4171/qt/29}].


\vspace{-.3cm}
\bibitem[AMV18]{AMV18}
N. P. Armitage, E. Mele, and A. Vishwanath,
{\it Weyl and Dirac semimetals in three-dimensional solids},
Rev. Mod. Phys. {\bf 90} (2018), 015001,
[\href{https://doi.org/10.1103/RevModPhys.90.015001}{\tt doi:10.1103/RevModPhys.90.015001}].

\vspace{-.3cm}
\bibitem[ASW84]{ArovasSchriefferWilczek84}
D. P. Arovas, J. R. Schrieffer, and F. Wilczek,
{\it Fractional Statistics and the Quantum Hall Effect},
Phys. Rev. Lett. {\bf 53}  (1984), 722-723,
[\href{https://doi.org/10.1103/PhysRevLett.53.722}{\tt doi:10.1103/PhysRevLett.53.722}].

\vspace{-.3cm}
\bibitem[ASWZ85]{ASWZ85}
D. P. Arovas, R. Schrieffer, F. Wilczek, and A. Zee,
{\it Statistical mechanics of anyons}, Nucl. Phys. B {\bf 251} (1985), 117-126,
[\href{https://doi.org/10.1016/0550-3213(85)90252-4}{\tt doi:10.1016/0550-3213(85)90252-4}].

\vspace{-.3cm}
\bibitem[AF15]{AndoFu15}
Y. Ando and  L. Fu,
{\it Topological Crystalline Insulators and Topological Superconductors: From Concepts to Materials},
Ann. Rev. Cond. Mat. Phys. {\bf 6} (2015), 361-381, [\href{https://doi.org/10.1146/annurev-conmatphys-031214-014501}{\tt doi:10.1146/annurev-conmatphys-031214-014501}],
[\href{https://arxiv.org/abs/1501.00531}{\tt arXiv:1501.00531}].


\vspace{-.3cm}
\bibitem[Ar02]{Arveson02}
W. Arveson,
{\it A Short Course on Spectral Theory},
Grad. Texts Math. {\bf 209} Springer, 2002,
[\href{https://link.springer.com/book/10.1007/b97227}{\tt doi:10.1007/b97227}].

\vspace{-.3cm}
\bibitem[AA67]{AtiyahAnderson67}
M. Atiyah and D. W. Anderson,
{\it K-theory},
W. A. Benjamin, Inc., New York-Amsterdam, 1967, \newline
[\href{https://ncatlab.org/nlab/files/AtiyahKTheory.pdf}{\tt ISBN-13:978-0201407921}].


\vspace{-.3cm}
\bibitem[AS04]{AtiyahSegal04}
M. Atiyah and G. Segal,
{\it Twisted K-theory},
Ukr. Math. Bull. {\bf 1}  (2004),  291-334,
[\href{https://arxiv.org/abs/math/0407054}{\tt arXiv:math/0407054}], \newline
[\href{http://iamm.su/en/journals/j879/?VID=10}{\tt iamm.su/en/journals/j879/?VID=10}].


\vspace{-.3cm}
\bibitem[AtSi69]{AtiyahSinger69}
M. F. Atiyah and I. M. Singer,
{\it Index theory for skew-adjoint Fredholm operators},
Publ. Math. IH{\'E}S {\bf 37} (1969), 5-26,
[\href{http://www.numdam.org/item/PMIHES_1969__37__5_0}{\tt doi:PMIHES\_1969\_\_37\_\_5\_0}].


\vspace{-.3cm}
\bibitem[At16]{Atteia16}
J. Atteia,
{\it Topology and electronic transport in Dirac systems under irradiation}, PhD Thesis, U. Bordeaux, 2016,
[\href{https://tel.archives-ouvertes.fr/tel-02426217}{\tt tel:02426217}].



\vspace{-.3cm}
\bibitem[BCMS93]{BCMS93}
G. A. Baker, G. S. Canright, S. B. Mulay, and  C. Sundberg,
{\it On the spectral problem for anyons},
Comm. Math. Phys.  {\bf 153} (1993), 277–295,
[\href{https://doi.org/10.1007/BF02096644}{\tt doi:10.1007/BF02096644}].


\vspace{-.3cm}
\bibitem[BPCW19]{BPCW14}
M. Barkeshli, P. Bonderson, M. Cheng,
and Z. Wang, {\it Symmetry Fractionalization, Defects, and Gauging of Topological Phases},
Phys. Rev. B {\bf 100} (2019) 115147 ,
[\href{https://doi.org/10.1103/PhysRevB.100.115147}{\tt doi:10.1103/PhysRevB.100.115147}],
[\href{https://arxiv.org/abs/1410.4540}{\tt arXiv:1410.4540}].

\vspace{-.3cm}
\bibitem[BP20]{BarlasProdan19}
Y. Barlas and E. Prodan,
{\it Topological braiding of non-Abelian mid-gap defects in classical meta-materials},
Phys. Rev. Lett. {\bf 124} (2020) 146801,
[\href{https://doi.org/10.1103/PhysRevLett.124.146801}{\tt doi:10.1103/PhysRevLett.124.146801}],
[\href{https://arxiv.org/abs/1903.00463}{\tt arXiv:1903.00463}].







\vspace{-.3cm}
\bibitem[Be84]{Berry84}
M. V. Berry,
{\it Quantal phase factors accompanying adiabatic changes}, Proc. R. Soc. Lond. A {\bf 392} (1984) 45–57,
[\href{https://doi.org/10.1098/rspa.1984.0023}{\tt doi:10.1098/rspa.1984.0023}],
[\href{https://www.jstor.org/stable/2397741}{\tt jstor:2397741}].

\vspace{-.3cm}
\bibitem[Bi1910]{Bieberbach1910}
L. Bieberbach,
{\it {\"U}ber die Bewegungsgruppen des $n$ dimensionalen Euklidischen Raumes mit einem endlichen Fundamentalbereich},
Nachr. G{\"o}tt. (1910), 75-84,
[\href{https://eudml.org/doc/58754}{\tt dml:58754}].

\vspace{-.3cm}
\bibitem[Bl86]{Blackadar86}
B. Blackadar,
{\it K-Theory for Operator Algebras},
Cambridge University Press, 1999, \newline
[\href{https://link.springer.com/book/10.1007/978-1-4613-9572-0}{\tt doi:10.1007/978-1-4613-9572-0}].


\vspace{-.3cm}
\bibitem[BB77]{BleekerBooss}
D. D. Bleecker and B. Booss,
{\it Topology and Analysis. Introduction to the Atiyah-Singer index formula and gauge theoretic physics},
 Springer-Verlag, Berlin-New York, 1977,
[\href{https://doi.org/10.1007/978-1-4684-0627-6}{\tt
  ISBN:3-540-08451-7}].

\vspace{-.3cm}
\bibitem[Bo21]{Bonderson21}
P. Bonderson,
{\it Measuring Topological Order}, Phys. Rev. Research {\bf 3} (2021)  033110,
[\href{https://arxiv.org/abs/2102.05677}{\tt arXiv:2102.05677}], \newline
[\href{https://doi.org/10.1103/PhysRevResearch.3.033110}{\tt doi:10.1103/PhysRevResearch.3.033110}].

\vspace{-.3cm}
\bibitem[Bon18]{Bongaarts18}
P. J. M. Bongaarts,
{\it The electron-positron field, coupled to external electromagnetic potentials, as an elementary $C^\ast$-algebra theory},
Phys. Lett. B {\bf 779} (2018), 420-424,
[\href{https://doi.org/10.1016/0003-4916(70)90007-2}{\tt doi:10.1016/0003-4916(70)90007-2}].


\vspace{-.3cm}
\bibitem[Bor18]{Borcherding18}
D. Borcherding,
{\it Non-Abelian quasi-particles in electronic systems},  PhD thesis Hannover (2018)
[\href{https://doi.org/10.15488/4280}{\tt doi:10.15488/4280}]


\vspace{-.3cm}
\bibitem[BF28]{BornFock28}
M. Born and V. A. Fock,
{\it Beweis des Adiabatensatzes},
Zeitschr. Phys. {\bf 51} (1928), 165–180, \newline
[\href{https://doi.org/10.1007/BF01343193}{\tt doi:10.1007/BF01343193}].

\vspace{-.3cm}
\bibitem[BW${}^+$20]{BWSWYB20}
A. Bouhon, Q.-S. Wu, R.-J. Slager, H. Weng, O. V. Yazyev, and T. Bzdu{\v s}ek,
{\it Non-Abelian reciprocal braiding of Weyl points and its manifestation in $\mathrm{ZrTe}$},
Nature Phys. {\bf 16} (2020), 1137–1143,
[\href{https://arxiv.org/abs/1907.10611}{\tt arXiv:1907.10611}], \newline
[\href{https://doi.org/10.1038/s41567-020-0967-9}{\tt doi:10.1038/s41567-020-0967-9}].

\vspace{-3mm}
\bibitem[BNV13]{BNV13}
U. Bunke, T. Nikolaus, and M. V{\"o}lkl,
{\it Differential cohomology theories as sheaves of spectra},
J. Homotopy Relat. Struct. {\bf 11} (2016), 1-66,
[\href{https://doi.org/10.1007/s40062-014-0092-5}{\tt doi:10.1007/s40062-014-0092-5}],
[\href{https://arxiv.org/abs/1311.3188}{\tt arXiv:1311.3188}].


\vspace{-.3cm}
\bibitem[BS12]{BunkeSchick11}
U. Bunke and T. Schick,
{\it Differential K-theory. A survey},
in C. B{\"a}r et al. (eds.)
{\it Global Differential Geometry},
Proc. Math. {\bf 17} Springer (2012) 303-357,
[\href{https://link.springer.com/chapter/10.1007/978-3-642-22842-1_11}{\tt doi:10.1007/978-3-642-22842-1\_11}],
[\href{https://arxiv.org/abs/1011.6663}{\tt arXiv:1011.6663}].

\vspace{-.3cm}
\bibitem[BHB11]{BHB11}
A. A. Burkov, M. D. Hook, and L. Balents,
{\it Topological nodal semimetals},
Phys. Rev. B {\bf 84} (2011) 235126, \newline
[\href{https://doi.org/10.1103/PhysRevB.84.235126}{\tt doi:10.1103/PhysRevB.84.235126}],
[\href{https://arxiv.org/abs/1110.1089}{\tt arXiv:1110.1089}].


\vspace{-.3cm}
\bibitem[CHO82]{CareyHurstOBrien82}
A. L. Carey, C. A. Hurst, and D. M. O’Brien, {\it Automorphisms of the canonical anticommutation relations and index theory},
J. Funct. Anal. {\bf 48}  (1982), 360-393,
[\href{https://doi.org/10.1016/0022-1236(82)90092-1}{\tt doi:10.1016/0022-1236(82)90092-1}].


\vspace{-.3cm}
\bibitem[CS21]{CayssolFuchs20}
J. Cayssol and J.-N. Fuchs,
{\it Topological and geometrical aspects of band theory},
J. Phys. Mater. {\bf 4} (2021) 034007,
[\href{https://doi.org/10.1088/2515-7639/abf0b5}{\tt doi:10.1088/2515-7639/abf0b5}],
[\href{https://arxiv.org/abs/2012.11941}{\tt arXiv:2012.11941}].


\vspace{-.3cm}
\bibitem[CLBFN15]{CLBFN15}
C. Cesare, A. J. Landahl, D. Bacon, S. T. Flammia and A. Neels,
{\it Adiabatic topological quantum computing}, Phys. Rev. A {\bf 92} (2015) 012336
[\href{https://arxiv.org/abs/1406.2690}{\tt arXiv:1406.2690}]
[\href{https://doi.org/10.1103/PhysRevA.92.012336}{\tt doi:10.1103/PhysRevA.92.012336}]

\vspace{-.3cm}
\bibitem[CN08]{ChangNiu08}
M.-C. Chang and Q. Niu,
{\it Berry curvature, orbital moment, and effective quantum theory of electrons in electromagnetic fields},
J. Phys.: Condens. Matter {\bf 20} (2008) 193202,
[\href{https://iopscience.iop.org/article/10.1088/0953-8984/20/19/193202}{\tt doi:10.1088/0953-8984/20/19/193202}].


\vspace{-.3cm}
\bibitem[Ch86]{Charlap86}
L. S. Charlap,
{\it Bieberbach Groups and Flat Manifolds},
Springer, 1986,
[\href{https://doi.org/10.1007/978-1-4613-8687-2}{\tt doi:10.1007/978-1-4613-8687-2}].

\vspace{-.3cm}
\bibitem[CBSM22]{CBSM21}
S. Chen, A. Bouhon, R.-J. Slager and B. Monserrat,
{\it Non-Abelian braiding of Weyl nodes via symmetry-constrained phase transitions},
Phys. Rev. B {\bf 105} (2022) L081117,
[\href{https://doi.org/10.1103/PhysRevB.105.L081117}{\tt doi:10.1103/PhysRevB.105.L081117}], \newline
[\href{https://arxiv.org/abs/2108.10330}{\tt arXiv:2108.10330}].

\vspace{-.3cm}
\bibitem[CGDS11]{CGDS11}
M. Cheng, V. Galitski, and S. Das Sarma,
{\it Non-adiabatic Effects in the Braiding of Non-Abelian Anyons in Topological Superconductors},
Phys. Rev. B {\bf 84} (2011) 104529,
[\href{https://doi.org/10.1103/PhysRevB.84.104529}{\tt doi:10.1103/PhysRevB.84.104529}], \newline
[\href{https://arxiv.org/abs/1106.2549}{\tt arXiv:1106.2549}].



\vspace{-.3cm}
\bibitem[CGLW13]{CGLW13}
X. Chen, Z.-C. Gu, Z.-X. Liu, and X.-G.Wen,
{\it Symmetry protected topological orders and the group cohomology of their symmetry group},
Phys. Rev. B {\bf 87} (2013) 155114,
[\href{https://doi.org/10.1103/PhysRevB.87.155114}{\tt doi:10.1103/PhysRevB.87.155114}], \newline
[\href{https://arxiv.org/abs/1106.4772}{\tt arXiv:1106.4772}].


\vspace{-.3cm}
\bibitem[CGLW12]{CGLW12}
X. Chen, Z.-C. Gu, Z.-X. Liu, and X.-G. Wen,
{\it Symmetry protected topological orders and the group cohomology of their symmetry group},
Science {\bf 338} (2012), 1604-1606,
[\href{https://doi.org/10.1103/PhysRevB.87.155114}{\tt doi:10.1103/PhysRevB.87.155114}].

\vspace{-.3cm}
\bibitem[CGW10]{ChenGuWen10}
X Chen, Z.-C. Gu and X.-G. Wen,
{\it Local unitary transformation, long-range quantum entanglement, wave function renormalization, and topological order}, Phys. Rev. B {\bf 82} (2010) 155138,
[\href{https://arxiv.org/abs/1004.3835}{\tt arXiv:1004.3835}], \newline
[\href{https://doi.org/10.1103/PhysRevB.82.155138}{\tt
https://doi.org/10.1103/PhysRevB.82.155138}].


\vspace{-.3cm}
\bibitem[CLW11]{CLW11}
X. Chen, Z.-X. Liu, and X.-G. Wen,
{\it Two-dimensional symmetry-protected topological orders and their protected gapless edge
excitations}, Phys. Rev. B {\bf 84} (2011) 235141,
[\href{https://doi.org/10.1103/PhysRevB.84.235141}{\tt doi:10.1103/PhysRevB.84.235141}], \newline
[\href{https://arxiv.org/abs/1106.4752}{\tt arXiv:1106.4752}].


\vspace{-.3cm}
\bibitem[CWWH89]{CWWH89}
Y.-H. Chen, F. Wilczek, E. Witten, and B. Halperin, {\it On Anyon Superconductivity},
Int. J. Mod. Phys. B {\bf 03} (1989), 1001-1067,
[\href{https://doi.org/10.1142/S0217979289000725}{\tt doi:10.1142/S0217979289000725}].


\vspace{-.3cm}
\bibitem[CS14]{ChiuSchnyder14}
C.-K. Chiu and  A. P. Schnyder,
{\it Classification of reflection-symmetry-protected topological semimetals and nodal superconductors},
Phys. Rev. B {\bf 90} (2014) 205136,
[\href{https://doi.org/10.1103/PhysRevB.90.205136}{\tt doi:10.1103/PhysRevB.90.205136}].


\vspace{-.3cm}
\bibitem[CTSR16]{CTSR16}
C.-K. Chiu, J. C. Y. Teo, A. P. Schnyder, and S. Ryu,
{\it Classification of topological quantum matter with symmetries},
Rev. Mod. Phys. {\bf 88} (2016) 035005,
[\href{https://doi.org/10.1103/RevModPhys.88.035005}{\tt doi:10.1103/RevModPhys.88.035005}],
[\href{https://arxiv.org/abs/1505.03535}{\tt arXiv:1505.03535}].


\vspace{-.3cm}
\bibitem[CLM12]{CLM12}
G. Y. Cho, Y.-M. Lu, and J. E. Moore,
{\it Gapless edge states of background field theory and translation-symmetric $\mathbb{Z}_2$ spin liquids},
Phys. Rev. B {\bf 86} 125101 (2012)
[\href{https://doi.org/10.1103/PhysRevB.86.125101}{\tt doi:10.1103/PhysRevB.86.125101}]

\vspace{-3mm}
\bibitem[CSS21]{CSS21}
D. Corfield, H. Sati, and U. Schreiber,
{\it Fundamental weight systems are quantum states} (2021),
[\href{https://arxiv.org/abs/2105.02871}{\tt arXiv:2105.02871}].





\vspace{-.3cm}
\bibitem[CLR21]{CLR21}
T. Creutzig, S. Lentner, and M. Rupert,
{\it Characterizing braided tensor categories associated to logarithmic vertex operator algebras},
[\href{https://arxiv.org/abs/2104.13262}{\tt arXiv:2104.13262}].


\vspace{-3mm}
\bibitem[CR13]{CreutzigRidoutII13}
T. Creutzig and D. Ridout,
{\it Modular Data and Verlinde Formulae for Fractional Level WZW Models II},
Nucl. Phys.  {\bf B875}  (2013), 423-458,
[\href{https://doi.org/10.1016/j.nuclphysb.2013.07.008}{\tt doi:10.1016/j.nuclphysb.2013.07.008}],
[\href{https://arxiv.org/abs/1306.4388}{\tt arXiv:1306.4388}].



\vspace{-.3cm}
\bibitem[DSFN15]{DasSarmaFreedmanNayak15}
S. Das Sarma, M. Freedman and C. Nayak,
{\it Majorana Zero Modes and Topological Quantum Computation}, npj Quantum Information {\bf 1} (2015) 15001,
[\href{ https://doi.org/10.1038/npjqi.2015.1}{\tt doi:10.1038/npjqi.2015.1}].


\vspace{-.3cm}
\bibitem[DJMM90]{DJMM90}
E. Date, M. Jimbo, A. Matsuo, and T. Miwa, {\it Hypergeometric-type integrals and the
$\mathfrak{sl}(2,\mathbb{C})$-Knizhnik-Zamolodchikov equation}, Int. J. Mod. Phys. B {\bf 04}  (1990), 1049-1057,
[\href{https://doi.org/10.1142/S0217979290000528}{\tt doi:10.1142/S0217979290000528}].

\vspace{-.3cm}
\bibitem[DMV03]{DMV03}
G. Date, M. V. N. Murthy, and R. Vathsan,
{\it Classical and Quantum Mechanics of Anyons}, \newline
[\href{https://arxiv.org/abs/cond-mat/0302019}{\tt arXiv:cond-mat/0302019}].

\vspace{-3mm}
\bibitem[De70]{Deligne70}
P. Deligne,
{\it Equations diff{\'e}rentielles {\`a} points singuliers r{\'e}guliers},
Lecture Notes  Math. {\bf 163},
Springer, Berlin, 1970, \newline
[\href{https://publications.ias.edu/node/355}{\tt publications.ias:355}].


\vspace{-.3cm}
\bibitem[DTC21]{DTC}
 Delft Topology Course ($+$ video by D. Haldane), {\it Haldane model, Berry curvature, and Chern number}, \newline
 [\href{https://topocondmat.org/w4_haldane/haldane_model.html}{\tt topocondmat.org/w4\_haldane/haldane\_model.html}].

\vspace{-.3cm}
\bibitem[DFT97]{DFT97}
G. Dell’Antonio, R. Figari, and A. Teta,
{\it Statistics in Space Dimension Two},
Lett. Math. Phys. {\bf 40} (1997), 235–256,
[\href{https://doi.org/10.1023/A:1007361832622}{\tt doi:10.1023/A:1007361832622}].

\vspace{-.3cm}
\bibitem[DNL11]{DeNittisLein11}
G. De Nittis and M. Lein,
{\it Exponentially Localized Wannier Functions in Periodic Zero Flux Magnetic Fields},
J. Math. Phys. {\bf 52} (2011) 112103,
[\href{https://doi.org/10.1063/1.3657344}{\tt doi:10.1063/1.3657344}],
[\href{https://arxiv.org/abs/1108.5651}{\tt arXiv:1108.5651}].

\vspace{-.3cm}
\bibitem[Di04]{Dimca04}
A. Dimca,
{\it Sheaves in Topology}, Universitext, Springer (2004)
[\href{https://doi.org/10.1007/978-3-642-18868-8}{\tt doi:10.1007/978-3-642-18868-8}]


\vspace{-3mm}
\bibitem[Ea12]{Eades12}
H. Eades,
{\it Type Theory and Applications}, 2012,
[\href{https://metatheorem.org/includes/pubs/comp.pdf}{\tt metatheorem.org/includes/pubs/comp.pdf}]



\vspace{-.3cm}
\bibitem[Ei90]{Einarsson90}
T. Einarsson,
{\it Fractional statistics on a torus},
Phys. Rev. Lett. {\bf 64}  (1990), 1995--1998, \newline
[\href{https://doi.org/10.1103/PhysRevLett.64.1995}{\tt doi:10.1103/PhysRevLett.64.1995}].


\vspace{-3mm}
\bibitem[ESV92]{ESV92}
H. Esnault, V. Schechtman, and E. Viehweg,
{\it Cohomology of local systems on the complement of hyperplanes}, Invent. Math. {\bf 109} (1992),  557-561,
[\href{https://doi.org/10.1007/BF01232443}{\tt doi:10.1007/BF01232443}].


\vspace{-.3cm}
\bibitem[En86]{Engel86}
P. Engel,
{\it Geometric Crystallography}, D. Reidel Publishing
(1986),
[\href{https://doi.org/10.1007/978-94-009-4760-3}{\tt doi:10.1007/978-94-009-4760-3}].

\vspace{-.3cm}
\bibitem[EMS04]{EngelMichelSenechal04}
P. Engel, L. Michel, and M. Senechal,
{\it Lattice Geometry}, IHES/P/04/45 (2004),
[\href{https://cds.cern.ch/record/859509}{\tt cds:/859509}].



\vspace{-3mm}
\bibitem[EFK98]{EtingofFrenkelKirillov98}
P. I. Etingof, I. Frenkel, and A. A Kirillov,
{\it Lectures on Representation Theory and Knizhnik-Zamolodchikov Equations},
Math. Surv. monogr. {\bf 58},
Amer. Math. Soc., Providence, RI, 1998,
[\href{https://bookstore.ams.org/surv-58}
{\tt ams.org/surv-58}].


\vspace{-.3cm}
\bibitem[EGNO15]{EGNO15}
P. Etingof, S. Gelaki, D. Nikshych, and V. Ostrik,
{\it Tensor Categories},
 Math. Surv.  Monogr.  {\bf 205}, Amer. Math. Soc., 2015,
[\href{https://bookstore.ams.org/surv-205}{\tt ISBN:978-1-4704-3441-0}].



\vspace{-.3cm}
\bibitem[FWDF16]{FWDZ16}
C. Fang, H. Weng, X. Dai, and Z. Fang,
{\it Topological nodal line semimetals},
Chinese Phys. B {\bf 25} (2016) 117106,
[\href{https://doi.org/10.1088/1674-1056/25/11/117106}{\tt doi:10.1088/1674-1056/25/11/117106}],
[\href{https://arxiv.org/abs/1609.05414}{\tt arXiv:1609.05414}].


\vspace{-.3cm}
\bibitem[Fa81]{Farkas81}
D. R. Farkas,
{\it Crystallographic groups and their mathematics},
Rocky Mountain J. Math. {\bf 11} (1981), 511-552, \newline
[\href{https://projecteuclid.org/euclid.rmjm/1250128489}{\tt doi:euclid.rmjm/1250128489}].


\vspace{-.3cm}
\bibitem[FSV94]{FeiginSchechtmanVarchenko94}
B.  Feigin, V. Schechtman, and A. Varchenko,
{\it On algebraic equations satisfied by hypergeometric correlators in WZW models. I.},
Comm. Math. Phys. {\bf 163} (1994), 173–184,
[\href{https://doi.org/10.1007/BF02101739}{\tt doi:10.1007/BF02101739}].


\vspace{-.3cm}
\bibitem[FZWY21]{FZWY21}
X. Feng, J. Zhu, W. Wu, and S. A. Yang,
{\it Two-dimensional topological semimetals}, Chin. Phys. B {\bf 30} (2021)  107304,
[\href{https://iopscience.iop.org/article/10.1088/1674-1056/ac1f0c}{\tt doi:10.1088/1674-1056/ac1f0c}],
[\href{https://arxiv.org/abs/2103.13772}{\tt arXiv:2103.13772}].

\vspace{-.3cm}
\bibitem[FS18]{FieldSimula18}
B. Field and  T. Simula,
{\it Introduction to topological quantum computation with non-Abelian anyons},
Quant. Sci. Techn. {\bf 4} (2018) 045004,
[\href{https://iopscience.iop.org/article/10.1088/2058-9565/aacad2}{\tt doi:10.1088/2058-9565/aacad2}],
[\href{https://arxiv.org/abs/1802.06176}{\tt arXiv:1802.06176}].


\vspace{-.3cm}
\bibitem[FMP14a]{FiorenzaMonacoPanati14a}
D. Fiorenza, D. Monaco, and G. Panati,
{\it Construction of real-valued localized composite Wannier functions for insulators},
Ann. Henri Poincar{\'e} {\bf 17}  (2016), 63-97,
[\href{https://doi.org/10.1007/s00023-015-0400-6}{\tt doi:10.1007/s00023-015-0400-6}],
[\href{https://arxiv.org/abs/1408.0527}{\tt arXiv:1408.0527}].

\vspace{-.3cm}
\bibitem[FMP14b]{FiorenzaMonacoPanati14b}
D. Fiorenza, D. Monaco, and G. Panati,
{\it $\mathbb{Z}_2$ invariants of topological insulators as geometric obstructions},
Commun. Math. Phys. {\bf 343} (2016), 1115-1157,
[\href{https://doi.org/10.1007/s00220-015-2552-0}{\tt doi:10.1007/s00220-015-2552-0}],
[\href{https://arxiv.org/abs/1408.1030}{\tt arXiv:1408.1030}].

\vspace{-3mm}
\bibitem[FSS20-Cha]{FSS20Character}
D. Fiorenza, H. Sati, and U. Schreiber,
{\it The character map in (twisted differential) non-abelian cohomology},
[\href{https://arxiv.org/abs/2009.11909}{\tt arXiv:2009.11909}].


\vspace{-.3cm}
\bibitem[Fl03]{Flohr03}
M. Flohr,
{\it Bits and pieces in logarithmic conformal field theory}, Int. J. Mod. Phys. A {\bf 18} (2003), 4497-4591, \newline
[\href{https://doi.org/10.1142/S0217751X03016859}{\tt doi:10.1142/S0217751X03016859}],
[\href{https://arxiv.org/abs/hep-th/0111228}{\tt arXiv:hep-th/0111228}].

\vspace{-.3cm}
\bibitem[Fr08]{Fredenhagen08}
S. Fredenhagen,
{\it Physical Background to the K-Theory Classification of D-Branes: Introduction and References}
in: {\it Basic Bundle Theory and K-Cohomology Invariants}, Lect. Notes Phys., Springer, 2008, \newline
[\href{https://doi.org/10.1007/978-3-540-74956-1_1}{\tt doi:10.1007/978-3-540-74956-1\_1}].


\vspace{-.3cm}
\bibitem[FGR96]{FGR96}
K. Fredenhagen, M. Gaberdiel, and  S. M. Rüger,
{\it Scattering states of plektons (particles with braid group statistics) in 2+1 dimensional quantum field theory},
Commun. Math. Phys. {\bf 175} (1996), 319–335,
[\href{https://doi.org/10.1007/BF02102411}{\tt doi:10.1007/BF02102411}].

\vspace{-.3cm}
\bibitem[Fr12]{Freed12}
D. S. Freed,
{\it On Wigner's theorem},
Geom.  \& Topol. Monogr.  {\bf 18} (2012), 83–89,
[\href{http://dx.doi.org/10.2140/gtm.2012.18.83}{\tt doi:10.2140/gtm.2012.18.83}],
[\href{https://arxiv.org/abs/1112.2133}{\tt arXiv:1112.2133}].

\vspace{-.3cm}
\bibitem[Fr14]{Freed14}
D. S. Freed,
{\it Short-range entanglement and invertible field theories},
[\href{https://arxiv.org/abs/1406.7278}{\tt arXiv:1406.7278}].

\vspace{-.3cm}
\bibitem[FH16]{FreedHopkins16}
D. S. Freed and M. J. Hopkins,
{\it Reflection positivity and invertible topological phases},
Geom. Topol. {\bf 25} (2021), 1165-1330,
[\href{https://doi.org/10.2140/gt.2021.25.1165}{\tt doi:10.2140/gt.2021.25.1165}],
[\href{https://arxiv.org/abs/1604.06527}{\tt arXiv:1604.06527}].

\vspace{-3mm}
\bibitem[FHT07]{FreedHopkinsTeleman02ComplexCoefficients}
D. Freed, M. Hopkins, and C. Teleman,
{\it Twisted equivariant K-theory with complex coefficients},
J. Top. {\bf 1} (2007), 16-44,
[\href{https://doi.org/10.1112/jtopol/jtm001}{\tt doi:10.1112/jtopol/jtm001}],
[\href{https://arxiv.org/abs/math/0206257}{\tt arXiv:math/0206257}].


\vspace{-.3cm}
\bibitem[FHT11]{FreedHopkinsTelement07TwistedKTheoryI}
D. S. Freed, M. J. Hopkins and C. Teleman,
{\it Loop groups and twisted K-theory I},
J. Topol. {\bf 4}  (2011), 737-798, \newline
[\href{https://doi.org/10.1112/jtopol/jtr019}{\tt doi:10.1112/jtopol/jtr019}],
[\href{https://arxiv.org/abs/0711.1906}{\tt arXiv:0711.1906}].


\vspace{-.3cm}
\bibitem[FM12]{FreedMoore12}
D. S. Freed and G. W. Moore,
{\it Twisted equivariant matter},
Ann. Henri Poincar{\'e} {\bf 14} (2013), 1927–2023, \newline
[\href{https://doi.org/10.1007/s00023-013-0236-x}{\tt doi:10.1007/s00023-013-0236-x}],
[\href{https://arxiv.org/abs/1208.5055}{\tt arXiv:1208.5055}].

\vspace{-3mm}
\bibitem[FKLW03]{FKLW01}
M. Freedman, A. Kitaev, M. Larsen, and Z. Wang,
{\it Topological quantum computation},
Bull. Amer. Math. Soc. {\bf 40} (2003), 31-38,
[\href{https://doi.org/10.1090/S0273-0979-02-00964-3}{\tt doi:10.1090/S0273-0979-02-00964-3}],
[\href{https://arxiv.org/abs/quant-ph/0101025}{\tt arXiv:quant-ph/0101025}].

\vspace{-3mm}
\bibitem[FLW02]{FLW02}
M. Freedman, M. Larsen, and Z. Wang,
{\it A Modular Functor Which is Universal for Quantum Computation},
Commun. Math. Phys. {\bf  227} (2002), 605–622,
[\href{https://doi.org/10.1007/s002200200645}{\tt
doi:10.1007/s002200200645}],
[\href{https://arxiv.org/abs/quant-ph/0001108}{\tt
arXiv:quant-ph/0001108}].


\vspace{-.3cm}
\bibitem[FGM90]{FGM90}
J. Fr{\"o}hlich, F. Gabbiani, and P. Marchetti,
{\it Braid statistics in three-dimensional local quantum field theory}, in: H. C. Lee (ed.)
{\it Physics, Geometry and Topology},
NATO ASI Series {\bf 238} Springer, 1990, \newline
[\href{https://doi.org/10.1007/978-1-4615-3802-8_2}{\tt doi:10.1007/978-1-4615-3802-8\_2}].

\vspace{-.3cm}
\bibitem[FC13]{FruchartCapentier13}
M. Fruchart and D. Carpentier,
{\it An introduction to topological insulators},
Comptes Rendus Physique {\bf 14} (2013), 779-815,
[\href{https://doi.org/10.1016/j.crhy.2013.09.013}{\tt doi:10.1016/j.crhy.2013.09.013}],
[\href{https://arxiv.org/abs/1310.0255}{\tt arXiv:1310.0255}].

\vspace{-.3cm}
\bibitem[Fu11]{Fu11}
L. Fu,
{\it Topological Crystalline Insulators},
Phys. Rev. Lett. {\bf 106} (2011) 106802, \newline
[\href{https://doi.org/10.1103/PhysRevLett.106.106802}{\tt doi:10.1103/PhysRevLett.106.106802}].




\vspace{-.3cm}
\bibitem[FPGM10]{FPGM10}
J. N. Fuchs, F. Pi{\'e}chon, M. O. Goerbig, and G. Montambaux,
{\it Topological Berry phase and semiclassical quantization of cyclotron orbits for two dimensional electrons in coupled band models},
Eur. Phys. J. B {\bf 77} (2010), 351–362,
[\href{https://doi.org/10.1140/epjb/e2010-00259-2}{\tt doi:10.1140/epjb/e2010-00259-2}],
[\href{https://arxiv.org/abs/1006.5632}{\tt arXiv:1006.5632}].

\vspace{-.3cm}
\bibitem[Ga01]{Gaberdiel01}
M R. Gaberdiel,
{\it Fusion rules and logarithmic representations of a WZW model at fractional level},
Nucl. Phys. B {\bf 618} (2001), 407-436,
[\href{https://doi.org/10.1016/S0550-3213(01)00490-4}{\tt doi:10.1016/S0550-3213(01)00490-4}],
[\href{https://arxiv.org/abs/hep-th/0105046}{\tt arXiv:hep-th/0105046}].


\vspace{-.3cm}
\bibitem[Ga03]{Gaberdiel03}
M. R. Gaberdiel,
{\it An algebraic approach to logarithmic conformal field theory},
Int. J. Mod. Phys. A {\bf 18} (2003), 4593-4638,
[\href{https://doi.org/10.1142/S0217751X03016860}{\tt doi:10.1142/S0217751X03016860}],
[\href{https://arxiv.org/abs/hep-th/0111260}{\tt arXiv:hep-th/0111260}].


\vspace{-.3cm}
\bibitem[GK96]{GaberdielKausch96}
M. R. Gaberdiel and H. G. Kausch,
{\it A Rational Logarithmic Conformal Field Theory},
Phys. Lett. B {\bf 386} (1996), 131-137,
[\href{https://doi.org/10.1016/0370-2693(96)00949-5}{\tt doi:10.1016/0370-2693(96)00949-5}],
[\href{https://arxiv.org/abs/hep-th/9606050}{\tt arXiv:hep-th/9606050}].



\vspace{-.3cm}
\bibitem[GVKR19]{GVKR19}
H. Gao, J. W. F. Venderbos, Y. Kim, and A. M. Rappe,
{\it Topological Semimetals from first-principles},
Ann. Rev. Mat. Res. {\bf 49} (2019), 153-183,
[\href{https://doi.org/10.1146/annurev-matsci-070218-010049}{\tt doi:10.1146/annurev-matsci-070218-010049}],
[\href{https://arxiv.org/abs/1810.08186}{\tt arXiv:1810.08186}].



\vspace{-.3cm}
\bibitem[GAT${}^+$13]{GATHLTW13}
C. Gils, E. Ardonne, S. Trebst, D. A. Huse, A. W. W. Ludwig, M. Troyer, and Z. Wang,
{\it Anyonic quantum spin chains: Spin-1 generalizations and topological stability},
Phys. Rev. B{\bf 87} (2013), 235120, \newline
[\href{https://doi.org/10.1103/PhysRevB.87.235120}{\tt doi:10.1103/PhysRevB.87.235120}],
[\href{https://arxiv.org/abs/1303.4290}{\tt arXiv:1303.4290}].


\vspace{-.3cm}
\bibitem[Gi04]{Girvin04}
S. M. Girvin,
{\it Introduction to the Fractional Quantum Hall Effect},
S{\'e}m. Poincar{\'e} {\bf 2} (2004) 53–74; reprinted in: {\it The Quantum Hall Effect}, Progress Math. Phys. {\bf 45} Birkh{\"a}user, 2005,
[\href{https://doi.org/10.1007/3-7643-7393-8_4}{\tt doi:10.1007/3-7643-7393-8\_4}].


\vspace{-.3cm}
\bibitem[GMS81]{GMS81}
G. A. Goldin, R. Menikoff, and D. H. Sharp,
{\it Representations of a local current algebra in nonsimply connected space and the Aharonov–Bohm effect},
J. Math. Phys. {\bf 22} (1981) 1664-1668,
[\href{https://doi.org/10.1063/1.525110}{\tt doi:10.1063/1.525110}].



\vspace{-3mm}
\bibitem[GS17]{GS-tAHSS}
D. Grady and H. Sati,
{\it Twisted differential generalized cohomology theories and their Atiyah-Hirzebruch spectral sequence},
Algebr. Geom. Topol. {\bf 19} (2019), 2899-2960,
[\href{https://doi.org/10.2140/agt.2019.19.2899}{\tt
doi:10.2140/agt.2019.19.2899}],
\href{https://arxiv.org/abs/1711.06650}{[{\tt arXiv:1711.06650}]}.



\vspace{-3mm}
\bibitem[GS18]{GS-Deligne}
D. Grady and H. Sati,
{\it Twisted smooth Deligne cohomology},
Ann. Global Anal. Geom. {\bf 53} (2018), 445-466,
\newline
[\href{https://doi.org/10.1007/s10455-017-9583-z}
{\tt doi:10.1007/s10455-017-9583-z}],
[\href{https://arxiv.org/abs/1706.02742}{\tt arXiv:1706.02742}].



\vspace{-3mm}
\bibitem[GS19]{GS-RR}
D. Grady and H. Sati,
{\it Ramond-Ramond fields and twisted differential K-theory},
[\href{https://arxiv.org/abs/1903.08843}{\tt arXiv:1903.08843}].


\vspace{-.3cm}
\bibitem[GW92]{GreiterWilczek92}
M. Greiter and F. Wilczek,
{\it Exact solutions and the adiabatic heuristic for quantum Hall states},
Nucl. Phys. B {\bf 370}  (1992), 577-600,
[\href{https://doi.org/10.1016/0550-3213(92)90424-A}{\tt doi:10.1016/0550-3213(92)90424-A}].




\vspace{-.3cm}
\bibitem[Gr13]{Grover13}
T. Grover,
{\it Entanglement entropy and strongly correlated topological matter},
Mod. Phys. Lett. A {\bf 28} (2013) 1330001,
[\href{https://doi.org/10.1142/S0217732313300012}{\tt doi:10.1142/S0217732313300012}].

\vspace{-.3cm}
\bibitem[GHL21]{GuHaghighatLiu21}
X. Gu, B. Haghighat, and Y. Liu,
{\it Ising- and Fibonacci-Anyons from KZ-equations},
[\href{https://arxiv.org/abs/2112.07195}{\tt arXiv:2112.07195}].



\vspace{-.3cm}
\bibitem[Gu00]{Gukov}
S. Gukov,
{\it K-Theory, Reality, and Orientifolds},
Commun. Math. Phys. {\bf 210} (2000), 621-639, \newline
[\href{https://doi.org/10.1007/s002200050793}{\tt
doi:10.1007/s002200050793}],
[\href{https://arxiv.org/abs/hep-th/9901042}{\tt
arXiv:hep-th/9901042}].



\vspace{-.3cm}
\bibitem[GLG88]{GLG88}
Y. Guo, J.-M. Langlois, and W. A. Goddard,
{\it Electronic Structure and Valence-Bond Band Structure of Cuprate Superconducting Materials}, Science, New Series {\bf 239} (1988),
896-899,
[\href{https://www.jstor.org/stable/1700316}{\tt jstor:1700316}].




\vspace{-.3cm}
\bibitem[GW09]{GuWen09}
Z.-C. Gu and  X.-G. Wen,
{\it Tensor-Entanglement-Filtering Renormalization Approach and Symmetry Protected Topological Order},
Phys. Rev. B {\bf 80} (2009) 155131,
[\href{https://doi.org/10.1103/PhysRevB.85.075125}{\tt doi:10.1103/PhysRevB.85.075125}],
[\href{https://arxiv.org/abs/0903.1069}{\tt arXiv:0903.1069}].



\vspace{-.3cm}
\bibitem[GW14]{GuWen12}
Z.-C. Gu and X.-G. Wen,
{\it Symmetry-protected topological orders for interacting fermions -- fermionic topological non-linear sigma-models and a group super-cohomology theory},
Phys. Rev. B {\bf 90}  (2014) 115141, [\href{https://arxiv.org/abs/1201.2648}{\tt arXiv:1201.2648}], \newline
[\href{https://doi.org/10.1103/PhysRevB.90.115141}{\tt doi:10.1103/PhysRevB.90.115141}].


\vspace{-.3cm}
\bibitem[GFN97]{GFN97}
V. Gurarie, M. Flohr, and C. Nayak,
{\it The Haldane-Rezayi Quantum Hall State and Conformal Field Theory},
Nucl. Phys. B {\bf 498} (1997), 513-538,
[\href{https://doi.org/10.1016/S0550-3213(97)00351-9}{\tt doi:10.1016/S0550-3213(97)00351-9}],
[\href{https://arxiv.org/abs/cond-mat/9701212}{\tt arXiv:cond-mat/9701212}].


\vspace{-.3cm}
\bibitem[HLS05]{HLS05}
C. Hainzl, M. Lewin, and  É. S{\'e}r{\'e},
{\it Existence of a Stable Polarized Vacuum in the Bogoliubov-Dirac-Fock Approximation},
Comm. Math. Phys. {\bf 257} (2005), 515–562,
[\href{https://doi.org/10.1007/s00220-005-1343-4}{\tt doi:10.1007/s00220-005-1343-4}].

\vspace{-.3cm}
\bibitem[Ha15]{Haldane15}
F. D. M. Haldane,
{\it Model for a Quantum Hall Effect without Landau Levels: Condensed-Matter Realization of the ``Parity Anomaly''}, Phys. Rev. Lett. {\bf 61} (2015),
2015-2018,
[\href{https://doi.org/10.1103/PhysRevLett.61.2015}{\tt doi:10.1103/PhysRevLett.61.2015}].

\vspace{-.3cm}
\bibitem[Hal84]{Halperin84}
B. I. Halperin,
{\it Statistics of Quasiparticles and the Hierarchy of Fractional Quantized Hall States},
Phys. Rev. Lett. {\bf 52} (1984) 1583-1586,
[\href{https://doi.org/10.1103/PhysRevLett.52.1583}{\tt doi:10.1103/PhysRevLett.52.1583}].

\vspace{-.3cm}
\bibitem[HLS18]{HartnollLucasSachdev18}
S. Hartnoll, A. Lucas, and S. Sachdev,
{\it Holographic quantum matter},
MIT Press, 2018,
[\href{https://mitpress.mit.edu/books/holographic-quantum-matter}{\tt
ISBN:9780262038430}],
[\href{https://arxiv.org/abs/1612.07324}{\tt arXiv:1612.07324}].


\vspace{-.3cm}
\bibitem[Hi1903]{Hilton1903}
H. Hilton,
{\it Mathematical crystallography and the theory of groups of movements},
Clarendon Press, Oxford, 1903,
[\href{https://archive.org/details/mathematicalcry03hiltgoog/page/n6/mode/2up}{\tt archive.org/details/mathematicalcry03hiltgoog/page/n6/mode/2up}].




\vspace{-.3cm}
\bibitem[HH92]{HosotaniHo92}
Y. Hosotani and  C.-L. Ho,
{\it Anyons on a Torus},
AIP Conf. Proc.  {\bf 272} (1992) 1466,
[\href{https://doi.org/10.1063/1.43444}{\tt doi:10.1063/1.43444}],
[\href{https://arxiv.org/abs/hep-th/9210112}{\tt arXiv:hep-th/9210112}].


\vspace{-.3cm}
\bibitem[HGZ21]{HGZ21}
J. Hu, J. Gu and W. Zhang,
{\it Bloch’s band structures of a pair of interacting electrons in simple one- and two-dimensional lattices},
Phys. Lett. A {\bf 414} (2021) 127634,
[\href{https://doi.org/10.1016/j.physleta.2021.127634}{\tt doi:10.1016/j.physleta.2021.127634}].


\vspace{-.3cm}
\bibitem[HJJS08]{HJJS08}
D. Husem{\"o}ller, M. Joachim, B. Jur{\v c}o, and M. Schottenloher,
{\it Basic Bundle Theory and K-Cohomology Invariants},
Lecture Notes in Physics {\bf 726}, Springer, 2008,
[\href{https://link.springer.com/book/10.1007/978-3-540-74956-1}{\tt doi:10.1007/978-3-540-74956-1}].


\vspace{-.3cm}
\bibitem[IL92]{IengoLechner92}
R. Iengo and K. Lechner,
{\it Anyon quantum mechanics and Chern-Simons theory}, Phys. Rep. {\bf 213} 4 (1992) 179-269 [\href{https://doi.org/10.1016/0370-1573(92)90039-3}{\tt doi:10.1016/0370-1573(92)90039-3}]

\vspace{-.3cm}
\bibitem[IIS90]{IIS90}
T. Imbo, C. S. Imbo, and E. C. G. Sudarshan,
{\it Identical particles, exotic statistics and braid groups}, Phys. Lett. B {\bf 234} (1990), 103-107. [\href{https://doi.org/10.1016/0370-2693(90)92010-G}{\tt doi:10.1016/0370-2693(90)92010-G}].

\vspace{-.3cm}
\bibitem[Ino98]{Ino98}
K. Ino,
{\it Modular Invariants in the Fractional Quantum Hall Effect},
Nucl. Phys. B {\bf 532} (1998), 783-806, \newline
[\href{https://doi.org/10.1016/S0550-3213(98)00598-7}{\tt doi:10.1016/S0550-3213(98)00598-7}],
[\href{https://arxiv.org/abs/cond-mat/9804198}{\tt arXiv:cond-mat/9804198}].

\vspace{-.3cm}
\bibitem[J\"a65]{Janich65}
K. J\"anich,
{\it Vektorraumb\"undel und der Raum der Fredholm-Operatoren
(Vector space bundles and the space of Fredholm operators)},
Math. Ann. {\bf 161} (1965), 129-142,
[\href{https://doi.org/10.1007/BF01360851}{\tt
doi:10.1007/BF01360851}].


\vspace{-.3cm}
\bibitem[JBL${}^+$21]{JBL21}
B. Jiang, A. Bouhon, Z.-K. Lin, X. Zhou, B. Hou, F. Li, R.-J. Slager and J.-H. Jiang,
{\it Observation of non-Abelian topological semimetals and their phase transitions}, Nature Physics {\bf 17} (2021) 1239-1246
[\href{https://arxiv.org/abs/2104.13397}{\tt arXiv:2104.13397}]
[\href{https://doi.org/10.1038/s41567-021-01340-x}{\tt doi:10.1038/s41567-021-01340-x}]


\vspace{-.3cm}
\bibitem[Jin${}^+$20]{JinEtAl20}
Y. J. Jin, B. B. Zheng, X. L. Xiao, Z. J. Chen, Y. Xu, and H. Xu,
{\it Two-dimensional Dirac Semimetals without Inversion Symmetry},
Phys. Rev. Lett. {\bf 125}  (2020) 116402,
[\href{https://doi.org/10.1103/PhysRevLett.125.116402}{\tt doi:10.1103/PhysRevLett.125.116402}], \newline
[\href{https://arxiv.org/abs/2008.10175}{\tt arXiv:2008.10175}].


\vspace{-.3cm}
\bibitem[JS21]{JohansenSimula20}
E. G. Johansen and T. Simula,
{\it Fibonacci anyons versus Majorana fermions}
PRX Quantum {\bf 2} (2021) 010334, \newline
[\href{https://doi.org/10.1103/PRXQuantum.2.010334}{\tt
doi:10.1103/PRXQuantum.2.010334}],
[\href{https://arxiv.org/abs/2008.10790}{\tt arXiv:2008.10790}].


\vspace{-.3cm}
\bibitem[JRN21]{JRN21}
N. B. Joseph, S. Roy, and A. Narayan, {\it Tunable topology and berry curvature dipole in transition metal dichalcogenide Janus monolayers}, Mater. Res. Express {\bf 8} (2021) 124001,
[\href{https://iopscience.iop.org/article/10.1088/2053-1591/ac440b}{\tt doi:10.1088/2053-1591/ac440b}].

\vspace{-.3cm}
\bibitem[KV20]{KangVafek20}
J. Kang and O. Vafek,
{\it Non-Abelian Dirac node braiding and near-degeneracy of correlated phases at odd integer filling in magic angle twisted bilayer graphene},
Phys. Rev. B {\bf 102} (2020) 035161
[\href{https://arxiv.org/abs/2002.10360}{\tt arXiv:2002.10360}]
[\href{https://doi.org/10.1103/PhysRevB.102.035161}{\tt doi:10.1103/PhysRevB.102.035161}]

\vspace{-.3cm}
\bibitem[Ka70]{Karoubi70}
M. Karoubi,
{\it Espaces Classifiants en K-Th{\'e}orie}, Trans. Amer. Math. Soc. {\bf 147}  (1970), 75-115, \newline
[\href{https://doi.org/10.2307/1995218}{\tt doi:10.2307/1995218}].


\vspace{-.3cm}
\bibitem[Ka78]{Karoubi78}
M. Karoubi,
{\it K-Theory -- An introduction},
Grundl. Math. Wiss. {\bf 226}, Springer, Berlin, 1978, \newline
[\href{https://link.springer.com/book/10.1007/978-3-540-79890-3}{\tt doi:10.1007/978-3-540-79890-3}].

\vspace{-3mm}
\bibitem[Ka87]{Karoubi87}
M. Karoubi,
{\it Homologie Cyclique et K-Theorie}, Ast{\'e}risque {\bf 149} (1987),
[\href{http://www.numdam.org/item/AST_1987__149__1_0}{\tt numdam:AST\_1987\_\_149\_\_1\_0}].

\vspace{-3mm}
\bibitem[Ka90]{Karoubi90}
M. Karoubi,
{\it Th{\'e}orie G{\'e}n{\'e}rale des Classes Caract{\'e}ristiques Secondaires},
K-Theory {\bf 4}  (1990), 55-87,  \newline
[\href{http://dx.doi.org/10.1007/BF00534193}{\tt doi:10.1007/BF00534193}].

\vspace{-.3cm}
\bibitem[KMM20]{KMM20}
A. Kareekunnan, M. Muruganathan, and H. Mizuta,
{\it Manipulating Berry curvature in hBN/bilayer graphene commensurate heterostructures},
Phys. Rev. B {\bf 101} (2020) 195406,
[\href{https://doi.org/10.1103/PhysRevB.101.195406}{\tt doi:10.1103/PhysRevB.101.195406}].


\vspace{-.3cm}
\bibitem[Ka50]{Kato50}
T. Kato,
{\it On the Adiabatic Theorem of Quantum Mechanics},
J. Phys. Soc. Jpn. {\bf 5} (1950), 435-439, \newline
[\href{https://journals.jps.jp/doi/10.1143/JPSJ.5.435}{\tt doi:10.1143/JPSJ.5.435}].


\vspace{-3mm}
\bibitem[KR19]{KawasetsuRidout19}
K. Kawasetsu and D.  Ridout,
{\it Relaxed highest-weight modules I: rank 1 cases},
Commun. Math. Phys. {\bf 368} (2019), 627–663,
[\href{https://doi.org/10.1007/s00220-019-03305-x}{\tt doi:10.1007/s00220-019-03305-x}],
[\href{https://arxiv.org/abs/1803.01989}{\tt arXiv:1803.01989}].

\vspace{-3mm}
\bibitem[KR22]{KawesetsuRidout22}
K. Kawasetsu and D. Ridout,
{\it Relaxed highest-weight modules II: classifications for affine vertex algebras}, Comm. Contemp. Math. {\bf 24} 5 (2022) 2150037,
[\href{https://arxiv.org/abs/1906.02935}{\tt arXiv:1906.02935}].


\vspace{-3mm}
\bibitem[Ki03]{Kitaev03}
A. Kitaev,
{\it Fault-tolerant quantum computation by anyons},
Annals Phys. {\bf 303} (2003), 2-30,  \newline
[\href{https://doi.org/10.1016/S0003-4916(02)00018-0}{\tt doi:10.1016/S0003-4916(02)00018-0}],
[\href{https://arxiv.org/abs/quant-ph/9707021}{\tt arXiv:quant-ph/9707021}].

\vspace{-.3cm}
\bibitem[Ki06]{Kitaev06}
A. Kitaev,
{\it Anyons in an exactly solved model and beyond}, Ann. Phys. {\bf 321} 1 (2006) 2-111,
\newline
[\href{https://doi.org/10.1016/j.aop.2005.10.005}{\tt doi:10.1016/j.aop.2005.10.005}],
[\href{https://arxiv.org/abs/cond-mat/0506438}{\tt
arXiv:cond-mat/0506438}].


\vspace{-.3cm}
\bibitem[Ki09]{Kitaev09}
A. Kitaev,
{\it Periodic table for topological insulators and superconductors},
AIP Conference Proceedings {\bf 1134} (2009) 22,
[\href{https://doi.org/10.1063/1.3149495}{\tt doi:10.1063/1.3149495}],
[\href{https://arxiv.org/abs/0901.2686}{\tt arXiv:0901.2686}].


\vspace{-.3cm}
\bibitem[Ki11]{Kitaev11}
A. Kitaev,
{\it Toward Topological Classification of Phases with Short-range Entanglement}, talk at KITP (2011), \newline
[\href{https://online.kitp.ucsb.edu/online/topomat11/kitaev}{\tt online.kitp.ucsb.edu/online/topomat11/kitaev}].


\vspace{-.3cm}
\bibitem[Ki13]{Kitaev13}
A. Kitaev,
{\it On the Classification of Short-Range Entangled States}, talk at Simons Center SCGP (2013), \newline
[\href{https://scgp.stonybrook.edu/archives/7874}{\tt scgp.stonybrook.edu/archives/7874}].

\vspace{-.3cm}
\bibitem[KP06]{KitaevPreskill06}
A. Kitaev and J. Preskill,
{\it Topological entanglement entropy}, Phys. Rev. Lett. {\bf 96} (2006) 110404, \newline
[\href{https://arxiv.org/abs/hep-th/0510092}{\tt arXiv:hep-th/0510092}].


\vspace{-.3cm}
\bibitem[Ki53]{Kittel53}
C. Kittel,
{\it Introduction to Solid State Physics},
Wiley, 1953,
[\href{https://www.wiley.com/en-us/Introduction+to+Solid+State+Physics\%2C+8th+Edition-p-9780471415268}{\tt ISBN:978-0-471-41526-8}].

\vspace{-.3cm}
\bibitem[KS77]{KS77}
M. Klaus and G. Scharf,
{\it The regular external field problem in quantum electrodynamics}, Helv. Phys. Acta {\bf 50} (1977),
779-802, [\href{http://doi.org/10.5169/seals-114890}{\tt doi:10.5169/seals-114890}].

\vspace{-.3cm}
\bibitem[Ko96]{Kochman96}
S. Kochman,
{\it Bordism, Stable Homotopy and Adams Spectral Sequences},
Fields Institute Monographs,
Amer. Math. Soc., 1996,
[\href{https://cdsweb.cern.ch/record/2264210}{\tt cds:2264210}].


\vspace{-.3cm}
\bibitem[Ko02]{Kohno02}
T. Kohno,
{\it Conformal field theory and topology},
Transl Math. Monogr. {\bf 210},
Amer. Math. Soc., Providence, RI, 2002,
[\href{https://bookstore.ams.org/mmono-210}{\tt ams:mmono-210}].

\vspace{-.3cm}
\bibitem[Kou${}^+$21]{KouwenhovenEtAl21}
L. P. Kouwenhoven et al.,
{\it Retraction Note: Quantized Majorana conductance},
Nature {\bf 591}  (2021) E30, \newline
[\href{https://doi.org/10.1038/s41586-021-03373-x}{\tt doi:10.1038/s41586-021-03373-x}].

\vspace{-.3cm}
\bibitem[Kou${}^+$22]{KouwenhovenEtAl22}
L. P. Kouwenhoven et al.
{\it Retraction Note: Epitaxy of advanced nanowire quantum devices},
Nature {\bf 604} (2022) 786,
[\href{https://doi.org/10.1038/s41586-022-04704-2}{\tt doi:10.1038/s41586-022-04704-2}].


\vspace{-.3cm}
\bibitem[KdBvWKS17]{KdBvWKS17}
J. Kruthoff, J. de Boer, J. van Wezel, C. L. Kane, and R.-J. Slager
{\it Topological Classification of Crystalline Insulators through Band Structure Combinatorics},
Phys. Rev. X {\bf 7} (2017) 041069, \newline
[\href{https://journals.aps.org/prx/abstract/10.1103/PhysRevX.7.041069}{\tt doi:10.1103/PhysRevX.7.041069}].


\vspace{-.3cm}
\bibitem[Ku98]{Kulikov98}
V. S. Kulikov,
{\it Mixed Hodge Structures and Singularities}, Cambridge University Press, 1998, \newline
[\href{https://doi.org/10.1017/CBO9780511758928}{\tt doi:10.1017/CBO9780511758928}].




\vspace{-.3cm}
\bibitem[La19]{Lan19}
T. Lan,
{\it Matrix formulation for non-Abelian families},
Phys. Rev. B {\bf 100} (2019) 241102,
[\href{https://arxiv.org/abs/1908.02599}{\tt arXiv:1908.02599}],
[\href{https://doi.org/10.1103/PhysRevB.100.241102}{\tt doi:10.1103/PhysRevB.100.241102}].

\vspace{-.3cm}
\bibitem[La17]{Landsman17}
K. Landsman,
{\it Foundations of quantum theory -- From classical concepts to Operator algebras},
Springer Open, 2017,
[\href{https://link.springer.com/book/10.1007/978-3-319-51777-3}{\tt doi:10.1007/978-3-319-51777-3}].


\vspace{-.3cm}
\bibitem[LCD86]{LanczosClarkDerrick86}
C. Lanczos, R. C. Clark, and G. H. Derrick (eds.),
{\it Mathematical Methods in Solid State and Superfluid Theory},
Springer, Berlin, 1986,
[\href{https://doi.org/10.1007/978-1-4899-6435-9}{\tt doi:10.1007/978-1-4899-6435-9}].


\vspace{-.3cm}
\bibitem[Lau83]{Laughlin83}
R. B. Laughlin,
{\it Anomalous Quantum Hall Effect: An Incompressible Quantum Fluid with Fractionally Charged Excitations}, Phys. Rev. Lett. {\bf 50} (1983), 1395-1398,
[\href{https://doi.org/10.1103/PhysRevLett.50.1395}{\tt doi:10.1103/PhysRevLett.50.1395}].

\vspace{-.3cm}
\bibitem[Le16]{Lehnert16}
R. Lehnert,
{\it CPT Symmetry and Its Violation},
Symmetry {\bf 8} (2016) 114,
[\href{https://doi.org/10.3390/sym8110114}{\tt doi:10.3390/sym8110114}].


\vspace{-.3cm}
\bibitem[Le92]{Lerda92}
A. Lerda,
{\it Anyons -- Quantum Mechanics of Particles with Fractional Statistics},
Lect. Notes Phys. {\bf 14}, Springer, Berlin, 1992,
[\href{https://link.springer.com/book/10.1007/978-3-540-47466-1}{\tt doi:10.1007/978-3-540-47466-1}].

\vspace{-.3cm}
\bibitem[LW06]{LevinWen06}
M. Levin and X.-G. Wen,
{\it Detecting topological order in a ground state wave function},
Phys. Rev. Lett. {\bf 96}  (2006) 110405,
[\href{https://doi.org/10.1103/PhysRevLett.96.110405}{\tt doi:10.1103/PhysRevLett.96.110405}],
[\href{https://arxiv.org/abs/cond-mat/0510613}{\tt arXiv:cond-mat/0510613}].


\vspace{-.3cm}
\bibitem[Li06]{Li06}
S. S. Li,
{\it Semiconductor Physical Electronics}, Springer, 2006, [\href{https://doi.org/10.1007/0-387-37766-2_4}{\tt doi:10.1007/0-387-37766-2\_4}].




\vspace{-.3cm}
\bibitem[LWQG20]{LieWangQiuGao20}
C.-W. Liu, Z. Wang, R. L. J. Qiu, and X. P. A. Gao,
{\it Development of topological insulator and topological crystalline insulator nanostructures},
Nanotechnology {\bf 31} (2020) 192001,
[\href{https://iopscience.iop.org/article/10.1088/1361-6528/ab6dfc}{\tt doi:10.1088/1361-6528/ab6dfc}].


\vspace{-.3cm}
\bibitem[Li${}^+$20]{LieEtAl20}
J. Li, Z. Zhang, C. Wang, H. Huang, B.-L. Gu, and W. Duan,
{\it Topological semimetals from the perspective of first-principles calculations},
J. Appl. Phys.  {\bf 128} (2020) 191101,
[\href{https://doi.org/10.1063/5.0025396}{\tt doi:10.1063/5.0025396}].

\vspace{-.3cm}
\bibitem[LP93]{LoPreskill93}
H.-K. Lo and J. Preskill,
{\it Non-Abelian vortices and non-Abelian statistics},
Phys. Rev. D {\bf 48} (1993),  4821-4834,
[\href{https://doi.org/10.1103/PhysRevD.48.4821}{\tt doi:10.1103/PhysRevD.48.4821}].

\vspace{-3mm}
\bibitem[Lo94]{Lott94}
J. Lott,
{\it R/Z Index Theory},
Comm. Anal.  Geom.  {\bf 2} (1994), 279-311,
[\href{https://dx.doi.org/10.4310/CAG.1994.v2.n2.a6}{\tt doi:10.4310/CAG.1994.v2.n2.a6}].

\vspace{-.3cm}
\bibitem[LL03]{LuoLuo03}
S.-l. Luo and Y.-F. Luo,
{\it Correlation and Entanglement}, Acta Math. Appl. Sinica {\bf 19} (2003), 581–598, \newline
[\href{https://doi.org/10.1007/s10255-003-0133-z}{\tt doi:10.1007/s10255-003-0133-z}].

\vspace{-.3cm}
\bibitem[LV22]{LuVijay22}
T.-C. Lu and S. Vijay,
{\it Characterizing Long-Range Entanglement in a Mixed State Through an Emergent Order on the Entangling Surface},
[\href{https://arxiv.org/abs/2201.07792}{\tt arXiv:2201.07792}].



\vspace{-.3cm}
\bibitem[MHZ11]{MHZ11}
J. Maciejko, T. L. Hughes, and S.-C. Zhang,
{\it The Quantum Spin Hall Effect}, Ann. Rev. Condensed Matter Physics {\bf 2} (2011), 31-53,
[\href{https://doi.org/10.1146/annurev-conmatphys-062910-140538}{\tt doi:10.1146/annurev-conmatphys-062910-140538}],
\newline
[\url{http://physics.gu.se/~tfkhj/QSHreview.2011.pdf}].

\vspace{-.3cm}
\bibitem[Ma19]{Mansuripur19}
M. Mansuripur,
{\it Spin-orbit coupling in the hydrogen atom, the Thomas precession, and the exact solution of Dirac’s equation}, Spintronics XII, Proc.
 SPIE {\bf 11090} (2019) 110901X,
[\href{https://doi.org/10.1117/12.2529885}{\tt doi:10.1117/12.2529885}], \newline
[\href{https://arxiv.org/abs/1909.07333}{\tt arXiv:1909.07333}].


\vspace{-.3cm}
\bibitem[MMN21]{MasakiMizushimaNitta21}
Y. Masaki, T. Mizushima, and M. Nitta,
{\it Non-Abelian Half-Quantum Vortices in ${}^3P_2$ Topological Superfluids}, \newline
[\href{https://arxiv.org/abs/2107.02448}{\tt arXiv:2107.02448}].




\vspace{-.3cm}
\bibitem[MT16]{MathaiThiang15Semimetals}
V. Mathai and G. C. Thiang,
{\it Differential Topology of Semimetals},
Commun. Math. Phys. {\bf 355} (2017), 561–602,
[\href{https://doi.org/10.1007/s00220-017-2965-z}{\tt doi:10.1007/s00220-017-2965-z}],
[\href{https://arxiv.org/abs/1611.08961}{\tt arXiv:1611.08961}].

\vspace{-.3cm}
\bibitem[MT17]{MathaiThiang16}
V. Mathai and  G. Ch. Thiang,
{\it Global topology of Weyl semimetals and Fermi arcs}, J. Phys. A: Math. Theor. {\bf 50} (2017) 11LT01,
[\href{https://doi.org/10.1088/1751-8121/aa59b2}{\tt doi:10.1088/1751-8121/aa59b2}],
[\href{https://arxiv.org/abs/1607.02242}{\tt arXiv:1607.02242}].

\vspace{-.3cm}
\bibitem[MPSS19]{MPSS19}
T. Mawson, T. Petersen, J. Slingerland, and T. Simula,
{\it Braiding and fusion of non-Abelian vortex anyons}, Phys. Rev. Lett. {\bf 123} (2019) 140404,
[\href{https://doi.org/10.1103/PhysRevLett.123.140404}{\tt doi:10.1103/PhysRevLett.123.140404}].


\vspace{-.3cm}
\bibitem[MMB${}^+$19]{MMBDRSC19}
G. C. Ménard, A. Mesaros, C. Brun, F. Debontridder, D. Roditchev, P. Simon, and T. Cren,
{\it Isolated pairs of Majorana zero modes in a disordered superconducting lead monolayer},
Nature Commun. {\bf 10} (2019) 2587, \newline
[\href{https://doi.org/10.1038/s41467-019-10397-5}{\tt doi:10.1038/s41467-019-10397-5}].



\vspace{-.3cm}
\bibitem[Mer79]{Mermin79}
N. D. Mermin,
{\it The topological theory of defects in ordered media}, Rev. Mod. Phys. {\bf 51} (1979), 591-648, \newline  [\href{https://doi.org/10.1103/RevModPhys.51.591}{\tt doi:10.1103/RevModPhys.51.591}].

\vspace{-.3cm}
\bibitem[Mi72]{Miller72}
W. Miller,
{\it Symmetry Groups and Their Applications},
Pure and Applied Math. {\bf 50} (1972), 16-60, \newline
[\href{https://doi.org/10.1016/S0079-8169(08)60959-9}{\tt doi:10.1016/S0079-8169(08)60959-9}].


\vspace{-.3cm}
\bibitem[MiSt74]{MilnorStasheff74}
J. Milnor and J. Stasheff,
{\it Characteristic classes},
Princeton Univ. Press (1974)
[\href{https://doi.org/10.1515/9781400881826}{\tt doi:10.1515/9781400881826}]


\vspace{-3mm}
\bibitem[MR91]{MooreRead91}
G. Moore and  N. Read,
{\it Nonabelions in the fractional quantum hall effect}, Nucl. Phys. B {\bf 360} (1991), 362-396, \newline
[\href{https://doi.org/10.1016/0550-3213(91)90407-O}{\tt doi:10.1016/0550-3213(91)90407-O}].


\vspace{-.3cm}
\bibitem[MF13]{MorimotoFurusaki13}
T. Morimoto and A. Furusaki,
{\it Topological classification with additional symmetries from Clifford algebras},
Phys. Rev. B {\bf 88} (2013) 125129,
[\href{https://doi.org/10.1103/PhysRevB.88.125129}{\tt doi:10.1103/PhysRevB.88.125129}],
[\href{https://arxiv.org/abs/1306.2505}{\tt arXiv:1306.2505}].




\vspace{-.3cm}
\bibitem[MM21]{MoessnerMoore21}
R. Moessner and J. Moore,
{\it Topological Phases of Matter},
Cambridge University Press, 2021, \newline
[\href{https://doi.org/10.1017/9781316226308}{\tt doi:10.1017/9781316226308}].

\vspace{-.3cm}
\bibitem[MP15]{MonacoPanati16}
D. Monaco and G. Panati,
{\it Symmetry and localization in periodic crystals: triviality of Bloch bundles with a fermionic time-reversal symmetry},
Acta Appl. Math. {\bf 137} (2015), 185-203,
[\href{https://doi.org/10.1007/s10440-014-9995-8}{\tt doi:10.1007/s10440-014-9995-8}], \newline
[\href{https://arxiv.org/abs/1601.02906}{\tt arXiv:1601.02906}].


\vspace{-.3cm}
\bibitem[MPPS18]{MPPS16}
D. Monaco, G. Panati, A. Pisante, and S. Teufel,
{\it Optimal decay of Wannier functions in Chern and Quantum Hall insulators}, Comm. Math. Phys. {\bf 359} (2018), 61-100,
[\href{https://doi.org/10.1007/s00220-017-3067-7}{\tt doi:10.1007/s00220-017-3067-7}],
[\href{https://arxiv.org/abs/1612.09552}{\tt arXiv:1612.09552}].

\vspace{-.3cm}
\bibitem[MuSc95]{MundSchrader95}
J. Mund and R. Schrader,
{\it Hilbert Spaces for Nonrelativistic and Relativistic “Free” Plektons (Particles with Braid Group Statistics)},
in: {\it Advances in dynamical systems and quantum physics} (Capri, 1993), World Scientific (1995), pp. 235-259,
[\href{https://arxiv.org/abs/hep-th/9310054}{\tt arXiv:hep-th/9310054}].

\vspace{-.3cm}
\bibitem[Mu00]{Munkres00}
J. Munkres,
{\it Topology},
Prentice Hall, 2000,
[{\tt ISBN:9780131816299}].





\vspace{-.3cm}
\bibitem[MuSh09]{MurthyShankar09}
M. V. N. Murthy and  R. Shankar, {\it Exclusion Statistics: From Pauli to Haldane}, Institute of Mathematical Sciences report No. 120,
2009, [\href{https://www.imsc.res.in/xmlui/handle/123456789/334}{\tt dspace:123456789/334}].



\vspace{-.3cm}
\bibitem[LM77]{LeinaasMyrheim77}
J. M. Leinaas and J. Myrheim,
{\it On the theory of identical particles}, Nuovo Cim B {\bf 37} (1977), 1–23, \newline
[\href{https://doi.org/10.1007/BF02727953}{\tt doi:10.1007/BF02727953}].

\vspace{-.3cm}
\bibitem[Na03]{Nakahara03}
M. Nakahara,
{\it Geometry, Topology and Physics},
Institute of Physics Publishing, Bristol, 2003, \newline
[\href{https://doi.org/10.1201/9781315275826}{\tt doi:10.1201/9781315275826}].






\vspace{-3mm}
\bibitem[NSSFS08]{NSSFS08}
C. Nayak, S. H. Simon, A. Stern, M. Freedman, and S. Das Sarma,
{\it Non-Abelian anyons and topological quantum computation},
Rev. Mod. Phys. {\bf  80} (2008), 1083-1159,
[\href{https://doi.org/10.1103/RevModPhys.80.1083}{\tt
doi:10.1103/RevModPhys.80.1083}].

\vspace{-.3cm}
\bibitem[Nen80]{Nenciu80}
G. Nenciu,
{\it On the adiabatic theorem of quantum mechanics},
J. Phys. A: Math. Gen. {\bf 13} (1980) L15, \newline
[\href{https://iopscience.iop.org/article/10.1088/0305-4470/13/2/002}{\tt doi:10.1088/0305-4470/13/2/002}].

\vspace{-.3cm}
\bibitem[Ni02a]{Nichols02}
A. Nichols,
{\it Extended chiral algebras in the $\mathrm{SU}(2)_0$ WZNW model},
J. High Energy Phys. {\bf 04} (2002), \newline
[\href{https://iopscience.iop.org/article/10.1088/1126-6708/2002/04/056}{\tt doi:10.1088/1126-6708/2002/04/056}],
[\href{https://arxiv.org/abs/hep-th/0112094}{\tt arXiv:hep-th/0112094}].

\vspace{-.3cm}
\bibitem[Ni02b]{NicholsThesis02}
A. Nichols,
{\it $\mathrm{SU}(2)_0$ Logarithmic Conformal Field Theories},
PhD thesis, Oxford, 2012,
[\href{https://arxiv.org/abs/hep-th/0210070}{\tt arXiv:hep-th/0210070}].

\vspace{-.3cm}
\bibitem[NG${}^+$04]{NGMJZDGF04}
K. Novoselov, A. Geim, S. V. Morozov, D. Jiang, Y. Zhang, S.V. Dubonos, I.V. Grigorieva, and A.A. Firsov,
{\it Electric field effect in atomically thin carbon films}, Science {\bf 306} (2004), 666-669,
[\href{http://dx.doi.org/10.1126/science.1102896}{\tt doi:10.1126/science.1102896}],
[\href{https://arxiv.org/abs/cond-mat/0410550}{\tt arXiv:cond-mat/0410550}].


\vspace{-.3cm}
\bibitem[PP94]{PakuliakPerelomov94}
S. Pakuliak, A. Perelomov,
{\it Relation Between Hyperelliptic Integrals},
Mod. Phys. Lett. {\bf 9} (1994), 1791-1798, \newline
[\href{https://doi.org/10.1142/S0217732394001647}{\tt doi:10.1142/S0217732394001647}].

\vspace{-.3cm}
\bibitem[Pa06]{Panati06}
G. Panati,
{\it Triviality of Bloch and Bloch-Dirac bundles},
Ann. Henri Poincar{\'e} {\bf 8} (2007), 995-1011, \newline
[\href{https://doi.org/10.1007/s00023-007-0326-8}{\tt doi:10.1007/s00023-007-0326-8}],
[\href{https://arxiv.org/abs/math-ph/0601034}{\tt arXiv:math-ph/0601034}].



\vspace{-.3cm}
\bibitem[PGZO22]{PGZO22}
H. Park, W. Gao, X. Zhang, and S. S. Oh,
{\it Nodal lines in momentum space: topological invariants and recent realizations in photonic and other systems}, Nanophotonics
{\bf 11} 11 (2022) 2779-2801,
[\href{https://doi.org/10.1515/nanoph-2021-0692}{\tt doi:10.1515/nanoph-2021-0692}].


\vspace{-.3cm}
\bibitem[PBSM22]{PBSM22}
B. Peng, A. Bouhon, R.-J. Slager, and B. Monserrat,
{\it Multi-gap topology and non-Abelian braiding of phonons from first principles},
Phys. Rev. B {\bf 105} (2022) 085115, [\href{https://doi.org/10.1103/PhysRevB.105.085115}{\tt doi:10.1103/PhysRevB.105.085115}], \newline
[\href{https://arxiv.org/abs/2111.05872}{\tt arXiv:2111.05872}].



\vspace{-.3cm}
\bibitem[PBMS22]{PBMS22}
B. Peng, A. Bouhon, B. Monserrat, and R.-J. Slager,
{\it Phonons as a platform for non-Abelian braiding and its manifestation in layered silicates},
Nature Commun. {\bf 13} (2022) 423,
[\href{https://doi.org/10.1038/s41467-022-28046-9}{\tt doi:10.1038/s41467-022-28046-9}], \newline
[\href{https://arxiv.org/abs/2105.08733}{\tt arXiv:2105.08733}].




\vspace{-.3cm}
\bibitem[PRFM16]{PRFM16}
F. Pi{\'e}chon, A. Raoux, J.-N. Fuchs, and G. Montambaux,
{\it Geometric orbital susceptibility: quantum metric without Berry curvature},
Phys. Rev. B {\bf 94} (2016)  134423,
[\href{https://doi.org/10.1103/PhysRevB.94.134423}{\tt doi:10.1103/PhysRevB.94.134423}],
[\href{https://arxiv.org/abs/1605.01258}{\tt arXiv:1605.01258}].





\vspace{-.3cm}
\bibitem[PBTO12]{PBTO09}
F. Pollmann, E. Berg, A. M. Turner, and M. Oshikawa,
{\it Symmetry protection of topological order in one-dimensional quantum spin systems}, Phys. Rev. B {\bf 85} (2012)  075125,
\newline
[\href{https://doi.org/10.1103/PhysRevB.85.075125}{\tt doi:10.1103/PhysRevB.85.075125}],
[\href{https://arxiv.org/abs/0909.4059}{\tt arXiv:0909.4059}].







\vspace{-.3cm}
\bibitem[PJ21]{PuJain21}
S. Pu and J. K. Jain,
{\it Composite anyons on a torus},
Phys. Rev. B {\bf 104} (2021) 115135,
[\href{https://arxiv.org/abs/2106.15705}{\tt arXiv:2106.15705}], \newline
[\href{https://doi.org/10.1103/PhysRevB.104.115135}{\tt doi:10.1103/PhysRevB.104.115135}].



\vspace{-3mm}
\bibitem[Rao16]{Rao16}
S. Rao,
{\it Introduction to abelian and non-abelian anyons},
In: {\it Topology and Condensed Matter} Phys. Texts \& Read.  in Phys.  Sci. {\bf 19} Springer (2017), 399-437,
[\href{https://doi.org/10.1007/978-981-10-6841-6_16}{\tt doi:10.1007/978-981-10-6841-6\_16}],
[\href{https://arxiv.org/abs/1610.09260}{\tt arXiv:1610.09260}].

\vspace{-.3cm}
\bibitem[RR99]{ReadRezayi99}
N. Read and E. Rezayi,
{\it Beyond paired quantum Hall states: Parafermions and incompressible states in the first excited Landau level},
Phys. Rev. B {\bf 59} (1999) 8084,
[\href{https://doi.org/10.1103/PhysRevB.59.8084}{\tt doi:10.1103/PhysRevB.59.8084}].


\vspace{-.3cm}
\bibitem[RS78]{ReedSimon78}
M. Reed and B. Simon,
{\it Methods of Modern Mathematical Physics
--
IV: Analysis of Operators},
Academic Press, 1978,
[\href{https://www.elsevier.com/books/iv-analysis-of-operators/reed/978-0-08-057045-7}{\tt
ISBN:9780080570457}].


\vspace{-3mm}
\bibitem[Ri10]{Ridout10}
D. Ridout,
{\it Fractional Level WZW Models as Logarithmic CFTs}, seminar notes, Melbourne, 2010,
\newline
[\url{https://researchers.ms.unimelb.edu.au/~dridout@unimelb/seminars/100225.pdf}]



\vspace{-3mm}
\bibitem[Ri20]{Ridout20}
D. Ridout,
{\it Fractional-level WZW models}, seminar notes,
Melbourne, 2020, \newline
[\url{https://researchers.ms.unimelb.edu.au/~dridout@unimelb/seminars/200205.pdf}]




\vspace{-.3cm}
\bibitem[RLL09]{RordamLarsenLaustsen09}
M. Rørdam, F. Larsen, and N. Laustsen,
{\it An Introduction to K-Theory for $C^\ast$-Algebras}, Cambridge University  Press, 2009,
[\href{https://doi.org/10.1017/CBO9780511623806}{\tt doi:10.1017/CBO9780511623806}].

\vspace{-.3cm}
\bibitem[R{\"o}04]{Roessler04}
U. R{\"o}{\ss}ler,
{\it Solid State Theory: An Introduction},
Springer (2004, 2009),
[\href{https://link.springer.com/book/10.1007/978-3-540-92762-4}{\tt doi:10.1007/978-3-540-92762-4}].

\vspace{-.3cm}
\bibitem[RW18]{RowellWang18}
E. C. Rowell and Z. Wang,
{\it Mathematics of Topological Quantum Computing},
Bull. Amer. Math. Soc. {\bf 55} (2018), 183-238,
[\href{https://doi.org/10.1090/bull/1605}{\tt doi:10.1090/bull/1605}],
[\href{https://arxiv.org/abs/1705.06206}{\tt arXiv:1705.06206}].




\vspace{-.3cm}
\bibitem[Run]{Runkel}
I. Runkel,
{\it Algebra in Braided Tensor Categories and Conformal Field Theory},
\newline
[\href{https://www.math.uni-hamburg.de/home/runkel/PDF/alg.pdf}{\tt www.math.uni-hamburg.de/home/runkel/PDF/alg.pdf}]



\vspace{-.3cm}
\bibitem[SRN15]{SarmaFreedmanNayak15}
S. D. Sarma, M. Freedman and C. Nayak,
{\it Majorana zero modes and topological quantum computation},
npj Quantum Inf. {\bf 1} (2015) 15001,
[\href{https://doi.org/10.1038/npjqi.2015.1}{\tt doi:10.1038/npjqi.2015.1}].

\vspace{-3mm}
\bibitem[SS19-Tad]{SS19TadpoleCancellation}
H. Sati and U. Schreiber,
{\it Equivariant Cohomotopy implies orientifold tadpole cancellation},
J. Geom. Phys. {\bf 156} (2020), 103775,
[\href{https://doi.org/10.1016/j.geomphys.2020.103775}{\tt doi:10.1016/j.geomphys.2020.103775}],
[\href{https://arxiv.org/abs/1909.12277}{\tt arXiv:1909.12277}].


\vspace{-3mm}
\bibitem[SS20-Orb]{SS20OrbifoldCohomology}
H. Sati and U. Schreiber,
{\it Proper Orbifold Cohomology},
[\href{https://arxiv.org/abs/2008.01101}{\tt arXiv:2008.01101}].



\vspace{-3mm}
\bibitem[SS21-Bun]{SS21EPB}
H. Sati and U. Schreiber,
{\it Equivariant principal $\infty$-bundles},
[\href{https://arxiv.org/abs/2112.13654}{\tt arXiv:2112.13654}].

\vspace{-.3cm}
\bibitem[SS21-MF]{SS21MF}
H. Sati and U. Schreiber,
{\it M/F-Theory as $M\!f$-Theory},
Reviews in Mathematical Physics
{\bf 35} 10 (2023)
[\href{https://arxiv.org/abs/2103.01877}{\tt arXiv:2103.01877}]
[\href{https://doi.org/10.1142/S0129055X23500289}{\tt doi:10.1142/S0129055X23500289}].



\vspace{-.3cm}
\bibitem[SS22-Conf]{SS22ConfigurationSpaces}
H. Sati and U. Schreiber,
{\it Differential Cohomotopy implies
intersecting brane observables via configuration spaces and chord diagrams},
Adv. Theor. Math. Phys. {\bf 26} 4 (2022),
[\href{https://www.intlpress.com/site/pub/pages/journals/items/atmp/_home/acceptedpapers/index.php}{\tt ISSN:1095-0753}],
[\href{https://arxiv.org/abs/1912.10425}{\tt arXiv:1912.10425}].

\vspace{-.3cm}
\bibitem[SS22-Any]{SS22AnyonicDefectBranes}
H. Sati and U. Schreiber,
{\it Anyonic defect branes in TED-K-theory},
Reviews in Mathematical Physics
{\bf 35} 06 (2023) 2350009
[\href{https://arxiv.org/abs/2203.11838}{\tt arXiv:2203.11838}]
[\href{https://doi.org/10.1142/S0129055X23500095}{\tt doi:10.1142/S0129055X23500095}].



\vspace{-3mm}
\bibitem[SS22-TED]{SS22TED}
H. Sati and U. Schreiber,
{\it Twisted equivariant differential non-abelian cohomology},
in preparation.

\vspace{-.3cm}
\bibitem[SS22-TQC]{SS22TQC}
H. Sati and U. Schreiber,
{\it Topological Quantum Programming in {\tt TED-K}},
PlanQC {\bf 2022} 33 (2022)
\newline
[\href{https://arxiv.org/abs/2209.08331}{\tt arXiv:2209.08331}]
[\href{https://ncatlab.org/schreiber/show/Topological+Quantum+Programming+in+TED-K}{\tt ncatlab.org/schreiber/show/Topological+Quantum+Programming+in+TED-K}]

\vspace{-.3cm}
\bibitem[SS22-Surv]{TalkNotes}
U. Schreiber with H. Sati,
{\it Anyonic Defect Branes in TED-K-Theory},
talk at RIND Sem. MathPhys \& String Theory (May 2022),
[\href{https://ncatlab.org/schreiber/files/DefectBranes-slides_220509f.pdf}{\tt ncatlab.org/schreiber/files/DefectBranes-slides\_220509f.pdf}]



\vspace{-.3cm}
\bibitem[SV90]{SchechtmanVarchenko90}
V. V. Schechtman and A. N. Varchenko,
{\it Hypergeometric solutions of Knizhnik-Zamolodchikov equations},
Lett. Math. Phys. {\bf 20} (1990), 279–283,
[\href{https://doi.org/10.1007/BF00626523}{\tt doi:10.1007/BF00626523}].



\vspace{-.3cm}
\bibitem[Schn18]{Schnyder18}
A. P. Schnyder,
{\it Accidental and symmetry-enforced band crossings in topological semimetals}, lecture notes, 2018,
\newline
[\url{https://www.fkf.mpg.de/6431357/topo_lecture_notes_schnyder_TMS18.pdf}]


\vspace{-.3cm}
\bibitem[Schn20]{Schnyder20}
A. P. Schnyder,
{\it Topological semimetals},
lecture notes, 2020,
\newline
[\url{https://www.fkf.mpg.de/7143520/topological_semimetals.pdf}]











\vspace{-.3cm}
\bibitem[See04]{Seeger04}
K. Seeger,
{\it Semiconductor Physics},
Adv. Texts Phys, Springer (2004),
[\href{https://doi.org/10.1007/978-3-662-09855-4}{\tt doi:10.1007/978-3-662-09855-4}].



\vspace{-.3cm}
\bibitem[SC14]{ShenCha14}
J. Shen and J. J. Cha,
{\it Topological Crystalline Insulator Nanostructures},
Nanoscale {\bf 6}  (2014), 14133-14140, \newline
[\href{https://doi.org/10.1039/c4nr05124f}{\tt doi:10.1039/c4nr05124f}].


\vspace{-.3cm}
\bibitem[SSG17]{ShiozakiSatoGomi17}
K. Shiozaki, M. Sato, and K. Gomi,
{\it Topological Crystalline Materials -- General Formulation, Module Structure, and Wallpaper Groups}, Phys. Rev. B {\bf 95} (2017) 235425,
[\href{https://doi.org/10.1103/PhysRevB.95.235425}{\tt doi:10.1103/PhysRevB.95.235425}], \newline
[\href{https://arxiv.org/abs/1701.08725}{\tt arXiv:1701.08725}].










\vspace{-.3cm}
\bibitem[Si83]{Simon83}
B. Simon,
{\it Holonomy, the Quantum Adiabatic Theorem, and Berry's Phase},
Phys. Rev. Lett. {\bf 51} (1983), 2167--2170,
[\href{https://doi.org/10.1103/PhysRevLett.51.2167}{\tt doi:10.1103/PhysRevLett.51.2167}].

\vspace{-.3cm}
\bibitem[SiSu08]{SimonsSullivan08}
J. Simons and D. Sullivan,
{\it Structured vector bundles define differential K-theory}, Ast{\'e}risque {\bf 321} (2008) 1-3, \newline
[\href{http://www.numdam.org/item/?id=AST_2008__321__1_0}{{\tt numdam:AST\_2008\_\_321\_\_1\_0}(abridged)}],
[\href{https://arxiv.org/abs/0810.4935}{\tt arXiv:0810.4935}].

\vspace{-.3cm}
\bibitem[SMJZ13]{SMJZ13}
R.-J. Slager, A. Mesaros, V. Juricic, and J. Zaanen,
{\it The space group classification of topological band insulators},
Nature Phys. {\bf 9} (2013), 98-102,
[\href{https://doi.org/10.1038/nphys2513}{\tt doi:10.1038/nphys2513}],
[\href{https://arxiv.org/abs/1209.2610}{\tt arXiv:1209.2610}].

\vspace{-.3cm}
\bibitem[SlVa19]{SlinkinVarchenko18}
A. Slinkin and A. Varchenko,
{\it Twisted de Rham Complex on Line and Singular Vectors in $\widehat{\mathfrak{sl}(2)}$ Verma Modules},
SIGMA {\bf 15} (2019) 075,
[\href{https://doi.org/10.3842/SIGMA.2019.075}{\tt doi:10.3842/SIGMA.2019.075}],
[\href{https://arxiv.org/abs/1812.09791}{\tt arXiv:1812.09791}].


\vspace{-.3cm}
\bibitem[Sm93]{Smirnov93}
F. A. Smirnov,
{\it Form factors, deformed Knizhnik-Zamolodchikov equations and finite-gap integration}, Commun.  Math. Phys. {\bf 155} (1993), 459–487,
[\href{https://doi.org/10.1007/BF02096723}{\tt doi:10.1007/BF02096723}],
[\href{https://arxiv.org/abs/hep-th/9210052}{\tt arXiv:hep-th/9210052}].

\vspace{-.3cm}
\bibitem[SF15]{SodemannFu15}
I. Sodemann and L. Fu,
{\it Quantum Nonlinear Hall Effect Induced by Berry Curvature Dipole in Time-Reversal Invariant Materials},
Phys. Rev. Lett. 115 (2015) 216806,
[\href{https://doi.org/10.1103/PhysRevLett.115.216806}{\tt doi:10.1103/PhysRevLett.115.216806}].



\vspace{-.3cm}
\bibitem[SSL15]{SSL15}
J. C. W. Song, P. Samutpraphoot, and L. S. Levitov,
{\it Topological Bloch Bands in Graphene Superlattices}, Proc. Nat. Acad. Sci.  {\bf 112} (2015), 10879-10883,
[\href{https://doi.org/10.1073/pnas.1424760112}{\tt doi:10.1073/pnas.1424760112}],
[\href{https://arxiv.org/abs/1404.4019}{\tt arXiv:1404.4019}].


\vspace{-.3cm}
\bibitem[St20]{Stanescu20}
T. D. Stanescu,
{\it Introduction to Topological Quantum Matter \& Quantum Computation},
CRC Press, 2020, \newline
[\href{https://www.routledge.com/Introduction-to-Topological-Quantum-Matter--Quantum-Computation/Stanescu/p/book/9780367574116}{\tt ISBN:9780367574116}].

\vspace{-.3cm}
\bibitem[SdBKP18]{SdBKP18}
L. Stehouwer, J. de Boer, J. Kruthoff, and H. Posthuma, {\it Classification of crystalline topological insulators through K-theory},
Adv. Theor. Math. Phys. {\bf 25} (2021), 723-775,
[\href{ https://dx.doi.org/10.4310/ATMP.2021.v25.n3.a3}{\tt
 doi:10.4310/ATMP.2021.v25.n3.a3}], \newline
[\href{https://arxiv.org/abs/1811.02592}{\tt arXiv:1811.02592}].

\vspace{-.3cm}
\bibitem[St98]{Strange98}
P. Strange,
{\it Relativistic Quantum Mechanics -- with applications in condensed matter and atomic physics},
Cambridge University Press, 1998,
[\href{https://doi.org/10.1017/CBO9780511622755}{\tt doi:10.1017/CBO9780511622755}].


\vspace{-.3cm}
\bibitem[SW01]{SW01}
R. F. Streater and A. S. Wightman,
{\it PCT, Spin and Statistics, and All That},
Princeton University Press (2001), \newline
[\href{https://doi.org/10.1515/9781400884230}{\tt doi:10.1515/9781400884230}].

\vspace{-3mm}
\bibitem[Su18]{Su18}
L. Su,
{\it Fractional Quantum Hall States with Conformal
Field Theories}, course notes,
U. Chicago, 2018, \newline
[\url{https://homes.psd.uchicago.edu/~sethi/Teaching/P483-W2018/FQH_CFT.pdf}]


\vspace{-.3cm}
\bibitem[SO82]{SzaboOstlund82}
A. Szabo and N. S. Ostlund,
{\it Modern Quantum Chemistry – Introduction to Advanced Electronic Structure Theory}, Macmillan (1982), McGraw-Hill (1989), Dover (1996),
\newline
[\href{https://chemistlibrary.files.wordpress.com/2015/02/modern-quantum-chemistry.pdf}{\tt chemistlibrary.files.wordpress.com/2015/02/modern-quantum-chemistry.pdf}]






\vspace{-.3cm}
\bibitem[Th92]{Thaller92}
B. Thaller,
{\it The Dirac Equation},
Springer, 1992,
[\href{https://link.springer.com/book/10.1007/978-3-662-02753-0}{\tt doi:10.1007/978-3-662-02753-0}].





\vspace{-.3cm}
\bibitem[Th15]{Thiang15Iso}
G. C. Thiang,
{\it Topological phases: isomorphism, homotopy and K-theory},
Int. J. Geom. Methods Mod. Phys. {\bf 12} (2015) 1550098,
[\href{https://doi.org/10.1142/S021988781550098X}{\tt doi:10.1142/S021988781550098X}],
[\href{https://arxiv.org/abs/1412.4191}{\tt arXiv:1412.4191}].


\vspace{-.3cm}
\bibitem[Th16]{Thiang14}
G. C. Thiang,
{\it On the K-theoretic classification of topological phases of matter},
Ann. Henri Poincar{\'e} {\bf 17} (2016), 757-794,
[\href{https://doi.org/10.1007/s00023-015-0418-9}{\tt doi:10.1007/s00023-015-0418-9}],
[\href{https://arxiv.org/abs/1406.7366}{\tt arXiv:1406.7366}].

\vspace{-.3cm}
\bibitem[TB20]{TiwariBzdusek20}
A. Tiwari and T. Bzdu{\v s}ek,
{\it Non-Abelian topology of nodal-line rings in PT-symmetric systems},
Phys. Rev. B {\bf 101} (2020) 195130,
[\href{https://doi.org/10.1103/PhysRevB.101.195130}{\tt doi:10.1103/PhysRevB.101.195130}].

\vspace{-.3cm}
\bibitem[To20]{Tolcachier20}
A. Tolcachier,
{\it Holonomy groups of compact flat solvmanifolds}, Geom.  Ded.  {\bf 209} (2020), 95–117, \newline
[\href{https://doi.org/10.1007/s10711-020-00524-8}{\tt doi:10.1007/s10711-020-00524-8}],
[\href{https://arxiv.org/abs/1907.02021}{\tt arXiv:1907.02021}].


\vspace{-.3cm}
\bibitem[tD08]{tomDieck08}
T. tom Dieck,
{\it Algebraic topology},
Europ. Math. Soc., Zurich, 2008,
[\href{https://www.ems-ph.org/books/book.php?proj_nr=86}{\tt doi:10.4171/048}].

\vspace{-.3cm}
\bibitem[To17]{Tong17}
D. Tong,
{\it Lectures on solid state physics} (2017),
[\href{http://www.damtp.cam.ac.uk/user/tong/solidstate.html}{\tt damtp.cam.ac.uk/user/tong/solidstate.html}].

\vspace{-.3cm}
\bibitem[TTWL08]{TTWL08}
S. Trebst, M. Troyer, Z. Wang, and A. W. W. Ludwig,
{\it A short introduction to Fibonacci anyon models},
Prog. Theor. Phys. Supp. {\bf 176} (2008), 384-407,
[\href{https://doi.org/10.1143/PTPS.176.384}{\tt doi:10.1143/PTPS.176.384}],
[\href{https://arxiv.org/abs/0902.3275}{\tt arXiv:0902.3275}].

\vspace{-.3cm}
\bibitem[Ts14]{Tsvelik14}
A. M. Tsvelik,
{\it An integrable model with parafermion zero energy modes},
Phys Rev. Lett. {\bf 113} (2014) 066401,
[\href{https://doi.org/10.1103/PhysRevLett.113.066401}{\tt doi:10.1103/PhysRevLett.113.066401}],
[\href{https://arxiv.org/abs/1404.2840}{\tt arXiv:1404.2840}].



\vspace{-.3cm}
\bibitem[TV13]{TurnerVishwanath13}
A. M. Turner and A. Vishwanath,
{\it Beyond Band Insulators: Topology of Semi-metals and Interacting Phases},
in:
{Topological Insulators},
(M. Franz and L. Molenkamp eds.),
Contemporary Concepts of Condensed Matter Science {\bf 6}, Elsevier (2013), 293-324,
[\href{https://www.sciencedirect.com/bookseries/contemporary-concepts-of-condensed-matter-science/vol/6/suppl/C}{\tt ISBN:978-0-444-63314-9}],
[\href{https://arxiv.org/abs/1301.0330}{\tt arXiv:1301.0330}].

\vspace{-3mm}
\bibitem[TX06]{TuXu06}
J.-L. Tu and P. Xu,
{\it Chern character for twisted K-theory of orbifolds},
Adv. Math. {\bf 207} (2006), 455-483, \newline
[\href{https://doi.org/10.1016/j.aim.2005.12.001}{\tt doi:10.1016/j.aim.2005.12.001}],
[\href{https://arxiv.org/abs/math/0505267}{\tt arXiv:math/0505267}].






\vspace{-.3cm}
\bibitem[Val21]{Valera21}
S. J. Valera,
{\it Fusion Structure from Exchange Symmetry in (2+1)-Dimensions},
An. Phys. {\bf 429} (2021) 168471
[\href{https://arxiv.org/abs/2004.06282}{\tt arXiv:2004.06282}][\href{https://doi.org/10.1016/j.aop.2021.168471}{\tt doi:10.1016/j.aop.2021.168471}]

\vspace{-.3cm}
\bibitem[Van18]{Vanderbilt18}
D. Vanderbilt,
{\it Berry Phases in Electronic Structure Theory -- Electric Polarization, Orbital Magnetization and Topological Insulators},
Cambridge University Press (2018),
[\href{https://doi.org/10.1017/9781316662205}{\tt doi:10.1017/9781316662205}].




\vspace{-3mm}
\bibitem[Va95]{Varchenko95}
A. Varchenko,
{\it Multidimensional Hypergeometric Functions and Representation Theory of Lie Algebras and Quantum Groups},
Adv. Ser. Math. Phys. {\bf 21},
World Scientific, Singapore, 1995,
[\href{https://doi.org/10.1142/2467}{\tt doi:10.1142/2467}].


\vspace{-.3cm}
\bibitem[Vo03I]{Voisin03I}
C. Voisin,  {\it Hodge theory and Complex algebraic geometry I}, translated by L. Schneps,
Cambridge University Press (2002/3),
[\href{https://doi.org/10.1017/CBO9780511615344}{\tt doi:10.1017/CBO9780511615344}].




\vspace{-.3cm}
\bibitem[Wal47]{Wallace47}
P. R. Wallace,
{\it The Band Theory of Graphite},
Phys. Rev. {\bf 71} (1947), 622--634,
[\href{https://doi.org/10.1103/PhysRev.71.622}{\tt doi:10.1103/PhysRev.71.622}].


\vspace{-3mm}
\bibitem[Wan10]{Wang10}
Z. Wang,
{\it Topological Quantum Computation},
CBMS Regional Conference Series in Mathematics {\bf 112}, Amer. Math. Soc., 2010,
[\href{http://www.ams.org/publications/authors/books/postpub/cbms-112}{\tt ISBN-13:9780821849309}].

\vspace{-.3cm}
\bibitem[Wan18]{Wang17}
Z. Wang,
{\it Beyond Anyons},
Mod. Phys. Lett. A {\bf 33} 28 (2018) 1830011
[\href{https://arxiv.org/abs/1710.00464}{\tt arXiv:1710.00464}]
[\href{https://doi.org/10.1142/S0217732318300112}{\tt doi:10.1142/S0217732318300112}]

\vspace{-.3cm}
\bibitem[WS14]{WangSenthil14}
C. Wang and  T. Senthil,
{\it Interacting fermionic topological insulators/superconductors in three dimensions}, Phys. Rev. B {\bf 89}  (2014)
195124,  [\href{https://doi.org/10.1103/PhysRevB.89.195124}{\tt doi:10.1103/PhysRevB.89.195124}],
[\href{https://arxiv.org/abs/1401.1142}{\tt arXiv:1401.1142}].


\vspace{-.3cm}
\bibitem[WXL17]{WangXuLai17}
G.  Wang, H. Xu, and Y.-C. Lai,
{\it Mechanical topological semimetals with massless quasiparticles and a finite Berry curvature},
Phys. Rev. B {\bf 95} (2017) 235159,
[\href{https://doi.org/10.1103/PhysRevB.95.235159}{\tt doi:10.1103/PhysRevB.95.235159}].



\vspace{-.3cm}
\bibitem[WZ${}^+$12]{WZLLJHD12}
N. O. Weiss, H. Zhou, L. Liao, Y. Liu, S. Jiang, Y. Huang, and X. Duan,
{\it Graphene: an emerging electronic material}, Adv. Mater. {\bf 24} (2012),
5782-825,
[\href{https://doi.org/10.1002/adma.201201482}{\tt doi:10.1002/adma.201201482}].


\vspace{-.3cm}
\bibitem[Wen89]{Wen89}
X.-G. Wen,
{\it Vacuum degeneracy of chiral spin states in compactified space},
Phys. Rev. B {\bf 40} (1989), 7387--7390,
[\href{https://doi.org/10.1103/PhysRevB.40.7387}{\tt doi:10.1103/PhysRevB.40.7387}].

\vspace{-.3cm}
\bibitem[Wen91a]{Wen91Review}
X.-G. Wen,
{\it Topological orders and Chern-Simons theory in strongly correlated quantum liquid}, Int. J.  Mod. Phys. B {\bf 05}  (1991),
1641-1648,
[\href{https://doi.org/10.1142/S0217979291001541}{\tt doi:10.1142/S0217979291001541}].

\vspace{-.3cm}
\bibitem[Wen91b]{Wen91}
X.-G. Wen,
{\it Non-Abelian statistics in the fractional quantum Hall states},
Phys. Rev. Lett. {\bf 66} (1991), 802--805, \newline
[\href{https://doi.org/10.1103/PhysRevLett.66.802}{\tt doi:10.1103/PhysRevLett.66.802}].


\vspace{-.3cm}
\bibitem[Wen93]{Wen93}
X.-G. Wen,
{\it Topological order and edge structure of $\nu = 1/2$ quantum Hall state},
Phys. Rev. Lett. {\bf 70} (1993), 355--358,
[\href{https://doi.org/10.1103/PhysRevLett.70.355}{\tt doi:10.1103/PhysRevLett.70.355}].


\vspace{-.3cm}
\bibitem[Wen95]{Wen95}
X.-G. Wen,
{\it Topological orders and Edge excitations in FQH states},
Adv.  Phys. {\bf 44} (1995), 405-473, \newline
[\href{https://doi.org/10.1080/00018739500101566}{\tt doi:10.1080/00018739500101566}],
[\href{https://arxiv.org/abs/cond-mat/9506066}{\tt arXiv:cond-mat/9506066}].


\vspace{-.3cm}
\bibitem[WN90]{WenNiu90}
X.-G. Wen and Q. Niu,
{\it Ground-state degeneracy of the fractional quantum Hall states in the presence of a random potential and on high-genus Riemann surfaces},
Phys. Rev. B {\bf 41} (1990), 9377--9396, \newline
[\href{https://link.aps.org/doi/10.1103/PhysRevB.41.9377}{\tt
doi:10.1103/PhysRevB.41.9377}].


\vspace{-.3cm}
\bibitem[Wil82a]{Wilczek82a}
F. Wilczek,
{\it Magnetic Flux, Angular Momentum, and Statistics},
Phys. Rev. Lett. {\bf 48} (1982), 1144-1146, \newline
[\href{https://doi.org/10.1103/PhysRevLett.48.1144}{\tt doi:10.1103/PhysRevLett.48.1144}].

\vspace{-.3cm}
\bibitem[Wil82b]{Wilczek82b}
F. Wilczek,
{\it Quantum Mechanics of Fractional-Spin Particles}, Phys. Rev. Lett. {\bf 49} (1982), 957-959,  \newline
[\href{https://doi.org/10.1103/PhysRevLett.49.957}{\tt doi:10.1103/PhysRevLett.49.957}].

\vspace{-.3cm}
\bibitem[Wil90]{Wilczek90}
F. Wilczek,
{\it Fractional Statistics and Anyon Superconductivity}, World Scientific, Singapore, 1990,
\newline  [\href{https://doi.org/10.1142/0961}{\tt doi:10.1142/0961}].

\vspace{-.3cm}
\bibitem[Wil91]{wilczek91}
F. Wilczek,
{\it States of Anyon Matter},
Int. J. Mod. Phys. B {\bf 05}  (1991), 1273-1312, \newline
[\href{https://doi.org/10.1142/S0217979291000626}{\tt doi:10.1142/S0217979291000626}].

\vspace{-.3cm}
\bibitem[WZ84]{WilczekZee84}
F. Wilczek and A. Zee,
{\it Appearance of gauge structure in simple dynamical systems}, Phys. Rev. Lett. {\bf 52} (1984), 2111-2114,
[\href{https://doi.org/10.1103/PhysRevLett.52.2111}{\tt doi:10.1103/PhysRevLett.52.2111}].

\vspace{-.3cm}
\bibitem[Wi01]{Witten01}
E. Witten,
{\it Overview Of K-Theory Applied To Strings}, Int. J. Mod. Phys. A{\bf 16} (2001), 693-706,
\newline
[\href{https://doi.org/10.1142/S0217751X01003822}{\tt
doi:10.1142/S0217751X01003822}],
[\href{https://arxiv.org/abs/hep-th/0007175}{\tt arXiv:hep-th/0007175}].

\vspace{-.3cm}
\bibitem[Wu84]{Wu84}
Y.-S. Wu,
{\it Multiparticle Quantum Mechanics Obeying Fractional Statistics},
Phys. Rev. Lett. {\bf 53} (1984), 111-114,
[\href{https://doi.org/10.1103/PhysRevLett.53.111}{\tt doi:10.1103/PhysRevLett.53.111}],
[\href{https://core.ac.uk/download/pdf/276286925.pdf}{\tt core:276286925}].


\vspace{-.3cm}
\bibitem[XCN10]{XCN10}
D. Xiao, M.-C. Chang, and Q. Niu,
{\it Berry Phase Effects on Electronic Properties},
Rev. Mod. Phys. {\bf 82} (2010), 1959-2007,
[\href{https://doi.org/10.1103/RevModPhys.82.1959}{\tt doi:10.1103/RevModPhys.82.1959}],
[\href{https://arxiv.org/abs/0907.2021}{\tt arXiv:0907.2021}].


\vspace{-.3cm}
\bibitem[YN14]{YangNagaosa14}
B.-J. Yang and N. Nagaosa,
{\it Classification of stable three-dimensional Dirac semimetals with nontrivial topology},
Nature Comm. {\bf 5} (2014) 4898,
[\href{https://doi.org/10.1038/ncomms5898}{\tt doi:10.1038/ncomms5898}].

\vspace{-.3cm}
\bibitem[YXL14]{YXL14}
F. Yang, X. Xu, and R.-B. Liu,
{\it Giant Faraday rotation induced by Berry phase in bilayer graphene under strong terahertz fields},
New J. Phys. {\bf 16} (2014) 043014,
[\href{https://iopscience.iop.org/article/10.1088/1367-2630/16/4/043014}{\tt doi:10.1088/1367-2630/16/4/043014}],
[\href{https://arxiv.org/abs/1307.7987}{\tt arXiv:1307.7987}].



\vspace{-3mm}
\bibitem[ZLSS15]{ZaanenLiuSunSchalm15}
J. Zaanen, Y. Liu, Y.-W. Sun, and K. Schalm,
{\it Holographic Duality in Condensed Matter Physics},
Cambridge University Press (2015),
[\href{https://doi.org/10.1017/CBO9781139942492}{\tt doi:10.1017/CBO9781139942492}].

\vspace{-3mm}
\bibitem[Za21]{Zaanen21}
J. Zaanen,
{\it Lectures on quantum supreme matter},
[\href{https://arxiv.org/abs/2110.00961}{\tt arXiv:2110.00961}].


\vspace{-.3cm}
\bibitem[Zak89]{Zak89}
J. Zak,
{\it Berry’s phase for energy bands in solids},
Phys. Rev. Lett. {\bf 62} (1989) 2747-2750,
\newline
[\href{https://doi.org/10.1103/PhysRevLett.62.2747}{\tt doi:10.1103/PhysRevLett.62.2747}].


\vspace{-.3cm}
\bibitem[ZCZW19]{ZCZW19}
B. Zeng, X. Chen, D.-L. Zhou, and X.-G. Wen,
{\it Quantum Information Meets Quantum Matter – From Quantum Entanglement to Topological Phases of Many-Body Systems},
Quantum Science and Technology (QST),
Springer, 2019,
[\href{https://doi.org/10.1007/978-1-4939-9084-9}{\tt doi:10.1007/978-1-4939-9084-9}],
[\href{https://arxiv.org/abs/1508.02595}{\tt arXiv:1508.02595}].


\vspace{-.3cm}
\bibitem[ZNT21]{ZNT21}
C. Zeng, S. Nandy, and S. Tewari,
{\it Nonlinear transport in Weyl semimetals induced by Berry curvature dipole},
Phys. Rev. B {\bf 103} (2021) 245119,
[\href{https://doi.org/10.1103/PhysRevB.103.245119}{\tt doi:10.1103/PhysRevB.103.245119}].

\vspace{-.3cm}
\bibitem[ZWXT21]{ZWXT21}
H.-C. Zhang, Y.-H. Wu, T. Xiang, and H.-H. Tu,
{\it Chiral conformal field theory for topological states and the anyon eigenbasis on the torus},
[\href{https://arxiv.org/abs/2107.02596}{\tt arXiv:2107.02596}].


\vspace{-.3cm}
\bibitem[ZZ${}^+$16]{ZZLXZZ16}
D.-W. Zhang, Y. X. Zhao, R.-B. Liu, Z.-Y. Xue, S.-L. Zhu, and Z. D. Wang,
{\it Quantum simulation of exotic PT-invariant topological nodal loop bands with ultracold atoms in an optical lattice},
Phys. Rev. A {\bf 93} (2016) 043617,
[\href{https://doi.org/10.1103/PhysRevA.93.043617}{\tt doi:10.1103/PhysRevA.93.043617}],
[\href{https://arxiv.org/abs/1601.00371}{\tt arXiv:1601.00371}].



\end{thebibliography}
\end{document}